%% file: main.tex
\documentclass[a4paper,11pt]{article}

\input{MyDef.tex}

\newcommand{\ff}{\mathfrak{f}}

\usepackage{fontawesome5} 

\usepackage[misc]{ifsym} 

\usepackage{tikz-feynman}

\usetikzlibrary{decorations.markings}

\usetikzlibrary{decorations.pathmorphing,decorations.markings}

\usetikzlibrary{patterns}

\allowdisplaybreaks

\usetikzlibrary{decorations.markings}

\tikzset{
    partial ellipse/.style args={#1:#2:#3}{
        insert path={+ (#1:#3) arc (#1:#2:#3)}
    }
}

\tikzset{->-/.style={decoration={
  markings,
  mark=at position #1 with {\arrow{>}}},postaction={decorate}}}
\tikzset{-<-/.style={decoration={
  markings,
  mark=at position #1 with {\arrow{<}}},postaction={decorate}}}

\begin{document}

\flushbottom

\title{Generalized comodule tube algebras for boundary and domain wall defects of (2+1)D topological order}

\author[a,\star]{Zhian Jia\orcidlink{0000-0001-8588-173X},}

\author[b,\dagger]{Sheng Tan\orcidlink{0009-0008-3318-9942}}

\affiliation[a]{Institute of Quantum Physics, School of Physics, Central South University,
Changsha 418003, China}

\affiliation[b]{School of Mathematical Sciences, Capital Normal University, Beijing 100048, China}

\affiliation[\star]{Email: \href{mailto:giannjia@foxmail.com}{giannjia@foxmail.com}}

\affiliation[\dagger]{Email: \href{mailto:tansheng2018@outlook.com}{tansheng2018@outlook.com}}


\abstract{
The tube algebra, which carries the structure of a $C^*$ weak Hopf algebra, is a fundamental tool for characterizing topological excitations in topological phases. In this work, we generalize the tube algebra framework to codimension-2 defects in $(2+1)$D gapped phases described by Turaev--Viro--Barrett--Westbury TQFTs, with particular emphasis on Levin--Wen string-net models. We show that such codimension-2 defects are described by a \emph{defect tube algebra}, which naturally carries the structure of a comodule algebra over the weak Hopf tube algebra associated with the topological excitations. Boundary and domain wall defects are then characterized by representations of the corresponding boundary and domain wall defect tube algebras. In particular, a domain wall defect tube algebra can be regarded as a generalized Drinfeld double of the tube algebras associated with the two adjacent boundary defects. More generally, for a domain wall joining $N$ bulk phases, the associated domain wall defect tube algebra can be viewed as an $N$-tuple algebra. This perspective extends naturally to general $k$-defects, namely codimension-2 defects joining $k$ codimension-1 defects, for which the resulting generalized tube algebra carries the structure of a multicomodule algebra.
}


\maketitle

\input{Ch-Introduction}

\input{Ch-Motivation}

\input{Ch-BdTube}

\input{Ch-DWTube}

\input{Ch-DWTwist}

\input{Ch-Multicomodule}

\input{Ch-DWMultimodule}

\input{Ch-Examples}

\section{Conclusion and Outlook}
\label{sec:Discussion}

In this paper we have developed a unified algebraic framework for codimension-$2$ defects in $(2+1)$D topological order, built on a family of generalized tube algebras attached to defects for boundaries and domain walls. The organizing principle throughout is that the tube algebra of a codimension-$2$ defect is not merely an algebra: it carries a (multi)comodule structure over the weak Hopf tube algebras, and the defect data are recovered from its representation theory. This yields an explicit Tannaka--Krein type reconstruction for comodule algebras.

Several directions merit further exploration.
\begin{enumerate}
    \item \emph{Fermionic and symmetry-enriched tube algebras.} For fermionic and symmetry-enriched string-net models, analogous tube-algebraic characterizations of the corresponding topological phases are expected to hold. The fermionic case requires a graded (super) refinement of the comodule algebra structure. The symmetry enriched case requires a $G$-enriched structure.
    \item \emph{Higher-dimensional generalizations.} For $(3+1)$D and higher-dimensional string-net models, a tube-algebra-like structure should likewise emerge. Understanding how the higher algebraic version of comodule and multicomodule structures lift to this setting---and how they encode fusion $2$-categorical data---is an important open problem.
    \item \emph{Quantum $N$-tuple algebras.} The quantum $N$-tuple algebra introduced in this work deserves a systematic investigation as an algebraic structure in its own right, beyond its origin from topological constructions. Several fundamental questions remain open, including its intrinsic algebraic characterization, the classification of its representations, the possible existence of an $R$-matrix-like structure governing the exchange of $N$-fold junctions, and the formulation of an appropriate weak Hopf-algebraic framework that generalizes the notion of the quantum double.
   \item \emph{Non-invertible symmetries and SymTFT.} The generalized tube algebras developed in this work also naturally arise in the study of SymTFTs with defects. A systematic investigation of their role in this context would be an interesting direction for future research.
\end{enumerate}

\begin{acknowledgments}
Z.J. is supported by start-up grant of Central South University (Grant No. 502045031) and by the Frontier Interdisciplinary Direction ``Quantum Technologies'' Project of Central South University (Grant No. 506010805).
S.T. is supported by Beijing Natural Science Foundation (Grant No.~1264052), a Beijing municipal research support program for returned overseas scholars, and start-up fund from Capital Normal University.
Z.J. acknowledges the hospitality of Hank Chen and Wei Cui during his visit to BIMSA, and thanks Hank Chen, Wei Cui, Mo Huang, Ran Luo, Ce Shen, Yilong Wang, and Zhi-Hao Zhang for valuable discussions during his stay in BIMSA. He also thanks Yi-Nan Wang for hosting his visit to Peking University and acknowledges beneficial discussions with Yi-Nan Wang and Yu-An Chen during his stay there. He especially thanks Liang Kong for many beneficial discussions on tube algebra and topological defects that initiate the project, and Zhengwei Liu for illustrating the tube algebra in alterfold TQFT framework. He also thanks Dominic Williamson for bringing several related works to his attention.
Part of the results of this work were presented at the ``BIMSA TQFT and Higher Symmetries Seminar'' and workshop ``Topological quantum groups, quantum symmetries and related topics'' in Harbin Institute of Technology.
S.T. acknowledges the helpful conversations with Zhengwei Liu, Hank Chen and Yu-An Chen.

\end{acknowledgments}

\appendix

\input{Ch-Appendix}

\bibliographystyle{apsrev4-1-title}
\bibliography{Jiabib}
\end{document}

%% file: MyDef.tex
\usepackage{mytheme}
\pdfoutput = 1
\usepackage[T1]{fontenc}

\usepackage{amsmath}
\usepackage{amssymb}
\usepackage{amsthm}
\usepackage{bm}
\usepackage{braket}
\usepackage{graphicx}
\usepackage{adjustbox}

\usepackage{mathdots}

\usepackage{multirow}
\usepackage{tabu}

\usepackage{booktabs} 
\usepackage{xcolor}   
\usepackage{hhline}   

\usepackage{orcidlink} 

\usepackage{tikz}
\usetikzlibrary{positioning,intersections}
\usetikzlibrary{calc}
\usetikzlibrary{shapes.geometric}
\usepackage{braids}
\usepackage{tikz-cd}
\usetikzlibrary{shapes.geometric,shapes.misc}
\usetikzlibrary{arrows,matrix,calc,scopes,decorations.markings,snakes}
\usetikzlibrary{arrows.meta}
\usetikzlibrary{intersections, patterns,fit} 
\usetikzlibrary{decorations.pathreplacing,angles,quotes}
\usetikzlibrary{tqft} 

\definecolor{darkblue}{RGB}{40, 85, 120} 

\newtheorem{theorem}{Theorem}[section]
\newtheorem{lemma}[theorem]{Lemma}

\newtheorem{proposition}[theorem]{Proposition}
\newtheorem{corollary}[theorem]{Corollary}

\theoremstyle{definition}
\newtheorem{definition}[theorem]{Definition}

\theoremstyle{remark}
\newtheorem{remark}[theorem]{Remark}

\usepackage[bbgreekl]{mathbbol} 

\usepackage{dsfont}

\usepackage{microtype}       

\DeclareMathAlphabet{\mathcal}{OMS}{cmsy}{m}{n}

\usepackage{euscript}

\newcommand\eB           {\EuScript{B}}
\newcommand\eC           {\EuScript{C}}
\newcommand\eD           {\EuScript{D}}
\newcommand\eE          {\EuScript{E}}

\newcommand\eK         {\EuScript{K}}

\newcommand\eM          {\EuScript{M}}
\newcommand\eN         {\EuScript{N}}

\newcommand\eP         {\EuScript{P}}

\newcommand\eX         {\EuScript{X}}

\newcommand{\Zbb}{\mathbb{Z}}
\newcommand{\Cbb}{\mathbb{C}}

\newcommand{\one}{\mathbb{1}}

\newcommand{\Rep}{\mathsf{Rep}}

\newcommand\Tr{\operatorname{Tr}}

\newcommand{\id}{\operatorname{id}}
\newcommand{\Irr}{\operatorname{Irr}}
\newcommand{\Hom}{\operatorname{Hom}}
\newcommand{\End}{\operatorname{End}}

\newcommand{\Tube}{\mathbf{Tube}}

\newcommand{\sfEnd}{\mathsf{End}}

\newcommand{\Fun}{\mathsf{Fun}}
\newcommand{\Vect}{\mathsf{Vect}}

\newcommand\ffa         {\mathfrak{a}}
\newcommand\ffb         {\mathfrak{b}}

\newcommand\fff         {\mathfrak{f}}

\newcommand\fC           {\mathfrak{C}}

\usepackage{mathtools,amssymb,scalerel}

\newcommand\biopencrossl{%
	\mathrel{\scalerel*{>\kern-.4\LMpt\joinrel\blacktriangleleft}{x}}}
\newcommand\biopencrossr{%
	\mathrel{\scalerel*{\blacktriangleright\joinrel\kern-.4\LMpt<}{x}}}
\newcommand\bicrosslr{%
	\mathrel{\scalerel*{\mathrel{\blacktriangleright}\joinrel\blacktriangleleft}{x}}}



%% file: Ch-Introduction.tex
\section{Introduction}
\label{sec:intro}
Anyons are particles in (2+1)D spacetime manifold that obey exchange statistics beyond those of bosons and fermions \cite{Leinaas1977anyon}. Upon exchanging two anyons, the many-body wavefunction acquires a phase factor more general than $e^{i\nu\pi}$ with $\nu=0,1$. For non-Abelian anyons, the exchange operation acts on the degenerate ground-state manifold through a noncommuting unitary transformation.
More broadly, anyons arise as quasiparticle excitations in topologically ordered phases \cite{Wen2004,simon2023topoligical}.
Beyond their fundamental theoretical significance, anyons are also of considerable interest for quantum information processing, since topological states are intrinsically robust against local perturbations and noise~\cite{Kitaev2003,wang2010topological,Nayak2008}.

Non-Abelian anyons play a central role in topological quantum computation. However, the experimental preparation and manipulation of non-Abelian anyons remain highly challenging with current technology.
To overcome these difficulties, a variety of alternative approaches have been proposed. One promising direction is to consider topological point defects, which can effectively exhibit non-Abelian properties even in systems without intrinsic non-Abelian anyonic excitations \cite{Bombin2010,Cheng2012defect,Sarma2015defect,cong2017defects,Kitaev2012boundary}.

For a non-chiral $(2+1)$D topological phase, its topological excitations are characterized by a unitary modular tensor category (UMTC) $\mathcal{Z}(\eC)$, which is the Drinfeld center of some unitary fusion category (UFC) $\eC$. 
It has been shown that these excitations can be characterized by the \emph{tube algebra} constructed from $\eC$.
The study of tube algebras originates in the pioneering work of Ocneanu; see, for example,~\cite{ocneanu1994chirality,ocneanu2001operator,izumi2000structure,muger2003subfactorsII,evans1998quantum,Neshveyev2018tube,popa2015cohomology}. Given a UFC $\eC$, we denote its tube algebra by $\Tube_{\eC}$. A fundamental result asserts that the representation category of the tube algebra is equivalent to the Drinfeld center,
\begin{equation}
\Rep(\Tube_{\eC}) \simeq \mathcal{Z}(\eC).
\end{equation}
Over the years, numerous generalizations and variants of tube algebras have been developed~\cite{Kitaev2012boundary,jones2001annular,hoek2019drinfeld,lan2025tube,bridgeman2020computing,bai2025weakhopf,jia2024weakTube,Barter2019domainwall,Barter2022computing,jia2025tube}. These algebras provide powerful algebraic frameworks for describing a wide range of categorical structures arising in mathematical physics.
From a conceptual viewpoint, tube algebras can be approached from at least three complementary perspectives: subfactor theory, tensor category theory, and topological order. Tube algebras and their variants have found broad applications in topological phases of matter. They provide algebraic descriptions of topological excitations, characterize charges in symmetry topological field theories (SymTFTs), and furnish computational tools for extracting categorical and topological data, etc.~\cite{Kitaev2012boundary,bai2025weakhopf,choi2026generalized,cordova2025representation,Christian2023,Bullivant2019tube,Lin2023tube,Hu2018full,Kawagoe2024tube,bridgeman2020computing,jia2024weakTube,Barter2019domainwall,Barter2022computing,jia2025tube, jia2025haagerup,jia2025Ising}.

\begin{figure}
    \centering
    \includegraphics[width=0.65\linewidth]{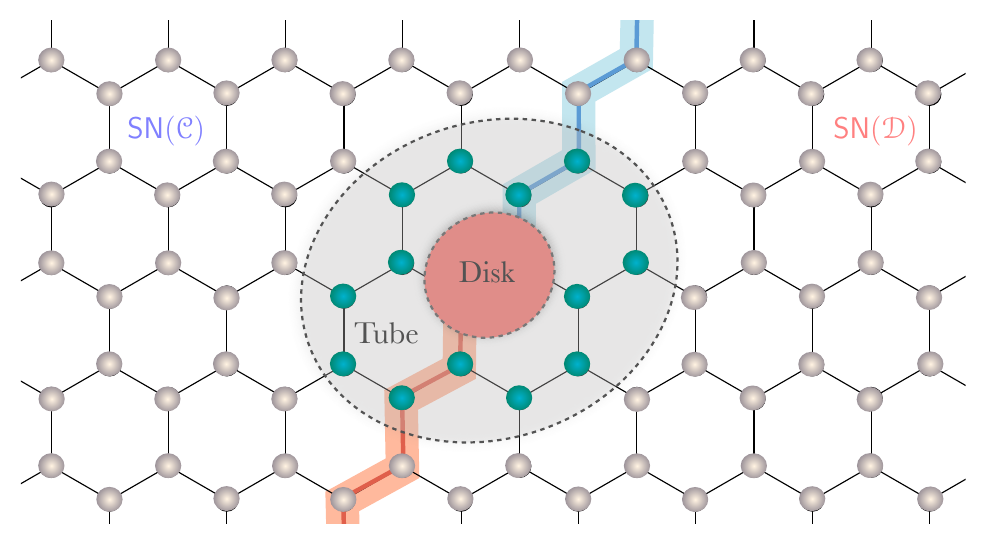}
   \caption{The point defect (disk region) connects two domain walls described by the $\eC|\eD$-bimodule categories $\eM$ and $\eN$, with tube regions surrounding the defect.}
    \label{fig:TubeRegion}
\end{figure}

The present work adopts the topological-order perspective, focusing in particular on tube algebras that characterize topological excitations and defects. The boundary tube algebra introduced by Kitaev and Kong~\cite{Kitaev2012boundary} was originally developed to describe gapped boundaries of the Levin--Wen string-net model. Since then, this perspective has inspired a variety of tube algebra constructions associated with bulk phases, domain walls, and higher-codimensional excitations~\cite{bridgeman2020computing,bai2025weakhopf,jia2024weakTube,Barter2019domainwall,Barter2022computing,jia2025tube,Bullivant2019tube}.
More recently, in~\cite{jia2025tube}, we showed that the Drinfeld quantum double of a weak Hopf algebra also admits a natural tube algebra interpretation. Furthermore, this construction can be generalized using multimodule categories, leading to the notion of the \emph{$N$-tuple algebra}.

Non-chiral topological orders can support both gapped boundaries and gapped domain walls, namely, extensions of the Hamiltonian to two-dimensional spatial manifolds with boundaries or line defects remain gapped, see, e.g., \cite{Beigi2011the,Kitaev2012boundary,Levin2013protected,Cong2017,wang2020electric,kapustin2011topological,Barkeshli2013theory,Lan2015gapped,hu2018boundary,Lan2020Gapped,jia2023boundary,Jia2023weak,jia2024weakTube,jia2025tube,jia2024generalized,Jia2026weakhopf}.
Point-like (codimension-2\footnote{Throughout this paper, codimension is defined relative to the spatial dimension of the spacetime manifold.}) defects on a given gapped boundary or domain wall (codimension-1 defect) may be viewed as topological excitations localized on the corresponding boundary or domain wall.
These excitations can be characterized through the weak Hopf boundary and domain wall tube algebras, in the sense that the excitations correspond to representations of the tube algebra, while their fusion rules are encoded in the coalgebra structure of the tube algebra~\cite{Kitaev2012boundary,bridgeman2023invertible,jia2024weakTube,jia2025tube,bai2025weakhopf,Liu2026alterfold}.

In this work, we investigate point defects in two-dimensional topological phases and uncover the underlying quantum algebraic symmetries associated with these defects. 
The most extensively studied examples of such defects are those occurring at the interface between rough and smooth boundaries in the toric code model. These defects exhibit non-Abelian behavior analogous to that of Ising anyons; see, e.g., Ref.~\cite{Sarma2015defect}. Their $\mathbb{Z}_N$ generalizations are known as parafermion zero modes~\cite{Cheng2012defect}.
While the algebraic structures associated with codimension-2 defects have been discussed briefly in Refs.~\cite{Kitaev2012boundary,bridgeman2020computing}, a complete characterization from the perspective of quantum algebra symmetries remains lacking.

We will show that \emph{the appropriate quantum algebraic structure for codimension-2 defects is a comodule algebra over weak Hopf algebras}. We then explicitly construct the corresponding comodule tube algebra. This construction may be viewed as a natural extension of Tannaka--Krein reconstruction to the setting of comodule algebras. Furthermore, the weak Hopf tube algebra arises as a special case of the comodule tube algebra when the codimension-2 defect connects two codimension-1 defects of the same type. In this sense, the comodule tube algebra provides a unified framework encompassing both topological excitations and more general codimension-2 defects.

\subsection{Main results of this work}

Our main contribution is the extension of the weak Hopf tube algebraic characterization of $(2+1)$D bosonic topological excitations to defects. We show that the corresponding defect tube algebra naturally carries the structure of a comodule algebra over the weak Hopf tube algebras describing topological excitations for the corresponding codimension-1 defects. Moreover, the domain wall defect tube algebra can be regarded as the quantum double\footnote{This notion is defined for weak Hopf algebras; for comodule algebras, since there is no coalgebra structure, the notion cannot be generalized straightforwardly. Instead, we provide a definition inspired by the defect tube algebra.} of the boundary defect tube algebra, for which we give a rigorous definition and establish the corresponding equivalence. This notion of quantum double can also be generalized to quantum $N$-tuple algebras, namely, an comodule algebra constructed from $N$ given comodule algebras.

In string-net models, violations of local stabilizer conditions ($B_f=1$ and $Q_v=1$) are interpreted as topological excitations. Given a local disk region, if the stabilizer conditions are violated within that region, we say that the region contains an excitation. For gapped boundaries and domain walls, boundary and domain wall excitations are defined analogously as localized disk regions on the corresponding boundary or domain wall; see Figure~\ref{fig:TubeRegion} for an illustration of the domain wall case. By considering a tube region surrounding such excitation disks and applying topological local moves to reduce it to a minimal tube, one obtains the tube algebra. Its multiplication is naturally defined by gluing two tube regions together. The coalgebra structure is induced by the fusion of topological excitations and may be understood geometrically as inserting a topological line and cutting a tube into two smaller tubes. The tube algebras for bulk excitations, boundary excitations and domain wall excitations are all weak Hopf algebras \cite{Kitaev2012boundary,bridgeman2023invertible,jia2024weakTube,jia2025tube,bai2025weakhopf}.

The above construction admits a natural generalization to topological codimension-2 defects \footnote{Topological excitations may be regarded as a special class of defects connecting two codimension-1 defects of the same type. More generally, defects can interpolate between codimension-1 defects of different types.}, see Refs.~\cite{Kitaev2012boundary,kong2012universal,williamson2017SETtube,morrison2012blob} for some early related discussions. The crucial distinction from the bulk case is that defects are not fusion-closed: while topological excitations can fuse with defects, and the fusion of two defects may yield topological excitations, defects themselves do not generally admit a fusion product. As a consequence, the coalgebra structure present in the bulk tube algebra no longer exists in the defect sector. Instead, we show that the defect tube algebra carries the structure of a comodule algebra over the bulk tube algebra. More precisely, it is equipped with a coaction of the bulk tube algebra that is compatible with the algebra multiplication. This coaction is induced by the fusion of bulk topological excitations with the defect, thereby replacing the coproduct associated with the fusion of bulk excitations.

Since topological excitations, viewed as codimension-2 defects, can only connect two codimension-1 defects of the same type, they do not exhibit the more general junction structure of arbitrary codimension-2 defects. In contrast, a general codimension-2 defect may connect $N$ distinct codimension-1 defects, a phenomenon with no analogue in the conventional tube algebra description of topological excitations. In this case, the corresponding defect tube algebra naturally carries a multicomodule structure and, as discussed in Section~\ref{section:multicomodule}, forms a multicomodule algebra. Moreover, every such defect can be transformed into a twist defect connecting a regular codimension-1 defect and the given codimension-1 defect associated with a non-regular module category.

More generally, one may consider multimodule categories over $N$ base UFCs, corresponding to codimension-1 domain walls connecting $N$ bulk phases; see Figure~\ref{fig:defect}. The associated defect tube algebra then acquires the structure of an \emph{$N$-tuple algebra}. The case $N=2$ reduces to the Drinfeld quantum double associated with comodule algebras, while higher values of $N$ provide a natural higher-arity generalization. To the best of our knowledge, these \emph{$N$-tuple algebras} have not been systematically studied in the existing literature. 

The main results of this work are summarized roughly in the following theorem.

\begin{theorem}
Codimension-2 defects in gapped phases described by Turaev--Viro--Barrett--Westbury TQFTs are characterized by comodule tube algebras (more precisely, these tube algebras are bicomodule algebras over the weak Hopf tube algebras for the corresponding excitations on the codimension-1 defect), and each such defect furnishes a representation of the corresponding comodule tube algebra. This framework applies uniformly to boundary defects\footnote{Boundary defects may be viewed as a special case of domain wall defects, since a boundary can be interpreted as a domain wall separating a topological phase from the vacuum.}, domain wall defects, and defects for domain wall connecting $N$ bulk phases, which are discussed in Sections~\ref{Sec:BoundaryTube}, \ref{Sec:DomainwallTube}, and \ref{sec:multi-moduletube}, respectively. More generally, codimension-2 defects connecting more than two codimension-1 defects are described by multicomodule algebras; see Section~\ref{section:multicomodule}.

We develop the theory of generalization of Drinfeld quantum double to comodule algebras. We prove that the domain wall defect tube algebra is equivalent to the  Drinfeld quantum double of boundary defect tube algebras. For domain walls connecting $N$ bulk phases, the associated defect tube algebra carries the structure of a bicomodule algebra. It may be regarded as an quantum $N$-tuple algebra of the $N$ boundary defect tube algebras; the special case $N=2$ reduces to the Drinfeld quantum double, see Section~\ref{sec:multi-moduletube}. 

We provide an explicit construction of the separability idempotent and use it to establish a generalized Schur orthogonality relation for defect tube algebras. Based on this, we prove the equivalence between the representation category of the defect tube algebra and the functor category of module categories over UFCs. These results are summarized in Table~\ref{tab:TubeAlg}.
\end{theorem}

\begin{table}[hb]
    \centering
    \resizebox{1.0\linewidth}{!}{%
    \begin{tabular}{|c|c|c|}
        \hline
        & Tube algebra description 
        & Categorical description \\
        \hline
        Boundary phase
        & $\mathbf{Tube}({_{\eC}}\eM)$
        & $\mathsf{Boundary}({_{\eC}}\eM)=\sfEnd_{\eC}(\eM)^{\rm op} \simeq \Rep(\mathbf{Tube}({_{\eC}}\eM))$ \\
        \hline 
        Boundary defect
        & $\mathbf{Tube}({_{\eC}}\eM;{_{\eC}}\eN)$
        & $\mathsf{Defect}({_{\eC}}\eM;{_{\eC}}\eN) = \Fun_{\eC}(\eM,\eN)^{\rm op} \simeq \Rep(\mathbf{Tube}({_{\eC}}\eM;{_{\eC}}\eN))$ \\ 
        \hline
        Domain wall phase
        & $\mathbf{Tube}({_{\eC}}\eM_{\eD})$
        & $\mathsf{Wall}({_{\eC}}\eM_{\eD})=\mathsf{End}_{\eC|\eD}(\eM)^{\rm op} \simeq \Rep(\mathbf{Tube}({_{\eC}}\eM_{\eD}))$ \\
        \hline 
        Domain wall defect
        & $\mathbf{Tube}({_{\eC}}\eM_{\eD};{_{\eC}}\eN_{\eD})$
        & $\mathsf{Defect}({_{\eC}}\eM_{\eD};{_{\eC}}\eN_{\eD}) =\Fun_{\eC|\eD}(\eM,\eN)^{\rm op} \simeq \Rep(\mathbf{Tube}({_{\eC}}\eM_{\eD};{_{\eC}}\eN_{\eD}))$ \\    
        \hline   
        Multimodule domain wall phase
        & $\mathbf{Tube}({_{\{\eC_{\alpha}\}}}\eM_{\{\eD_{\beta}\}})$
        & $\begin{aligned}
        \mathsf{Wall}({_{\{\eC_{\alpha}\}}}\eM_{\{\eD_{\beta}\}})
        = 
            \mathsf{End}_{{_{\{\eC_{\alpha}\}}}|{\{\eD_{\beta}\}}}(\eM)^{\rm op}
           \simeq \Rep(\mathbf{Tube}({_{\{\eC_{\alpha}\}}}\eM_{\{\eD_{\beta}\}})) 
        \end{aligned}$ \\
        \hline 
        Multimodule domain wall defect
        & 
        $\mathbf{Tube}({_{\{\eC_{\alpha}\}}}\eM_{\{\eD_{\beta}\}};
           {_{\{\eC_{\alpha}\}}}\eN_{\{\eD_{\beta}\}})$
        & 
        $\begin{aligned}
            &\mathsf{Defect}({_{\{\eC_{\alpha}\}}}\eM_{\{\eD_{\beta}\}};
            {_{\{\eC_{\alpha}\}}}\eN_{\{\eD_{\beta}\}}) 
             = 
            \Fun_{{_{\{\eC_{\alpha}\}}}|{\{\eD_{\beta}\}}}(\eM,
            \eN)^{\rm op}\\
          &  \simeq 
            \Rep(\mathbf{Tube}({_{\{\eC_{\alpha}\}}}\eM_{\{\eD_{\beta}\}};
            {_{\{\eC_{\alpha}\}}}\eN_{\{\eD_{\beta}\}}))
        \end{aligned}$ \\ 
        \hline 
      $N$-defect
        & $\mathbf{Tube}(\eM_1;\cdots;\eM_N)$
        & $\begin{aligned}
            &\mathsf{Defect}(\eM_1;\cdots;\eM_N) =\Fun_{\eC^{k_0,k_0+1}|\eC^{N,1}}(\eM_{\rm in},\eM_{\rm out})^{\rm op}
            \\
           &  \simeq \Rep(\mathbf{Tube}(\eM_1;\cdots;\eM_N))
        \end{aligned}$ \\ 
        \hline    
    \end{tabular}
    }
    \caption{Summary of the tube algebra descriptions of boundary and domain wall excitations and defects.}
    \label{tab:TubeAlg}
\end{table}

Let us make a few remarks before proceeding. Throughout this paper, $d$ denotes the spatial dimension and $D$ the spacetime dimension. Unless explicitly stated otherwise, all categories and algebras are understood to be defined over $\Cbb$. To keep the discussion accessible to physics readers, we adopt the tube-basis approach when constructing the generalized tube algebras. We note that an equivalent formulation can be given using internal Homs, along the lines of~\cite{bai2025weakhopf}. Our goal is to make the material as self-contained as possible, so we have tried to present the proofs clearly and without assuming too much prior expertise. For the reader's convenience, we also spell out the boundary, domain wall, and multimodule cases in detail. Although these share the same conceptual backbone and the presentation may be somewhat repetitive, we hope that treating them separately will help physics readers follow the main ideas more easily.

The paper is organized as follows. In Section~\ref{sec:AlgTheory}, we review the algebraic theory of codimension-2 defects on boundaries and domain walls and provide the motivation for the construction of defect tube algebras. In Section~\ref{Sec:BoundaryTube}, we construct the boundary defect tube algebra and show how boundary defects are realized as its representations. We provide an explicit construction of the defect tube algebra, prove its comodule algebra structure, and introduce the factorization map and antipode-like map together with their fundamental properties. We then derive the separability idempotent and establish the generalized Schur orthogonality relation. Based on these results, we prove the equivalence between the representation category of the defect tube algebra and the corresponding module functor category.
In Section~\ref{Sec:DomainwallTube}, we develop the analogous theory for domain walls and their associated defects. We further show that the domain wall defect tube algebra possesses a richer algebraic structure and can be interpreted as a generalized Drinfeld quantum double.
In Section~\ref{sec:dw-twist}, we specialize the domain wall construction to twist defects, derive a simplified presentation of their tube algebras, and explain how general domain wall defects can be reformulated as twist defects using relative tensor products.
In Section~\ref{section:multicomodule}, we introduce multicomodule tube algebras for codimension-2 defects connecting multiple codimension-1 defects and discuss the correspondence between general $N$-fold defects and twist defects.
In Section~\ref{sec:multi-moduletube}, we study multimodule domain walls and the associated $N$-tuple algebra structures. In this setting, the corresponding defect tube algebra possesses a structure that can be interpreted as a quantum $N$-tuple algebra, which reduces to the generalized Drinfeld quantum double for comodule algebras when $N=2$.
In Section~\ref{sec:exmp}, we illustrate our constructions through explicit examples. Finally, we conclude with a discussion in Section~\ref{sec:Discussion}. Additional background on string-net models and technical proofs are collected in the appendices.

%% file: Ch-Motivation.tex
\section{Boundary and domain wall defects in topological orders}
\label{sec:AlgTheory}
In this section, we review the basic theory of boundary and domain wall defects in topological quantum phases described by Turaev--Viro--Barrett--Westbury TQFTs, namely the $(2+1)$D Levin--Wen string-net models based on (multi)fusion categories~\cite{Levin2005,Kitaev2012boundary,kong2012universal,kirillov2011stringnet,Lin2021generalized,jia2024weakTube,turaev1992state,barrett1996invariants}. We also provide the physical motivation for the construction of the defect tube algebra.

\subsection{Algebraic theory of boundary defects} 

Let us start by reviewing the theory of boundary defects, as discussed in, e.g., \cite{Kitaev2012boundary,cong2017defects}, before elaborating on why comodule algebras play a pivotal role in characterizing these defects. 
A well-known example is the defect separating the smooth and rough boundaries of the $\Zbb_2$ toric code. This defect behaves as an Ising anyon, making it non-Abelian with a quantum dimension of $\sqrt{2}$.  
Algebraically, the bulk unitary fusion category (UFC) is $\eC=\Rep(\Zbb_2)$, and the two boundaries are characterized by the $\eC$-module categories $\eM_s=\Rep(\Zbb_2)$ and $\eM_r=\Vect_{\Cbb}$ \footnote{In this paper, when discussing boundaries, we assign an orientation such that the bulk lies on the left-hand side of the boundary when traversing along its positive direction; the module category is thus a left module category.}. 
Correspondingly, boundary defects are described by $\eC$-module functors (see Ref.~\cite{etingof2016tensor} for details about the notions of module functors):
\begin{itemize}
    \item The defect $\tau_+ \in \Fun_{\eC}(\eM_s,\eM_r) =: \eP_{0,1}$, which interpolates from the smooth boundary to the rough boundary.
    \item The defect $\tau_- \in \Fun_{\eC}(\eM_r,\eM_s) =: \eP_{1,0}$, which interpolates from the rough boundary to the smooth boundary.
\end{itemize}
Boundary excitations can be viewed as special types of defects that connect the same boundary \footnote{For any category $\eX$, the monoidal structure on $\End(\eX)$ is given by composition of functors, and we adopt the convention $\mathfrak{f}\otimes \mathfrak{g} := \mathfrak{g}\circ\mathfrak{f}$.}:
\begin{itemize}
    \item The smooth boundary excitations $I,E \in \Fun_{\eC}(\eM_s,\eM_s) =: \eP_{0,0}$.
    \item The rough boundary excitations $I,M \in \Fun_{\eC}(\eM_r,\eM_r) =: \eP_{1,1}$.
\end{itemize}
The fusion of defects is described by a multifusion category (for more details about multifusion category, see Refs.~\cite{jia2024weakTube,kong2019multifusion,etingof2016tensor})
\begin{equation}
    \eP = \bigoplus_{i,j=0,1} \eP_{i,j}.
\end{equation}
For defects $a \in \eP_{i,j}$ and $b \in \eP_{k,l}$, fusion is allowed only when $j = k$, in which case the resulting defect lies in $\eP_{i,l}$. For example,
\begin{equation}
    \tau_+ \otimes \tau_- \in \eP_{0,0}
\end{equation}
yields a boundary excitation on the smooth boundary. Similarly, $E \otimes \tau_+$ is allowed and results in a defect in $\eP_{0,1}$, whereas $M \otimes \tau_+$ and $M \otimes E$ are not allowed. We emphasize that fusion in this setting is direction-dependent: for instance, $E \otimes \tau_+$ is allowed, while $\tau_+ \otimes E$ is not. This asymmetry reflects the fact that the boundary carries a natural orientation.

The quantum symmetries underlying boundary excitations are weak Hopf symmetries, which can be explicitly reconstructed from the boundary tube algebra \cite{Kitaev2012boundary,bridgeman2023invertible,jia2024weakTube,jia2025tube}. In this sense, boundary excitations can be regarded as symmetry charges associated with these symmetries. For the smooth boundary, the symmetry $W_s$ gives rise to $\eP_{0,0} = \Rep(W_s)$, while for the rough boundary, the symmetry $W_r$ yields $\eP_{1,1} = \Rep(W_r)$.
In contrast, for boundary defects the corresponding theory is not yet well established. From the tube algebra perspective, we propose that the relevant symmetry structure is instead described by a comodule algebra over weak Hopf algebras. Namely, for $\eP_{0,1}$ there exists a comodule algebra $W_{s,r}$ over $W_s$ and $W_r$ such that $\eP_{0,1} = \Rep(W_{s,r})$. Similarly, for $\eP_{1,0}$ there is a comodule algebra $W_{r,s}$ such that $\eP_{1,0} = \Rep(W_{r,s})$. We will present an explicit tube algebra construction of these comodule algebras in Section~\ref{sec:ExpTubeToric}.

For general topological boundaries of a non-chiral topological phase, a similar picture holds.
Given a bulk UFC $\eC$ and two gapped boundaries determined by $\eC$-module categories $\eM$ and $\eN$, defects between the $\eM$-boundary and the $\eN$-boundary are characterized by the $\eC$-module functor category $\Fun_{\eC}(\eM,\eN)$;\footnote{One may equally well describe the boundary defect by $\Fun_{\eC}(\eM,\eN)^{\rm op}$, i.e., with the direction of all arrows reversed. As we will see in the string-net construction (Remark~\ref{remk:opFun}), this choice is in fact the more natural one. This convention does not substantially affect the physics, since if $\eC$ is a fusion category, then so is $\eC^{\rm op}$. We shall make use of both descriptions interchangeably in what follows.} see Fig.~\ref{fig:defect}. 
Topological excitations on a boundary can be identified with boundary defects on the same boundary. For an $\eM$-boundary, this correspondence is explicitly given by $\mathsf{Boundary}({_{\eC}}\eM) = \Fun_{\eC}(\eM,\eM)$.
A crucial observation is that boundary excitations can be absorbed by boundary defects. Concretely, for a boundary charge $X \in \Fun_{\eC}(\eM,\eM)$ and a boundary defect $\ff \in \Fun_{\eC}(\eM,\eN)$, one has $X \otimes \ff \in \Fun_{\eC}(\eM,\eN)$.
Mathematically, this fusion operation is defined by the composition of functors. From the formal definitions of boundary charges and boundary defects, it follows immediately that fusing a boundary charge with a boundary defect produces another boundary defect.
For a boundary charge $Y \in \Fun_{\eC}(\eN,\eN)$, however, fusion acts on the opposite side of the defect, yielding $\ff \otimes Y \in \Fun_{\eC}(\eM,\eN)$.
Physically, this reflects the fact that boundary charges are confined to propagate within their respective boundary phases: the $\eM$-boundary lies to the left of the defect, while the $\eN$-boundary lies to its right.

For a fixed bulk UFC $\eC$, there are finitely many topological boundary types, characterized by indecomposable $\eC$-module categories $\eM_i$ with $i = 0, \ldots, n-1$.
We accordingly define the defect categories $\eP_{i,j} := \Fun_{\eC}(\eM_i,\eM_j)$.
Collecting boundary excitations and boundary defects together, we obtain a multifusion category
\begin{equation}
    \eP^{\rm boundary}_{\eC} = \bigoplus_{i,j=0}^{n-1} \eP_{i,j}.
\end{equation}
The fusion of $a \in \eP_{i,j}$ and $b \in \eP_{k,l}$ is nonvanishing only when $j = k$.
The diagonal components $\eP_{i,i}$ describe boundary excitations on the $\eM_i$-boundary, while the off-diagonal components $\eP_{i,j}$ with $i \neq j$ correspond to defects between the $\eM_i$-boundary and the $\eM_j$-boundary.

\begin{figure}
    \centering
    \includegraphics[width=0.85\linewidth]{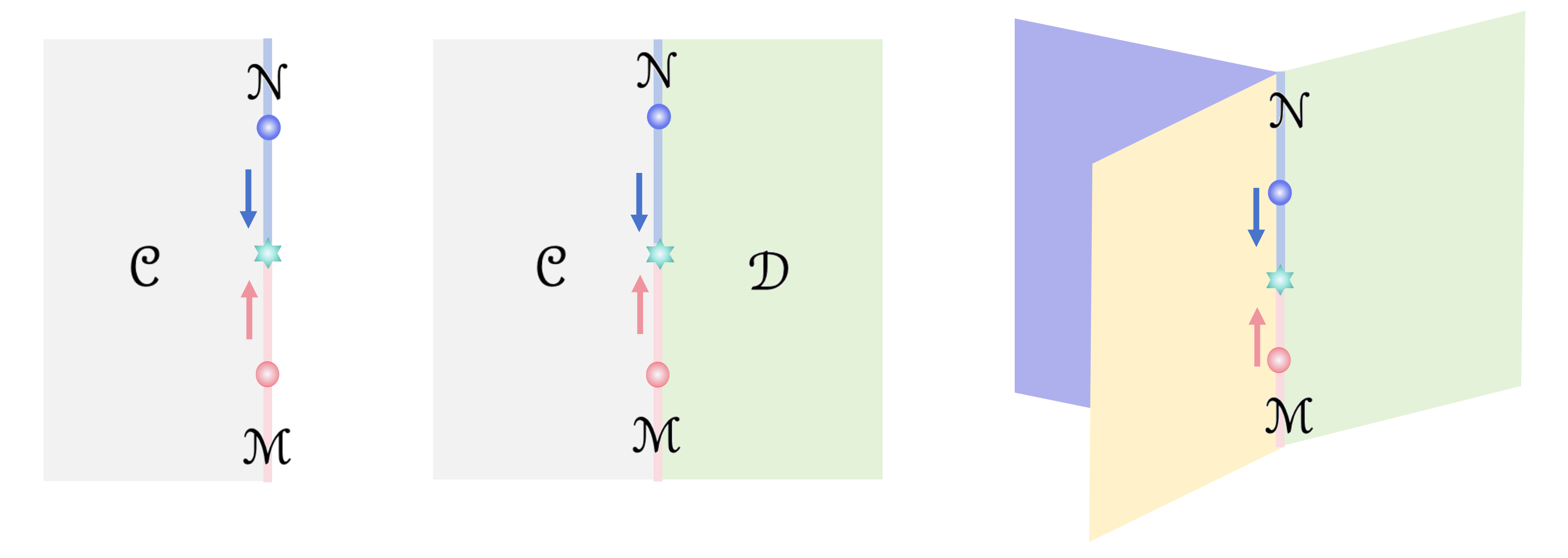}
    \caption{Defects corresponding to gapped boundaries, gapped domain walls, and multimodule domain walls. Fusion between defects, boundary charges, and domain wall charges occurs in one dimension, such that the direction of fusion plays a critical role.}
    \label{fig:defect}
\end{figure}


However, a direct algebraic characterization of all irreducible boundary excitations and boundary defects is generally difficult.
In the finite group case, where the bulk UFC is $\eC = \Vect_G$ \cite{cong2017defects,Cong2017}, defects between two gapped boundaries specified by subgroups $K_1, K_2 \subseteq G$ are classified by the set
\begin{equation}
\bigl\{ (T,R) : T \in K_1 \backslash G / K_2,\; R \in \Irr \left(\Rep( K_1\cap^{r_T} K_2) \right) \bigr\},
\end{equation}
where $K_1\cap^{r_T} K_2 := K_1 \cap r_T K_2 r_T^{-1}$ for a representative $r_T \in T$.
Moreover, the quantum dimension of a simple defect is given by
\begin{equation}
\operatorname{FPdim}(T,R)
=
\frac{\sqrt{|K_1||K_2|}}{\bigl| K_1\cap^{r_T} K_2 \bigr|}
\cdot \operatorname{dim}(R).
\end{equation}
Taking $K_1 = K_2 = K$ reproduces the classification of boundary excitations for the $K$-boundary.

For a general UFC $\eC$, the problem is considerably more challenging.
The boundary tube algebra provides a powerful approach to characterizing boundary excitations: one can reconstruct the boundary tube algebra $\mathbf{Tube}({_{\eC}\eM})$ \cite{Kitaev2012boundary,kong2012universal,bridgeman2023invertible,jia2024weakTube,jia2025tube}, which forms a weak Hopf algebra whose irreducible representations correspond to boundary charge types.
In contrast, boundary defects have not been systematically analyzed in the existing literature; see Refs.~\cite{Kitaev2012boundary,jia2024weakTube,jia2025tube} for brief discussions.

In this work, we argue that boundary defects, corresponding to the off-diagonal sectors $\eP_{i,j}$, are characterized by a defect tube algebra $\mathbf{Tube}({_{\eC}\eM_i};{_{\eC}\eM_j})$, which, as we will show, takes the structure of a $\mathbf{Tube}({_{\eC}\eM_j})|\mathbf{Tube}({_{\eC}\eM_i})$-bicomodule algebra\footnote{Recall that  $\mathbf{Tube}({_{\eC}\eM_k}) := \mathbf{Tube}({_{\eC}\eM_k};{_{\eC}\eM_k})$.}.
All irreducible boundary defects are then given by irreducible representations of $\mathbf{Tube}({_{\eC}\eM_i};{_{\eC}\eM_j})$.
In this way, we obtain a correspondence
\begin{equation}
    \eP_{i,j} \simeq \Rep\bigl(\mathbf{Tube}({_{\eC}\eM_i};{_{\eC}\eM_j})\bigr).
\end{equation}
This provides a complete algebraic characterization of boundary excitations and boundary defects from the quantum algebra symmetry perspective.

\subsection{Algebraic theory of domain wall defects}

Domain wall defects can be mapped to boundary defects via the folding trick.
For two non-chiral topological phases specified by input UFCs $\eC$ and $\eD$, 
a domain wall ${}_{\eC}\eM_{\eD}$ can be folded into a topological boundary 
of the phase with input UFC $\eC \boxtimes \eD^{\rm rev}$.\footnote{The reverse category \(\eD^{\mathrm{rev}}\) has the same objects and morphisms as \(\eD\), but with the reversed tensor product \(x\otimes_{\mathrm{rev}}y:=y\otimes x\).} 
Under this correspondence, the boundary is described by a left module category 
${}_{\eC \boxtimes \eD^{\rm rev}}\eM$.
Accordingly, a domain wall defect 
$\mathfrak{a} \in \Fun_{\eC|\eD}(\eM,\eN)$ 
is mapped to a boundary defect
\begin{equation}
    \mathfrak{a} \in \Fun_{\eC \boxtimes \eD^{\rm rev}}(\eM,\eN).
\end{equation}

Given two bulk phases characterized by UFCs $\eC$ and $\eD$, 
a domain wall ${}_{\eC}\eM^i_{\eD}$ supports a topological phase 
described by the domain wall tube algebra 
$\mathbf{Tube}({}_{\eC}\eM^i_{\eD})$. 
One has
\begin{equation}
  \eP_{i,i} 
  = \mathsf{Fun}({}_{\eC}\eM^i_{\eD},{}_{\eC}\eM^i_{\eD})
  \simeq 
  \Rep\!\left(\mathbf{Tube}({}_{\eC}\eM^i_{\eD})\right).
\end{equation}
See Refs.~\cite{jia2024weakTube,jia2025tube,bai2025weakhopf} for the detailed discussion.

As we will show, by an argument parallel to the boundary case, 
domain wall defects between ${}_{\eC}\eM^i_{\eD}$ and ${}_{\eC}\eM^j_{\eD}$ 
are characterized by the tube algebra 
$\mathbf{Tube}({}_{\eC}\eM^i_{\eD};{}_{\eC}\eM^j_{\eD})$. 
In contrast to the diagonal case, this algebra is no longer a weak Hopf algebra; 
instead, it forms a bicomodule algebra over the weak Hopf algebras 
$\mathbf{Tube}({}_{\eC}\eM^j_{\eD})$ 
and 
$\mathbf{Tube}({}_{\eC}\eM^i_{\eD})$. 
Correspondingly,
\begin{equation}
  \eP_{i,j} 
  = \mathsf{Fun}({}_{\eC}\eM^i_{\eD},{}_{\eC}\eM^j_{\eD})
  \simeq 
  \Rep\!\left(\mathbf{Tube}({}_{\eC}\eM^i_{\eD};{}_{\eC}\eM^j_{\eD})\right).
\end{equation}
Collecting all domain wall charges and defects, one obtains a multifusion category
\begin{equation}
    \eP^{\rm wall}_{\eC|\eD} = \bigoplus_{i,j=0}^{n-1} \eP_{i,j}.
\end{equation}
This category is fully characterized by the weak Hopf tube algebras together with the associated comodule tube algebras.

The tube algebra characterization of topological charges and defects becomes particularly rich in the presence of domain walls, where it admits a nontrivial structural interpretation.
Since a bulk phase can be regarded as the domain wall phase associated with the regular domain wall, the boundary--bulk correspondence implies that the bulk theory is given by the Drinfeld center of the boundary theory. In particular, one is led to
\begin{equation}
    \Rep\!\left(\mathbf{Tube}({}_{\eC}\eC_{\eC})\right)
    \simeq 
    \mathcal{Z}\!\left(
    \Rep\!\left(\mathbf{Tube}({}_{\eC}\eC)\right)
    \right).
\end{equation}
This suggests that the bulk tube algebra can be interpreted as a Drinfeld quantum double of the boundary tube algebra.

More generally, we propose that (here $\bicrosslr$ denotes a gluing operation, analogous to the bicrossed product appearing in the construction of the Drinfeld double)\footnote{Throughout, we abbreviate $\mathbf{Tube}({}_{\eC}\eM_{\eD};{}_{\eC}\eM_{\eD})$ as $\mathbf{Tube}({}_{\eC}\eM_{\eD})$, etc., to avoid clutter in the notation when there is no ambiguity.}
\begin{equation}
  \mathbf{Tube}({}_{\eC}\eM_{\eD})
  \simeq 
  \mathbf{Tube}_{\rm bd}(\eM_{\eD})
  \,\bicrosslr\,
  \mathbf{Tube}_{\rm bd}({}_{\eC}\eM),
\end{equation}
namely, the tube algebra associated with a domain wall can be obtained by gluing the boundary tube algebras on the two sides. 
This construction gives rise to a Drinfeld double--like structure for a pair of weak Hopf algebras.

This picture extends naturally to domain wall defects. In this case, the corresponding tube algebra can be expressed as
\begin{equation}
  \mathbf{Tube}({}_{\eC}\eM_{\eD};{}_{\eC}\eN_{\eD})
  \simeq 
  \mathbf{Tube}_{\rm bd}(\eM_{\eD};\eN_{\eD})
  \,\bicrosslr\,
  \mathbf{Tube}_{\rm bd}({}_{\eC}\eM;{}_{\eC}\eN),
\end{equation}
which endows a pair of comodule tube algebras with a Drinfeld double--like structure via the gluing operation.
To the best of our knowledge, this construction defines a new operation for comodule algebras that has not been previously discussed in the literature.

\subsection{Multimodule domain wall defects}

For domain walls connecting more than two bulk phases (see Fig.~\ref{fig:defect}), one needs to employ multimodule categories to characterize the domain wall~\cite{jia2025tube}. 
A multimodule category over UFCs $\{\eC_{\alpha}\}_{\alpha\in I}$ and $\{\eD_{\beta}\}_{\beta \in J}$ is denoted by $_{\{\eC_{\alpha}\}}\eM_{\{\eD_{\beta}\}}$. 
The reason for dividing the bulks into two groups is that each $(2+1)$D bulk has its own orientation, which may or may not align with the orientation of the domain wall.

The tube algebra can be extended to the multimodule domain wall $\mathbf{Tube}(_{\{\eC_{\alpha}\}}\eM_{\{\eD_{\beta}\}})$, which can be proven to be a weak Hopf algebra.
Fixing bulk phases and an indecomposable domain wall $_{\{\eC_{\alpha}\}}\eM_{\{\eD_{\beta}\}}^i$, we similarly can characterize the topological excitations in the domain wall via the corresponding tube algebra
\begin{equation}
\eP_{i,i} =   \Fun_{{\{\eC_{\alpha}\}}|{\{\eD_{\beta}\}}}(_{\{\eC_{\alpha}\}}\eM_{\{\eD_{\beta}\}}^i,{_{\{\eC_{\alpha}\}}}\eM_{\{\eD_{\beta}\}}^i) \simeq \Rep(\mathbf{Tube}(_{\{\eC_{\alpha}\}}\eM^i_{\{\eD_{\beta}\}})).
\end{equation}
An interesting aspect is that this type of tube algebra provides a natural generalization of the Drinfeld quantum double. 
For $n$ bulks, the resulting multimodule tube algebra can be regarded as an \emph{$N$-tuple algebra}~\cite{jia2025tube} for the $N$ boundary tube algebras:
\begin{equation}
 \mathbf{Tube}(_{\{\eC_{\alpha}\}}\eM_{\{\eD_{\beta}\}}) =   \left( \bicrosslr_{\alpha\in I} \mathbf{Tube}(_{\eC_{\alpha}}\eM) \right) \bicrosslr \left( \bicrosslr_{\beta\in J} \mathbf{Tube}(\eM_{\eD_{\beta}}) \right) 
\end{equation}
where $|I|+|J|=N$.
There is a hierarchical structure for these algebras: any $M$-tuple algebra with $M<N$ can be naturally embedded into the $N$-tuple algebra by assigning trivial strings to the bulks that do not appear in the $M$-tuple. In particular, all quantum doubles are contained within this $N$-tuple algebra.

As we will show later, the above discussion also holds for defects connecting multimodule domain walls. The only difference is that the resulting tube algebra 
$\mathbf{Tube}({_{\{\eC_{\alpha}\}}\eM^i_{\{\eD_{\beta}\}}};{_{\{\eC_{\alpha}\}}\eM^j_{\{\eD_{\beta}\}}})$
forms a comodule algebra.
The representation category of this defect tube algebra is equivalent to the domain wall defect category
\begin{equation}
\eP_{i,j} =   \Fun_{{\{\eC_{\alpha}\}}|{\{\eD_{\beta}\}}}(_{\{\eC_{\alpha}\}}\eM_{\{\eD_{\beta}\}}^i,{_{\{\eC_{\alpha}\}}}\eM_{\{\eD_{\beta}\}}^j) \simeq \Rep(\mathbf{Tube}(_{\{\eC_{\alpha}\}}\eM^i_{\{\eD_{\beta}\}};{_{\{\eC_{\alpha}\}}}\eM_{\{\eD_{\beta}\}}^j)).
\end{equation}
By combining domain wall excitations and defects between domain walls, we obtain a multifusion category
\begin{equation}
    \eP_{{\{\eC_{\alpha}\}}|{\{\eD_{\beta}\}}}^{\rm wall} = \bigoplus_{i,j=0}^{n-1} \eP_{i,j}.
\end{equation}
The domain wall weak Hopf tube algebra, together with the defect comodule tube algebra, provides a complete characterization of this multifusion category.

The tube algebra also provides a systematic method to construct $N$-tuple algebras for comodule algebras, a construction that has not been previously discussed in the literature.

\subsection{Lattice Hamiltonian for defects}

Following Refs.~\cite{Kitaev2012boundary,kong2012universal,jia2024weakTube,jia2025tube}, the Hamiltonian for a boundary defect can be constructed by regarding the defect as a dangling string in the string-net framework. For a given functor $\mathfrak{f} \in \Fun_{\eC}(\eM,\eN)$ and objects $x \in \eM$, $y \in \eN$, we interpret $\mathfrak{f}$ as a dangling string and define the fusion of $x$ with $\mathfrak{f}$ by $x \otimes \mathfrak{f} := \mathfrak{f}(x)$ (From this we see that $(\mathfrak{g}\circ \mathfrak{f}) (x)=\mathfrak{g} (\mathfrak{f} (x))$ can be interpreted as $x\otimes (\mathfrak{g}\circ \mathfrak{f})=(x\otimes \mathfrak{f})\otimes \mathfrak{g}$; we thus naturally have $\mathfrak{g}\circ \mathfrak{f} =\mathfrak{f}\otimes \mathfrak{g}$.). The corresponding fusion vertex is then of the following form
\begin{equation}
\langle x,\mathfrak{f}\rightarrow y;\alpha |\,\,   =  \begin{aligned}
		\begin{tikzpicture}
\draw[
  line width=1.6pt,
  cyan,
  postaction={
    decorate,
    decoration={
      markings,
      mark=at position 0.6 with {\arrow{latex}}
    }
  }
] (0,0) -- (0,1);            
\draw[
  decorate, line width=.6pt,
  decoration={snake, amplitude=2pt, segment length=6pt}
] (0,0) -- (1,0);
\draw[-latex, line width=.6pt]
  (1.05,0) -- (1,0);
\draw[line width=.6pt,
  black
] (1,0) -- (1.25,0);
\draw[
  line width=1.6pt,
  violet,
  postaction={
    decorate,
    decoration={
      markings,
      mark=at position 0.6 with {\arrow{latex}}
    }
  }
] (0,-1.0) -- (0,0);   

   \node[ line width=0.6pt, dashed, draw opacity=0.5] (a) at (0.3,0.7){$y$};
    \node[ line width=0.6pt, dashed, draw opacity=0.5] (a) at (0.3,-0.7){$x$};
    \node[ line width=0.6pt, dashed, draw opacity=0.5] (a) at (0.7,0.3){$\mathfrak{f}$};    
    \node[ line width=0.6pt, dashed, draw opacity=0.5] (a) at (-0.3,0){$\alpha$};    
    \node[ line width=0.6pt, dashed, draw opacity=0.5] at (0,0){$\scriptstyle \bullet$}; 
		\end{tikzpicture}
	\end{aligned}
\end{equation}
where a wavy line is used to emphasize the defect string and $\alpha$ is the label of basis vectors.
Accordingly, the vertex space associated with the defect is
$\mathcal{H}_v^{\mathfrak{f}}
=
\bigoplus_{x \in \Irr \eM,\; y \in \Irr \eN}
\Hom\bigl(\mathfrak{f}(x), y\bigr)$.
Note that the defect string is the only string permitted to have a dangling end.

\begin{remark}\label{remk:opFun}
 Given a natural transformation $\eta: \mathfrak{g}\to \mathfrak{f}$, it induces a map from the vertex space of $\mathfrak{f}$ to that of $\mathfrak{g}$ as follows:
    \begin{equation}
        \begin{aligned}
		\begin{tikzpicture}
\draw[
  line width=1.6pt,
  cyan,
  postaction={
    decorate,
    decoration={
      markings,
      mark=at position 0.6 with {\arrow{latex}}
    }
  }
] (0,0) -- (0,1);     
\draw[
  decorate, line width=.6pt,
  decoration={snake, amplitude=2pt, segment length=6pt}
] (0,0) -- (1,0);
\draw[
  decorate, line width=.6pt,
  decoration={snake, amplitude=2pt, segment length=6pt}
] (2.3,0) -- (1.3,0);
\draw[-latex, line width=.6pt]
  (1.05,0) -- (1,0);
\draw[line width=.6pt,
  black
] (1,0) -- (1.25,0);
\draw[-latex, line width=.6pt]
  (2.35,0) -- (2.3,0);
\draw[line width=.6pt,
  black
] (2.3,0) -- (2.55,0);
\draw[
  line width=1.6pt,
  violet,
  postaction={
    decorate,
    decoration={
      markings,
      mark=at position 0.6 with {\arrow{latex}}
    }
  }
] (0,-1.0) -- (0,0);   
   \node[ line width=0.6pt, dashed, draw opacity=0.5] (a) at (0.3,0.7){$y$};
    \node[ line width=0.6pt, dashed, draw opacity=0.5] (a) at (0.3,-0.7){$x$};
    \node[ line width=0.6pt, dashed, draw opacity=0.5] (a) at (0.7,0.3){$\mathfrak{f}$};  
      \node[ line width=0.6pt, dashed, draw opacity=0.5] (a) at (1.9,0.3){$\mathfrak{g}$};  
    \node[ line width=0.6pt, dashed, draw opacity=0.5] (a) at (-0.3,0){$\alpha$};   
        \node[ line width=0.6pt, dashed, draw opacity=0.5] (a) at (1.2,-0.3){$\eta$};  
    \node[ line width=0.6pt, dashed, draw opacity=0.5] at (0,0){$\scriptstyle \bullet$}; 
      \node[ line width=0.6pt, dashed, draw opacity=0.5] at (1.25,0){$\scriptstyle \bullet$}; 
		\end{tikzpicture}
	\end{aligned}
    \end{equation}
This contravariant nature suggests that it is more convenient to work with $\Fun(\eM,\eN)^{\rm op}$ when describing the codimension-2 defect. Further details can be found in Proposition~\ref{prop:Hom-Nat} and its proof.
\end{remark}

The vertex operator $Q_v$ for the defect is the projection to the allowed vertex configuration; the face operator $B_f$ is of similar form as that for the bulk but the fusion symbol for those involving domain wall and boundary strings should use the fusion symbols arising from the module category structure (this has a more natural interpretation using multifusion string-net \cite{jia2024weakTube}, since in this setting, objects in $\eM,\eN$ and module functors lie in the same multifusion category). The defect Hamiltonian is given by
\begin{equation}
    H_{\rm defect}^{\eM,\eN,\mathfrak{f}}=-\sum_f B_f -Q_{v}
\end{equation}
where $\eM$, $\eN$ label the domain wall or boundary and $\mathfrak{f}$ is the defect connecting two domain walls or boundaries.

%% file: Ch-DWTwist.tex
\section{Domain wall twist defects and their tube algebras}
\label{sec:dw-twist}

When the two bulk phases are both described by the same fusion category $\eC$, there exists a distinguished class of domain wall defects known as \emph{twist defects}~\cite{Kitaev2012boundary,jia2022electricmagnetic,Barkeshli2019symmetry}. Such defects interpolate between a nontrivial domain wall ${}_{\eC}\eM_{\eC}$ and the transparent domain wall ${}_{\eC}\eC_{\eC}$. Mathematically, they are described by bimodule functors
\begin{equation}
    \tau_- \in \Fun_{\eC|\eC}(\eM,\eC).
\end{equation}
A schematic illustration is given by
\begin{align}
    \begin{aligned}
        \begin{tikzpicture}[scale=1]
            \filldraw[black!60, fill=gray!15, dotted] (0,0) circle[radius=1.9];
            \draw[line width=1pt,dotted,gray] (90:0.32) -- (90:1.9);
            \draw[-latex,line width=1pt,dotted,gray] (90:0.32) -- (90:1.3);
            \draw[line width=1pt,cyan] (-90:0.32) -- (-90:1.9);
            \draw[-latex,line width=1pt,cyan] (-90:1.3) -- (-90:.9);
            \draw[line width=1pt,black] (0,0) circle (0.32);
            \fill[pattern=north east lines, pattern color=gray] (0,0) circle (0.32);
             \node[ line width=0.6pt, dashed, draw opacity=0.5] (a) at (180:1){$\eC$};
             \node[ line width=0.6pt, dashed, draw opacity=0.5] (a) at (0:1){$\eC$};
             \node[ line width=0.6pt, dashed, draw opacity=0.5] (a) at (270:2.2){$\eM$};
             \node[ line width=0.6pt, dashed, draw opacity=0.5] (a) at (90:2.2){$\eC$};
             \node[ line width=0.6pt, dashed, draw opacity=0.5] (a) at (0,0){$\tau_-$};
        \end{tikzpicture}
    \end{aligned}\;. 
\end{align}
Similarly, one may define the dual type of twist defect
\begin{equation}
    \tau_+ \in \Fun_{\eC|\eC}(\eC,\eM).
\end{equation}
For these special defects, the corresponding defect tube algebra admits a significant simplification and becomes more closely related to the conventional tube algebra of Ocneanu~\cite{ocneanu1994chirality,ocneanu2001operator,muger2003subfactorsII}.

Without loss of generality, let us focus on the defect type $\tau_-$. By considering a tubular neighborhood surrounding the defect, one obtains the tube algebra
\begin{equation}
\textbf{Tube}(\eM;\eC)
=
\operatorname{span}
\left\{
\begin{aligned}
        \begin{tikzpicture}[scale=0.65]
        \filldraw[black!60, fill=gray!15, dotted, even odd rule] (0,0) circle[radius=0.5] (0,0) circle[radius=1.5];
             \draw[line width=1pt,red] (0,0.5)--(0,1.5);
             \draw[line width=1pt,cyan] (0,-0.5)--(0,-1.5);
             \draw[red] (0,0.8) arc[start angle=90, end angle=270, radius=0.8];
             \draw[red] (0,1.3) arc[start angle=90, end angle=-90, radius=1.3];
            \node[ line width=0.6pt, dashed, draw opacity=0.5] (a) at (0,1.7){$\scriptstyle w$};
             \node[ line width=0.6pt, dashed, draw opacity=0.5] (a) at (0,-1.7){$\scriptstyle x$};
            \node[ line width=0.6pt, dashed, draw opacity=0.5] (a) at (-1,0){$\scriptstyle a$};
            \node[ line width=0.6pt, dashed, draw opacity=0.5] (a) at (1.1,0){$\scriptstyle b$};
            \node[ line width=0.6pt, dashed, draw opacity=0.5] (a) at (-0.2,-1){$\scriptstyle y$};
            \node[ line width=0.6pt, dashed, draw opacity=0.5] (a) at (-0.2,-1.35){$\scriptstyle \mu$};
            \node[ line width=0.6pt, dashed, draw opacity=0.5] (a) at (0.2,-0.8){$\scriptstyle \nu$};
            \node[ line width=0.6pt, dashed, draw opacity=0.5] (a) at (0,-0.3){$\scriptstyle z$};
            \node[ line width=0.6pt, dashed, draw opacity=0.5] (a) at (0,0.3){$\scriptstyle u$};
            \node[ line width=0.6pt, dashed, draw opacity=0.5] (a) at (-0.2,1){$\scriptstyle v$};
            \node[ line width=0.6pt, dashed, draw opacity=0.5] (a) at (-0.2,1.3){$\scriptstyle \gamma$};
            \node[ line width=0.6pt, dashed, draw opacity=0.5] (a) at (0.2,0.8){$\scriptstyle \zeta$};
        \end{tikzpicture}
    \end{aligned}
\right\}.
\end{equation}
Here $a,b,u,v,w\in\Irr(\eC)$ and $x,y,z\in\Irr(\eM)$, while the morphisms $\mu,\nu,\gamma,\zeta$ belong to the corresponding morphism spaces determined by the string labels. By the general discussion of the previous section, this tube algebra naturally carries the structure of a bicomodule algebra.
Furthermore, since the upper part of the tube lies entirely in the $\eC$ string-net sector, it can be reduced using the fusion-category relations. As a result, the tube algebra admits the simplified presentation
\begin{equation}
\mathbf{Tube}'(\eM;\eC)
=
\operatorname{span}
\left\{
\begin{aligned}
        \begin{tikzpicture}[scale=0.65]
        \filldraw[black!60, fill=gray!15, dotted, even odd rule] (0,0) circle[radius=0.5] (0,0) circle[radius=1.5];
             \draw[line width=1pt,cyan] (0,-0.5)--(0,-1.5);
             \draw[red] (0,0.8) arc[start angle=90, end angle=270, radius=0.8];
             \draw[red] (0,0.8) arc[start angle=90, end angle=-90, radius=1];
             \node[ line width=0.6pt, dashed, draw opacity=0.5] (a) at (0,-1.7){$\scriptstyle x$};
             \node[ line width=0.6pt, dashed, draw opacity=0.5] (a) at (0,1.1){$\scriptstyle a$};
             \node[ line width=0.6pt, dashed, draw opacity=0.5] (a) at (-0.2,-1){$\scriptstyle y$};
             \node[ line width=0.6pt, dashed, draw opacity=0.5] (a) at (-0.2,-1.35){$\scriptstyle \mu$};
             \node[ line width=0.6pt, dashed, draw opacity=0.5] (a) at (0.2,-0.8){$\scriptstyle \nu$};
             \node[ line width=0.6pt, dashed, draw opacity=0.5] (a) at (0,-0.3){$\scriptstyle z$};
        \end{tikzpicture}
    \end{aligned}
\right\}.
\end{equation}
The unit element of $\mathbf{Tube}'(\eM;\eC)$ is represented by the tube containing a single domain wall edge, while the multiplication is inherited from the original tube algebra and is defined by the gluing of tubes. 

The $\mathbf{Tube}'(\eM;\eC)$  is also a right $\mathbf{Tube}(\eM)$-comodule algebra, whose coaction is induced by the fusion structure 
\begin{equation}
    \begin{aligned}
        \begin{tikzpicture}[scale=0.65]
            \filldraw[black!60, fill=gray!15, dotted] (0,0.75) circle[radius=1.95]; 
            \draw[line width=1pt,cyan] (0,0.3)--(0,-1.2);
            \draw[line width=1pt,cyan] (0,0.3)--(0,1.2);
            \draw[fill=gray!15] (0,1.3) circle (0.2); 
            \fill[pattern=north east lines, pattern color=gray] (0,1.3) circle (0.2);
            \draw[fill=gray!15] (0,0.2) circle (0.2); 
            \fill[pattern=north east lines, pattern color=gray] (0,0.2) circle (0.2);
            \draw[red] (0,2) arc[start angle=90, end angle=270, radius=1.25];
            \draw[red] (0,2) arc[start angle=90, end angle=-90, radius=1.35];
            \node[ line width=0.6pt, dashed, draw opacity=0.5] (a) at (0.2,0.75){$\scriptstyle i$};
            \node[ line width=0.6pt, dashed, draw opacity=0.5] (a) at (-1.45,0.75){$\scriptstyle a$};
            \node[ line width=0.6pt, dashed, draw opacity=0.5] (a) at (0.5,0.2){$\scriptstyle \mathfrak{f}$};
            \node[ line width=0.6pt, dashed, draw opacity=0.5] (a) at (0.6,1.3){$\scriptstyle \tau_{-}$};                 
        \end{tikzpicture}
    \end{aligned}
    = \sum_{j,\sigma} \sqrt{ \frac{d_j}{d_ad_i}}\sum_{k,\rho}\sqrt{ \frac{d_k}{d_jd_{\bar{a}}}} \;\;
    \begin{aligned}
        \begin{tikzpicture}[scale=0.75]
            \filldraw[black!60, fill=gray!15, dotted] (0,1.8) circle[radius=1.2];
            \filldraw[black!60, fill=gray!15, dotted] (0,-0.6) circle[radius=1.2];
            \draw[line width=1pt,cyan] (0,-0.6)--(0,-1.8);
             \draw[line width=1pt,cyan] (0,-0.6)--(0,1.6);
            \draw[fill=gray!15] (0,1.8) circle (0.2); 
            \fill[pattern=north east lines, pattern color=gray] (0,1.8) circle (0.2);
            \draw[fill=gray!15] (0,-0.6) circle (0.2); 
            \fill[pattern=north east lines, pattern color=gray] (0,-0.6) circle (0.2);
            \draw[red] (0,2.4) arc[start angle=90, end angle=270, radius=0.6];
            \draw[red] (0,2.4) arc[start angle=90, end angle=-90, radius=0.75];
            \draw[red] (0,0) arc[start angle=90, end angle=270, radius=0.6];
            \draw[red] (0,0.3) arc[start angle=90, end angle=-90, radius=0.9];
            \node[ line width=0.6pt, dashed, draw opacity=0.5] (a) at (0.35,2.08){$\scriptstyle \tau_-$};  
            \node[ line width=0.6pt, dashed, draw opacity=0.5] (a) at (0.3,-0.8){$\scriptstyle \mathfrak{f}$};
            \node[ line width=0.6pt, dashed, draw opacity=0.5] (a) at (-0.8,-.6){$\scriptstyle a$}; 
            \node[ line width=0.6pt, dashed, draw opacity=0.5] (a) at (1,1.8){$\scriptstyle \bar{a}$};
            \node[ line width=0.6pt, dashed, draw opacity=0.5] (a) at (-0.8,1.8){$\scriptstyle a$};
            \node[ line width=0.6pt, dashed, draw opacity=0.5] (a) at (0.7,-.6){$\scriptstyle \bar{a}$};     
            \node[ line width=0.6pt, dashed, draw opacity=0.5] (a) at (0.2,1.2){$\scriptstyle \sigma$};  
            \node[ line width=0.6pt, dashed, draw opacity=0.5] (a) at (0.15,1.46){$\scriptstyle i$}; 
            \node[ line width=0.6pt, dashed, draw opacity=0.5] (a) at (0.2,0.75){$\scriptstyle \rho$};  
            \node[ line width=0.6pt, dashed, draw opacity=0.5] (a) at (0.2,0.45){$\scriptstyle \rho$};
            \node[ line width=0.6pt, dashed, draw opacity=0.5] (a) at (0.2,0){$\scriptstyle \sigma$}; 
            \node[ line width=0.6pt, dashed, draw opacity=0.5] (a) at (-0.18,1){$\scriptstyle j$};               
            \node[ line width=0.6pt, dashed, draw opacity=0.5] (a) at (-0.2,0.6){$\scriptstyle k$};   
            \node[ line width=0.6pt, dashed, draw opacity=0.5] (a) at (-0.18,0.2){$\scriptstyle j$}; 
            \node[ line width=0.6pt, dashed, draw opacity=0.5] (a) at (0.15,-0.25){$\scriptstyle i$};  
        \end{tikzpicture}
    \end{aligned}\;\; .
\end{equation}
Moreover, $\mathbf{Tube}(\eM;\eC)$ and $\mathbf{Tube}'(\eM;\eC)$ are Morita equivalent and therefore possess equivalent representation categories.

The preceding discussion shows that every domain wall twist defect may be viewed as a special class of domain wall defects whose associated tube algebra admits a particularly simple description. One may then ask whether the converse is also true: can an arbitrary domain wall defect be reformulated as a domain wall twist defect? The answer is affirmative, although establishing this correspondence requires additional categorical machinery.
The key ingredient is the relative tensor product of module categories over unitary fusion categories; we refer the reader to~\cite{etingof2010fusion} for a detailed account. Given a $\eC|\eD$-bimodule category $\eM$ and a $\eD|\eE$-bimodule category $\eN$, their fusion is described by the relative tensor product
\begin{equation}
    \eM\boxtimes_{\eD}\eN,
\end{equation}
which naturally inherits the structure of a $\eC|\eE$-bimodule category.
Using this construction, an arbitrary domain wall defect can be reinterpreted as a domain wall twist defect through the following procedure:
\begin{align}
    \begin{aligned}
        \begin{tikzpicture}[scale=1]
            \filldraw[black!60, fill=gray!15,dotted] (0,0) circle[radius=1.9];
            \draw[line width=1pt,cyan] (90:0.32) -- (90:1.9);
            \draw[-latex,line width=1pt,cyan] (90:0.32) -- (90:1.3);
            \draw[line width=1pt,violet] (-90:0.32) -- (-90:1.9);
            \draw[-latex,line width=1pt,violet] (-90:1.3) -- (-90:.9);
            \draw[line width=1pt,black] (0,0) circle (0.32);
            \fill[pattern=north east lines, pattern color=gray] (0,0) circle (0.32);
             \node[ line width=0.6pt, dashed, draw opacity=0.5] (a) at (180:1){$\eC$};
             \node[ line width=0.6pt, dashed, draw opacity=0.5] (a) at (0:1){$\eD$};
             \node[ line width=0.6pt, dashed, draw opacity=0.5] (a) at (270:2.2){$\eM$};
             \node[ line width=0.6pt, dashed, draw opacity=0.5] (a) at (90:2.2){$\eN$};
             \node[ line width=0.6pt, dashed, draw opacity=0.5] (a) at (0,0){$\tau$};
        \end{tikzpicture}
    \end{aligned} \mapsto     
    \begin{aligned}
        \begin{tikzpicture}[scale=1]
            \filldraw[black!60, fill=gray!15,dotted] (0,0) circle[radius=1.9];
            \draw[line width=1.5pt,blue] (90:0.32) -- (90:1.9);
            \draw[-latex,line width=1.5pt,blue] (90:0.32) -- (90:1.3);
            \draw[line width=1pt,black] (0,0) circle (0.32);
            \fill[pattern=north east lines, pattern color=gray] (0,0) circle (0.32);
             \node[ line width=0.6pt, dashed, draw opacity=0.5] (a) at (270:1){$\eD$};
             \node[ line width=0.6pt, dashed, draw opacity=0.5] (a) at (1,1){$\eM^{\rm op}\boxtimes_{\eC} \eN$};
             \node[ line width=0.6pt, dashed, draw opacity=0.5] (a) at (0,0){$\tau_+$};
        \end{tikzpicture}
    \end{aligned}
    \simeq 
        \begin{aligned}
        \begin{tikzpicture}[scale=1]
            \filldraw[black!60, fill=gray!15,dotted] (0,0) circle[radius=1.9];
            \draw[line width=1.5pt,brown] (90:-1.9) -- (90:-0.32);
            \draw[-latex,line width=1.5pt,brown] (90:-1.32) -- (90:-0.9);
            \draw[line width=1pt,black] (0,0) circle (0.32);
            \fill[pattern=north east lines, pattern color=gray] (0,0) circle (0.32);
             \node[ line width=0.6pt, dashed, draw opacity=0.5] (a) at (90:1){$\eC$};
             \node[ line width=0.6pt, dashed, draw opacity=0.5] (a) at (1,-1){$\eM\boxtimes_{\eD} \eN^{\rm op}$};
             \node[ line width=0.6pt, dashed, draw opacity=0.5] (a) at (0,0){$\tau_-$};
        \end{tikzpicture}
    \end{aligned}
\end{align}
where we use the fact $\eM^{\rm op}$ is a $\eD|\eC$-bimodule and $\Fun_{\eC|\eD}(\eM,\eN)\simeq  \eM^{\rm op}\boxtimes_{\eC} \eN$. This is fusion of two domain walls over the bulk $\eC$.

It is also worth noting that bulk excitations can be viewed as a special class of domain wall twist defects, namely defects interpolating between two transparent (regular) domain walls ${}_{\eC}\eC_{\eC}$. Consequently, the comodule tube algebra framework provides a unified description of both domain wall defects and bulk excitations. In this perspective, the conventional tube algebra arises as a special case of the comodule tube algebra associated with transparent domain walls.

%% file: Ch-Multicomodule.tex
\section{Multicomodule tube algebra for codimension-2 defects connecting multiple domain walls}
\label{section:multicomodule}

Let us now turn to a more general configuration of (2+1)D topological order with both codimension-1 and codimension-2 defects, as shown in Figure~\ref{fig:2dTOdef}. Here, several (2+1)D topological orders are separated by codimension-1 defects, and these defects are themselves joined by codimension-2 defects. A key feature of this network is the presence of codimension-2 defects that connect \(N\) codimension-1 defects. Considering the tubular region around such a defect naturally yields \emph{multicomodule tube algebra}.

\begin{figure}[h]
    \centering
    \includegraphics[width=0.6\linewidth]{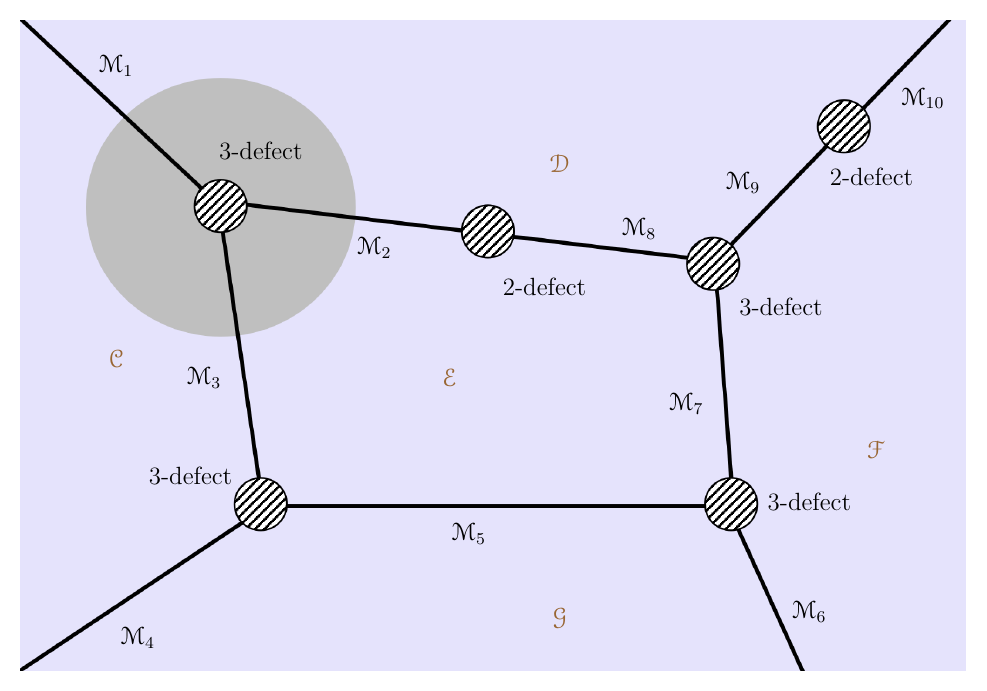}
\caption{General configuration of defects in a two-dimensional topological phase. A codimension-2 $N$-defect may connect $N$ codimension-1 defects. By considering the tube surrounding the $N$-defect, one obtains the associated multicomodule defect tube algebra.}
    \label{fig:2dTOdef}
\end{figure}

In general, bulk phases are labeled by UFCs. Codimension-1 defects separating bulk phases labeled by $\eC_i$ and $\eC_j$ are described by $\eC_j|\eC_i$-bimodule categories. Since codimension-1 defects are oriented, a defect with a fixed orientation is labeled by a bimodule category $\eM$, while reversing its orientation replaces $\eM$ with the opposite bimodule category $\eM^{\rm op}$. Correspondingly, if $\eM$ is a $\eC_j|\eC_i$-bimodule category, then $\eM^{\rm op}$ naturally acquires the structure of a $\eC_i|\eC_j$-bimodule category.

Let $N\geq 2$ be an integer. Consider a configuration consisting of $N$ codimension-1 defects, labeled by bimodule categories, meeting at a codimension-2 defect $\tau$. Up to isotopy, such a configuration may be represented as in Figure~\ref{fig:N_defect}(a). Let $\{\eM_i\}_{1\leq i\leq N}$ be a collection of bimodule categories and $\{\eC^{i,i+1}\}_{1\leq i\leq N}$ the corresponding bulk fusion categories\footnote{We adopt this superscript notation for convenience in describing the hierarchical structure that will appear later. It should not be confused with the standard notation for multifusion categories.}. Fix an integer $1\leq k_0\leq N$. For $1\leq i\leq k_0$, the bimodule category $\eM_i$ is taken to be a $\eC^{i,i+1}|\eC^{i-1,i}$-bimodule category, corresponding to a codimension-1 defect oriented away from the codimension-2 defect. For $k_0+1\leq i\leq N$, $\eM_i$ is a $\eC^{i-1,i}|\eC^{i,i+1}$-bimodule category, corresponding to a codimension-1 defect oriented toward the codimension-2 defect. Notice any other configuration can be turned into this by reversing the orientations of some codimension-1 defects.

Consider a tubular neighborhood surrounding the codimension-2 $N$-defect. By applying topological local moves, this neighborhood can be reduced to a canonical tube configuration, from which one constructs the associated defect tube algebra
\begin{equation}
    \mathbf{Tube}(\{(\eM_i)\}_{1\leq i\leq N}),
\end{equation}
or, more succinctly,
\begin{equation}
    \mathbf{Tube}(\eM_1;\eM_2;\ldots;\eM_N),
\end{equation}
where the underlying fusion categories are understood implicitly. As we shall show, this algebra is naturally equipped with multiple compatible comodule structures arising from the adjacent codimension-1 defects, leading to the notion of a \emph{multicomodule algebra}.

\begin{definition}[Multicomodule algebra]
Let $\{A_i\}_{i\in I}$ and $\{B_j\}_{j\in J}$ be families of weak Hopf algebras. An \(\{A_i\}_{i\in I}\vert\{B_j\}_{j\in J}\)-multicomodule algebra is an algebra $\mathfrak{C}$ equipped with left $A_i$-coactions $\beta_i:\mathfrak{C}\to A_i\otimes \mathfrak{C}$ and right $B_j$-coactions $\rho_j:\mathfrak{C}\to \mathfrak{C}\otimes B_j$, each of which endows $\mathfrak{C}$ with the structure of a comodule algebra, subject to the following compatibility conditions:
\begin{enumerate}
    \item For any $i,k\in I$, the left coactions commute:
    \begin{equation}
        (\id_{A_i}\otimes \beta_k)\circ \beta_i
        =
        (\sigma_{A_k,A_i}\otimes \id_{\mathfrak{C}})
        \circ
        (\id_{A_k}\otimes \beta_i)
        \circ
        \beta_k.
    \end{equation}
    Equivalently, $\mathfrak{C}$ is an $A_i|A_k^{\rm op}$ bicomodule algebra for every pair $i,k\in I$.

    \item For any $j,l\in J$, the right coactions commute:
    \begin{equation}
        (\rho_l\otimes \id_{B_j})\circ \rho_j
        =
        (\id_{\mathfrak{C}}\otimes \sigma_{B_j,B_l})
        \circ
        (\rho_j\otimes \id_{B_l})
        \circ
        \rho_l.
    \end{equation}
    Equivalently, $\mathfrak{C}$ is a $B_l|B_j^{\rm op}$ bicomodule algebra for every pair $l,j\in J$.

    \item For any $i\in I$ and $j\in J$, the left and right coactions are compatible:
    \begin{equation}
        (\id_{A_i}\otimes \rho_j)\circ \beta_i
        =
        (\beta_i\otimes \id_{B_j})\circ \rho_j.
    \end{equation}
    Equivalently, $\mathfrak{C}$ is an $A_i|B_j$ bicomodule algebra for every $i\in I$ and $j\in J$.
\end{enumerate}
Here, $\sigma_{A,B}:A\otimes B\to B\otimes A$ denotes the braiding defined by
\begin{equation}
    \sigma_{A,B}(a\otimes b)=b\otimes a.
\end{equation}
When $|I|=|J|=1$, this notion reduces to that of a bicomodule algebra; when $|I|+|J|=3$, one obtains a tricomodule algebra, etc.
\end{definition}

A string-net lattice model for this configuration can be constructed as a generalization of the model for 2-defects~\cite{jia2024weakTube}, although the defect data in the present setting are considerably more intricate. The key observation is that one must separately fuse the inward-oriented and outward-oriented bimodule categories, leading to the composite bimodule categories
\begin{align}
  & \eM_{\rm in}
  =
  \eM_{k_0+1}
  \boxtimes_{\eC^{k_0+1,k_0+2}}
  \eM_{k_0+2}
  \boxtimes_{\eC^{k_0+2,k_0+3}}
  \cdots
  \boxtimes_{\eC^{N-1,N}}
  \eM_N,\\
  & \eM_{\rm out}
  =
  \eM_{k_0}
  \boxtimes_{\eC^{k_0-1,k_0}}
  \eM_{k_0-1}
  \boxtimes_{\eC^{k_0-2,k_0-1}}
  \cdots
  \boxtimes_{\eC^{1,2}}
  \eM_1.
\end{align}
Both $\eM_{\rm in}$ and $\eM_{\rm out}$ naturally inherit the structure of $\eC^{k_0,k_0+1}|\eC^{N,1}$-bimodule categories. We refer the reader to~\cite{etingof2010fusion} for a detailed discussion of the tensor product of module categories.
In this framework, the $N$-defect is algebraically described by a bimodule functor
\begin{equation}
    \tau \in \Fun_{\eC^{k_0,k_0+1}|\eC^{N,1}}(\eM_{\rm in},\eM_{\rm out}).
\end{equation}
As we shall see, this categorical description is reflected naturally in the hierarchical structure of the associated multicomodule tube algebra.
Another equivalent description of the $N$-defect follows from the well-known equivalence
\begin{equation}
    \Fun_{\eC}(\eM,\eN)
    \simeq
    \eM^{\rm op}\boxtimes_{\eC}\eN,
\end{equation}
for right $\eC$-module categories $\eM$ and $\eN$. Applying this equivalence to the present configuration, one may reverse the orientations of the codimension-1 defects so that they are all oriented either toward or away from the codimension-2 defect. The defect $\tau$ can then be regarded as an object of the category of tensor product of bimodules
\begin{equation}
    \tau \in
    \eM^{\rm op}_1
    \boxtimes_{\eC^{1,2}}
    \eM^{\rm op}_2
    \boxtimes_{\eC^{2,3}}
    \cdots
    \boxtimes_{\eC^{k_0-1,k_0}}
    \eM^{\rm op}_{k_0}
    \boxtimes_{\eC^{k_0,k_0+1}}
    \eM_{k_0+1}
    \boxtimes_{\eC^{k_0+1,k_0+2}}
    \cdots
    \boxtimes_{\eC^{N-1,N}}
    \eM_N.
\end{equation}
This transforms an $N$-defect into a twist defect.


\subsection{Codimension-2  tri-defect comodule tube algebra}

We now consider the simplest nontrivial case, namely $N=3$, which already exhibits all essential features of the general construction. The corresponding algebraic structure will be referred to as the \emph{tri-defect comodule tube algebra} and denoted by $\mathbf{Tube}(\eM_1;\eM_2;\eM_3)$.
There are two orientation patterns for three domain walls meeting at a codimension-2 defect: either two domain walls are oriented toward the junction and one away from it, or one domain wall is oriented toward the junction and two away from it. These two configurations are related by reversing the orientations of a suitable subset of domain walls and replacing the corresponding bimodule categories by their opposite bimodule categories.
Throughout this subsection, $\eM_1$, $\eM_2$, and $\eM_3$ are taken to be $\eC^{1,2}|\eC^{3,1}$-, $\eC^{1,2}|\eC^{2,3}$-, and $\eC^{2,3}|\eC^{3,1}$-bimodule categories, respectively. In other words, we consider the configuration
\begin{align}
    \begin{aligned}
        \begin{tikzpicture}[scale=1]
            \filldraw[black!60, fill=gray!15, dotted] (0,0) circle[radius=1.9];
            \draw[line width=1pt,cyan] (90:0.32) -- (90:1.9);
            \draw[-latex,line width=1pt,cyan] (90:0.32) -- (90:1.3);
            \draw[line width=1pt,violet] (210:0.32) -- (210:1.9);
                        \draw[-latex,line width=1pt,violet] (210:1.9) -- (210:1.32);
            \draw[line width=1pt,magenta] (330:0.32) -- (330:1.9);
                        \draw[-latex,line width=1pt,magenta]  (330:1.9) -- (330:1.32);
            \draw[line width=1pt,black] (0,0) circle (0.32); 
            \fill[pattern=north east lines, pattern color=gray] (0,0) circle (0.32);
             \node[ line width=0.6pt, dashed, draw opacity=0.5] (a) at (150:1){$\eC^{1,2}$};
             \node[ line width=0.6pt, dashed, draw opacity=0.5] (a) at (270:1.4){$\eC^{2,3}$};
             \node[ line width=0.6pt, dashed, draw opacity=0.5] (a) at (30:1.4){$\eC^{3,1}$};
             \node[ line width=0.6pt, dashed, draw opacity=0.5] (a) at (330:2.2){$\eM_3$};
                          \node[ line width=0.6pt, dashed, draw opacity=0.5] (a) at (90:2.2){$\eM_1$};
                                                    \node[ line width=0.6pt, dashed, draw opacity=0.5] (a) at (210:2.2){$\eM_2$};
             \node[ line width=0.6pt, dashed, draw opacity=0.5] (a) at (0,0){$\tau$};
        \end{tikzpicture}
    \end{aligned}
\end{align}
Accordingly, we shall omit all orientation arrows hereafter whenever no confusion can arise.

In this scenario, following the same idea as that for 2-defects, we construct the tri-defect comodule tube algebra 
$\mathbf{Tube}(\eM_1;\eM_2;\eM_3)$, which encodes the interactions of three domain walls meeting at a defect point. 
As a vector space, it is of the form
\begin{equation} \label{eq:tri-tube-basis}
 \mathbf{Tube}(\eM_1;\eM_2;\eM_3)=\operatorname{span} \left\{  \begin{aligned}
        \begin{tikzpicture}[scale=0.7]
            \filldraw[black!60, fill=gray!15, dotted, even odd rule] (0,0) circle[radius=0.5] (0,0) circle[radius=1.9];
            \draw[line width=1pt,cyan] (90:0.5) -- (90:1.9);
            \draw[line width=1pt,violet] (210:0.5) -- (210:1.9);
            \draw[line width=1pt,magenta] (330:0.5) -- (330:1.9);
            \draw[line width=.6pt,red] (90:0.8) arc (90:210:0.8);
            \draw[line width=.6pt,teal] (210:1.2) arc (210:330:1.2);
            \draw[line width=.6pt,blue] (90:1.6) arc (90:-30:1.6);
            \node[ line width=0.6pt, dashed, draw opacity=0.5] (a) at (150:1){$\scriptstyle a$};
            \node[ line width=0.6pt, dashed, draw opacity=0.5] (a) at (270:1.4){$\scriptstyle b$};
            \node[ line width=0.6pt, dashed, draw opacity=0.5] (a) at (30:1.4){$\scriptstyle c$};
            \node[ line width=0.6pt, dashed, draw opacity=0.5] (a) at (90:0.3){$\scriptstyle u$};
            \node[ line width=0.6pt, dashed, draw opacity=0.5] (a) at (0.2,1.2){$\scriptstyle v$};
            \node[ line width=0.6pt, dashed, draw opacity=0.5] (a) at (90:2.1){$\scriptstyle w$};
            \node[ line width=0.6pt, dashed, draw opacity=0.5] (a) at (-0.2,1.5){$\scriptstyle \gamma$};
            \node[ line width=0.6pt, dashed, draw opacity=0.5] (a) at (0.2,0.8){$\scriptstyle \zeta$};
            \node[ line width=0.6pt, dashed, draw opacity=0.5] (a) at (210:2.1){$\scriptstyle x$};
            \node[ line width=0.6pt, dashed, draw opacity=0.5] (a) at (-0.95,-0.3){$\scriptstyle y$};
            \node[ line width=0.6pt, dashed, draw opacity=0.5] (a) at (210:0.3){$\scriptstyle z$};
            \node[ line width=0.6pt, dashed, draw opacity=0.5] (a) at (-1.2,-0.5){$\scriptstyle \mu$};
            \node[ line width=0.6pt, dashed, draw opacity=0.5] (a) at (-0.6,-0.6){$\scriptstyle \nu$};
            \node[ line width=0.6pt, dashed, draw opacity=0.5] (a) at (330:2.1){$\scriptstyle r$};
            \node[ line width=0.6pt, dashed, draw opacity=0.5] (a) at (1.1,-0.9){$\scriptstyle s$};
            \node[ line width=0.6pt, dashed, draw opacity=0.5] (a) at (330:0.3){$\scriptstyle t$};
            \node[ line width=0.6pt, dashed, draw opacity=0.5] (a) at (1.1,-0.4){$\scriptstyle \sigma$};
            \node[ line width=0.6pt, dashed, draw opacity=0.5] (a) at (1.35,-1.05){$\scriptstyle \tau$};
        \end{tikzpicture}
    \end{aligned}\right\}\;.
\end{equation}
These diagrams can be viewed as annular configurations with three marked radial segments corresponding to the domain walls. 
More precisely, the colored radial lines correspond to simple objects in the bimodule categories:
the cyan, violet, and magenta walls are labeled by objects in 
$\Irr(\eM_1)$, $\Irr(\eM_2)$, and $\Irr(\eM_3)$, respectively. 
The annular regions between consecutive walls carry simple objects 
$a\in \Irr(\eC^{1,2})$, $b\in \Irr(\eC^{2,3})$, and $c\in \Irr(\eC^{3,1})$.
The additional labels $\gamma,\zeta,\mu,\nu,\sigma,\tau$ 
encode morphism data compatible with the bimodule structures, ensuring that the entire diagram represents a well-defined element in the comodule tube algebra.

The algebra structure is defined in the standard diagrammatic manner.
The unit is given by the sum of all diagrams in which $a,b,c$ are the unit objects, and the wall labels range over simple objects in the respective bimodule categories. 
The multiplication is defined by gluing of compatible diagrams, as illustrated below:
\begin{align}
    & \mu\left(    \begin{aligned}
        \begin{tikzpicture}[scale=0.7]
            \filldraw[black!60, fill=gray!15, dotted, even odd rule] (0,0) circle[radius=0.5] (0,0) circle[radius=1.9];
            \draw[line width=1pt,cyan] (90:0.5) -- (90:1.9);
            \draw[line width=1pt,violet] (210:0.5) -- (210:1.9);
            \draw[line width=1pt,magenta] (330:0.5) -- (330:1.9);
            \draw[line width=.6pt,red] (90:0.8) arc (90:210:0.8);
            \draw[line width=.6pt,teal] (210:1.2) arc (210:330:1.2);
            \draw[line width=.6pt,blue] (90:1.6) arc (90:-30:1.6);
            \node[ line width=0.6pt, dashed, draw opacity=0.5] (a) at (150:1){$\scriptstyle a$};
            \node[ line width=0.6pt, dashed, draw opacity=0.5] (a) at (270:1.4){$\scriptstyle b$};
            \node[ line width=0.6pt, dashed, draw opacity=0.5] (a) at (30:1.4){$\scriptstyle c$};
            \node[ line width=0.6pt, dashed, draw opacity=0.5] (a) at (90:0.3){$\scriptstyle u$};
            \node[ line width=0.6pt, dashed, draw opacity=0.5] (a) at (0.2,1.2){$\scriptstyle v$};
            \node[ line width=0.6pt, dashed, draw opacity=0.5] (a) at (90:2.1){$\scriptstyle w$};
            \node[ line width=0.6pt, dashed, draw opacity=0.5] (a) at (-0.2,1.5){$\scriptstyle \gamma$};
            \node[ line width=0.6pt, dashed, draw opacity=0.5] (a) at (0.2,0.8){$\scriptstyle \zeta$};
            \node[ line width=0.6pt, dashed, draw opacity=0.5] (a) at (210:2.1){$\scriptstyle x$};
            \node[ line width=0.6pt, dashed, draw opacity=0.5] (a) at (-0.95,-0.3){$\scriptstyle y$};
            \node[ line width=0.6pt, dashed, draw opacity=0.5] (a) at (210:0.3){$\scriptstyle z$};
            \node[ line width=0.6pt, dashed, draw opacity=0.5] (a) at (-1.2,-0.5){$\scriptstyle \mu$};
            \node[ line width=0.6pt, dashed, draw opacity=0.5] (a) at (-0.6,-0.6){$\scriptstyle \nu$};
            \node[ line width=0.6pt, dashed, draw opacity=0.5] (a) at (330:2.1){$\scriptstyle r$};
            \node[ line width=0.6pt, dashed, draw opacity=0.5] (a) at (1.1,-0.9){$\scriptstyle s$};
            \node[ line width=0.6pt, dashed, draw opacity=0.5] (a) at (330:0.3){$\scriptstyle t$};
            \node[ line width=0.6pt, dashed, draw opacity=0.5] (a) at (1.1,-0.4){$\scriptstyle \sigma$};
            \node[ line width=0.6pt, dashed, draw opacity=0.5] (a) at (1.35,-1.05){$\scriptstyle \tau$};
        \end{tikzpicture}
    \end{aligned}
    \otimes 
        \begin{aligned}
        \begin{tikzpicture}[scale=0.7]
            \filldraw[black!60, fill=gray!15, dotted, even odd rule] (0,0) circle[radius=0.5] (0,0) circle[radius=1.9];
            \draw[line width=1pt,cyan] (90:0.5) -- (90:1.9);
            \draw[line width=1pt,violet] (210:0.5) -- (210:1.9);
            \draw[line width=1pt,magenta] (330:0.5) -- (330:1.9);
            \draw[line width=.6pt,red] (90:0.8) arc (90:210:0.8);
            \draw[line width=.6pt,teal] (210:1.2) arc (210:330:1.2);
            \draw[line width=.6pt,blue] (90:1.6) arc (90:-30:1.6);
            \node[ line width=0.6pt, dashed, draw opacity=0.5] (a) at (150:1){$\scriptstyle a'$};
            \node[ line width=0.6pt, dashed, draw opacity=0.5] (a) at (270:1.4){$\scriptstyle b'$};
            \node[ line width=0.6pt, dashed, draw opacity=0.5] (a) at (30:1.33){$\scriptstyle c'$};
            \node[ line width=0.6pt, dashed, draw opacity=0.5] (a) at (90:0.3){$\scriptstyle u'$};
            \node[ line width=0.6pt, dashed, draw opacity=0.5] (a) at (0.3,1.2){$\scriptstyle v'$};
            \node[ line width=0.6pt, dashed, draw opacity=0.5] (a) at (90:2.1){$\scriptstyle w'$};
            \node[ line width=0.6pt, dashed, draw opacity=0.5] (a) at (-0.25,1.5){$\scriptstyle \gamma'$};
            \node[ line width=0.6pt, dashed, draw opacity=0.5] (a) at (0.3,0.8){$\scriptstyle \zeta'$};
            \node[ line width=0.6pt, dashed, draw opacity=0.5] (a) at (210:2.1){$\scriptstyle x'$};
            \node[ line width=0.6pt, dashed, draw opacity=0.5] (a) at (-0.95,-0.3){$\scriptstyle y'$};
            \node[ line width=0.6pt, dashed, draw opacity=0.5] (a) at (210:0.2){$\scriptstyle z'$};
            \node[ line width=0.6pt, dashed, draw opacity=0.5] (a) at (-1.3,-0.5){$\scriptstyle \mu'$};
            \node[ line width=0.6pt, dashed, draw opacity=0.5] (a) at (-0.6,-0.6){$\scriptstyle \nu'$};
            \node[ line width=0.6pt, dashed, draw opacity=0.5] (a) at (330:2.1){$\scriptstyle r'$};
            \node[ line width=0.6pt, dashed, draw opacity=0.5] (a) at (1.1,-0.9){$\scriptstyle s'$};
            \node[ line width=0.6pt, dashed, draw opacity=0.5] (a) at (330:0.3){$\scriptstyle t'$};
            \node[ line width=0.6pt, dashed, draw opacity=0.5] (a) at (1.1,-0.3){$\scriptstyle \sigma'$};
            \node[ line width=0.6pt, dashed, draw opacity=0.5] (a) at (1.35,-1.05){$\scriptstyle \tau'$};
        \end{tikzpicture}
    \end{aligned}\right) \nonumber \\
    =&\;\, \delta_{u,w'}\delta_{z,x'}\delta_{t,r'}     \;\;
    \begin{aligned}
        \begin{tikzpicture}[scale=0.7]
            \filldraw[black!60, fill=gray!15, dotted, even odd rule] (0,0) circle[radius=0.5] (0,0) circle[radius=3.1];
            \draw[line width=1pt,cyan] (90:0.5) -- (90:3.1);
            \draw[line width=1pt,violet] (210:0.5) -- (210:3.1);
            \draw[line width=1pt,magenta] (330:0.5) -- (330:3.1);
            \draw[line width=.6pt,red] (90:0.8) arc (90:210:0.8);
            \draw[line width=.6pt,teal] (210:1.2) arc (210:330:1.2);
            \draw[line width=.6pt,blue] (90:1.6) arc (90:-30:1.6);
            \draw[line width=.6pt,red] (90:2.0) arc (90:210:2.0);
            \draw[line width=.6pt,teal] (210:2.4) arc (210:330:2.4);
            \draw[line width=.6pt,blue] (90:2.8) arc (90:-30:2.8);
            \node[ line width=0.6pt, dashed, draw opacity=0.5] (a) at (150:1){$\scriptstyle a'$};
            \node[ line width=0.6pt, dashed, draw opacity=0.5] (a) at (270:1.4){$\scriptstyle b'$};
            \node[ line width=0.6pt, dashed, draw opacity=0.5] (a) at (30:1.33){$\scriptstyle c'$};
            \node[ line width=0.6pt, dashed, draw opacity=0.5] (a) at (90:0.3){$\scriptstyle u'$};
            \node[ line width=0.6pt, dashed, draw opacity=0.5] (a) at (0.3,1.2){$\scriptstyle v'$};
            \node[ line width=0.6pt, dashed, draw opacity=0.5] (a) at (-0.25,1.5){$\scriptstyle \gamma'$};
            \node[ line width=0.6pt, dashed, draw opacity=0.5] (a) at (0.3,0.8){$\scriptstyle \zeta'$};
            \node[ line width=0.6pt, dashed, draw opacity=0.5] (a) at (-0.95,-0.3){$\scriptstyle y'$};
            \node[ line width=0.6pt, dashed, draw opacity=0.5] (a) at (210:0.2){$\scriptstyle z'$};
            \node[ line width=0.6pt, dashed, draw opacity=0.5] (a) at (-1.3,-0.5){$\scriptstyle \mu'$};
            \node[ line width=0.6pt, dashed, draw opacity=0.5] (a) at (-0.6,-0.6){$\scriptstyle \nu'$};
            \node[ line width=0.6pt, dashed, draw opacity=0.5] (a) at (1.1,-0.9){$\scriptstyle s'$};
            \node[ line width=0.6pt, dashed, draw opacity=0.5] (a) at (330:0.3){$\scriptstyle t'$};
            \node[ line width=0.6pt, dashed, draw opacity=0.5] (a) at (1.1,-0.3){$\scriptstyle \sigma'$};
            \node[ line width=0.6pt, dashed, draw opacity=0.5] (a) at (1.35,-1.05){$\scriptstyle \tau'$};
            \node[ line width=0.6pt, dashed, draw opacity=0.5] (a) at (150:2.2){$\scriptstyle a$};
            \node[ line width=0.6pt, dashed, draw opacity=0.5] (a) at (270:2.6){$\scriptstyle b$};
            \node[ line width=0.6pt, dashed, draw opacity=0.5] (a) at (30:2.6){$\scriptstyle c$};
            \node[ line width=0.6pt, dashed, draw opacity=0.5] (a) at (0.2,1.8){$\scriptstyle u$};
            \node[ line width=0.6pt, dashed, draw opacity=0.5] (a) at (-0.2,2.4){$\scriptstyle v$};
            \node[ line width=0.6pt, dashed, draw opacity=0.5] (a) at (90:3.3){$\scriptstyle w$};
            \node[ line width=0.6pt, dashed, draw opacity=0.5] (a) at (-0.2,2.8){$\scriptstyle \gamma$};
            \node[ line width=0.6pt, dashed, draw opacity=0.5] (a) at (0.2,2.2){$\scriptstyle \zeta$};
            \node[ line width=0.6pt, dashed, draw opacity=0.5] (a) at (210:3.3){$\scriptstyle x$};
            \node[ line width=0.6pt, dashed, draw opacity=0.5] (a) at (-2,-0.9){$\scriptstyle y$};
            \node[ line width=0.6pt, dashed, draw opacity=0.5] (a) at (-1.3,-1){$\scriptstyle z$};
            \node[ line width=0.6pt, dashed, draw opacity=0.5] (a) at (-2.26,-1.05){$\scriptstyle \mu$};
            \node[ line width=0.6pt, dashed, draw opacity=0.5] (a) at (-1.6,-1.16){$\scriptstyle \nu$};
            \node[ line width=0.6pt, dashed, draw opacity=0.5] (a) at (330:3.3){$\scriptstyle r$};
            \node[ line width=0.6pt, dashed, draw opacity=0.5] (a) at (2.16,-1.48){$\scriptstyle s$};
            \node[ line width=0.6pt, dashed, draw opacity=0.5] (a) at (1.8,-0.8){$\scriptstyle t$};
            \node[ line width=0.6pt, dashed, draw opacity=0.5] (a) at (2.16,-0.99){$\scriptstyle \sigma$};
            \node[ line width=0.6pt, dashed, draw opacity=0.5] (a) at (2.4,-1.6){$\scriptstyle \tau$};
        \end{tikzpicture}
    \end{aligned}\;.
\end{align}
It is clear that $\mathbf{Tube}(\eM_1;\eM_2;\eM_3)$ is an algebra with the above structure morphisms.

There are also comodule structures of the tri-defect tube algebra.
Consider the fusion of a domain wall excitation $\mathfrak{f}$ with a tri-defect $\tau$. 
Intuitively, this fusion can be characterized via a natural action of the tri-defect comodule tube algebra 
$\mathbf{Tube}(\eM_1;\eM_2;\eM_3)$, arising from its coaction structure relative to the corresponding domain wall tube algebra. 
This action admits a diagrammatic realization as follows:
\begin{align}
    \begin{aligned}
        \begin{tikzpicture}[scale=1]
            \filldraw[black!60, fill=gray!15, dotted] (0,0) circle[radius=1.9];
            \draw[line width=1pt,cyan] (90:0.12) -- (90:1.9);
            \draw[line width=1pt,violet] (210:0.12) -- (210:1.9);
            \draw[line width=1pt,magenta] (330:0.12) -- (330:0.48);
            \draw[line width=1pt,magenta] (330:0.72) -- (330:1.9);
            \draw[line width=.6pt,red] (90:0.8) arc (90:210:0.8);
            \draw[line width=.6pt,teal] (210:1.2) arc (210:330:1.2);
            \draw[line width=.6pt,blue] (90:1.6) arc (90:-30:1.6);
            \draw[line width=1.2pt,black] (0,0) circle (0.12); 
            \fill[pattern=north east lines, pattern color=gray] (0,0) circle (0.12);
            \draw[line width=1.2pt,black] (330:0.6) circle (0.12); 
            \fill[pattern=north east lines, pattern color=gray] (330:0.6) circle (0.12);
            \node[ line width=0.6pt, dashed, draw opacity=0.5] (a) at (150:1){$\scriptstyle a$};
            \node[ line width=0.6pt, dashed, draw opacity=0.5] (a) at (270:1.4){$\scriptstyle b$};
            \node[ line width=0.6pt, dashed, draw opacity=0.5] (a) at (30:1.4){$\scriptstyle c$};
            \node[ line width=0.6pt, dashed, draw opacity=0.5] (a) at (0.2,-0.36){$\small i$};
            \node[ line width=0.6pt, dashed, draw opacity=0.5] (a) at (0.5,-0.67){$\small \mathfrak{f}$};
            \node[ line width=0.6pt, dashed, draw opacity=0.5] (a) at (0.3,0.1){$\small \tau$};
        \end{tikzpicture}
    \end{aligned} \; = \;  \sum_{j,\sigma} \sqrt{ \frac{d_j}{d_ad_i}}\sum_{k,\rho}\sqrt{ \frac{d_k}{d_jd_b}} \;\;  \begin{aligned}
        \begin{tikzpicture}[scale=0.8]
            \filldraw[black!60, fill=gray!15, dotted] (0,0) circle[radius=1.9];
            \draw[line width=1pt,cyan] (90:0.15) -- (90:1.9);
            \draw[line width=1pt,violet] (210:0.15) -- (210:1.9);
            \draw[line width=1pt,magenta] (330:0.15) -- (330:1.9);
            \draw[line width=.6pt,red] (90:0.8) arc (90:210:0.8);
            \draw[line width=.6pt,teal] (210:1.2) arc (210:330:1.2);
            \draw[line width=.6pt,blue] (90:1.6) arc (90:-30:1.6);
            \draw[line width=1.2pt,black] (0,0) circle (0.15); 
            \fill[pattern=north east lines, pattern color=gray] (0,0) circle (0.15);
            \begin{scope}[shift={(330:3.8)}]
                \filldraw[black!60, fill=gray!15, dotted] (0,0) circle[radius=1.9];
                \draw[line width=1pt,magenta] (330:0.15) -- (330:1.9);
                \draw[line width=1pt,magenta] (150:0.15) -- (150:1.9);
                \draw[line width=.6pt,teal] (150:1.2) arc (150:330:1.2);
                \draw[line width=.6pt,blue] (150:1.6) arc (150:-30:1.6);
                \draw[line width=1.2pt,black] (0,0) circle (0.15); 
                \fill[pattern=north east lines, pattern color=gray] (0,0) circle (0.15);
                \node[ line width=0.6pt, dashed, draw opacity=0.5] (a) at (0,-0.45){$\small \mathfrak{f}$};
                \node[ line width=0.6pt, dashed, draw opacity=0.5] (a) at (240:1.4){$\scriptstyle b$};
                \node[ line width=0.6pt, dashed, draw opacity=0.5] (a) at (60:1.4){$\scriptstyle c$};
            \end{scope}
            \node[ line width=0.6pt, dashed, draw opacity=0.5] (a) at (0.35,0.15){$\small \tau$};
            \node[ line width=0.6pt, dashed, draw opacity=0.5] (a) at (150:1){$\scriptstyle a$};
            \node[ line width=0.6pt, dashed, draw opacity=0.5] (a) at (270:1.4){$\scriptstyle b$};
            \node[ line width=0.6pt, dashed, draw opacity=0.5] (a) at (30:1.4){$\scriptstyle c$};
            \node[ line width=0.6pt, dashed, draw opacity=0.5] (a) at (1.55,-1.2){$\small k$};
            \node[ line width=0.6pt, dashed, draw opacity=0.5] (a) at (1.35,-0.4){$\small j$};
            \node[ line width=0.6pt, dashed, draw opacity=0.5] (a) at (2.25,-0.95){$\small j$};
            \node[ line width=0.6pt, dashed, draw opacity=0.5] (a) at (0.4,-0.55){$\small i$};
            \node[ line width=0.6pt, dashed, draw opacity=0.5] (a) at (2.6,-1.8){$\small i$};
            \node[ line width=0.6pt, dashed, draw opacity=0.5] (a) at (1.2,-1){$\small \rho$};
            \node[ line width=0.6pt, dashed, draw opacity=0.5] (a) at (1.9,-1.4){$\small \rho$};
            \node[ line width=0.6pt, dashed, draw opacity=0.5] (a) at (1.1,-0.3){$\small \sigma$};
            \node[ line width=0.6pt, dashed, draw opacity=0.5] (a) at (2.6,-1.2){$\small \sigma$};
        \end{tikzpicture}
    \end{aligned}
\end{align}
which inspires the definition of the right coaction 
\begin{equation}
    \rho_3:\mathbf{Tube}(\eM_1;\eM_2;\eM_3)\to \mathbf{Tube}(\eM_1;\eM_2;\eM_3) \otimes \mathbf{Tube}(\eM_3)
\end{equation}
as defined by 
\begin{equation}
    \rho_3\left(\begin{aligned}
        \begin{tikzpicture}[scale=0.7]
            \filldraw[black!60, fill=gray!15, dotted, even odd rule] (0,0) circle[radius=0.5] (0,0) circle[radius=1.9];
            \draw[line width=1pt,cyan] (90:0.5) -- (90:1.9);
            \draw[line width=1pt,violet] (210:0.5) -- (210:1.9);
            \draw[line width=1pt,magenta] (330:0.5) -- (330:1.9);
            \draw[line width=.6pt,red] (90:0.8) arc (90:210:0.8);
            \draw[line width=.6pt,teal] (210:1.2) arc (210:330:1.2);
            \draw[line width=.6pt,blue] (90:1.6) arc (90:-30:1.6);
            \node[ line width=0.6pt, dashed, draw opacity=0.5] (a) at (150:1){$\scriptstyle a$};
            \node[ line width=0.6pt, dashed, draw opacity=0.5] (a) at (270:1.4){$\scriptstyle b$};
            \node[ line width=0.6pt, dashed, draw opacity=0.5] (a) at (30:1.4){$\scriptstyle c$};
            \node[ line width=0.6pt, dashed, draw opacity=0.5] (a) at (90:0.3){$\scriptstyle u$};
            \node[ line width=0.6pt, dashed, draw opacity=0.5] (a) at (0.2,1.2){$\scriptstyle v$};
            \node[ line width=0.6pt, dashed, draw opacity=0.5] (a) at (90:2.1){$\scriptstyle w$};
            \node[ line width=0.6pt, dashed, draw opacity=0.5] (a) at (-0.2,1.5){$\scriptstyle \gamma$};
            \node[ line width=0.6pt, dashed, draw opacity=0.5] (a) at (0.2,0.8){$\scriptstyle \zeta$};
            \node[ line width=0.6pt, dashed, draw opacity=0.5] (a) at (210:2.1){$\scriptstyle x$};
            \node[ line width=0.6pt, dashed, draw opacity=0.5] (a) at (-0.95,-0.3){$\scriptstyle y$};
            \node[ line width=0.6pt, dashed, draw opacity=0.5] (a) at (210:0.3){$\scriptstyle z$};
            \node[ line width=0.6pt, dashed, draw opacity=0.5] (a) at (-1.2,-0.5){$\scriptstyle \mu$};
            \node[ line width=0.6pt, dashed, draw opacity=0.5] (a) at (-0.6,-0.6){$\scriptstyle \nu$};
            \node[ line width=0.6pt, dashed, draw opacity=0.5] (a) at (330:2.1){$\scriptstyle r$};
            \node[ line width=0.6pt, dashed, draw opacity=0.5] (a) at (1.1,-0.9){$\scriptstyle s$};
            \node[ line width=0.6pt, dashed, draw opacity=0.5] (a) at (330:0.3){$\scriptstyle t$};
            \node[ line width=0.6pt, dashed, draw opacity=0.5] (a) at (1.1,-0.3){$\scriptstyle \theta$};
            \node[ line width=0.6pt, dashed, draw opacity=0.5] (a) at (1.35,-1.05){$\scriptstyle \tau$};
        \end{tikzpicture}
    \end{aligned}\right) \; = \;  \sum_{i,j,k,\sigma,\rho} \sqrt{ \frac{d_k}{d_ad_id_b}} \;\; \begin{aligned}
        \begin{tikzpicture}[scale=0.7]
            \filldraw[black!60, fill=gray!15, dotted, even odd rule] (0,0) circle[radius=0.5] (0,0) circle[radius=1.9];
            \draw[line width=1pt,cyan] (90:0.5) -- (90:1.9);
            \draw[line width=1pt,violet] (210:0.5) -- (210:1.9);
            \draw[line width=1pt,magenta] (330:0.5) -- (330:1.9);
            \draw[line width=.6pt,red] (90:0.8) arc (90:210:0.8);
            \draw[line width=.6pt,teal] (210:1.2) arc (210:330:1.2);
            \draw[line width=.6pt,blue] (90:1.6) arc (90:-30:1.6);
            \node[ line width=0.6pt, dashed, draw opacity=0.5] (a) at (150:1){$\scriptstyle a$};
            \node[ line width=0.6pt, dashed, draw opacity=0.5] (a) at (270:1.4){$\scriptstyle b$};
            \node[ line width=0.6pt, dashed, draw opacity=0.5] (a) at (30:1.4){$\scriptstyle c$};
            \node[ line width=0.6pt, dashed, draw opacity=0.5] (a) at (90:0.3){$\scriptstyle u$};
            \node[ line width=0.6pt, dashed, draw opacity=0.5] (a) at (0.2,1.2){$\scriptstyle v$};
            \node[ line width=0.6pt, dashed, draw opacity=0.5] (a) at (90:2.1){$\scriptstyle w$};
            \node[ line width=0.6pt, dashed, draw opacity=0.5] (a) at (-0.2,1.5){$\scriptstyle \gamma$};
            \node[ line width=0.6pt, dashed, draw opacity=0.5] (a) at (0.2,0.8){$\scriptstyle \zeta$};
            \node[ line width=0.6pt, dashed, draw opacity=0.5] (a) at (210:2.1){$\scriptstyle x$};
            \node[ line width=0.6pt, dashed, draw opacity=0.5] (a) at (-0.95,-0.3){$\scriptstyle y$};
            \node[ line width=0.6pt, dashed, draw opacity=0.5] (a) at (210:0.3){$\scriptstyle z$};
            \node[ line width=0.6pt, dashed, draw opacity=0.5] (a) at (-1.2,-0.5){$\scriptstyle \mu$};
            \node[ line width=0.6pt, dashed, draw opacity=0.5] (a) at (-0.6,-0.6){$\scriptstyle \nu$};
            \node[ line width=0.6pt, dashed, draw opacity=0.5] (a) at (330:2.1){$\scriptstyle k$};
            \node[ line width=0.6pt, dashed, draw opacity=0.5] (a) at (1.1,-0.9){$\scriptstyle j$};
            \node[ line width=0.6pt, dashed, draw opacity=0.5] (a) at (330:0.3){$\scriptstyle i$};
            \node[ line width=0.6pt, dashed, draw opacity=0.5] (a) at (1.1,-0.4){$\scriptstyle \sigma$};
            \node[ line width=0.6pt, dashed, draw opacity=0.5] (a) at (1.35,-1.05){$\scriptstyle \rho$};
        \end{tikzpicture} \end{aligned}
        \otimes 
        \begin{aligned}
        \begin{tikzpicture}[scale=0.7]
            \filldraw[black!60, fill=gray!15, dotted, even odd rule] (0,0) circle[radius=0.5] (0,0) circle[radius=1.9];
            \draw[line width=1pt,magenta] (150:0.5) -- (150:1.9);
            \draw[line width=1pt,magenta] (330:0.5) -- (330:1.9);
            \draw[line width=.6pt,teal] (150:1.2) arc (150:330:1.2);
            \draw[line width=.6pt,blue] (150:1.6) arc (150:-30:1.6);
            \node[ line width=0.6pt, dashed, draw opacity=0.5] (a) at (240:1.4){$\scriptstyle b$};
            \node[ line width=0.6pt, dashed, draw opacity=0.5] (a) at (60:1.4){$\scriptstyle c$};
            \node[ line width=0.6pt, dashed, draw opacity=0.5] (a) at (150:2.1){$\scriptstyle k$};
            \node[ line width=0.6pt, dashed, draw opacity=0.5] (a) at (-1.33,0.45){$\scriptstyle j$};
            \node[ line width=0.6pt, dashed, draw opacity=0.5] (a) at (150:0.3){$\scriptstyle i$};
            \node[ line width=0.6pt, dashed, draw opacity=0.5] (a) at (-0.9,0.8){$\scriptstyle \sigma$};
            \node[ line width=0.6pt, dashed, draw opacity=0.5] (a) at (-1.57,0.6){$\scriptstyle \rho$};
            \node[ line width=0.6pt, dashed, draw opacity=0.5] (a) at (330:2.1){$\scriptstyle r$};
            \node[ line width=0.6pt, dashed, draw opacity=0.5] (a) at (1.1,-0.9){$\scriptstyle s$};
            \node[ line width=0.6pt, dashed, draw opacity=0.5] (a) at (330:0.3){$\scriptstyle t$};
            \node[ line width=0.6pt, dashed, draw opacity=0.5] (a) at (1.1,-0.3){$\scriptstyle \theta$};
            \node[ line width=0.6pt, dashed, draw opacity=0.5] (a) at (1.35,-1.05){$\scriptstyle \tau$};
        \end{tikzpicture}
    \end{aligned}\;. \label{eq:right-tri-comodule}
\end{equation}
Analogously, the same constructions give rise to the following left and right coactions
\begin{align}
    \beta_1:&\mathbf{Tube}(\eM_1;\eM_2;\eM_3)\to \mathbf{Tube}(\eM_1)\otimes \mathbf{Tube}(\eM_1;\eM_2;\eM_3), \\
    \rho_2:&\mathbf{Tube}(\eM_1;\eM_2;\eM_3)\to \mathbf{Tube}(\eM_1;\eM_2;\eM_3) \otimes \mathbf{Tube}(\eM_2). 
\end{align}
The left or right construction is determined by the orientation of the corresponding domain wall. 

To summarize, we have the following result.
\begin{proposition}
    The morphisms $\beta_1,\rho_2,\rho_3$ endow $\mathbf{Tube}(\eM_1;\eM_2;\eM_3)$ with the structure of a $\mathbf{Tube}(\eM_1)|\{\mathbf{Tube}(\eM_2),\mathbf{Tube}(\eM_3)\}$ tricomodule algebra. The codimension-2 tri-defects can be characterized by the representations of $\mathbf{Tube}(\eM_1)|\{\mathbf{Tube}(\eM_2),\mathbf{Tube}(\eM_3)\}$. The fusion of codimension-2 tri-defects and domain wall excitations are captured by the coaction of $\mathbf{Tube}(\eM_1)|\{\mathbf{Tube}(\eM_2),\mathbf{Tube}(\eM_3)\}$.
\end{proposition}


\subsection{Codimension-2 $N$-defect comodule tube algebra}

\begin{figure}[t]
    \centering
    \includegraphics[width=0.8\textwidth]{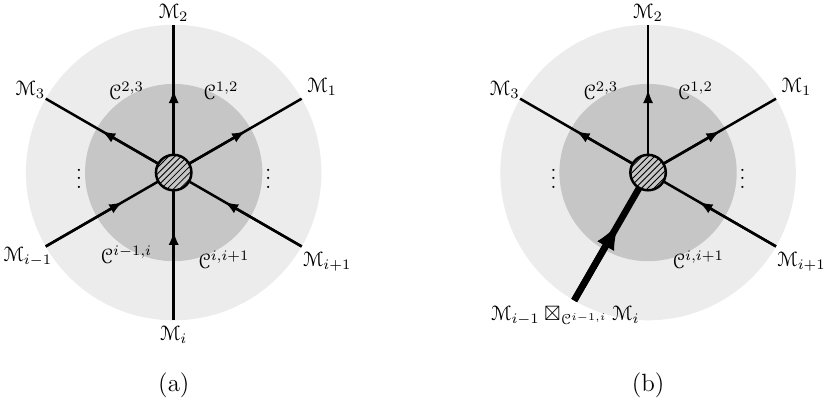} 
    \caption{(a) Schematic representation of an $N$-defect. 
    The central defect is attached to $N$ domain walls labeled by bimodule categories $\eM_i$, arranged cyclically around the defect. 
    The region between two adjacent domain walls $\eM_i$ and $\eM_{i+1}$ is labeled by the fusion category $\eC^{i,i+1}$, with indices understood modulo $N$. 
    (b) Fusion of two adjacent domain walls $\eM_{i-1}$ and $\eM_i$. 
    Here $\eM_{i-1}$ carries a right $\eC^{i-1,i}$-module structure and $\eM_i$ carries a left $\eC^{i-1,i}$-module structure, so the two domain walls can be fused along the intermediate fusion category $\eC^{i-1,i}$ via Deligne tensor product, producing a $\eC^{i-2,i-1}|\eC^{i,i+1}$ domain wall $\eM_{i-1}\boxtimes_{\eC^{i-1,i}}\eM_i$.  }
    \label{fig:N_defect}
\end{figure}

The construction of tri-defect comodule tube algebra can be directly generalized to the general case for $N$-defect. Specifically, we consider the following algebra 
\begin{equation} \label{eq:N-tube-basis}
  \mathbf{Tube}(\{\eM_i\}_{1\leq i\leq N})  :=\operatorname{span}\left\{  \begin{aligned}
        \begin{tikzpicture}[scale=0.7]
            \filldraw[black!60, fill=gray!15, dotted, even odd rule] (0,0) circle[radius=0.5] (0,0) circle[radius=1.9];
            \draw[line width=1pt,violet] (30:0.5) -- (30:1.9);
            \draw[line width=1pt,cyan] (90:0.5) -- (90:1.9);
            \draw[line width=1pt,olive] (150:0.5) -- (150:1.9);
            \draw[line width=1pt,lime!50!black] (210:0.5) -- (210:1.9);
            \draw[line width=1pt,brown] (270:0.5) -- (270:1.9);
            \draw[line width=1pt,magenta] (330:0.5) -- (330:1.9);
            \draw[line width=.6pt,red] (90:0.8) arc (90:150:0.8);
            \draw[line width=.6pt,teal] (30:1.5) arc (30:90:1.5);
            \draw[line width=.6pt,purple] (210:0.9) arc (210:270:0.9);
            \draw[line width=.6pt,blue] (270:1.5) arc (270:330:1.5);
            \node[ line width=1.2pt, dashed, draw opacity=0.5] (a) at (175:1.2){$ \vdots$}; 
            \node[ line width=1.2pt, dashed, draw opacity=0.5] (a) at (5:1.3){$ \vdots$};
        \end{tikzpicture}
    \end{aligned}\right\}.
\end{equation}
In these tube diagrams, the radial walls are labeled by simple objects in the corresponding bimodule categories, while the annular arcs are labeled by simple objects in the corresponding bulk fusion categories. 
The vertices are decorated by bimodule morphisms compatible with the adjacent fusion and module structures.

Analogously to the tri-defect comodule tube algebra, $\mathbf{Tube}(\{\eM_i\}_{1\leq i\leq N})$ carries a natural algebra structure defined diagrammatically. 
The unit is given by the sum of all diagrams in which all arcs are the unit objects, and the wall labels range over simple objects in the respective bimodule categories, and the multiplication is given by gluing compatible diagrams along matching domain walls.  
Furthermore, for each $1\leq i \leq k_0$, $\mathbf{Tube}(\{\eM_i\}_{1\leq i\leq N})$ possesses a natural left comodule structure over the domain wall tube algebra $\mathbf{Tube}(\eM_i)$, whereas for each $k_0+1\leq j\leq N$, it admits a natural right comodule structure over $\mathbf{Tube}(\eM_j)$. The left and right comodule structure are canonically associated with the chosen orientation convention, and are realized by the morphisms
\begin{align}
    \beta_i&:\mathbf{Tube}(\{\eM_i\}_{1\leq i\leq N}) \to \mathbf{Tube}(\eM_i)\otimes \mathbf{Tube}(\{\eM_i\}_{1\leq i\leq N}) , \\
    \rho_j&:\mathbf{Tube}(\{\eM_i\}_{1\leq i\leq N}) \to \mathbf{Tube}(\{\eM_i\}_{1\leq i\leq N}) \otimes \mathbf{Tube}(\eM_j),
\end{align}
which are analogous to Eq.~\eqref{eq:right-tri-comodule}. 

\begin{proposition}
    The morphisms $\{\beta_i,\rho_j\}$ endow $\mathbf{Tube}(\{\eM_i\}_{1\leq i\leq N})$ with the structure of a $\{\mathbf{Tube}(\eM_i)\}_{i=1}^{k_0}|\{\mathbf{Tube}(\eM_j)\}_{j=k_0+1}^N$-multicomodule algebra. Its representation category $\Rep(\mathbf{Tube}(\{\eM_i\}_{1\leq i\leq N}))$ is equivalent to the bimodule functor category $\Fun(\eM_{\rm in}, \eM_{\rm out})$.
\end{proposition}

One of the most important features of $\mathbf{Tube}(\{\eM_i\}_{1\leq i \leq M})$ is that it possesses a hierarchical structure, see Figure~\ref{fig:N_defect}. An $N$-defect multicomodule tube algebra can be transformed into an $(N-1)$-defect multicomodule tube algebra by fusing two nearby domain walls $\eM_{i-1}\boxtimes_{\eC^{i-1,i}}\eM_{i}$, yielding $\mathbf{Tube}(\{\eM_l\}_{1\leq l \leq i-1}; \eM_{i-1}\boxtimes_{\eC^{i-1,i}}\eM_{i};\{\eM_l\}_{i+1\leq l \leq N} )$, as illustrated in Figure~\ref{fig:N_defect}. The two resulting multicomodule tube algebras are categorically Morita equivalent, in the sense that their representation categories are equivalent. Upon iterating this procedure, we find that all such generalized tube algebras are categorically Morita equivalent to the twist defect tube algebra, and hence all describe the same type $N$-defects.

%% file: Ch-DWMultimodule.tex
\section{Defect tube algebra for domain wall connecting many bulks}
\label{sec:multi-moduletube}

In Ref.~\cite{jia2025tube}, we constructed a tube algebra for a multimodule domain wall at which $N$ bulk topological phases meet. Such a wall is described mathematically by a multimodule category equipped with compatible actions of $N$ fusion categories. For $N=1$, this construction recovers the boundary tube algebra, whereas for $N=2$ it reduces to the ordinary domain wall tube algebra. The multimodule tube algebra therefore provides a common generalization of these two familiar constructions.

The purpose of this section is to extend this framework from a single multimodule domain wall to defects between two such walls. We will
construct the corresponding multimodule defect tube algebra, establish its bicomodule algebra structure, and relate its finite-dimensional
representations to multimodule functors. We will then show that the gluing construction gives rise to a hierarchy of comodule tube algebras organized by an $N$-tuple algebra structure. 

We begin by recalling the relevant categorical notions. Let $I$ and $J$ be finite index sets, and let $\{\eC_\alpha\}_{\alpha\in I}$ and $\{\eD_\beta\}_{\beta\in J}$ be two families of fusion categories. A $\{\eC_\alpha\}_{\alpha\in I}|\{\eD_\beta\}_{\beta\in J}$-multimodule category consists of an abelian category $\eM$ equipped with left actions $\eC_\alpha\otimes\eM\to\eM$, $\alpha\in I$, and right actions $\eM\otimes\eD_\beta\to\eM$, $\beta\in J$, together with coherent associativity and unit isomorphisms. The actions associated with different fusion categories are required to commute coherently. More precisely, these compatibility data make $\eM$ simultaneously a $\eC_\alpha|\eD_\beta$-bimodule category for every $\alpha\in I$ and $\beta\in J$, a $\eC_\alpha|\eC_{\alpha'}^{\rm rev}$-bimodule category for distinct $\alpha,\alpha'\in I$, and a $\eD_\beta^{\rm rev}|\eD_{\beta'}$-bimodule category for distinct $\beta,\beta'\in J$. These pairwise bimodule structures are required to be mutually compatible. Here for a monoidal category $\eC$, its reverse monoidal category $\eC^{\rm rev}$ has the same underlying category and tensor unit as $\eC$, but its tensor product is defined by $X\otimes^{\rm rev}Y:=Y\otimes X$. The preceding data may equivalently be packaged as a $(\boxtimes _\alpha \eC_\alpha)|(\boxtimes_\beta\eD_\beta)$-bimodule category structure on $\eM$. By folding the right actions to the left, this is also equivalent to a left $(\boxtimes _\alpha \eC_\alpha)\boxtimes(\boxtimes_\beta\eD_\beta^{\rm rev})$-module category structure. Equivalently, by folding the left actions to the right, $\eM$ may be regarded as a right $(\boxtimes _\alpha \eC_\alpha^{\rm rev})\boxtimes(\boxtimes_\beta\eD_\beta)$-module category.

Physically, a multimodule category describes a one-dimensional gapped domain wall
along which several two-dimensional string-net phases meet. We first fix an orientation of the domain wall. Each adjacent oriented surface carries a string-net model with input fusion category $\eC_\alpha$ or $\eD_\beta$, according to whether its orientation induced by the right-hand rule agrees or disagrees with that of the wall. Correspondingly, the strings in these bulk regions act on $\eM$ from the left or from the right.
A lattice realization is obtained by placing the usual Levin--Wen lattice in each bulk region and labeling the edges along the common domain wall by simple objects of $\eM$. The wall vertices are labeled by morphisms in the corresponding
multimodule Hom-spaces. Away from the wall, the ordinary bulk vertex and face
projectors remain unchanged, while near the wall they are replaced by projectors defined using the multimodule actions and their compatible F-moves. These mutually commuting local projectors define an exactly solvable gapped Hamiltonian. See \cite[Section~3.2]{jia2025tube} for more details. A localized violation of the projector conditions gives an excitation or defect on the domain wall. Surrounding such a region by a tube and reducing it by topological local moves motivates the multimodule tube and defect tube algebras constructed below.

\subsection{Generalized tube algebras for multimodule domain wall}

We first recall the tube algebra associated with a single multimodule domain wall. As established in Ref.~\cite{jia2025tube}, this multimodule tube algebra admits a weak Hopf algebra structure. It may be viewed as an $N$-fold generalization of the Drinfeld quantum double construction: rather than combining two boundary tube algebras, it is assembled by gluing together the boundary tube algebras associated with all $N$ bulk regions. For this reason, we refer to it as a \emph{tube quantum $N$-tuple algebra}. Its lower order components are obtained by assigning tensor unit labels to selected bulk strings, giving a natural hierarchy of weak Hopf tube subalgebras. As we will show below, the defect tube algebra between two multimodule domain walls is naturally a bicomodule algebra over their respective multimodule tube algebras.

Suppose $|I|=k$, $|J|=l$, and $N=k+l$. 
Consider a multimodule category domain wall determined by a  $\{\eC_{\alpha}\}_{\alpha \in I}|\{\eD_{\beta}\}_{\beta \in J}$-multimodule category $\eM$. 
The associated domain wall tube algebra $\mathbf{Tube}({_{\{\eC_{\alpha}\}}} \eM_{\{\eD_{\beta}\}})$ is the vector space
\begin{equation} \label{eq:mul_tube_basis} 
    \mathbf{Tube}({_{\{\eC_{\alpha}\}}} \eM_{\{\eD_{\beta}\}}) = \operatorname{span}\left\{\begin{aligned}
        \begin{tikzpicture}[scale=0.6]
            \filldraw[black!60, fill=gray!15, dotted, even odd rule] (0,0) circle[radius=0.5] (0,0) circle[radius=2.6];
            \draw[line width=1pt,violet] (0,0.5)--(0,2.6);
            \draw[line width=1pt,violet] (0,-0.5)--(0,-2.6);
            \draw[line width=0.6pt,red] (0,0.8) arc[start angle=90, end angle=270, radius=0.8];
            \draw[line width=0.6pt,red] (0,1.5) arc[start angle=90, end angle=270, radius=1.5];
            \draw[line width=0.6pt,blue] (0,1.7) arc[start angle=90, end angle=-90, radius=1.7];
            \draw[line width=0.6pt,blue] (0,2.4) arc[start angle=90, end angle=-90, radius=2.4];
            \node[ line width=0.6pt, dashed, draw opacity=0.5] (a) at (-0.3,1.7){$\scriptstyle \zeta_1$};
            \node[ line width=0.6pt, dashed, draw opacity=0.5] (a) at (-0.3,2.4){$\scriptstyle \zeta_l$};
            \node[ line width=0.6pt, dashed, draw opacity=0.5] (a) at (-0.3,-1.75){$\scriptstyle \mu_1$};
            \node[ line width=0.6pt, dashed, draw opacity=0.5] (a) at (-0.3,-2.4){$\scriptstyle \mu_l$};
            \node[ line width=0.6pt, dashed, draw opacity=0.5] (a) at (-1.1,0){{\color{red}$\scriptstyle \cdots$}};
            \node[ line width=0.6pt, dashed, draw opacity=0.5] (a) at (1.4,0){$\scriptstyle b_1$};
            \node[ line width=0.6pt, dashed, draw opacity=0.5] (a) at (2.1,0){{\color{blue}$\scriptstyle \cdots$}};
            \node[ line width=0.6pt, dashed, draw opacity=0.5] (a) at (2.7,0){$\scriptstyle b_l$};
            \node[ line width=0.6pt, dashed, draw opacity=0.5] (a) at (-.4,0){$\scriptstyle a_1$};
            \node[ line width=0.6pt, dashed, draw opacity=0.5] (a) at (-1.8,0){$\scriptstyle a_k$};
            \node[ line width=0.6pt, dashed, draw opacity=0.5] (a) at (0.3,-1.5){$\scriptstyle \nu_k$};
            \node[ line width=0.6pt, dashed, draw opacity=0.5] (a) at (0.3,-0.8){$\scriptstyle \nu_1$};
            \node[ line width=0.6pt, dashed, draw opacity=0.5] (a) at (0,-0.3){$\scriptstyle z_1$};
            \node[ line width=0.6pt, dashed, draw opacity=0.5] (a) at (0,-2.8){$\scriptstyle w_{l+1}$};
            \node[ line width=0.6pt, dashed, draw opacity=0.5] (a) at (0,0.3){$\scriptstyle u_1$};
            \node[ line width=0.6pt, dashed, draw opacity=0.5] (a) at (0,2.8){$\scriptstyle v_{l+1}$};
            \node[ line width=0.6pt, dashed, draw opacity=0.5] (a) at (0.3,1.4){$\scriptstyle \gamma_k$};
            \node[ line width=0.6pt, dashed, draw opacity=0.5] (a) at (0.3,0.8){$\scriptstyle \gamma_1$};
        \end{tikzpicture}
    \end{aligned}\right\}
\end{equation}
where $a_i\in \Irr(\eC_i),b_j\in \Irr(\eD_j)$, $u_i,z_i,v_j,w_j \in \Irr(\eM)$.
In the upper half, the wall labels are ordered as $u_1,\cdots,u_k$, $v_1,\cdots,v_{l+1}$ from bottom to top, while in the lower half, the wall labels are ordered as $z_1,\cdots,z_k$, $w_1,\cdots,w_{l+1}$ from top to bottom. The wall vertex labels are from the Hom-spaces listed as follows: 
\begin{gather*}
    \gamma_i\in \Hom_\eM(a_i\otimes u_i,u_{i+1}),\;\;\nu_i\in \Hom_\eM(z_{i+1},a_i\otimes z_i), \;\;1\leq i \leq k-1, \\
    \zeta_j\in \Hom_\eM(v_j\otimes b_j,v_{j+1}),\;\; \mu_j\in \Hom_\eM(w_{j+1},w_j\otimes b_j),\;\; 1\leq j\leq l, \\
    \gamma_k\in \Hom_\eM(a_k\otimes u_k,v_1),\;\; \nu_k\in \Hom_\eM(w_1,a_k\otimes z_k). 
\end{gather*}
One implicit condition is that these Hom-spaces are required to be nonzero.

In \cite{jia2025tube}, we showed that $\mathbf{Tube}({_{\{\eC_{\alpha}\}}} \eM_{\{\eD_{\beta}\}} )$ can be equipped with a weak Hopf algebra structure\,\footnote{Note that the $*$-structure is constructed based on the assumption that $d_m^L=d_m^R=d_m$ for all $m\in \eM$ (e.g. the pseudo-unitary case).}. 
Furthermore, it can be seen that this algebra naturally contains many weak Hopf tube subalgebras.
More precisely, by setting certain bulk strings to the tensor unit — i.e., replacing the corresponding simple objects \(a_i \in \Irr(\eC_i)\) or \(b_j \in \Irr(\eD_j)\) with the unit object (drawn as a dotted string) — the tube algebra reduces to a weak Hopf subalgebra. 
For example, omitting the fusion categories \(\eC_\alpha\) and \(\eD_\beta\) in this way yields a weak Hopf algebra embedding 
\[
\mathbf{Tube}\big({}_{\eC_1 \boxtimes \cdots \boxtimes \hat{\eC}_\alpha \boxtimes \cdots \boxtimes \eC_{k}} 
\eM_{\eD_1 \boxtimes \cdots \boxtimes \hat{\eD}_\beta \boxtimes \cdots \boxtimes \eD_{l}}\big) \xhookrightarrow{\hspace{0.4cm}} \mathbf{Tube}({_{\{\eC_{\alpha}\}}} \eM_{\{\eD_{\beta}\}} ),
\]
where the hat \(\hat{\bullet}\) indicates removal of the corresponding category. 
This observation reveals a hierarchical structure: when \(N = |I| + |J|\), the full algebra can be interpreted as a \emph{tube quantum \(N\)-tuple algebra}. 
Every quantum \(k\)-tuple algebra (\(k \le N\)) obtained from the same multimodule data appears as a weak Hopf subalgebra of this \(N\)-tuple algebra. 
Using the folding trick, there is a correspondence of the multimodule tube algebra and boundary tube algebra 
\[
    \mathbf{Tube}({_{\{\eC_{\alpha}\}}} \eM_{\{\eD_{\beta}\}} ) \; \longleftrightarrow \; \mathbf{Tube}({_{(\boxtimes_{\alpha} \eC_{\alpha}) \boxtimes (\boxtimes_{\beta} \eD_{\beta}^{\rm rev})}} \eM). 
\]
The folding process is by simply folding all domain wall tubes appearing in the defining data of the weak Hopf domain wall tube algebra into a single boundary tube. We also studied its Morita theory: all multimodule tube algebras constructed are Morita equivalent, namely, their representation categories are equivalent to each other.

We now consider defects between two multimodule domain walls described
by multimodule categories $\eM$ and $\eN$ over the same families $\{\eC_\alpha\}_{\alpha\in I}$ and $\{\eD_\beta\}_{\beta\in J}$. The corresponding \emph{multimodule defect tube algebra} is the following vector space 
\begin{equation} \label{eq:mul_defect_tube_basis}  
   \mathbf{Tube}({_{\{\eC_{\alpha}\}}} \eM_{\{\eD_{\beta}\}};{_{\{\eC_{\alpha}\}}} \eN_{\{\eD_{\beta}\}}) = \operatorname{span} \left\{\begin{aligned}
        \begin{tikzpicture}[scale=0.6]
            \filldraw[black!60, fill=gray!15, dotted, even odd rule] (0,0) circle[radius=0.5] (0,0) circle[radius=2.6];
            \draw[line width=1pt,cyan] (0,0.5)--(0,2.6);
            \draw[line width=1pt,violet] (0,-0.5)--(0,-2.6);
            \draw[line width=0.6pt,red] (0,0.8) arc[start angle=90, end angle=270, radius=0.8];
            \draw[line width=0.6pt,red] (0,1.5) arc[start angle=90, end angle=270, radius=1.5];
            \draw[line width=0.6pt,blue] (0,1.7) arc[start angle=90, end angle=-90, radius=1.7];
            \draw[line width=0.6pt,blue] (0,2.4) arc[start angle=90, end angle=-90, radius=2.4];
            \node[ line width=0.6pt, dashed, draw opacity=0.5] (a) at (-0.3,1.7){$\scriptstyle \zeta_1$};
            \node[ line width=0.6pt, dashed, draw opacity=0.5] (a) at (-0.3,2.4){$\scriptstyle \zeta_l$};
            \node[ line width=0.6pt, dashed, draw opacity=0.5] (a) at (-0.3,-1.75){$\scriptstyle \mu_1$};
            \node[ line width=0.6pt, dashed, draw opacity=0.5] (a) at (-0.3,-2.4){$\scriptstyle \mu_l$};
            \node[ line width=0.6pt, dashed, draw opacity=0.5] (a) at (-1.1,0){{\color{red}$\scriptstyle \cdots$}};
            \node[ line width=0.6pt, dashed, draw opacity=0.5] (a) at (1.4,0){$\scriptstyle b_1$};
            \node[ line width=0.6pt, dashed, draw opacity=0.5] (a) at (2.1,0){{\color{blue}$\scriptstyle \cdots$}};
            \node[ line width=0.6pt, dashed, draw opacity=0.5] (a) at (2.7,0){$\scriptstyle b_l$};
            \node[ line width=0.6pt, dashed, draw opacity=0.5] (a) at (-.4,0){$\scriptstyle a_1$};
            \node[ line width=0.6pt, dashed, draw opacity=0.5] (a) at (-1.8,0){$\scriptstyle a_k$};
            \node[ line width=0.6pt, dashed, draw opacity=0.5] (a) at (0.3,-1.5){$\scriptstyle \nu_k$};
            \node[ line width=0.6pt, dashed, draw opacity=0.5] (a) at (0.3,-0.8){$\scriptstyle \nu_1$};
            \node[ line width=0.6pt, dashed, draw opacity=0.5] (a) at (0,-0.3){$\scriptstyle z_1$};
            \node[ line width=0.6pt, dashed, draw opacity=0.5] (a) at (0,-2.8){$\scriptstyle w_{l+1}$};
            \node[ line width=0.6pt, dashed, draw opacity=0.5] (a) at (0,0.3){$\scriptstyle u_1$};
             \node[ line width=0.6pt, dashed, draw opacity=0.5] (a) at (0,2.8){$\scriptstyle v_{l+1}$};
            \node[ line width=0.6pt, dashed, draw opacity=0.5] (a) at (0.3,1.4){$\scriptstyle \gamma_k$};
            \node[ line width=0.6pt, dashed, draw opacity=0.5] (a) at (0.3,0.8){$\scriptstyle \gamma_1$};
        \end{tikzpicture}
    \end{aligned}\right\}\;.
\end{equation}
Here $a_i\in \Irr(\eC_i)$, and $b_j\in \Irr(\eD_j)$. In the upper half, the wall labels $u_1,\ldots,u_k$, $v_1,\ldots,v_{l+1}$ (from bottom to top) are taken from $\Irr(\eN)$, while in the lower half, the labels $z_1,\ldots,z_k$, $w_1,\ldots,w_{l+1}$ (from top to bottom) are taken from $\Irr(\eM)$.
Therefore, the vertices in the upper half are labeled by morphisms in $\Hom_\eN$, whereas those in the lower half are labeled by morphisms in $\Hom_\eM$. 

A direct consequence is that the multimodule defect tube algebra has an algebra structure. First, the unit is given by 
\begin{equation} \label{eq:mul_defect_tube_unit}
    1=\sum_{x\in\Irr(\eM)}\sum_{z\in\Irr(\eN)} \;\begin{aligned}
        \begin{tikzpicture}[scale=0.65]
             \filldraw[black!60, fill=gray!15, dotted, even odd rule] (0,0) circle[radius=0.5] (0,0) circle[radius=2];
             \draw[line width=1pt,cyan] (0,0.5)--(0,2);
             \draw[line width=1pt,violet] (0,-0.5)--(0,-2);
             \draw[red, thick, dotted] (0,0.8) arc[start angle=90, end angle=270, radius=0.8];
             \draw[red, thick, dotted] (0,1.2) arc[start angle=90, end angle=270, radius=1.2];
             \draw[blue,thick,dotted] (0,1.4) arc[start angle=90, end angle=-90, radius=1.4];
             \draw[blue,thick,dotted] (0,1.8) arc[start angle=90, end angle=-90, radius=1.8];
            \node[ line width=0.6pt, dashed, draw opacity=0.5] (a) at (-0.2,-1.6){$\scriptstyle x$};
            \node[ line width=0.6pt, dashed, draw opacity=0.5] (a) at (-0.2,1.6){$\scriptstyle z$};
        \end{tikzpicture}
    \end{aligned}\;\;, 
\end{equation}
where dotted lines are labeled by multiple unit objects in respective base fusion categories. 
Second, the multiplication map $\mu$ is defined by gluing corresponding wall labels as 
\begin{align}
 &     \mu\left( \begin{aligned}
    \begin{tikzpicture}[scale=0.6]
    \filldraw[black!60, fill=gray!15, dotted, even odd rule] (0,0) circle[radius=0.5] (0,0) circle[radius=2.6];
         \draw[line width=1pt,cyan] (0,0.5)--(0,2.6);
         \draw[line width=1pt,violet] (0,-0.5)--(0,-2.6);
         \draw[line width=0.6pt,red] (0,0.8) arc[start angle=90, end angle=270, radius=0.8];
         \draw[line width=0.6pt,red] (0,1.5) arc[start angle=90, end angle=270, radius=1.5];
        \draw[line width=0.6pt,blue] (0,1.7) arc[start angle=90, end angle=-90, radius=1.7];
        \draw[line width=0.6pt,blue] (0,2.4) arc[start angle=90, end angle=-90, radius=2.4];
        \node[ line width=0.6pt, dashed, draw opacity=0.5] (a) at (-0.3,1.7){$\scriptstyle \zeta_1$};
         \node[ line width=0.6pt, dashed, draw opacity=0.5] (a) at (-0.3,2.4){$\scriptstyle \zeta_l$};
         \node[ line width=0.6pt, dashed, draw opacity=0.5] (a) at (-0.3,-1.75){$\scriptstyle \mu_1$};
        \node[ line width=0.6pt, dashed, draw opacity=0.5] (a) at (-0.3,-2.4){$\scriptstyle \mu_l$};
        \node[ line width=0.6pt, dashed, draw opacity=0.5] (a) at (-1.1,0){{\color{red}$\scriptstyle \cdots$}};
        \node[ line width=0.6pt, dashed, draw opacity=0.5] (a) at (1.4,0){$\scriptstyle b_1$};
         \node[ line width=0.6pt, dashed, draw opacity=0.5] (a) at (2.1,0){{\color{blue}$\scriptstyle \cdots$}};
        \node[ line width=0.6pt, dashed, draw opacity=0.5] (a) at (2.7,0){$\scriptstyle b_l$};
        \node[ line width=0.6pt, dashed, draw opacity=0.5] (a) at (-.4,0){$\scriptstyle a_1$};
        \node[ line width=0.6pt, dashed, draw opacity=0.5] (a) at (-1.8,0){$\scriptstyle a_k$};
        \node[ line width=0.6pt, dashed, draw opacity=0.5] (a) at (0.3,-1.5){$\scriptstyle \nu_k$};
        \node[ line width=0.6pt, dashed, draw opacity=0.5] (a) at (0.3,-0.8){$\scriptstyle \nu_1$};
        \node[ line width=0.6pt, dashed, draw opacity=0.5] (a) at (0,-0.3){$\scriptstyle z_1$};
        \node[ line width=0.6pt, dashed, draw opacity=0.5] (a) at (0,-2.8){$\scriptstyle w_{l+1}$};
        \node[ line width=0.6pt, dashed, draw opacity=0.5] (a) at (0,0.3){$\scriptstyle u_1$};
         \node[ line width=0.6pt, dashed, draw opacity=0.5] (a) at (0,2.8){$\scriptstyle v_{l+1}$};
        \node[ line width=0.6pt, dashed, draw opacity=0.5] (a) at (0.3,1.4){$\scriptstyle \gamma_k$};
        \node[ line width=0.6pt, dashed, draw opacity=0.5] (a) at (0.3,0.8){$\scriptstyle \gamma_1$};
    \end{tikzpicture}
\end{aligned}\otimes 
    \begin{aligned}
    \begin{tikzpicture}[scale=0.6]
    \filldraw[black!60, fill=gray!15, dotted, even odd rule] (0,0) circle[radius=0.5] (0,0) circle[radius=2.6];
         \draw[line width=1pt,cyan] (0,0.5)--(0,2.6);
         \draw[line width=1pt,violet] (0,-0.5)--(0,-2.6);
         \draw[line width=0.6pt,red] (0,0.8) arc[start angle=90, end angle=270, radius=0.8];
         \draw[line width=0.6pt,red] (0,1.5) arc[start angle=90, end angle=270, radius=1.5];
        \draw[line width=0.6pt,blue] (0,1.7) arc[start angle=90, end angle=-90, radius=1.7];
        \draw[line width=0.6pt,blue] (0,2.4) arc[start angle=90, end angle=-90, radius=2.4];
        \node[ line width=0.6pt, dashed, draw opacity=0.5] (a) at (-0.3,1.7){$\scriptstyle \zeta_1'$};
         \node[ line width=0.6pt, dashed, draw opacity=0.5] (a) at (-0.3,2.4){$\scriptstyle \zeta_l'$};
         \node[ line width=0.6pt, dashed, draw opacity=0.5] (a) at (-0.3,-1.75){$\scriptstyle \mu_1'$};
        \node[ line width=0.6pt, dashed, draw opacity=0.5] (a) at (-0.3,-2.4){$\scriptstyle \mu_l'$};
        \node[ line width=0.6pt, dashed, draw opacity=0.5] (a) at (-1.1,0){{\color{red}$\scriptstyle \cdots$}};
        \node[ line width=0.6pt, dashed, draw opacity=0.5] (a) at (1.4,0){$\scriptstyle b_1'$};
         \node[ line width=0.6pt, dashed, draw opacity=0.5] (a) at (2.1,0){{\color{blue}$\scriptstyle \cdots$}};
        \node[ line width=0.6pt, dashed, draw opacity=0.5] (a) at (2.7,0){$\scriptstyle b_l'$};
        \node[ line width=0.6pt, dashed, draw opacity=0.5] (a) at (-.4,0){$\scriptstyle a_1'$};
        \node[ line width=0.6pt, dashed, draw opacity=0.5] (a) at (-1.8,0){$\scriptstyle a_k'$};
        \node[ line width=0.6pt, dashed, draw opacity=0.5] (a) at (0.3,-1.5){$\scriptstyle \nu_k'$};
        \node[ line width=0.6pt, dashed, draw opacity=0.5] (a) at (0.3,-0.8){$\scriptstyle \nu_1'$};
        \node[ line width=0.6pt, dashed, draw opacity=0.5] (a) at (0,-0.3){$\scriptstyle z_1'$};
        \node[ line width=0.6pt, dashed, draw opacity=0.5] (a) at (0,-2.95){$\scriptstyle w_{l+1}'$};
        \node[ line width=0.6pt, dashed, draw opacity=0.5] (a) at (0,0.3){$\scriptstyle u_1'$};
         \node[ line width=0.6pt, dashed, draw opacity=0.5] (a) at (0,2.95){$\scriptstyle v_{l+1}'$};
        \node[ line width=0.6pt, dashed, draw opacity=0.5] (a) at (0.3,1.4){$\scriptstyle \gamma_k'$};
        \node[ line width=0.6pt, dashed, draw opacity=0.5] (a) at (0.3,0.8){$\scriptstyle \gamma_1'$};
    \end{tikzpicture}
\end{aligned}\right) \nonumber \\
    & = \delta_{u_1,v'_{l+1}}\delta_{z_1,w'_{l+1}}\; 
    \begin{aligned}
        \begin{tikzpicture}[scale=0.65]
            \filldraw[black!60, fill=gray!15, dotted, even odd rule] (0,0) circle[radius=0.5] (0,0) circle[radius=3.8];
            \draw[line width=1pt,cyan] (0,0.5)--(0,3.8);
            \draw[line width=1pt,violet] (0,-0.5)--(0,-3.8);
            \draw[line width=0.6pt,red] (0,0.8) arc[start angle=90, end angle=270, radius=0.8];
            \draw[line width=0.6pt,red] (0,1.3) arc[start angle=90, end angle=270, radius=1.3];
            \draw[line width=0.6pt,blue] (0,1.5) arc[start angle=90, end angle=-90, radius=1.5];
            \draw[line width=0.6pt,blue] (0,2) arc[start angle=90, end angle=-90, radius=2];
            \draw[line width=0.6pt,red] (0,2.4) arc[start angle=90, end angle=270, radius=2.4];
            \draw[line width=0.6pt,red] (0,2.9) arc[start angle=90, end angle=270, radius=2.9];
            \draw[line width=0.6pt,blue] (0,3.1) arc[start angle=90, end angle=-90, radius=3.1];
            \draw[line width=0.6pt,blue] (0,3.6) arc[start angle=90, end angle=-90, radius=3.6];
            \node[ line width=0.6pt, dashed, draw opacity=0.5] (a) at (-0.3,1.5){$\scriptstyle \zeta_1'$};
            \node[ line width=0.6pt, dashed, draw opacity=0.5] (a) at (-0.3,2.05){$\scriptstyle \zeta_l'$};
            \node[ line width=0.6pt, dashed, draw opacity=0.5] (a) at (-0.3,-1.5){$\scriptstyle \mu_1'$};
            \node[ line width=0.6pt, dashed, draw opacity=0.5] (a) at (-0.3,-2.05){$\scriptstyle \mu_l'$};
            \node[ line width=0.6pt, dashed, draw opacity=0.5] (a) at (-1,0){{\color{red}$\scriptstyle \cdots$}}; 
            \node[ line width=0.6pt, dashed, draw opacity=0.5] (a) at (1.25,0){$\scriptstyle b_1'$};
            \node[ line width=0.6pt, dashed, draw opacity=0.5] (a) at (1.8,0){{\color{blue}$\scriptstyle \cdots$}};
            \node[ line width=0.6pt, dashed, draw opacity=0.5] (a) at (2.25,0){$\scriptstyle b_l'$};
            \node[ line width=0.6pt, dashed, draw opacity=0.5] (a) at (-.5,0){$\scriptstyle a_1'$};
            \node[ line width=0.6pt, dashed, draw opacity=0.5] (a) at (-1.55,0){$\scriptstyle a_k'$};
            \node[ line width=0.6pt, dashed, draw opacity=0.5] (a) at (0.3,-1.2){$\scriptstyle \nu_k'$};
            \node[ line width=0.6pt, dashed, draw opacity=0.5] (a) at (0.3,-0.7){$\scriptstyle \nu_1'$};
            \node[ line width=0.6pt, dashed, draw opacity=0.5] (a) at (0,-0.25){$\scriptstyle z'_1$};
            \node[ line width=0.6pt, dashed, draw opacity=0.5] (a) at (0.3,-2.2){$\scriptstyle z_1$};
            \node[ line width=0.6pt, dashed, draw opacity=0.5] (a) at (0,0.3){$\scriptstyle u'_1$};
            \node[ line width=0.6pt, dashed, draw opacity=0.5] (a) at (0.3,2.15){$\scriptstyle u_1$};
            \node[ line width=0.6pt, dashed, draw opacity=0.5] (a) at (0.3,1.3){$\scriptstyle \gamma_k'$};
            \node[ line width=0.6pt, dashed, draw opacity=0.5] (a) at (0.3,0.8){$\scriptstyle \gamma_1'$};
            \node[ line width=0.6pt, dashed, draw opacity=0.5] (a) at (0.3,-2.5){$\scriptstyle \nu_1$};
            \node[ line width=0.6pt, dashed, draw opacity=0.5] (a) at (0.3,-2.9){$\scriptstyle \nu_k$};
            \node[ line width=0.6pt, dashed, draw opacity=0.5] (a) at (-0.3,-3.1){$\scriptstyle \mu_1$};
            \node[ line width=0.6pt, dashed, draw opacity=0.5] (a) at (-0.3,-3.6){$\scriptstyle \mu_l$};
            \node[ line width=0.6pt, dashed, draw opacity=0.5] (a) at (0,-4){$\scriptstyle w_{l+1}$};
            \node[ line width=0.6pt, dashed, draw opacity=0.5] (a) at (0.3,2.5){$\scriptstyle \gamma_1$};
            \node[ line width=0.6pt, dashed, draw opacity=0.5] (a) at (0.3,2.9){$\scriptstyle \gamma_k$};
            \node[ line width=0.6pt, dashed, draw opacity=0.5] (a) at (-0.3,3.1){$\scriptstyle \zeta_1$};
            \node[ line width=0.6pt, dashed, draw opacity=0.5] (a) at (-0.3,3.5){$\scriptstyle \zeta_l$};
            \node[ line width=0.6pt, dashed, draw opacity=0.5] (a) at (0,4){$\scriptstyle v_{l+1}$};
            \node[ line width=0.6pt, dashed, draw opacity=0.5] (a) at (-2.6,0){{\color{red}$\scriptstyle \cdots$}}; 
            \node[ line width=0.6pt, dashed, draw opacity=0.5] (a) at (2.9,0){$\scriptstyle b_1$};
            \node[ line width=0.6pt, dashed, draw opacity=0.5] (a) at (3.4,0){{\color{blue}$\scriptstyle \cdots$}};
            \node[ line width=0.6pt, dashed, draw opacity=0.5] (a) at (3.85,0){$\scriptstyle b_l$};
            \node[ line width=0.6pt, dashed, draw opacity=0.5] (a) at (-2.1,0){$\scriptstyle a_1$};
            \node[ line width=0.6pt, dashed, draw opacity=0.5] (a) at (-3.2,0){$\scriptstyle a_k$};
        \end{tikzpicture}
    \end{aligned}\;. \label{eq:Ntube_prod_1}
\end{align}
However, $\mathbf{Tube}({_{\{\eC_{\alpha}\}}} \eM_{\{\eD_{\beta}\}};{_{\{\eC_{\alpha}\}}} \eN_{\{\eD_{\beta}\}})$ has no natural coalgebra structure. Nevertheless, it carries comodule structures:
\begin{align*}
    \beta&: \mathbf{Tube}({_{\{\eC_{\alpha}\}}} \eM_{\{\eD_{\beta}\}};{_{\{\eC_{\alpha}\}}} \eN_{\{\eD_{\beta}\}}) \to \mathbf{Tube}({_{\{\eC_{\alpha}\}}} \eN_{\{\eD_{\beta}\}})\otimes \mathbf{Tube}({_{\{\eC_{\alpha}\}}} \eM_{\{\eD_{\beta}\}};{_{\{\eC_{\alpha}\}}} \eN_{\{\eD_{\beta}\}}), \\
    \rho&: \mathbf{Tube}({_{\{\eC_{\alpha}\}}} \eM_{\{\eD_{\beta}\}};{_{\{\eC_{\alpha}\}}} \eN_{\{\eD_{\beta}\}}) \to \mathbf{Tube}({_{\{\eC_{\alpha}\}}} \eM_{\{\eD_{\beta}\}};{_{\{\eC_{\alpha}\}}} \eN_{\{\eD_{\beta}\}}) \otimes \mathbf{Tube}({_{\{\eC_{\alpha}\}}} \eM_{\{\eD_{\beta}\}}).
\end{align*}
More specifically, the left coaction $\beta$ is defined as  
\begin{align}
    \beta\left(\begin{aligned}
    \begin{tikzpicture}[scale=0.6]
        \filldraw[black!60, fill=gray!15, dotted, even odd rule] (0,0) circle[radius=0.5] (0,0) circle[radius=2.6];
        \draw[line width=1pt,cyan] (0,0.5)--(0,2.6);
        \draw[line width=1pt,violet] (0,-0.5)--(0,-2.6);
        \draw[line width=.6pt,red] (0,0.8) arc[start angle=90, end angle=270, radius=0.8];
        \draw[line width=.6pt,red] (0,1.5) arc[start angle=90, end angle=270, radius=1.5];
        \draw[line width=.6pt,blue] (0,1.7) arc[start angle=90, end angle=-90, radius=1.7];
        \draw[line width=.6pt,blue] (0,2.4) arc[start angle=90, end angle=-90, radius=2.4];
        \node[ line width=0.6pt, dashed, draw opacity=0.5] (a) at (-0.3,1.7){$\scriptstyle \zeta_1$};
        \node[ line width=0.6pt, dashed, draw opacity=0.5] (a) at (-0.3,2.4){$\scriptstyle \zeta_l$};
        \node[ line width=0.6pt, dashed, draw opacity=0.5] (a) at (-0.3,-1.75){$\scriptstyle \mu_1$};
        \node[ line width=0.6pt, dashed, draw opacity=0.5] (a) at (-0.3,-2.4){$\scriptstyle \mu_l$};
        \node[ line width=0.6pt, dashed, draw opacity=0.5] (a) at (-1.1,0){{\color{red}$\scriptstyle \cdots$}};
        \node[ line width=0.6pt, dashed, draw opacity=0.5] (a) at (1.4,0){$\scriptstyle b_1$};
        \node[ line width=0.6pt, dashed, draw opacity=0.5] (a) at (2.1,0){{\color{blue}$\scriptstyle \cdots$}};
        \node[ line width=0.6pt, dashed, draw opacity=0.5] (a) at (2.7,0){$\scriptstyle b_l$};
        \node[ line width=0.6pt, dashed, draw opacity=0.5] (a) at (-.4,0){$\scriptstyle a_1$};
        \node[ line width=0.6pt, dashed, draw opacity=0.5] (a) at (-1.8,0){$\scriptstyle a_k$};
        \node[ line width=0.6pt, dashed, draw opacity=0.5] (a) at (0.3,-1.5){$\scriptstyle \nu_k$};
        \node[ line width=0.6pt, dashed, draw opacity=0.5] (a) at (0.3,-0.8){$\scriptstyle \nu_1$};
        \node[ line width=0.6pt, dashed, draw opacity=0.5] (a) at (0,-0.3){$\scriptstyle z_1$};
        \node[ line width=0.6pt, dashed, draw opacity=0.5] (a) at (0,-2.8){$\scriptstyle w_{l+1}$};
        \node[ line width=0.6pt, dashed, draw opacity=0.5] (a) at (0,0.3){$\scriptstyle u_1$};
        \node[ line width=0.6pt, dashed, draw opacity=0.5] (a) at (0,2.8){$\scriptstyle v_{l+1}$};
        \node[ line width=0.6pt, dashed, draw opacity=0.5] (a) at (0.3,1.4){$\scriptstyle \gamma_k$};
        \node[ line width=0.6pt, dashed, draw opacity=0.5] (a) at (0.3,0.8){$\scriptstyle \gamma_1$};
    \end{tikzpicture}
\end{aligned}\right) & = \; \sum_{s_i,t_j}\sum_{\sigma_i,\rho_j}  \sqrt{\frac{d_{t_{l+1}}}{d_{a_1}\cdots d_{a_k} d_{s_1} d_{b_1}\cdots d_{b_l}}}   \nonumber \\
& \quad \quad \begin{aligned}
    \begin{tikzpicture}[scale=0.6]
    \filldraw[black!60, fill=gray!15, dotted, even odd rule] (0,0) circle[radius=0.5] (0,0) circle[radius=2.6];
         \draw[line width=1pt,cyan] (0,0.5)--(0,2.6);
         \draw[line width=1pt,cyan] (0,-0.5)--(0,-2.6);
         \draw[line width=.6pt,red] (0,0.8) arc[start angle=90, end angle=270, radius=0.8];
         \draw[line width=.6pt,red] (0,1.5) arc[start angle=90, end angle=270, radius=1.5];
        \draw[line width=.6pt,blue] (0,1.7) arc[start angle=90, end angle=-90, radius=1.7];
        \draw[line width=.6pt,blue] (0,2.4) arc[start angle=90, end angle=-90, radius=2.4];
        \node[ line width=0.6pt, dashed, draw opacity=0.5] (a) at (-0.3,1.7){$\scriptstyle \zeta_1$};
         \node[ line width=0.6pt, dashed, draw opacity=0.5] (a) at (-0.3,2.4){$\scriptstyle \zeta_l$};
         \node[ line width=0.6pt, dashed, draw opacity=0.5] (a) at (-0.3,-1.7){$\scriptstyle \rho_1$};
        \node[ line width=0.6pt, dashed, draw opacity=0.5] (a) at (-0.3,-2.4){$\scriptstyle \rho_l$};
        \node[ line width=0.6pt, dashed, draw opacity=0.5] (a) at (-1.1,0){{\color{red}$\scriptstyle \cdots$}};
        \node[ line width=0.6pt, dashed, draw opacity=0.5] (a) at (1.4,0){$\scriptstyle b_1$};
         \node[ line width=0.6pt, dashed, draw opacity=0.5] (a) at (2.1,0){{\color{blue}$\scriptstyle \cdots$}};
        \node[ line width=0.6pt, dashed, draw opacity=0.5] (a) at (2.7,0){$\scriptstyle b_l$};
        \node[ line width=0.6pt, dashed, draw opacity=0.5] (a) at (-.4,0){$\scriptstyle a_1$};
        \node[ line width=0.6pt, dashed, draw opacity=0.5] (a) at (-1.8,0){$\scriptstyle a_k$};
        \node[ line width=0.6pt, dashed, draw opacity=0.5] (a) at (0.35,-1.5){$\scriptstyle \sigma_k$};
        \node[ line width=0.6pt, dashed, draw opacity=0.5] (a) at (0.35,-0.8){$\scriptstyle \sigma_1$};
        \node[ line width=0.6pt, dashed, draw opacity=0.5] (a) at (0,-0.3){$\scriptstyle s_1$};
        \node[ line width=0.6pt, dashed, draw opacity=0.5] (a) at (0,-2.8){$\scriptstyle t_{l+1}$};
        \node[ line width=0.6pt, dashed, draw opacity=0.5] (a) at (0,0.3){$\scriptstyle u_1$};
         \node[ line width=0.6pt, dashed, draw opacity=0.5] (a) at (0,2.8){$\scriptstyle v_{l+1}$};
        \node[ line width=0.6pt, dashed, draw opacity=0.5] (a) at (0.3,1.4){$\scriptstyle \gamma_k$};
        \node[ line width=0.6pt, dashed, draw opacity=0.5] (a) at (0.3,0.8){$\scriptstyle \gamma_1$};
    \end{tikzpicture}
\end{aligned} \otimes 
\begin{aligned}
    \begin{tikzpicture}[scale=0.6]
    \filldraw[black!60, fill=gray!15, dotted, even odd rule] (0,0) circle[radius=0.5] (0,0) circle[radius=2.6];
         \draw[line width=1pt,cyan] (0,0.5)--(0,2.6);
         \draw[line width=1pt,violet] (0,-0.5)--(0,-2.6);
         \draw[line width=.6pt,red] (0,0.8) arc[start angle=90, end angle=270, radius=0.8];
         \draw[line width=.6pt,red] (0,1.5) arc[start angle=90, end angle=270, radius=1.5];
        \draw[line width=.6pt,blue] (0,1.7) arc[start angle=90, end angle=-90, radius=1.7];
        \draw[line width=.6pt,blue] (0,2.4) arc[start angle=90, end angle=-90, radius=2.4];
        \node[ line width=0.6pt, dashed, draw opacity=0.5] (a) at (-0.3,1.7){$\scriptstyle \rho_1$};
         \node[ line width=0.6pt, dashed, draw opacity=0.5] (a) at (-0.3,2.4){$\scriptstyle \rho_l$};
         \node[ line width=0.6pt, dashed, draw opacity=0.5] (a) at (-0.3,-1.75){$\scriptstyle \mu_1$};
        \node[ line width=0.6pt, dashed, draw opacity=0.5] (a) at (-0.3,-2.4){$\scriptstyle \mu_l$};
        \node[ line width=0.6pt, dashed, draw opacity=0.5] (a) at (-1.1,0){{\color{red}$\scriptstyle \cdots$}};
        \node[ line width=0.6pt, dashed, draw opacity=0.5] (a) at (1.4,0){$\scriptstyle b_1$};
         \node[ line width=0.6pt, dashed, draw opacity=0.5] (a) at (2.1,0){{\color{blue}$\scriptstyle \cdots$}};
        \node[ line width=0.6pt, dashed, draw opacity=0.5] (a) at (2.7,0){$\scriptstyle b_l$};
        \node[ line width=0.6pt, dashed, draw opacity=0.5] (a) at (-.4,0){$\scriptstyle a_1$};
        \node[ line width=0.6pt, dashed, draw opacity=0.5] (a) at (-1.8,0){$\scriptstyle a_k$};
        \node[ line width=0.6pt, dashed, draw opacity=0.5] (a) at (0.3,-1.5){$\scriptstyle \nu_k$};
        \node[ line width=0.6pt, dashed, draw opacity=0.5] (a) at (0.3,-0.8){$\scriptstyle \nu_1$};
        \node[ line width=0.6pt, dashed, draw opacity=0.5] (a) at (0,-0.3){$\scriptstyle z_1$};
        \node[ line width=0.6pt, dashed, draw opacity=0.5] (a) at (0,-2.8){$\scriptstyle w_{l+1}$};
        \node[ line width=0.6pt, dashed, draw opacity=0.5] (a) at (0,0.3){$\scriptstyle s_1$};
         \node[ line width=0.6pt, dashed, draw opacity=0.5] (a) at (0,2.8){$\scriptstyle t_{l+1}$};
        \node[ line width=0.6pt, dashed, draw opacity=0.5] (a) at (0.35,1.4){$\scriptstyle \sigma_k$};
        \node[ line width=0.6pt, dashed, draw opacity=0.5] (a) at (0.35,0.8){$\scriptstyle \sigma_1$};
    \end{tikzpicture}
\end{aligned}\;, \label{eq:tube_coprod}
\end{align}
where in the lower half of the first component the wall labels are ordered as $s_1,\cdots,s_k$, $t_1,\cdots,t_{l+1}$ from top to bottom, and in the upper half of the second component the wall labels are ordered as $s_1,\cdots,s_k$, $t_1,\cdots,t_{l+1}$ from bottom to top. The map $\rho$ is analogously defined. This construction is motivated by the same reason as in Remark~\ref{rmk:coaction}: inserting a wall labeled by $\sum_{s_1\in \Irr(\eN)}s_1$ and applying sequential F-moves. 
Intuitively, the left coaction $\beta$ corresponds to fusing a multimodule domain wall excitation from the $\eN$ side onto the defect, while the right coaction $\rho$ corresponds to fusing a multimodule domain wall excitation from the $\eM$ side. 
The comodule algebra structure then encodes how these fusions are compatible with the defect algebra multiplication.

\begin{theorem}
    The linear morphisms $\beta$ and $\rho$ endow  $\mathbf{Tube}({_{\{\eC_{\alpha}\}}} \eM_{\{\eD_{\beta}\}};{_{\{\eC_{\alpha}\}}} \eN_{\{\eD_{\beta}\}})$ with the structure of a $\mathbf{Tube}({_{\{\eC_{\alpha}\}}} \eN_{\{\eD_{\beta}\}})|\mathbf{Tube}({_{\{\eC_{\alpha}\}}} \eM_{\{\eD_{\beta}\}})$-bicomodule algebra. 
\end{theorem}

\begin{proof}
    The only nontrivial part is to show that $\mathbf{Tube}({_{\{\eC_{\alpha}\}}} \eM_{\{\eD_{\beta}\}};{_{\{\eC_{\alpha}\}}} \eN_{\{\eD_{\beta}\}})$ is a left (resp. right) comodule algebra with respect to $\beta$ (resp. $\rho$), i.e., $\beta(XX')=\beta(X)\beta(X')$. However, this presents no essential difficulty, as the proof is almost the same as that for the domain wall defect tube algebra, except that the structural constants involved (e.g., the multimodule F-symbols) are more complicated and require careful handling.  
    Alternatively, this can be understood from a different perspective using Remark~\ref{rmk:coaction}, where the coaction $\beta$ is obtained by inserting  a wall labeled by $\sum_{s_1\in \Irr(\eN)}s_1$ and then applying F-moves. We omit the detailed verification here. 
\end{proof}

\begin{remark}[Morita theory]
One can also develop a Morita theory for multimodule defect tube algebras in a completely parallel way to the boundary and domain wall cases. 
For a multimodule defect tube algebra, one may introduce generalized tube spaces with prescribed numbers of external and internal legs in each bulk sector:
\begin{equation} \label{eq:tube_space}
    \mathbf{T}^{\mathbf{s},\mathbf{r};\mathbf{m},\mathbf{n}} = \operatorname{span}_\mathbb{C}
    \left\{
    \begin{aligned}
        \begin{tikzpicture}[scale=0.6]
            \filldraw[black!60, fill=gray!15, dotted, even odd rule] (0,0) circle[radius=0.5] (0,0) circle[radius=2.6];
            \draw[line width=1pt,cyan] (0,0.5)--(0,2.6);
            \draw[line width=1pt,violet] (0,-0.5)--(0,-2.6);
            \draw[red] (0,0.8) arc[start angle=90, end angle=270, radius=0.8];
            \draw[red] (0,1.3) arc[start angle=90, end angle=270, radius=1.3];
            \draw[blue] (0,1.7) arc[start angle=90, end angle=-90, radius=1.7];
            \draw[blue] (0,2.2) arc[start angle=90, end angle=-90, radius=2.2];
            \draw[red,line width=.6pt] (120:0.8) -- (120:1.1);
            \draw[red,line width=.6pt] (150:0.8) -- (150:0.5);
            \draw[red,line width=.6pt] (210:0.8) -- (210:1.1);
            \draw[red,line width=.6pt] (240:0.8) -- (240:0.5);
            \draw[red,line width=.6pt] (105:1.6) -- (105:1.3);
            \draw[red,line width=.6pt] (135:1) -- (135:1.3);
            \draw[red,line width=.6pt] (165:1.6) -- (165:1.3);
            \draw[red,line width=.6pt] (195:1) -- (195:1.3);
            \draw[red,line width=.6pt] (225:1.6) -- (225:1.3);
            \draw[red,line width=.6pt] (255:1) -- (255:1.3);
            \draw[blue,line width=.6pt] (60:1.7) -- (60:1.4);
            \draw[blue,line width=.6pt] (30:1.7) -- (30:2);
            \draw[blue,line width=.6pt] (0:1.7) -- (0:1.4);
            \draw[blue,line width=.6pt] (-30:1.7) -- (-30:2);
            \draw[blue,line width=.6pt] (-60:1.7) -- (-60:1.4);
            \draw[blue,line width=.6pt] (75:2.2) -- (75:2.5);
            \draw[blue,line width=.6pt] (45:2.2) -- (45:1.9);
            \draw[blue,line width=.6pt] (15:2.2) -- (15:2.5);
            \draw[blue,line width=.6pt] (-15:2.2) -- (-15:1.9);
            \draw[blue,line width=.6pt] (-45:2.2) -- (-45:2.5);
            \draw[blue,line width=.6pt] (-75:2.2) -- (-75:1.9);
            \node[ line width=0.6pt, dashed, draw opacity=0.5] (a) at (-1,0){{\color{red}$\scriptstyle \cdots$}};
            \node[ line width=0.6pt, dashed, draw opacity=0.5] (a) at (2,0){{\color{blue}$\scriptstyle \cdots$}};
            \node[ line width=0.6pt, dashed, draw opacity=0.5] (a) at (-1,1.3){{\color{red}$\scriptstyle s_{1}$}};
            \node[ line width=0.6pt, dashed, draw opacity=0.5] (a) at (-0.93,0.45){{\color{red}$\scriptstyle r_1$}};
            \node[ line width=0.6pt, dashed, draw opacity=0.5] (a) at (-0.55,-0.8){{\color{red}$\scriptstyle s_{k}$}};
            \node[ line width=0.6pt, dashed, draw opacity=0.5] (a) at (-0.5,-0.1){{\color{red}$\scriptstyle r_k$}};
            \node[ line width=0.6pt, dashed, draw opacity=0.5] (a) at (1.1,2.2){{\color{blue}$\scriptstyle m_{1}$}};
            \node[ line width=0.6pt, dashed, draw opacity=0.5] (a) at (1.1,1.65){{\color{blue}$\scriptstyle n_1$}};
            \node[ line width=0.6pt, dashed, draw opacity=0.5] (a) at (1.58,-1.15){{\color{blue}$\scriptstyle m_{l}$}};
            \node[ line width=0.6pt, dashed, draw opacity=0.5] (a) at (0.8,-1){{\color{blue}$\scriptstyle n_l$}};
        \end{tikzpicture}
    \end{aligned}\; : \; \text{edge}\in \Irr, \text{vertex}\in \Hom 
    \right\},
\end{equation}
where $\mathbf{s}=(s_1,\cdots,s_{k}),  \mathbf{r} = (r_1,\cdots,r_k) \in \mathbb{Z}^{k}_{\geq 0}$, and $\mathbf{m} = (m_1,\cdots,m_l), \mathbf{n} = (n_1,\cdots,n_l)\in \mathbb{Z}^{l}_{\geq 0}$. 
When these numbers agree sectorwise ($\mathbf{s} = \mathbf{r}$, $\mathbf{m} = \mathbf{n}$), the corresponding tube space carries an algebra structure; in general, such tube spaces form bimodules between the associated algebras. 
By stacking tubes and imposing the matching conditions independently in each bulk sector, one obtains Morita contexts relating tube algebras with different choices of auxiliary legs. 
Thus the Morita equivalence class of a multimodule defect tube algebra is independent of these auxiliary choices and depends only on the underlying multimodule defect data.
\end{remark}

\subsection{Multimodule domain wall defects as representations of defect tube algebra}

Before turning to the disk space construction, we record the algebraic ingredient needed for the representation theory of the multimodule defect tube algebra. The factorization map and the antipode-like map introduced for boundary and domain wall defect tube algebras admit direct multimodule analogues. Given three multimodule categories $\eM,\eK$, and $\eN$ over the same families $\{\eC_\alpha\}_{\alpha\in I}$ and $\{\eD_\beta\}_{\beta\in J}$, one has a factorization map
\begin{align*}
    \delta_{\eM,\eK,\eN}:&\; 
    \Tube({}_{\{\eC_\alpha\}}\eM_{\{\eD_\beta\}};{}_{\{\eC_\alpha\}}\eN_{\{\eD_\beta\}}) \\
    & \to
    \Tube({}_{\{\eC_\alpha\}}\eK_{\{\eD_\beta\}};{}_{\{\eC_\alpha\}}\eN_{\{\eD_\beta\}})
    \otimes
    \Tube({}_{\{\eC_\alpha\}}\eM_{\{\eD_\beta\}};{}_{\{\eC_\alpha\}}\eK_{\{\eD_\beta\}})
\end{align*}
and an antipode-like map
\begin{equation}
    \notag
    s_{\eM,\eN}:
    \Tube({}_{\{\eC_\alpha\}}\eM_{\{\eD_\beta\}};{}_{\{\eC_\alpha\}}\eN_{\{\eD_\beta\}})
    \to
    \Tube({}_{\{\eC_\alpha\}}\eN_{\{\eD_\beta\}};{}_{\{\eC_\alpha\}}\eM_{\{\eD_\beta\}}).
\end{equation}
Graphically, the factorization map is obtained by inserting an intermediate multimodule wall labeled by $\eK$, resolving the identity in every bulk sector, and cutting the original tube into two compatible defect tubes. The antipode-like map is obtained by the same reversal operation used in the boundary and domain wall cases, applied simultaneously to all bulk labels and multimodule vertices. Since these constructions are direct extensions of those given previously, we omit their explicit expansions in the multimodule tube basis.

The resulting maps satisfy the multimodule analogues of the structural identities established in the boundary and domain wall settings. In particular, successive factorizations are compatible, the factorization maps preserve multiplication, the antipode-like maps reverse multiplication, and the corresponding counital and factorization identities continue to hold. The factorization maps are also compatible with the left and right coactions introduced above. Their graphical verifications use the same resolutions of the identity, F-moves, and loop removals as before, now carried out simultaneously in the $\eC_\alpha$- and $\eD_\beta$-bulks. The coherence conditions of the multimodule category ensure that the order in which these local moves are performed does not affect the resulting diagram.

These structural maps allow the separability idempotent of the multimodule defect tube algebra to be constructed in the same way as in the boundary and domain wall cases.

\begin{theorem}
    \label{thm:sep-idem-multimodule}
    Let $\lambda \in \Tube({}_{\{\eC_\alpha\}}\eM_{\{\eD_\beta\}})$ be the Haar integral of the weak Hopf multimodule tube algebra. The multimodule defect tube algebra $\Tube({}_{\{\eC_\alpha\}}\eM_{\{\eD_\beta\}};{}_{\{\eC_\alpha\}}\eN_{\{\eD_\beta\}})$ has a separability idempotent given by
    \begin{equation}
        \label{eq:sep-idem-multimodule}
        \Upsilon :=\left(s_{\eN,\eM}\otimes\id\right)\circ \delta_{\eM,\eN,\eM} \left(\lambda\right)=\sum_{\langle\lambda\rangle}
        s_{\eN,\eM}\left(\lambda^{\langle 1\rangle}\right) \otimes \lambda^{\langle 2\rangle}.
    \end{equation}
    Consequently, $\Tube({}_{\{\eC_\alpha\}}\eM_{\{\eD_\beta\}};{}_{\{\eC_\alpha\}}\eN_{\{\eD_\beta\}})$ is separable and semisimple.
\end{theorem}

\begin{proof}
    The proof is parallel to those of boundary and domain wall cases. The invariance of the Haar integral, together with the structural identities of the factorization map and the antipode-like map, gives the centrality
    condition. The normalization is verified by composing the two tensor factors. The calculation is the same as in the boundary and
    domain wall cases, except that the resolutions of the identity and the corresponding loop removals are performed in every bulk sector. The mutual coherence of the multimodule actions allows these reductions to be carried out independently, and the product of their coefficients is precisely the coefficient appearing in the Haar integral $\lambda$. Hence $\Upsilon$ is a separability
    idempotent. Since $\Tube({}_{\{\eC_\alpha\}}\eM_{\{\eD_\beta\}};{}_{\{\eC_\alpha\}}\eN_{\{\eD_\beta\}})$ is finite-dimensional,
    separability implies semisimplicity.
\end{proof}

The preceding theorem supplies the averaging element required for generalized Schur orthogonality and, at the same time, ensures the semisimplicity of the multimodule defect tube algebra. We now turn to its finite-dimensional representations. As in the boundary and domain wall settings, a multimodule functor gives rise to a disk space on which the defect tube algebra acts. We will then identify intertwiners of these representations with multimodule natural transformations and establish the resulting equivalence of categories. To be precise, for a multimodule functor $\mathfrak f:\eM\to\eN$, define the disk space
\begin{equation}
    \mathcal{H}^\fff = \bigoplus_{s\in \Irr(\eM)}\bigoplus_{t\in \Irr(\eN)} \Hom_{\eN}(\fff(s),t) = \operatorname{span}\left\{\begin{aligned}
        \begin{tikzpicture}[scale=0.85]
            \filldraw[black!60, fill=gray!15, dotted] (0,0) circle (1); 
            \draw[cyan, line width=1pt] (0,0) -- (0,1); 
            \draw[cyan, line width=1pt, -latex] (0,0) -- (0,0.66); 
            \draw[violet, line width=1pt] (0,-1) -- (0,0); 
            \draw[violet, line width=1pt, -latex] (0,-1) -- (0,-0.4); 
            \draw[black, line width=0.6pt, -latex] (1,0) -- (0.55,0); 
            \draw[black, line width=0.6pt] (1,0) -- (0.75,0); 
            \draw[decorate, line width=0.6pt, decoration={snake, amplitude=2pt, segment length=6pt}] (0,0) -- (0.6,0);
            \node[ line width=0.6pt, dashed, draw opacity=0.5] at (-0.2,0.5){$\scriptstyle t$};
            \node[ line width=0.6pt, dashed, draw opacity=0.5] at (-0.2,-0.6){$\scriptstyle s$};
            \node[ line width=0.6pt, dashed, draw opacity=0.5] at (0.6,-0.3){$\scriptstyle \mathfrak{f}$};
            \node[ line width=0.6pt, dashed, draw opacity=0.5] at (-0.3,0){$\scriptstyle \alpha$};
            \node[ line width=0.6pt, dashed, draw opacity=0.5] at (0,0){$\scriptstyle \bullet$};
        \end{tikzpicture}
    \end{aligned}\right\}.
\end{equation}
As before, its elements can be represented in a disk as shown above. For simplicity, we omit arrows in what follows. 
The action of the multimodule defect tube algebra $\mathbf{Tube}({_{\{\eC_{\alpha}\}}} \eM_{\{\eD_{\beta}\}};{_{\{\eC_{\alpha}\}}} \eN_{\{\eD_{\beta}\}})$ on the disk vector space $\mathcal{H}^\mathfrak{f}$ is the same as that defined before. More precisely, by filling the disk representation of the morphism inside the tube representation of an algebra element, it is given as
\begin{equation}
    \notag
    \begin{aligned}
        \begin{tikzpicture}[scale=0.6]
            \filldraw[black!60, fill=gray!15, dotted, even odd rule] (0,0) circle[radius=0.5] (0,0) circle[radius=2.6];
            \draw[line width=1pt,cyan] (0,0.5)--(0,2.6);
            \draw[line width=1pt,violet] (0,-0.5)--(0,-2.6);
            \draw[line width=0.6pt,red] (0,0.8) arc[start angle=90, end angle=270, radius=0.8];
            \draw[line width=0.6pt,red] (0,1.5) arc[start angle=90, end angle=270, radius=1.5];
            \draw[line width=0.6pt,blue] (0,1.7) arc[start angle=90, end angle=-90, radius=1.7];
            \draw[line width=0.6pt,blue] (0,2.4) arc[start angle=90, end angle=-90, radius=2.4];
            \node[ line width=0.6pt, dashed, draw opacity=0.5] (a) at (-0.3,1.7){$\scriptstyle \zeta_1$};
            \node[ line width=0.6pt, dashed, draw opacity=0.5] (a) at (-0.3,2.4){$\scriptstyle \zeta_l$};
            \node[ line width=0.6pt, dashed, draw opacity=0.5] (a) at (-0.3,-1.75){$\scriptstyle \mu_1$};
            \node[ line width=0.6pt, dashed, draw opacity=0.5] (a) at (-0.3,-2.4){$\scriptstyle \mu_l$};
            \node[ line width=0.6pt, dashed, draw opacity=0.5] (a) at (-1.1,0){{\color{red}$\scriptstyle \cdots$}};
            \node[ line width=0.6pt, dashed, draw opacity=0.5] (a) at (1.4,0){$\scriptstyle b_1$};
            \node[ line width=0.6pt, dashed, draw opacity=0.5] (a) at (2.1,0){{\color{blue}$\scriptstyle \cdots$}};
            \node[ line width=0.6pt, dashed, draw opacity=0.5] (a) at (2.7,0){$\scriptstyle b_l$};
            \node[ line width=0.6pt, dashed, draw opacity=0.5] (a) at (-.4,0){$\scriptstyle a_1$};
            \node[ line width=0.6pt, dashed, draw opacity=0.5] (a) at (-1.8,0){$\scriptstyle a_k$};
            \node[ line width=0.6pt, dashed, draw opacity=0.5] (a) at (0.3,-1.5){$\scriptstyle \nu_k$};
            \node[ line width=0.6pt, dashed, draw opacity=0.5] (a) at (0.3,-0.8){$\scriptstyle \nu_1$};
            \node[ line width=0.6pt, dashed, draw opacity=0.5] (a) at (0,-0.3){$\scriptstyle z_1$};
            \node[ line width=0.6pt, dashed, draw opacity=0.5] (a) at (0,-2.8){$\scriptstyle w_{l+1}$};
            \node[ line width=0.6pt, dashed, draw opacity=0.5] (a) at (0,0.3){$\scriptstyle u_1$};
             \node[ line width=0.6pt, dashed, draw opacity=0.5] (a) at (0,2.8){$\scriptstyle v_{l+1}$};
            \node[ line width=0.6pt, dashed, draw opacity=0.5] (a) at (0.3,1.4){$\scriptstyle \gamma_k$};
            \node[ line width=0.6pt, dashed, draw opacity=0.5] (a) at (0.3,0.8){$\scriptstyle \gamma_1$};
        \end{tikzpicture}
    \end{aligned}
    \;\triangleright \; \begin{aligned}
        \begin{tikzpicture}[scale=0.85]
            \filldraw[black!60, fill=gray!15, dotted] (0,0) circle (1); 
            \draw[cyan, line width=1pt] (0,0) -- (0,1);  
            \draw[violet, line width=1pt] (0,-1) -- (0,0);  
            \draw[decorate, line width=0.6pt, decoration={snake, amplitude=2pt, segment length=6pt}] (0,0) -- (1,0);
            \node[ line width=0.6pt, dashed, draw opacity=0.5] at (-0.2,0.5){$\scriptstyle t$};
            \node[ line width=0.6pt, dashed, draw opacity=0.5] at (-0.2,-0.6){$\scriptstyle s$};
            \node[ line width=0.6pt, dashed, draw opacity=0.5] at (0.6,-0.3){$\scriptstyle \mathfrak{f}$};
            \node[ line width=0.6pt, dashed, draw opacity=0.5] at (-0.3,0){$\scriptstyle \alpha$};
            \node[ line width=0.6pt, dashed, draw opacity=0.5] at (0,0){$\scriptstyle \bullet$};
        \end{tikzpicture}
    \end{aligned} \; = \delta_{u_1,t}\delta_{z_1,s} \; \begin{aligned}
        \begin{tikzpicture}[scale=0.6]
            \filldraw[black!60, fill=gray!15, dotted] (0,0) circle[radius=2.6];
            \draw[line width=1pt,cyan] (0,0)--(0,2.6);
            \draw[line width=1pt,violet] (0,0)--(0,-2.6);
            \draw[line width=0.6pt,red] (0,0.8) arc[start angle=90, end angle=270, radius=0.8];
            \draw[line width=0.6pt,red] (0,1.5) arc[start angle=90, end angle=270, radius=1.5];
            \draw[line width=0.6pt,blue] (0,1.7) arc[start angle=90, end angle=-90, radius=1.7];
            \draw[line width=0.6pt,blue] (0,2.4) arc[start angle=90, end angle=-90, radius=2.4];
            \draw[decorate, line width=0.6pt, decoration={snake, amplitude=2pt, segment length=6pt}] (0,0) -- (1.2,0); 
            \node[ line width=0.6pt, dashed, draw opacity=0.5] (a) at (-0.3,1.7){$\scriptstyle \zeta_1$};
            \node[ line width=0.6pt, dashed, draw opacity=0.5] (a) at (-0.3,2.4){$\scriptstyle \zeta_l$};
            \node[ line width=0.6pt, dashed, draw opacity=0.5] (a) at (-0.3,-1.75){$\scriptstyle \mu_1$};
            \node[ line width=0.6pt, dashed, draw opacity=0.5] (a) at (-0.3,-2.4){$\scriptstyle \mu_l$};
            \node[ line width=0.6pt, dashed, draw opacity=0.5] (a) at (-1.1,0){{\color{red}$\scriptstyle \cdots$}};
            \node[ line width=0.6pt, dashed, draw opacity=0.5] (a) at (1.4,0.4){$\scriptstyle b_1$};
            \node[ line width=0.6pt, dashed, draw opacity=0.5] (a) at (2.1,0){{\color{blue}$\scriptstyle \cdots$}};
            \node[ line width=0.6pt, dashed, draw opacity=0.5] (a) at (2.7,0){$\scriptstyle b_l$};
            \node[ line width=0.6pt, dashed, draw opacity=0.5] (a) at (-0.9,0.6){$\scriptstyle a_1$};
            \node[ line width=0.6pt, dashed, draw opacity=0.5] (a) at (-1.8,0){$\scriptstyle a_k$};
            \node[ line width=0.6pt, dashed, draw opacity=0.5] (a) at (0.3,-1.5){$\scriptstyle \nu_k$};
            \node[ line width=0.6pt, dashed, draw opacity=0.5] (a) at (0.3,-0.8){$\scriptstyle \nu_1$};
            \node[ line width=0.6pt, dashed, draw opacity=0.5] (a) at (-0.25,-0.4){$\scriptstyle z_1$};
            \node[ line width=0.6pt, dashed, draw opacity=0.5] (a) at (0,-2.8){$\scriptstyle w_{l+1}$};
            \node[ line width=0.6pt, dashed, draw opacity=0.5] (a) at (-0.25,0.4){$\scriptstyle u_1$};
             \node[ line width=0.6pt, dashed, draw opacity=0.5] (a) at (0,2.8){$\scriptstyle v_{l+1}$};
            \node[ line width=0.6pt, dashed, draw opacity=0.5] (a) at (0.3,1.4){$\scriptstyle \gamma_k$};
            \node[ line width=0.6pt, dashed, draw opacity=0.5] (a) at (0.3,0.8){$\scriptstyle \gamma_1$};            
            \node[ line width=0.6pt, dashed, draw opacity=0.5] at (0.6,-0.4){$\scriptstyle \mathfrak{f}$};
            \node[ line width=0.6pt, dashed, draw opacity=0.5] at (-0.3,0){$\scriptstyle \alpha$};
            \node[ line width=0.6pt, dashed, draw opacity=0.5] at (0,0){$\scriptstyle \bullet$};
        \end{tikzpicture}
    \end{aligned} \;. 
\end{equation}
We denote this action by $\rho_\mathfrak{f}$. To evaluate the right-hand side, we need to operate the topological move inside the multifusion category $\eP$ which is constructed similar as before. In the case here, one introduces the folded fusion category $\eB:=(\boxtimes_{\alpha\in I}\eC_{\alpha} ) \boxtimes (\boxtimes_{\beta\in J}\eD_{\beta}^{\mathrm{rev}})$. 
Both $\eM$ and $\eN$ are treated as left $\eB$-module categories, and
$\mathfrak f$ becomes a $\eB$-module functor.  All labels involved in
the filled diagram can therefore be treated inside the multifusion
category
\begin{equation*}
    \eP := \bigoplus_{i,j=0}^{2} \Fun_{\eB}(\eM_i,\eM_j), \quad \eM_0=\eB,\quad \eM_1=\eM,\quad \eM_2=\eN.
\end{equation*}

\begin{proposition}
    The map $\rho_\mathfrak{f}:\mathbf{Tube}({_{\{\eC_{\alpha}\}}} \eM_{\{\eD_{\beta}\}};{_{\{\eC_{\alpha}\}}} \eN_{\{\eD_{\beta}\}}) \to \End(\mathcal{H}^\mathfrak{f})$ endows $\mathcal{H}^\mathfrak{f}$ with a structure of a representation over the multimodule defect tube algebra. 
\end{proposition}

\begin{proof}
    It only needs to illustrate $\rho_{\mathfrak f}(XY)= \rho_{\mathfrak f}(X)\rho_{\mathfrak f}(Y)$.  Let $X$ and $Y$ be composable tube algebra elements.  Filling the product $XY$ by a disk amounts to first gluing the two annuli and then evaluating the resulting diagram.  Filling by $Y$ first and subsequently by $X$ gives the same composite multimodule diagram. Algebraically, the two evaluations differ only by a sequence of associators in $\eP$, and their equality follows from the pentagon coherence of that category. Consequently, this proves that $\rho_{\mathfrak f}$ is a unital algebra homomorphism.
\end{proof}

The disk space assignment reverses natural transformations of functors.  If $\eta:\mathfrak g\Rightarrow\mathfrak f$ is a multimodule natural transformation, define
\begin{equation*}
    T_{\eta}: \mathcal H^{\mathfrak f}\to\mathcal H^{\mathfrak g},
    \quad T_{\eta}(\alpha) := \alpha\circ\eta_s \quad \text{for } \alpha\in\Hom_{\eN}(\mathfrak f(s),t).
\end{equation*}
The multimodule naturality of $\eta$ permits its coupon to pass
through every bulk string of a tube, so $T_{\eta}$ commutes with the
action defined by $\rho_\mathfrak{f}$.

\begin{proposition}
    \label{prop:Hom-Nat-MMDW}
    For multimodule functors $\mathfrak f,\mathfrak g:\eM\to\eN$,
    disk space intertwiners are naturally identified with
    multimodule natural transformations in the reverse direction:
    \begin{equation}
        \Hom_{\mathbf{Tube}({_{\{\eC_{\alpha}\}}} \eM_{\{\eD_{\beta}\}};{_{\{\eC_{\alpha}\}}} \eN_{\{\eD_{\beta}\}})}
        \bigl(\mathcal H^{\mathfrak f},\mathcal H^{\mathfrak g}
        \bigr) \simeq\operatorname{Nat}_{\{\eC_\alpha\}|\{\eD_\beta\}} (\mathfrak g,\mathfrak f).
        \label{eq:Hom-Nat-MMDW}
    \end{equation}
    Therefore, the functor $\Phi_{\mathrm{multi}}: \Fun_{\{\eC_\alpha\}|\{\eD_\beta\}}(\eM,\eN) \to \Rep(\mathbf{Tube}({_{\{\eC_{\alpha}\}}} \eM_{\{\eD_{\beta}\}};{_{\{\eC_{\alpha}\}}} \eN_{\{\eD_{\beta}\}}))$, $\mathfrak{f} \mapsto (\mathcal{H}^{\mathfrak{f}},\rho_\mathfrak{f})$, is fully faithful and contravariant.
\end{proposition}

\begin{proof}
    For simplicity, denote $A^{\rm multi}_{\eM,\eN} := \mathbf{Tube}({_{\{\eC_{\alpha}\}}} \eM_{\{\eD_{\beta}\}};{_{\{\eC_{\alpha}\}}} \eN_{\{\eD_{\beta}\}})$. Simultaneously fold all bulk regions meeting the wall.  The multimodule categories $\eM$ and $\eN$ then become left module categories over the fusion category $\eB=(\boxtimes_{\alpha\in I}\eC_{\alpha}) \boxtimes ( \boxtimes_{\beta\in J}\eD_{\beta}^{\mathrm{rev}} )$, and one obtains a canonical equivalence
    \begin{equation*}
        \Fun_{\{\eC_\alpha\}|\{\eD_\beta\}}(\eM,\eN) \simeq \Fun_{\eB}(\eM,\eN).
    \end{equation*}
    The same folding turns a multimodule tube into the corresponding
    boundary tube for the $\eB$-module categories; we denote the resulting tube algebra as $A_\partial:=\Tube({}_\eB\eM;{}_\eB\eN)$. Proposition~\ref{prop:Hom-Nat}, applied to
    $\eB$, therefore gives
    \begin{equation*}
        \Hom_{A^{\rm multi}_{\eM,\eN}}
        \bigl(\mathcal H^{\mathfrak f},\mathcal H^{\mathfrak g}\bigr) 
        \simeq  \Hom_{A_\partial} \bigl(\mathcal H^{\mathfrak f},\mathcal H^{\mathfrak g}\bigr)  \simeq \operatorname{Nat}_{\eB}(\mathfrak g,\mathfrak f) 
        \simeq \operatorname{Nat}_{\{\eC_\alpha\}|\{\eD_\beta\}}(\mathfrak g,\mathfrak f).
    \end{equation*}
    This is the assertion. 
\end{proof}

\begin{proposition}[Generalized Schur orthogonality]
    \label{prop:Schur-orth-MMDW}
    Let $\mathfrak f$ and $\mathfrak g$ be simple objects of
    multimodule functors $\eM\to \eN$.  Then $\langle \chi_{\mathfrak f},\chi_{\mathfrak g}\rangle = \delta_{\mathfrak f,\mathfrak g}$.  In particular, the disk representation associated with a simple multimodule functor is irreducible.
\end{proposition}

\begin{proof}
    The inner product, under which the choice of an orthonormal basis $\{\tilde{\alpha}\}$ for $\mathcal{H}^{\mathfrak{f}}$, is defined the same as before. Using this, one computes the character $\chi_\mathfrak{f}(X) = \operatorname{tr}(\rho_{\mathfrak{f}}) = \sum_{\tilde{\alpha}}\langle \rho_{\mathfrak{f}}(X)\tilde{\alpha},\tilde{\alpha}\rangle$. 
    For the separability idempotent $\Upsilon$ of the multimodule defect tube algebra, the inner product of characters is equal to 
    \begin{equation}
        \bigl\langle \chi_{\mathfrak f},\chi_{\mathfrak g}\bigr\rangle = \sum_{\langle\Upsilon\rangle} \chi_{\mathfrak f} \bigl(\Upsilon^{\langle1\rangle}\bigr)\chi_{\mathfrak g}
        \bigl(\Upsilon^{\langle2\rangle}\bigr) = \dim
        \Hom_{A^{\mathrm{multi}}_{\eM,\eN}} \bigl(\mathcal H^{\mathfrak f}, \mathcal H^{\mathfrak g} \bigr).
    \end{equation}
    Proposition~\ref{prop:Hom-Nat-MMDW} identifies the last space with $\operatorname{Nat}_{\{\eC_\alpha\}|\{\eD_\beta\}} (\mathfrak g,\mathfrak f)$.  Schur's lemma for the semisimple multimodule functor category now gives the stated Kronecker
    delta.
\end{proof}

Denote 
\[
    A^{\rm multi}_{\eM,\eN} := \mathbf{Tube}({_{\{\eC_{\alpha}\}}} \eM_{\{\eD_{\beta}\}};{_{\{\eC_{\alpha}\}}} \eN_{\{\eD_{\beta}\}}). 
\]
The preceding results define a functor
\begin{equation}
    \Phi_{\mathrm{multi}}:
    \Fun_{\{\eC_\alpha\}|\{\eD_\beta\}}(\eM,\eN)^{\rm op}
    \to 
    \Rep\bigl(A_{\eM,\eN}^{\mathrm{multi}}\bigr),
    \qquad
    \mathfrak f
    \longmapsto
    \bigl(
        \mathcal H^{\mathfrak f},
        \rho_{\mathfrak f}
    \bigr).
    \label{eq:multimodule-disk-functor}
\end{equation}
Proposition~\ref{prop:Hom-Nat-MMDW} supplies its full faithfulness.  Essential surjectivity follows from folding, as recorded next.

\begin{theorem}
    \label{thm:Fun-Rep-equiv-MMDW}
    The functor $\Phi_{\mathrm{multi}}$ possesses a
    quasi-inverse
    \begin{equation}
        \Psi_{\mathrm{multi}}:
        \Rep\bigl(A_{\eM,\eN}^{\mathrm{multi}}\bigr)
        \to 
        \Fun_{\{\eC_\alpha\}|\{\eD_\beta\}}(\eM,\eN)^{\rm op}.
    \end{equation}
    Hence there is an equivalence 
    \begin{equation}
        \Rep\bigl(\Tube({}_{\{\eC_\alpha\}}\eM_{\{\eD_\beta\}};{}_{\{\eC_\alpha\}}\eN_{\{\eD_\beta\}})\bigr) \simeq \Fun_{\{\eC_\alpha\}|\{\eD_\beta\}}(\eM,\eN)^{\rm op}.
        \label{eq:multimodule-representation-equivalence}
    \end{equation}
\end{theorem}

\begin{proof}
    Folding all $\eC_{\alpha}$- and $\eD_{\beta}$-regions onto the same side of the wall produces the fusion category $\eB := ( \boxtimes_{\alpha\in I}\eC_{\alpha})\boxtimes (\boxtimes_{\beta\in J}\eD_{\beta}^{\mathrm{rev}})$. 
    Folding identifies multimodule functors and their natural transformations with $\eB$-module functors and their natural transformations. It therefore induces an equivalence between the corresponding opposite functor categories
    \begin{equation}
        \notag
        \mathsf{Fold}^{\mathrm{op}}:
        \Fun_{\{\eC_\alpha\}|\{\eD_\beta\}}(\eM,\eN)^{\mathrm{op}}
        \xrightarrow{\simeq}
        \Fun_{\eB}(\eM,\eN)^{\mathrm{op}}.
    \end{equation}
    Folding also yields an equivalence
    \begin{equation}
        \notag
        \mathsf{Fold}_{\mathrm{multi}}:
        \Rep\bigl(A_{\eM,\eN}^{\mathrm{multi}}\bigr)
        \xrightarrow{\simeq}
        \Rep\bigl(\Tube({}_{\eB}\eM;{}_{\eB}\eN) \bigr).
    \end{equation}
    When the chosen graphical presentations are identified strictly, this functor is obtained by transporting representations along the corresponding algebra isomorphism. More generally, it is the representation equivalence supplied by the Morita context of the folded tube algebras.

    By Theorem~\ref{thm:Fun-Rep-equiv}, let $\Phi_{\partial}^{\eB}$ and $\Psi_{\partial}^{\eB}$ denote the boundary disk space functor associated with $\eB$ and its quasi-inverse, respectively:
    \begin{equation}
        \notag
        \begin{tikzcd}[column sep=1.7em]
            \Fun_{\eB}(\eM,\eN)^{\mathrm{op}}
            \arrow[r,shift left=0.7ex,"\Phi_{\partial}^{\eB}"]
            &
            \Rep\bigl(\Tube({}_{\eB}\eM;{}_{\eB}\eN)\bigr).
            \arrow[l,shift left=0.7ex,"\Psi_{\partial}^{\eB}"]
        \end{tikzcd}
    \end{equation}
    Filling a multimodule tube before folding gives the same linear map as folding first and filling the resulting boundary tube. Consequently, the two disk space constructions are related by a natural isomorphism
    \begin{equation}
        \mathsf{Fold}_{\mathrm{multi}}
        \circ
        \Phi_{\mathrm{multi}}
        \simeq
        \Phi_{\partial}^{\eB}
        \circ
        \mathsf{Fold}^{\mathrm{op}}.
        \label{eq:folding-compatible-with-multimodule-disk}
    \end{equation}

    Define
    \begin{equation*}
        \Psi_{\mathrm{multi}}:=\bigl(\mathsf{Fold}^{\mathrm{op}}\bigr)^{-1}\circ \Psi_{\partial}^{\eB}\circ\mathsf{Fold}_{\mathrm{multi}}.
    \end{equation*}
    This is a covariant functor
    \begin{equation*}
        \Psi_{\mathrm{multi}}:\Rep\bigl(A_{\eM,\eN}^{\mathrm{multi}}\bigr)\to\Fun_{\{\eC_\alpha\}|\{\eD_\beta\}} (\eM,\eN)^{\mathrm{op}}.
    \end{equation*}
    Using \eqref{eq:folding-compatible-with-multimodule-disk} and the quasi-inverse relations for $\Phi_{\partial}^{\eB}$ and $\Psi_{\partial}^{\eB}$, we obtain
    \begin{equation*}
        \Psi_{\mathrm{multi}}\circ\Phi_{\mathrm{multi}}
        \simeq\id_{\Fun_{\{\eC_\alpha\}|\{\eD_\beta\}} (\eM,\eN)^{\mathrm{op}}}.
    \end{equation*}
    Similarly, since $\mathsf{Fold}_{\mathrm{multi}}$ is an equivalence, we have
    \begin{equation*}
        \Phi_{\mathrm{multi}}\circ\Psi_{\mathrm{multi}}
        \simeq\id_{\Rep(A_{\eM,\eN}^{\mathrm{multi}})}.
    \end{equation*}
    Thus $\Psi_{\mathrm{multi}}$ is a quasi-inverse of $\Phi_{\mathrm{multi}}$, establishing \eqref{eq:multimodule-representation-equivalence}.
\end{proof}

\begin{corollary}
    \label{cor:multimodule-wall-excitations}
    Setting $\eN=\eM$ identifies the category of excitations on the multimodule wall with the representation category of its tube
    algebra:
    \begin{equation}
        \Rep\bigl(\Tube({}_{\{\eC_\alpha\}}\eM_{\{\eD_\beta\}})\bigr)\simeq
        \Fun_{\{\eC_\alpha\}|\{\eD_\beta\}}(\eM,\eM)^{\rm op}.
    \end{equation}
    Equipping $\Fun_{\{\eC_\alpha\}|\{\eD_\beta\}}(\eM,\eM)^{\rm op}$ with the monoidal structure induced by composition of multimodule endofunctors, and the representation category with the monoidal structure determined by the weak Hopf algebra coproduct, the above equivalence becomes a monoidal equivalence. 
\end{corollary}

\begin{corollary}
The category $\Rep(\Tube({}_{\{\eC_\alpha\}}\eM_{\{\eD_\beta\}};{}_{\{\eC_\alpha\}}\eN_{\{\eD_\beta\}}))$ is a bimodule category over the UFCs $\Rep(\Tube({}_{\{\eC_\alpha\}}\eM_{\{\eD_\beta\}}))$ and $\Rep(\Tube({}_{\{\eC_\alpha\}}\eN_{\{\eD_\beta\}}))$. Moreover, under the convention adopted in Remark~\ref{rmk:optubeConvention},  defining the subcategories
$\eP_{1,1}= \Rep(\Tube({}_{\{\eC_\alpha\}}\eM_{\{\eD_\beta\}}))$,
$\eP_{1,2}= \Rep(\Tube({}_{\{\eC_\alpha\}}\eM_{\{\eD_\beta\}};{}_{\{\eC_\alpha\}}\eN_{\{\eD_\beta\}}))$,
$\eP_{2,1}= \Rep(\Tube({}_{\{\eC_\alpha\}}\eN_{\{\eD_\beta\}};{}_{\{\eC_\alpha\}}\eM_{\{\eD_\beta\}}))$, and
$\eP_{2,2}= \Rep(\Tube({}_{\{\eC_\alpha\}}\eN_{\{\eD_\beta\}}))$,
then  $\eP=\bigoplus_{i,j=1}^{2}\eP_{i,j}$ forms a multifusion category.
\end{corollary}

\begin{remark}
    The number of bulk phases affects the number of strings and
    multimodule associators entering a basis tube, but it does not
    introduce additional representation-theoretic sectors beyond the
    multimodule functors.  Formula
    \eqref{eq:multimodule-representation-equivalence} therefore
    extends the boundary and ordinary domain wall classifications to
    an arbitrary finite collection of phases meeting along the same
    wall.  The cases
    \begin{equation}
        |I|+|J|=1
        \qquad\text{and}\qquad
        |I|+|J|=2
    \end{equation}
    recover, respectively, the boundary and the
    domain wall constructions.
\end{remark}

\subsection{$N$-tuple algebra of comodule tube algebras}
\label{sec:Ntuplealge}

The multimodule tube algebra retains information from every bulk phase incident on the wall.  Its lower order pieces are obtained by making some of those phases invisible to the tube, or equivalently, by requiring the corresponding bulk strings to carry tensor unit labels.  For an ordinary multimodule wall this procedure produces a hierarchy of weak Hopf subalgebras.  We now show that the same mechanism survives for defects between two multimodule walls, with weak Hopf algebras replaced by bicomodule tube algebras.

The general setup is the same: suppose ${}_{\{\eC_{\alpha}\}_{\alpha\in I}}\eM_{\{\eD_{\beta}\}_{\beta\in J}}$ and ${}_{\{\eC_{\alpha}\}_{\alpha\in I}}\eN_{\{\eD_{\beta}\}_{\beta\in J}}$ are two multimodule category over $\{\eC_{\alpha}\}_{\alpha\in I}|\{\eD_{\beta}\}_{\beta\in J}$, with $N=|I|+|J|$. 
For subsets $I'\subseteq I$ and $J'\subseteq J$, after forgetting the actions indexed by $(I\setminus I')\sqcup(J\setminus J')$, the categories $\eM$ and $\eN$ remain $\{\eC_{\alpha}\}_{\alpha\in I'}|\{\eD_{\beta}\}_{\beta\in J'}$-multimodule categories.  Thus the multimodule defect tube algebra 
\[
    \Tube({}_{\{\eC_{\alpha}\}_{I'}}\eM_{\{\eD_{\beta}\}_{J'}};{}_{\{\eC_{\alpha}\}_{I'}}\eN_{\{\eD_{\beta}\}_{J'}})
\]
is a bicomodule algebra over the two weak Hopf algebras
\[
    \Tube({}_{\{\eC_{\alpha}\}_{I'}}\eN_{\{\eD_{\beta}\}_{J'}})\quad \text{and}\quad \Tube({}_{\{\eC_{\alpha}\}_{I'}}\eM_{\{\eD_{\beta}\}_{J'}}). 
\]
We denote its coactions by $\beta_{I',J'}$ and $\rho_{I',J'}$.

Suppose that $I'\subseteq I''$ and $J'\subseteq J''$. 
Given a basis element 
\[
    X\in \Tube({}_{\{\eC_{\alpha}\}_{I'}}\eM_{\{\eD_{\beta}\}_{J'}};{}_{\{\eC_{\alpha}\}_{I'}}\eN_{\{\eD_{\beta}\}_{J'}}),
\]
insert a string labeled by $\mathbb{1}_{\alpha}\in\eC_{\alpha}$ for every
$\alpha\in I''\setminus I'$ and a string labeled by
$\mathbb{1}_{\beta}\in\eD_{\beta}$ for every
$\beta\in J''\setminus J'$.  The new vertices are decorated with the canonical unit constraints of the relevant multimodule category.
This operation defines a linear map
\begin{equation*}
    \iota_{I',J'}^{I'',J''}:
    \Tube({}_{\{\eC_{\alpha}\}_{I'}}\eM_{\{\eD_{\beta}\}_{J'}};{}_{\{\eC_{\alpha}\}_{I'}}\eN_{\{\eD_{\beta}\}_{J'}})\to \Tube({}_{\{\eC_{\alpha}\}_{I''}}\eM_{\{\eD_{\beta}\}_{J''}};{}_{\{\eC_{\alpha}\}_{I''}}\eN_{\{\eD_{\beta}\}_{J''}}).
\end{equation*}
The same prescription applied to an excitation tube gives weak Hopf algebra morphisms
\begin{equation*}
    \jmath_{I',J';\mathcal X}^{I'',J''}:
    \Tube({}_{\{\eC_{\alpha}\}_{I'}}\eN_{\{\eD_{\beta}\}_{J'}})
    \to 
    \Tube({}_{\{\eC_{\alpha}\}_{I'}}\eN_{\{\eD_{\beta}\}_{J'}}),
    \quad
    \mathcal X\in\{\eM,\eN\}.
\end{equation*}
Deleting the newly inserted unit strings recovers the original basis tube.  Hence both $\iota_{I',J'}^{I'',J''}$ and
$\jmath_{I',J';\mathcal X}^{I'',J''}$ are injective.

\begin{proposition}
    \label{prop:unit-insertion-bicomodule-embedding}
    The map $\iota_{I',J'}^{I'',J''}$ is a unital algebra embedding. Moreover, it intertwines the two coactions in the sense that
    \begin{align}
        \beta_{I'',J''} \circ \iota_{I',J'}^{I'',J''} & = \left(\jmath_{I',J';\eN}^{I'',J''} \otimes \iota_{I',J'}^{I'',J''}\right) \circ \beta_{I',J'}, \label{eq:coaction-compatible-embeddings-1}\\
        \rho_{I'',J''} \circ \iota_{I',J'}^{I'',J''} & = \left(\iota_{I',J'}^{I'',J''} \otimes \jmath_{I',J';\eM}^{I'',J''}\right) \circ \rho_{I',J'}. \label{eq:coaction-compatible-embeddings-2}
    \end{align}
    Consequently, the linear map $\iota_{I',J'}^{I'',J''}$ is an embedding of bicomodule algebras relative to the embedded domain wall tube algebras. 
\end{proposition}

\begin{proof}
    Multiplication in every algebra is defined by stacking compatible tube.  Tensor unit strings pass through this stacking operation by the unit axioms, and their unit coupons compose to unit coupons.  It follows that
    \begin{equation*}
        \iota_{I',J'}^{I'',J''}(XY) = \iota_{I',J'}^{I'',J''}(X) \iota_{I',J'}^{I'',J''}(Y).
    \end{equation*}
    Inserting unit strings also gives
    \begin{equation*}
        \iota_{I',J'}^{I'',J''} \bigl(1\bigr) = 1.
    \end{equation*}
    To check the first identity \eqref{eq:coaction-compatible-embeddings-1}, compute the left coaction by inserting an intermediate $\eN$-wall and resolving its simple labels.  An inactive bulk line remains labeled by its tensor unit throughout this operation, and every F-move involving that line reduces to a unit move.  One therefore obtains
    the same pair of tubes whether the coaction is performed before or after unit insertion.  This is precisely the first equality \eqref{eq:coaction-compatible-embeddings-1}.  Resolving an intermediate $\eM$-wall instead proves the second equality \eqref{eq:coaction-compatible-embeddings-2}.
\end{proof}

 For each $\alpha\in I$, the boundary defect tube algebra $\Tube_{\rm bd}(_{\eC_\alpha}\eM;{}_{\eC_\alpha}\eN)$ embeds via the embedding $\iota_{\{\alpha\},\emptyset}^{I,J}$ into the multimodule defect tube algebra, yielding a bicomodule subalgebra $\Tube_{L}(_{\eC_\alpha}\eM;{}_{\eC_\alpha}\eN)$:
\[
    \Tube_{\rm bd}({}_{\eC_\alpha}\eM;{}_{\eC_\alpha}\eN)\simeq \Tube_{L}({}_{\eC_\alpha}\eM;{}_{\eC_\alpha}\eN)  \xhookrightarrow{\hspace{0.4cm}} \Tube({}_{\{\eC_{\alpha}\}_{I}}\eM_{\{\eD_{\beta}\}_{J}};{}_{\{\eC_{\alpha}\}_{I}}\eN_{\{\eD_{\beta}\}_{J}}).
\]
Similarly for each $\beta\in J$, we obtain an embedding of bicomodule subalgebras 
\[
    \Tube_{\rm bd}(\eM_{\eD_\beta};\eN_{\eD_\beta})\simeq \Tube_{R}(\eM_{\eD_\beta};\eN_{\eD_\beta})  \xhookrightarrow{\hspace{0.4cm}} \Tube({}_{\{\eC_{\alpha}\}_{I}}\eM_{\{\eD_{\beta}\}_{J}};{}_{\{\eC_{\alpha}\}_{I}}\eN_{\{\eD_{\beta}\}_{J}}).
\]
By virtue of these embeddings, we are able to define an $N$-tuple generalization of the gluing map, analogous to \eqref{eq:gluing-map-tuple}, which is realized by taking the product of the embedded elements. The resulting comodule algebra is named as \emph{N-tuple algebra of comodule tube algebras} associated with $(I,J)$, and is denoted as 
\begin{equation*}
    (\bicrosslr_{\alpha\in I}\Tube_{L}({}_{\eC_\alpha}\eM;{}_{\eC_\alpha}\eN))\bicrosslr(\bicrosslr_{\beta\in J}\Tube_{R}(\eM_{\eD_\beta};\eN_{\eD_\beta})). 
\end{equation*}

\begin{theorem}[Hierarchy structure]
    \label{thm:hierarchy-multi-module-defect-tubes}
    For each $I'\subseteq I$ and $J'\subseteq J$, denote $r=|I'|+|J'|$ and $N=|I|+|J|$. Then the $r$-tuple algebra
    associated with $(I',J')$
    embeds into the $N$-tuple algebra associated with $(I,J)$ as bicomodule subalgebra. 
    Hence the $N$-tuple algebra contains a canonical $r$-tuple bicomodule subalgebra for every choice of $r$ incident bulk phases, where $0\leq r\leq N$. 
    In particular, the $N$-tuple algebra admits a hierarchical structure with respect to these bicomodule subalgebras.
\end{theorem}

\begin{proof}
    Proposition~\ref{prop:unit-insertion-bicomodule-embedding} gives an injective bicomodule algebra morphism for every inclusion $I'\subseteq I$ and $J'\subseteq J$. Let us write $S=(I',J')$, $T=(I'',J'')$  and define $S\subset T$ if $I'\subset I''$, $J'\subset J''$. If $S\subseteq T\subseteq U$, it is easy to see that the embeddings satisfy
    \[
        \iota_{T}^{U}\circ\iota_{S}^{T}= \iota_{S}^{U},\quad\iota_{S}^{S}=
        \id.
    \]
    Thus the embedding is independent of the order in which the missing bulk strings are added.  Therefore every chain of subsets gives the same embedded copy inside the $N$-tuple algebra.  
\end{proof}

%% file: Ch-Examples.tex
\section{Examples of comodule tube algebras}
\label{sec:exmp}
In this section, we give several examples of defect tube algebras. We will mainly focus on the boundary defect tube algebra, the bulk and domain wall case can be obtained similarly.

\subsection{Defect tube algebra for toric code}
\label{sec:ExpTubeToric}

In the Levin--Wen realization of the toric code, one may take the input fusion category to be $\eC=\Rep(\mathbb Z_2)$. 
We consider a gapped boundary of toric-code phase. The input fusion category for the bulk is $\eC=\Rep(\mathbb Z_2)$. Its irreducible objects are $\Irr(\eC)=\{\mathbb{1},\varepsilon\}$, where $\varepsilon$ is the nontrivial one-dimensional sign representation over $\mathbb{Z}_2$. The fusion rule is $\varepsilon\otimes \varepsilon = \mathbb{1}$. 
The category $\eM_s=\Rep(\mathbb{Z}_2)$, viewing as a $\Rep(\mathbb Z_2)$-module category with actions given by tensor product in $\Rep(\mathbb{Z}_2)$, describes the smooth boundary. 
The category $\eM_r=\Vect$ describes the rough boundary, whose action is provided by $a\triangleright x=F(a)\otimes_{\mathbb{C}} x$, with $F:\Rep(\mathbb Z_2)\to \Vect$ the forgetful functor. Clearly, $\mathbb{1}=\mathbb{C}$ is the unique irreducible object of $\Vect$. 

The boundary defect tube algebra $\mathbf{Tube}({}_\eC\eM_r,{}_\eC\eM_s)$ is spanned by 
\begin{equation}
    \begin{aligned}
    \begin{tikzpicture}[scale=0.65]
        \fill[gray!15] (0,1.5) arc[start angle=90, end angle=270, radius=1.5] -- (0,-0.5) arc[start angle=270, end angle=90, radius=0.5] -- cycle;
        \draw[dotted] (0,1.5) arc[start angle=90, end angle=270, radius=1.5]; 
        \draw[dotted] (0,0.5) arc[start angle=90, end angle=270, radius=0.5];
        \draw[red] (0,1) arc[start angle=90, end angle=270, radius=1]; 
        \draw[line width=1pt,cyan] (0,0.5)--(0,1.5);
        \draw[line width=1pt,violet] (0,-0.5)--(0,-1.5);
        \node[ line width=0.6pt, dashed, draw opacity=0.5] (a) at (-1.2,0){$\scriptstyle \mathbb{1}$};
        \node[ line width=0.6pt, dashed, draw opacity=0.5] (a) at (0.25,-1.3){$\scriptstyle \mathbb{1}$};
        \node[ line width=0.6pt, dashed, draw opacity=0.5] (a) at (0.25,-0.6){$\scriptstyle \mathbb{1}$};
        \node[ line width=0.6pt, dashed, draw opacity=0.5] (a) at (0.25,0.6){$\scriptstyle \mathbb{1}$};
        \node[ line width=0.6pt, dashed, draw opacity=0.5] (a) at (0.25,1.3){$\scriptstyle \mathbb{1}$};
    \end{tikzpicture}
    \end{aligned}\;, \quad 
    \begin{aligned}
    \begin{tikzpicture}[scale=0.65]
        \fill[gray!15] (0,1.5) arc[start angle=90, end angle=270, radius=1.5] -- (0,-0.5) arc[start angle=270, end angle=90, radius=0.5] -- cycle;
        \draw[dotted] (0,1.5) arc[start angle=90, end angle=270, radius=1.5]; 
        \draw[dotted] (0,0.5) arc[start angle=90, end angle=270, radius=0.5];
        \draw[red] (0,1) arc[start angle=90, end angle=270, radius=1]; 
        \draw[line width=1pt,cyan] (0,0.5)--(0,1.5);
        \draw[line width=1pt,violet] (0,-0.5)--(0,-1.5);
        \node[ line width=0.6pt, dashed, draw opacity=0.5] (a) at (-1.2,0){$\scriptstyle \mathbb{1}$};
        \node[ line width=0.6pt, dashed, draw opacity=0.5] (a) at (0.25,-1.3){$\scriptstyle \mathbb{1}$};
        \node[ line width=0.6pt, dashed, draw opacity=0.5] (a) at (0.25,-0.6){$\scriptstyle \mathbb{1}$};
        \node[ line width=0.6pt, dashed, draw opacity=0.5] (a) at (0.25,0.6){$\scriptstyle \varepsilon$};
        \node[ line width=0.6pt, dashed, draw opacity=0.5] (a) at (0.25,1.3){$\scriptstyle \varepsilon$};
    \end{tikzpicture}
    \end{aligned}\;, \quad 
    \begin{aligned}
    \begin{tikzpicture}[scale=0.65]
        \fill[gray!15] (0,1.5) arc[start angle=90, end angle=270, radius=1.5] -- (0,-0.5) arc[start angle=270, end angle=90, radius=0.5] -- cycle;
        \draw[dotted] (0,1.5) arc[start angle=90, end angle=270, radius=1.5]; 
        \draw[dotted] (0,0.5) arc[start angle=90, end angle=270, radius=0.5];
        \draw[red] (0,1) arc[start angle=90, end angle=270, radius=1]; 
        \draw[line width=1pt,cyan] (0,0.5)--(0,1.5);
        \draw[line width=1pt,violet] (0,-0.5)--(0,-1.5);
        \node[ line width=0.6pt, dashed, draw opacity=0.5] (a) at (-1.2,0){$\scriptstyle \varepsilon$};
        \node[ line width=0.6pt, dashed, draw opacity=0.5] (a) at (0.25,-1.3){$\scriptstyle \mathbb{1}$};
        \node[ line width=0.6pt, dashed, draw opacity=0.5] (a) at (0.25,-0.6){$\scriptstyle \mathbb{1}$};
        \node[ line width=0.6pt, dashed, draw opacity=0.5] (a) at (0.25,0.6){$\scriptstyle \mathbb{1}$};
        \node[ line width=0.6pt, dashed, draw opacity=0.5] (a) at (0.25,1.3){$\scriptstyle \varepsilon$};
    \end{tikzpicture}
    \end{aligned}\;, \quad 
    \begin{aligned}
    \begin{tikzpicture}[scale=0.65]
        \fill[gray!15] (0,1.5) arc[start angle=90, end angle=270, radius=1.5] -- (0,-0.5) arc[start angle=270, end angle=90, radius=0.5] -- cycle;
        \draw[dotted] (0,1.5) arc[start angle=90, end angle=270, radius=1.5]; 
        \draw[dotted] (0,0.5) arc[start angle=90, end angle=270, radius=0.5];
        \draw[red] (0,1) arc[start angle=90, end angle=270, radius=1]; 
        \draw[line width=1pt,cyan] (0,0.5)--(0,1.5);
        \draw[line width=1pt,violet] (0,-0.5)--(0,-1.5);
        \node[ line width=0.6pt, dashed, draw opacity=0.5] (a) at (-1.2,0){$\scriptstyle \varepsilon$};
        \node[ line width=0.6pt, dashed, draw opacity=0.5] (a) at (0.25,-1.3){$\scriptstyle \mathbb{1}$};
        \node[ line width=0.6pt, dashed, draw opacity=0.5] (a) at (0.25,-0.6){$\scriptstyle \mathbb{1}$};
        \node[ line width=0.6pt, dashed, draw opacity=0.5] (a) at (0.25,0.6){$\scriptstyle \varepsilon$};
        \node[ line width=0.6pt, dashed, draw opacity=0.5] (a) at (0.25,1.3){$\scriptstyle \mathbb{1}$};
    \end{tikzpicture}
    \end{aligned}\;. 
\end{equation}
For later use, we denote the four basis elements above by
\begin{equation}
    T_{a,u}:=
    \begin{tikzpicture}[baseline=-0.45ex,scale=0.55]
        \fill[gray!15] (0,1.5) arc[start angle=90, end angle=270, radius=1.5] -- (0,-0.5) arc[start angle=270, end angle=90, radius=0.5] -- cycle;
        \draw[dotted] (0,1.5) arc[start angle=90, end angle=270, radius=1.5]; 
        \draw[dotted] (0,0.5) arc[start angle=90, end angle=270, radius=0.5];
        \draw[red] (0,1) arc[start angle=90, end angle=270, radius=1]; 
        \draw[line width=1pt,cyan] (0,0.5)--(0,1.5);
        \draw[line width=1pt,violet] (0,-0.5)--(0,-1.5);
        \node at (-1.2,0){$\scriptstyle a$};
        \node at (0.25,-1.3){$\scriptstyle \mathbb{1}$};
        \node at (0.25,-0.6){$\scriptstyle \mathbb{1}$};
        \node at (0.25,0.6){$\scriptstyle u$};
        \node at (0.4,1.3){$\scriptstyle au$};
    \end{tikzpicture},
    \quad a,u\in \{\mathbb{1},\varepsilon\}.
\end{equation}
Here $au$ denotes the tensor product in $\Rep(\mathbb Z_2)$. 
The multiplication is obtained from the gluing rule in the boundary defect tube algebra.
Since all relevant Hom-spaces are one-dimensional and all F-symbols are trivial, one gets
\begin{equation}
    T_{a,u}\,T_{a',u'}
    =
    \delta_{u,a'u'}\,T_{aa',u'}.
\end{equation}
Equivalently, the unit is
\begin{equation}
    1_{\mathbf{Tube}({}_\eC\eM_r;{}_\eC\eM_s)}
    =
    T_{\mathbb{1},\mathbb{1}}+T_{\mathbb{1},\varepsilon}.
\end{equation}
With the ordered basis $T_{\mathbb{1},\mathbb{1}}, T_{\mathbb{1},\varepsilon}, T_{\varepsilon,\mathbb{1}}, T_{\varepsilon,\varepsilon},$,
the multiplication table is
\begin{equation}
\begin{array}{c|cccc}
     & T_{\mathbb{1},\mathbb{1}} & T_{\mathbb{1},\varepsilon} 
     & T_{\varepsilon,\mathbb{1}} & T_{\varepsilon,\varepsilon} \\ \hline
T_{\mathbb{1},\mathbb{1}} 
     & T_{\mathbb{1},\mathbb{1}} & 0 & 0 & T_{\varepsilon,\varepsilon} \\
T_{\mathbb{1},\varepsilon} 
     & 0 & T_{\mathbb{1},\varepsilon} & T_{\varepsilon,\mathbb{1}} & 0 \\
T_{\varepsilon,\mathbb{1}} 
     & T_{\varepsilon,\mathbb{1}} & 0 & 0 & T_{\mathbb{1},\varepsilon} \\
T_{\varepsilon,\varepsilon} 
     & 0 & T_{\varepsilon,\varepsilon} & T_{\mathbb{1},\mathbb{1}} & 0 
\end{array}
\end{equation}
In particular, if we identify $T_{a,u}\leftrightarrow E_{au,u}$, 
where $E_{x,y}$ denotes the matrix unit indexed by $x,y\in\{\mathbb{1},\varepsilon\}$,
then $\Tube({}_\eC\eM_r;{}_\eC\eM_s)$ can be identified with the algebra $M_2(\mathbb{C})$ of $2\times 2$ matrices.

Next we describe the bicomodule structure explicitly. First consider the rough boundary
tube algebra $\mathbf{Tube}({}_\eC\eM_r)$. Since $\eM_r=\Vect$ has only one simple object, it has
basis
\begin{equation}
    R_a:=\begin{tikzpicture}[baseline=-0.45ex,scale=0.55]
        \fill[gray!15] (0,1.5) arc[start angle=90, end angle=270, radius=1.5] -- (0,-0.5) arc[start angle=270, end angle=90, radius=0.5] -- cycle;
        \draw[dotted] (0,1.5) arc[start angle=90, end angle=270, radius=1.5]; 
        \draw[dotted] (0,0.5) arc[start angle=90, end angle=270, radius=0.5];
        \draw[red] (0,1) arc[start angle=90, end angle=270, radius=1]; 
        \draw[line width=1pt,violet] (0,0.5)--(0,1.5);
        \draw[line width=1pt,violet] (0,-0.5)--(0,-1.5);
        \node at (-1.2,0){$\scriptstyle a$};
        \node at (0.25,-1.3){$\scriptstyle \mathbb{1}$};
        \node at (0.25,-0.6){$\scriptstyle \mathbb{1}$};
        \node at (0.25,0.6){$\scriptstyle \mathbb{1}$};
        \node at (0.25,1.3){$\scriptstyle \mathbb{1}$};
    \end{tikzpicture},\quad a\in\{\mathbb{1},\varepsilon\}.
\end{equation}
Its product, coproduct and antipode are
\begin{equation}
    R_aR_{a'}=R_{aa'},\quad
    \Delta_r(R_a)=R_a\otimes R_a,\quad
    \varepsilon_r(R_a)=1, \quad S(R_a) = R_a. 
\end{equation}
For the smooth boundary tube algebra $\mathbf{Tube}({}_\eC\eM_s)$, we use the basis
\begin{equation}
    S_{a;i,u}:=\begin{tikzpicture}[baseline=-0.45ex,scale=0.55]
        \fill[gray!15] (0,1.5) arc[start angle=90, end angle=270, radius=1.5] -- (0,-0.5) arc[start angle=270, end angle=90, radius=0.5] -- cycle;
        \draw[dotted] (0,1.5) arc[start angle=90, end angle=270, radius=1.5]; 
        \draw[dotted] (0,0.5) arc[start angle=90, end angle=270, radius=0.5];
        \draw[red] (0,1) arc[start angle=90, end angle=270, radius=1]; 
        \draw[line width=1pt,cyan] (0,0.5)--(0,1.5);
        \draw[line width=1pt,cyan] (0,-0.5)--(0,-1.5);
        \node at (-1.2,0){$\scriptstyle a$};
        \node at (0.25,-1.3){$\scriptstyle i$};
        \node at (0.4,-0.6){$\scriptstyle ai$};
        \node at (0.25,0.6){$\scriptstyle u$};
        \node at (0.4,1.3){$\scriptstyle au$};
    \end{tikzpicture},
    \quad a,i,u\in\{\mathbb{1},\varepsilon\}.
\end{equation}
Again, all multiplicity spaces are one-dimensional.

The left coaction
\begin{equation}
    \beta:\mathbf{Tube}({}_\eC\eM_r;{}_\eC\eM_s)
    \to
    \mathbf{Tube}({}_\eC\eM_s)\otimes \mathbf{Tube}({}_\eC\eM_r;{}_\eC\eM_s)
\end{equation}
is given by inserting a smooth boundary tube on the smooth side. Specializing the general
formula to the present case gives
\begin{equation}
    \beta(T_{a,u})
    =
    \sum_{i\in\{\mathbb{1},\varepsilon\}}
    S_{a;i,u}\otimes T_{a,i}.
\end{equation}
Explicitly,
\begin{equation}
\begin{aligned}
    \beta(T_{\mathbb{1},\mathbb{1}})
    &=
    S_{\mathbb{1};\mathbb{1},\mathbb{1}}\otimes T_{\mathbb{1},\mathbb{1}}
    +
    S_{\mathbb{1};\varepsilon,\mathbb{1}}\otimes T_{\mathbb{1},\varepsilon},\\
    \beta(T_{\mathbb{1},\varepsilon})
    &=
    S_{\mathbb{1};\mathbb{1},\varepsilon}\otimes T_{\mathbb{1},\mathbb{1}}
    +
    S_{\mathbb{1};\varepsilon,\varepsilon}\otimes T_{\mathbb{1},\varepsilon},\\
    \beta(T_{\varepsilon,\mathbb{1}})
    &=
    S_{\varepsilon;\mathbb{1},\mathbb{1}}\otimes T_{\varepsilon,\mathbb{1}}
    +
    S_{\varepsilon;\varepsilon,\mathbb{1}}\otimes T_{\varepsilon,\varepsilon},\\
    \beta(T_{\varepsilon,\varepsilon})
    &=
    S_{\varepsilon;\mathbb{1},\varepsilon}\otimes T_{\varepsilon,\mathbb{1}}
    +
    S_{\varepsilon;\varepsilon,\varepsilon}\otimes T_{\varepsilon,\varepsilon}.
\end{aligned}
\end{equation}

The right coaction
\begin{equation}
    \rho:\mathbf{Tube}({}_\eC\eM_r;{}_\eC\eM_s)
    \to
    \mathbf{Tube}({}_\eC\eM_r;{}_\eC\eM_s)\otimes \mathbf{Tube}({}_\eC\eM_r)
\end{equation}
is even simpler, because the rough boundary has only one simple object. One obtains
\begin{equation}
    \rho(T_{a,u})=T_{a,u}\otimes R_a.
\end{equation}
Therefore,
\begin{equation}
\begin{aligned}
    \rho(T_{\mathbb{1},\mathbb{1}})
    =T_{\mathbb{1},\mathbb{1}}\otimes R_{\mathbb{1}},\;
    \rho(T_{\mathbb{1},\varepsilon})
    =T_{\mathbb{1},\varepsilon}\otimes R_{\mathbb{1}},\;
    \rho(T_{\varepsilon,\mathbb{1}})
    =T_{\varepsilon,\mathbb{1}}\otimes R_{\varepsilon},\;
    \rho(T_{\varepsilon,\varepsilon})
    =T_{\varepsilon,\varepsilon}\otimes R_{\varepsilon}.
\end{aligned}
\end{equation}

It is immediate from the formula
\begin{equation}
    T_{a,u}T_{a',u'}=\delta_{u,a'u'}T_{aa',u'}
\end{equation}
that the right coaction is multiplicative:
\begin{equation}
    \rho(T_{a,u}T_{a',u'})
    =
    \rho(T_{a,u})\rho(T_{a',u'}).
\end{equation}
The same conclusion holds for the left coaction:
\begin{equation}
    \beta(T_{a,u}T_{a',u'})
    =
    \beta(T_{a,u})\beta(T_{a',u'}).
\end{equation}
Thus $\mathbf{Tube}({}_\eC\eM_r;{}_\eC\eM_s)$ is a
$\mathbf{Tube}({}_\eC\eM_s)|\mathbf{Tube}({}_\eC\eM_r)$ bicomodule algebra.

Finally, since
\begin{equation}
    \mathbf{Tube}({}_\eC\eM_r;{}_\eC\eM_s)\cong M_2(\mathbb C),
\end{equation}
it has a unique irreducible representation. This agrees with the categorical description
\begin{equation}
    \Rep\bigl(\mathbf{Tube}({}_\eC\eM_r;{}_\eC\eM_s)\bigr)
    \simeq
    \Fun_\eC(\eM_r,\eM_s),
\end{equation}
which contains a single simple boundary defect, namely the rough-to-smooth boundary
defect. Physically this is the familiar toric-code boundary-changing defect between the
rough and smooth boundaries, whose quantum dimension is $\sqrt{2}$.

\subsection{Defect tube algebra for finite group}
 
Associated to a finite group $G$ is a fusion category $\underline{G}=\Vect_G$. Its simple objects are labeled by elements of $G$, and we write the simple object corresponding to $g\in G$ again as $g$. The morphism spaces are given by $\Hom(g,h) = \delta_{g,h}\mathbb{C}$, and the tensor product is determined by the group multiplication $g\otimes h=gh$. The tensor unit is the identity element $e\in G$, and the dual of $g$ is $g^{-1}$. 

An indecomposable left $\underline{G}$-module category is determined by a pair $(L,\psi)$ where $L\leq G$ is a subgroup and $\psi\in Z^2(L,\mathbb{C}^\times)$ is a 2-cocycle on $L$ \cite[Example 7.4.10]{etingof2016tensor}. We denote it as $\eM(L,\psi)$. The set of simple objects is indexed by the set of cosets $G/L$ with the action given by $g \triangleright g'L = gg'L$. The component $\psi$ determines the associativity constraint of the module structure. To illustrate this, we fix a set of representatives for $G/L = \{\overline{g_1},\cdots,\overline{g_r}\}$, with $\overline{g_i}=g_iL$. For each $g\in G$ and $g_i$, the coset $\overline{gg_i} = \overline{g_j}$ for some representative $g_j$. Then there is a unique $\ell_{g,\overline{g_i}}\in L$ such that $gg_i=g_j\ell_{g,\overline{g_i}}$. 
Then the associator $\alpha_{g,h,\overline{g_i}}:(g\triangleright h)\triangleright \overline{g_i}\to (g\triangleright (h\triangleright \overline{g_i}))$ is given by the scalar $\alpha_{g,g,\overline{g_i}}=\psi(\ell_{g,\overline{hg_i}},\ell_{h,\overline{g_i}})$.

Let us consider the boundary defect tube algebra between boundaries determined by $\eM(L,\psi)$ and $\eN(K,\phi)$. 
Suppose $G/L=\{\overline{g_1},\cdots,\overline{g_r}\}$, $G/K=\{\overline{h_1},\cdots,\overline{h_s}\}$, where $\overline{g_i}=g_iL,\overline{h_j}=h_jK$ with representatives $g_i,h_j\in G$ respectively.  
Then the corresponding boundary defect tube algebra is given as 
\begin{equation}
\mathbf{Tube}({}_{\underline{G}}\eM(L,\psi);{}_{\underline{G}}\eN(K,\phi)) 
=\operatorname{span}
\left\{    \begin{aligned}

    \end{aligned} \; :\; g\in G, 1\leq i\leq r, 1\leq j\leq s
    \right\}\,,
\end{equation}
which has dimension $|G||G/K||G/L|=\frac{|G|^3}{|K||L|}$. 
The weak Hopf algebra structure on the boundary tube algebra $\mathbf{Tube}(\eM(L,\psi))$ is 
\begin{align}
    1 & = \sum_{g\in G}\sum_{i=1}^r\sum_{j=1}^r\; \begin{aligned}
    %
    \end{aligned} \,. 
\end{align}
The weak Hopf algebra structure on $\mathbf{Tube}(\eN(K,\phi))$ can be explicitly written down similarly. 

As an algebra, the boundary defect tube algebra $\mathbf{Tube}({}_{\underline{G}}\eM(L,\psi);{}_{\underline{G}}\eN(K,\phi))$ has unit 
\begin{equation}
    1 = \sum_{g\in G}\sum_{i=1}^r\sum_{j=1}^s\; \begin{aligned}
    %
    \end{aligned}\;. 
\end{equation}
Here $\kappa_{h,\overline{h_l}} \in K$ is the unique element satisfying $hh_l=h_{n}\kappa_{h,\overline{h_l}}$ with $h_n$ the chosen representative for $hh_lK$.
Now it is ready to write down its coactions with respect to $\mathbf{Tube}({}_{\underline{G}}\eM(L,\psi))$ and $\mathbf{Tube}({}_{\underline{G}}\eN(K,\phi))$. To be precise, the left $\mathbf{Tube}({}_{\underline{G}}\eN(K,\phi))$-coaction is
\begin{equation}
    \beta\left(\begin{aligned}
    %
    \end{aligned}\;,
\end{equation}
and the right $\mathbf{Tube}({}_{\underline{G}}\eM(L,\psi))$-coaction is 
\begin{equation}
    \rho\left(\begin{aligned}
    %
    \end{aligned}\;. 
\end{equation}

For the domain wall case, we can use the folding trick to turn a $\underline{G}|\underline{G}$ domain wall into a $\underline{G\times G^{\rm op}}$ boundary. Here $G^{\rm op}$ is the opposite group. Thus the domain wall case is obtained from the boundary case by replacing $G$ with $G\times G^{\rm op}$. In particular, indecomposable domain walls are parametrized by pairs $(L,\psi)$, where $L\subset G\times G^{\rm op}$ and $\psi\in Z^2(L,\mathbb C^\times)$, and the simple objects are labeled by cosets $(G\times G^{\rm op})/L$.

%% file: Ch-Appendix.tex
\section{Review of string-net model}
\label{sec:SNreview}

Let us briefly review some basic aspects of the Levin-Wen string-net model, for more details, see, e.g., \cite{Levin2005,kirillov2011stringnet,lan2014topological,Lin2021generalized,jia2024weakTube,Kitaev2012boundary,hu2018boundary}. Let $\eD$ be the input UFC (or unitary multifusion category). In the diagrammatic representation, each edge is drawn as an upward-oriented line, analogous to anyon worldlines. For the fusion and splitting operations, one defines the vector spaces $V_c^{ab}=\Hom_{\EuScript{D}}(c,a\otimes b)$ and $V_{ab}^c=\Hom_{\EuScript{D}}(a\otimes b,c)$, which are spanned by the corresponding fusion and splitting vertex vectors. 

To normalize these local vertex vectors, we introduce a factor $Y_{c}^{ab}$ (which is usually chosen as $\sqrt{\frac{d_ad_b}{d_c}}$) in the following manner:
\begin{align} 
    \notag
	(Y_{c}^{ab})^{-1/2}
	\begin{aligned}
		\begin{tikzpicture}
			\draw[-latex,line width=.6pt,black] (0,-0.5) -- (0,-.1);
			\draw[line width=.6pt,black] (0,-0.4) -- (0,0);
			\draw[-latex,line width=.6pt,black](0,0)--(-0.3,0.3);
			\draw[line width=.6pt,black](-0.1,0.1)--(-0.4,0.4);
			\draw[-latex,line width=.6pt,black](0,0)--(0.3,0.3);
			\draw[line width=.6pt,black](0.1,0.1)--(0.4,0.4);
			\node[ line width=0.6pt, dashed, draw opacity=0.5] (a) at (-0.5,0.6){$a$};
			\node[ line width=0.6pt, dashed, draw opacity=0.5] (a) at (0.5,0.6){$b$};
			\node[ line width=0.6pt, dashed, draw opacity=0.5] (a) at (0,-0.7){$c$};
			\node[ line width=0.6pt, dashed, draw opacity=0.5] (a) at (0.3,-0.1){$\alpha$};
		\end{tikzpicture}
	\end{aligned}
	&=|c\to a,b;\alpha\rangle, \quad \quad \quad 
	 (Y_{c}^{ab})^{-1/2}
	\begin{aligned}
		\begin{tikzpicture}
			\draw[-latex,line width=.6pt,black] (0.4,-0.4) -- (0.1,-0.1);
			\draw[line width=.6pt,black] (0.3,-0.3) -- (0,0);
			\draw[-latex,line width=.6pt,black] (-0.4,-0.4) -- (-0.1,-0.1);
			\draw[line width=.6pt,black] (-0.3,-0.3) -- (0,0);
			\draw[-latex,line width=.6pt,black](0,0)--(0,0.3);
			\draw[line width=.6pt,black](0,0.1)--(0,0.5);
			\node[ line width=0.6pt, dashed, draw opacity=0.5] (a) at (0.4,-0.6){$b$};
			\node[ line width=0.6pt, dashed, draw opacity=0.5] (a) at (0,0.7){$c$};
			\node[ line width=0.6pt, dashed, draw opacity=0.5] (a) at (-0.4,-0.6){$a$};
			\node[ line width=0.6pt, dashed, draw opacity=0.5] (a) at (0.3,0.1){$\beta$};
		\end{tikzpicture}
	\end{aligned}
	=\langle a,b\to c;\beta|. \label{eq:vertexvec2}
\end{align}
The dimension $\dim V^{ab}_{c}=\dim V_{ab}^{c} =N_{ab}^c$ is called fusion multiplicity, and quantum dimension $d_a$ is the positive eigenvalue of fusion multiplicity matrix.

There are three types of basic topological moves that are crucial for constructing generalized tube algebras and for building the string-net lattice model:

\begin{enumerate}
    \item \emph{Loop move.}
There is a pairing between $V_{c}^{ab}$ and $V_{ab}^c$ given by 
\begin{equation}\label{eq:inner}
	\langle c\to a,b;\alpha | a,b\to c';\beta \rangle = \delta_{c,c'}\delta_{\alpha,\beta}.
\end{equation}
Diagrammatically, this is characterized by the loop move
\begin{equation} 
    \notag
\text{loop move:}\quad	\begin{aligned}
		\begin{tikzpicture}
			\draw[-latex,line width=.6pt,black](0,0)--(0,0.4);
			\draw[line width=.6pt,black](0,0.1)--(0,0.5);
			\draw[line width=.6pt,black](0,-0.3) circle (0.3);
			\draw[-latex,line width=.6pt,black](-0.3,-0.3)--(-.3,-.2);
			\draw[-latex,line width=.6pt,black](0.3,-0.3)--(.3,-.2);
			\draw[-latex,line width=.6pt,black](0,-1.1)--(0,-0.7);
			\draw[line width=.6pt,black](0,-1.1)--(0,-0.6);
			\node[ line width=0.6pt, dashed, draw opacity=0.5] (a) at (0.6,-0.3){$b$};
			\node[ line width=0.6pt, dashed, draw opacity=0.5] (a) at (0,0.75){$c'$};
			\node[ line width=0.6pt, dashed, draw opacity=0.5] (a) at (0,-1.3){$c$};
			\node[ line width=0.6pt, dashed, draw opacity=0.5] (a) at (-0.6,-0.3){$a$};
			\node[ line width=0.6pt, dashed, draw opacity=0.5] (a) at (0.3,0.2){$\beta$};
			\node[ line width=0.6pt, dashed, draw opacity=0.5] (a) at (-0.3,-0.8){$\alpha$};
		\end{tikzpicture}
	\end{aligned}
	=
	\delta_{c,c'}\delta_{\alpha,\beta} {Y_c^{ab}}
	\begin{aligned}
		\begin{tikzpicture}
			\draw[line width=.6pt,black](0,-1.1)--(0,0.5);
			\draw[-latex,line width=.6pt,black](0,-0.3)--(0,0);
			\node[ line width=0.6pt, dashed, draw opacity=0.5] (a) at (0,0.7){$c$};
		\end{tikzpicture}
	\end{aligned}. \label{eq:loopev}
\end{equation}
The left-hand side of Eq.~\eqref{eq:inner} is to be read from right to left, while the diagrams are to be read from bottom to top.

\item  \emph{Parallel move.}
The completeness relation for the fusion basis takes the form
\begin{equation}
    \notag
	\mathds{I}=\sum_{c,\alpha} |c\to a,b;\alpha\rangle \langle a,b \to c;\alpha|,
\end{equation}
where $\mathds{I}$ is identity over the linear space $\Hom_{\eD}(a\otimes b,a\otimes b)$.
Diagrammatically, this corresponds to the parallel move
 \begin{equation} 
     \notag
	\text{parallel move:}\quad\begin{aligned}
		\begin{tikzpicture}
			\draw[line width=.6pt,black](0,-1.1)--(0,0.5);
			\draw[-latex,line width=.6pt,black](0,-0.3)--(0,0);
			\node[ line width=0.6pt, dashed, draw opacity=0.5] (a) at (0,0.7){$a$};
			\draw[line width=.6pt,black](0.6,-1.1)--(0.6,0.5);
			\draw[-latex,line width=.6pt,black](0.6,-0.3)--(0.6,0);
			\node[ line width=0.6pt, dashed, draw opacity=0.5] (a) at (0.6,0.7){$b$};
		\end{tikzpicture}
	\end{aligned}
	=
	\sum_{c,\alpha}\frac{1}{Y_c^{ab}}
	\begin{aligned}
		\begin{tikzpicture}
			\draw[-latex,line width=.6pt,black] (0.4,-0.4) -- (0.1,-0.1);
			\draw[line width=.6pt,black] (0.3,-0.3) -- (0,0);
			\draw[-latex,line width=.6pt,black] (-0.4,-0.4) -- (-0.1,-0.1);
			\draw[line width=.6pt,black] (-0.3,-0.3) -- (0,0);
			\draw[-latex,line width=.6pt,black](0,0)--(0,0.45);
			\draw[line width=.6pt,black](0,0.1)--(0,0.7);
			\draw[line width=.6pt,black](-0.4,1.1)--(0,0.7);
			\draw[line width=.6pt,black](0.4,1.1)--(0,0.7);
			\draw[-latex,line width=.6pt,black](0,0.7)--(0.3,1);
			\draw[line width=.6pt,black](0.4,1.1)--(0,0.7);
			\draw[-latex,line width=.6pt,black](0,0.7)--(-0.3,1);
			\node[ line width=0.6pt, dashed, draw opacity=0.5] (a) at (0.4,-0.6){$b$};
			\node[ line width=0.6pt, dashed, draw opacity=0.5] (a) at (-0.4,-0.6){$a$};
			\node[ line width=0.6pt, dashed, draw opacity=0.5] (a) at (0.3,0.1){$\alpha$};
			\node[ line width=0.6pt, dashed, draw opacity=0.5] (a) at (0.3,0.6){$\alpha$};
			\node[ line width=0.6pt, dashed, draw opacity=0.5] (a) at (-0.2,0.3){$c$};
			\node[ line width=0.6pt, dashed, draw opacity=0.5] (a) at (0.4,1.3){$b$};
			\node[ line width=0.6pt, dashed, draw opacity=0.5] (a) at (-0.4,1.3){$a$};
		\end{tikzpicture}
	\end{aligned}.
\end{equation}

\item \emph{F-move.}
The associativity isomorphism $a: (i\otimes j)\otimes k \to i \otimes (j\otimes k)$ induces a linear map $F_{l}^{ijk}:\Hom_{\eD}(l,(i\otimes j)\otimes k)\to \Hom_{\eD}(l,i\otimes (j\otimes k))$. Similarly, $a^{\dagger}$ induces $(F_{l}^{ijk})^{\dagger}$. 
Diagrammatically, these maps are realized as F-moves:
\begin{gather*}
    \begin{aligned}
		\begin{tikzpicture}[scale=0.95]
			\draw[-latex,line width=.6pt,black] (0,-0.5) -- (0,-.1);
			\draw[line width=.6pt,black] (0,-0.4) -- (0,0);
			\draw[-latex,line width=.6pt,black](0,0)--(-0.7,0.7);
			\draw[line width=.6pt,black](0,0)--(-0.8,0.8);
			\draw[-latex,line width=.6pt,black](0.4,0.4)--(0.1,0.7);
			\draw[line width=.6pt,black](0.4,0.4)--(0,0.8);
			\draw[-latex,line width=.6pt,black](0,0)--(0.7,0.7);
			\draw[-latex,line width=.6pt,black](0,0)--(0.3,0.3);
			\draw[line width=.6pt,black](0.1,0.1)--(0.8,0.8);
			\node[ line width=0.6pt, dashed, draw opacity=0.5] (a) at (-0.8,1){$\scriptstyle i$};
			\node[ line width=0.6pt, dashed, draw opacity=0.5] (a) at (0,1){$\scriptstyle j$};
			\node[ line width=0.6pt, dashed, draw opacity=0.5] (a) at (0.8,1){$\scriptstyle k$};
			\node[ line width=0.6pt, dashed, draw opacity=0.5] (a) at (0,-0.7){$\scriptstyle l$};
			\node[ line width=0.6pt, dashed, draw opacity=0.5] (a) at (0.3,-0.1){$\scriptstyle \mu$};
			\node[ line width=0.6pt, dashed, draw opacity=0.5] (a) at (0,0.3){$\scriptstyle n$};
				\node[ line width=0.6pt, dashed, draw opacity=0.5] (a) at (0.6,0.3){$\scriptstyle \nu$};
		\end{tikzpicture}
	\end{aligned}
	 =\sum_{m,\alpha,\beta} [F^{ijk}_{l}]_{n\mu\nu}^{m\alpha\beta}
	\begin{aligned}
        \begin{tikzpicture}[scale=0.95]
			\draw[-latex,line width=.6pt,black] (0,-0.5) -- (0,-.1);
			\draw[line width=.6pt,black] (0,-0.4) -- (0,0);
			\draw[-latex,line width=.6pt,black](0,0)--(-0.3,0.3);
			\draw[line width=.6pt,black](-0.1,0.1)--(-0.4,0.4);
			\draw[-latex,line width=.6pt,black](-0.4,0.4)--(-0.7,0.7);
			\draw[line width=.6pt,black](-0.4,0.4)--(-0.8,0.8);
			\draw[-latex,line width=.6pt,black](-0.4,0.4)--(-0.1,0.7);
			\draw[line width=.6pt,black](-0.4,0.4)--(-0,0.8);
			\draw[-latex,line width=.6pt,black](0,0)--(0.7,0.7);
			\draw[line width=.6pt,black](0.1,0.1)--(0.8,0.8);
			\node[ line width=0.6pt, dashed, draw opacity=0.5] (a) at (-0.8,1){$\scriptstyle i$};
			\node[ line width=0.6pt, dashed, draw opacity=0.5] (a) at (0,1){$\scriptstyle j$};
			\node[ line width=0.6pt, dashed, draw opacity=0.5] (a) at (0.8,1){$\scriptstyle k$};
			\node[ line width=0.6pt, dashed, draw opacity=0.5] (a) at (0,-0.7){$\scriptstyle l$};
			\node[ line width=0.6pt, dashed, draw opacity=0.5] (a) at (0.3,-0.1){$\scriptstyle \alpha$};
			\node[ line width=0.6pt, dashed, draw opacity=0.5] (a) at (-0.7,0.3){$\scriptstyle \beta$};
			\node[ line width=0.6pt, dashed, draw opacity=0.5] (a) at (-0.4,0){$\scriptstyle m$};
		\end{tikzpicture}
	\end{aligned},\quad 
	\begin{aligned}
		\begin{tikzpicture}[scale=0.95]
			\draw[-latex,line width=.6pt,black] (0,-0.5) -- (0,-.1);
			\draw[line width=.6pt,black] (0,-0.4) -- (0,0);
			\draw[-latex,line width=.6pt,black](0,0)--(-0.3,0.3);
			\draw[line width=.6pt,black](-0.1,0.1)--(-0.4,0.4);
			\draw[-latex,line width=.6pt,black](-0.4,0.4)--(-0.7,0.7);
			\draw[line width=.6pt,black](-0.4,0.4)--(-0.8,0.8);
			\draw[-latex,line width=.6pt,black](-0.4,0.4)--(-0.1,0.7);
			\draw[line width=.6pt,black](-0.4,0.4)--(-0,0.8);
			\draw[-latex,line width=.6pt,black](0,0)--(0.7,0.7);
			\draw[line width=.6pt,black](0.1,0.1)--(0.8,0.8);
			\node[ line width=0.6pt, dashed, draw opacity=0.5] (a) at (-0.8,1){$\scriptstyle{i}$};
			\node[ line width=0.6pt, dashed, draw opacity=0.5] (a) at (0,1){$\scriptstyle{j}$};
			\node[ line width=0.6pt, dashed, draw opacity=0.5] (a) at (0.8,1){$\scriptstyle{k}$};
			\node[ line width=0.6pt, dashed, draw opacity=0.5] (a) at (0,-0.7){$\scriptstyle{l}$};
			\node[ line width=0.6pt, dashed, draw opacity=0.5] (a) at (0.3,-0.1){$\scriptstyle{\alpha}$};
			\node[ line width=0.6pt, dashed, draw opacity=0.5] (a) at (-0.7,0.3){$\scriptstyle{\beta}$};
			\node[ line width=0.6pt, dashed, draw opacity=0.5] (a) at (-0.4,0){$\scriptstyle m$};
		\end{tikzpicture}
	\end{aligned}
	 =\sum_{n,\mu,\nu} [(F^{ijk}_{l})^{-1}]^{n\mu\nu}_{m\alpha\beta}
	\begin{aligned}
		\begin{tikzpicture}[scale=0.95]
			\draw[-latex,line width=.6pt,black] (0,-0.5) -- (0,-.1);
			\draw[line width=.6pt,black] (0,-0.4) -- (0,0);
			\draw[-latex,line width=.6pt,black](0,0)--(-0.7,0.7);
			\draw[line width=.6pt,black](0,0)--(-0.8,0.8);
			\draw[-latex,line width=.6pt,black](0.4,0.4)--(0.1,0.7);
			\draw[line width=.6pt,black](0.4,0.4)--(0,0.8);
			\draw[-latex,line width=.6pt,black](0,0)--(0.7,0.7);
			\draw[-latex,line width=.6pt,black](0,0)--(0.3,0.3);
			\draw[line width=.6pt,black](0.1,0.1)--(0.8,0.8);
			\node[ line width=0.6pt, dashed, draw opacity=0.5] (a) at (-0.8,1){$\scriptstyle i$};
			\node[ line width=0.6pt, dashed, draw opacity=0.5] (a) at (0,1){$\scriptstyle j$};
			\node[ line width=0.6pt, dashed, draw opacity=0.5] (a) at (0.8,1){$\scriptstyle k$};
			\node[ line width=0.6pt, dashed, draw opacity=0.5] (a) at (0,-0.7){$\scriptstyle l$};
			\node[ line width=0.6pt, dashed, draw opacity=0.5] (a) at (0.3,-0.1){$\scriptstyle \mu$};
			\node[ line width=0.6pt, dashed, draw opacity=0.5] (a) at (0,0.3){$\scriptstyle n$};
				\node[ line width=0.6pt, dashed, draw opacity=0.5] (a) at (0.6,0.3){$\scriptstyle \nu$};
		\end{tikzpicture}
	\end{aligned}, \label{eq:F-move1} \\
      \begin{aligned}
        \begin{tikzpicture}[scale=0.95]
			\draw[-latex,line width=.6pt,black] (0,0) -- (0,0.4);
			\draw[line width=.6pt,black] (0,0.5) -- (0,0.1);
			\draw[-latex,line width=.6pt,black](-0.8,-0.8)--(-0.5,-0.5);
			\draw[line width=.6pt,black](-0.7,-0.7)--(0,0);
			\draw[-latex,line width=.6pt,black](0,-0.8)--(0.3,-0.5);
			\draw[line width=.6pt,black](0.4,-0.4)--(0,-0.8);
			\draw[-latex,line width=.6pt,black](0.8,-0.8)--(0.5,-0.5);
			\draw[line width=.6pt,black](0,0)--(0.7,-0.7);
			\node[ line width=0.6pt, dashed, draw opacity=0.5] (a) at (-0.8,-1){$\scriptstyle i$};
			\node[ line width=0.6pt, dashed, draw opacity=0.5] (a) at (0,-1){$\scriptstyle j$};
			\node[ line width=0.6pt, dashed, draw opacity=0.5] (a) at (0.8,-1){$\scriptstyle k$};
			\node[ line width=0.6pt, dashed, draw opacity=0.5] (a) at (0,0.7){$\scriptstyle l$};
			\node[ line width=0.6pt, dashed, draw opacity=0.5] (a) at (0.3,0.1){$\scriptstyle \mu$};
			\node[ line width=0.6pt, dashed, draw opacity=0.5] (a) at (0.65,-0.3){$\scriptstyle \nu$};
			\node[ line width=0.6pt, dashed, draw opacity=0.5] (a) at (0.05,-0.35){$\scriptstyle n$};
		\end{tikzpicture}
	\end{aligned}
	 =\sum_{m,\alpha,\beta} [F_{ijk}^{l}]_{n\mu\nu}^{m\alpha\beta}
	\begin{aligned}
		\begin{tikzpicture}[scale=0.95]
			\draw[-latex,line width=.6pt,black] (0,0) -- (0,0.4);
			\draw[line width=.6pt,black] (0,0.5) -- (0,0.1);
			\draw[-latex,line width=.6pt,black](-0.4,-0.4)--(-0.1,-0.1);
			\draw[line width=.6pt,black](-0.3,-0.3)--(0,0);
			\draw[-latex,line width=.6pt,black](-0.8,-0.8)--(-0.5,-0.5);
			\draw[line width=.6pt,black](-0.7,-0.7)--(-0.4,-0.4);
			\draw[-latex,line width=.6pt,black](0,-0.8)--(-0.3,-0.5);
			\draw[line width=.6pt,black](-0.4,-0.4)--(-0.1,-0.7);
			\draw[-latex,line width=.6pt,black](0.8,-0.8)--(0.5,-0.5);
			\draw[line width=.6pt,black](0,0)--(0.7,-0.7);
			\node[ line width=0.6pt, dashed, draw opacity=0.5] (a) at (-0.8,-1){$\scriptstyle i$};
			\node[ line width=0.6pt, dashed, draw opacity=0.5] (a) at (0,-1){$\scriptstyle j$};
			\node[ line width=0.6pt, dashed, draw opacity=0.5] (a) at (0.8,-1){$\scriptstyle k$};
			\node[ line width=0.6pt, dashed, draw opacity=0.5] (a) at (0,0.7){$\scriptstyle l$};
			\node[ line width=0.6pt, dashed, draw opacity=0.5] (a) at (0.3,0.1){$\scriptstyle \alpha$};
			\node[ line width=0.6pt, dashed, draw opacity=0.5] (a) at (-0.7,-0.3){$\scriptstyle \beta$};
			\node[ line width=0.6pt, dashed, draw opacity=0.5] (a) at (-0.4,0){$\scriptstyle m$};
		\end{tikzpicture}
	\end{aligned}, \quad 
 \begin{aligned}
		\begin{tikzpicture}[scale=0.95]
			\draw[-latex,line width=.6pt,black] (0,0) -- (0,0.4);
			\draw[line width=.6pt,black] (0,0.5) -- (0,0.1);
			\draw[-latex,line width=.6pt,black](-0.4,-0.4)--(-0.1,-0.1);
			\draw[line width=.6pt,black](-0.3,-0.3)--(0,0);
			\draw[-latex,line width=.6pt,black](-0.8,-0.8)--(-0.5,-0.5);
			\draw[line width=.6pt,black](-0.7,-0.7)--(-0.4,-0.4);
			\draw[-latex,line width=.6pt,black](0,-0.8)--(-0.3,-0.5);
			\draw[line width=.6pt,black](-0.4,-0.4)--(-0.1,-0.7);
			\draw[-latex,line width=.6pt,black](0.8,-0.8)--(0.5,-0.5);
			\draw[line width=.6pt,black](0,0)--(0.7,-0.7);
			\node[ line width=0.6pt, dashed, draw opacity=0.5] (a) at (-0.8,-1){$\scriptstyle i$};
			\node[ line width=0.6pt, dashed, draw opacity=0.5] (a) at (0,-1){$\scriptstyle j$};
			\node[ line width=0.6pt, dashed, draw opacity=0.5] (a) at (0.8,-1){$\scriptstyle k$};
			\node[ line width=0.6pt, dashed, draw opacity=0.5] (a) at (0,0.7){$\scriptstyle l$};
			\node[ line width=0.6pt, dashed, draw opacity=0.5] (a) at (0.3,0.1){$\scriptstyle \alpha$};
			\node[ line width=0.6pt, dashed, draw opacity=0.5] (a) at (-0.7,-0.3){$\scriptstyle \beta$};
			\node[ line width=0.6pt, dashed, draw opacity=0.5] (a) at (-0.4,0){$\scriptstyle m$};
		\end{tikzpicture}
	\end{aligned}
	 =\sum_{n,\mu,\nu} [(F_{ijk}^{l})^{-1}]^{n\mu\nu}_{m\alpha\beta}
	\begin{aligned}
		\begin{tikzpicture}[scale=0.95]
			\draw[-latex,line width=.6pt,black] (0,0) -- (0,0.4);
			\draw[line width=.6pt,black] (0,0.5) -- (0,0.1);
			\draw[-latex,line width=.6pt,black](-0.8,-0.8)--(-0.5,-0.5);
			\draw[line width=.6pt,black](-0.7,-0.7)--(0,0);
			\draw[-latex,line width=.6pt,black](0,-0.8)--(0.3,-0.5);
			\draw[line width=.6pt,black](0.4,-0.4)--(0,-0.8);
			\draw[-latex,line width=.6pt,black](0.8,-0.8)--(0.5,-0.5);
			\draw[line width=.6pt,black](0,0)--(0.7,-0.7);
			\node[ line width=0.6pt, dashed, draw opacity=0.5] (a) at (-0.8,-1){$\scriptstyle i$};
			\node[ line width=0.6pt, dashed, draw opacity=0.5] (a) at (0,-1){$\scriptstyle j$};
			\node[ line width=0.6pt, dashed, draw opacity=0.5] (a) at (0.8,-1){$\scriptstyle k$};
			\node[ line width=0.6pt, dashed, draw opacity=0.5] (a) at (0,0.7){$\scriptstyle l$};
			\node[ line width=0.6pt, dashed, draw opacity=0.5] (a) at (0.3,0.1){$\scriptstyle \mu$};
			\node[ line width=0.6pt, dashed, draw opacity=0.5] (a) at (0.65,-0.3){$\scriptstyle \nu$};
			\node[ line width=0.6pt, dashed, draw opacity=0.5] (a) at (0.05,-0.35){$\scriptstyle n$};
		\end{tikzpicture}
	\end{aligned}. 
\end{gather*}
\end{enumerate}
The topological local moves can be generalized in a straightforward manner to topological boundaries and domain walls by modifying the string labels and morphisms to take values in their respective categories.

For a string-net on a closed surface, there are two local stabilizers, $Q_v$ and $B_f$, assigned to the vertices and faces, respectively. These are all mutually commuting local projectors, and the Hamiltonian is given by
\begin{equation}
    \notag
    H=-\sum_v Q_v -\sum_f B_f.
\end{equation}
For a given vertex, we define $\delta_{a,b\to c}=1$ if $V_{ab}^c\neq 0$; otherwise, we set $\delta_{a,b\to c}=0$. Similarly, $\delta_{c\to a,b}$ is defined analogously. Using these, we introduce a vertex projector $Q_v$ as follows:
\begin{equation}
    \notag
    Q_v \big{|}  \begin{aligned}
		\begin{tikzpicture}
			\draw[-latex,line width=.6pt,black] (0,-0.5) -- (0,-.1);
			\draw[line width=.6pt,black] (0,-0.4) -- (0,0);
			\draw[-latex,line width=.6pt,black](0,0)--(-0.3,0.3);
			\draw[line width=.6pt,black](-0.1,0.1)--(-0.4,0.4);
			\draw[-latex,line width=.6pt,black](0,0)--(0.3,0.3);
			\draw[line width=.6pt,black](0.1,0.1)--(0.4,0.4);
			\node[ line width=0.6pt, dashed, draw opacity=0.5] (a) at (-0.5,0.6){$a$};
			\node[ line width=0.6pt, dashed, draw opacity=0.5] (a) at (0.5,0.6){$b$};
			\node[ line width=0.6pt, dashed, draw opacity=0.5] (a) at (0,-0.7){$c$};
			\node[ line width=0.6pt, dashed, draw opacity=0.5] (a) at (0.3,-0.1){$\alpha$};
		\end{tikzpicture}
	\end{aligned} \big{\rangle} =\delta_{c\to a,b} \,\big{|}\begin{aligned}
		\begin{tikzpicture}
			\draw[-latex,line width=.6pt,black] (0,-0.5) -- (0,-.1);
			\draw[line width=.6pt,black] (0,-0.4) -- (0,0);
			\draw[-latex,line width=.6pt,black](0,0)--(-0.3,0.3);
			\draw[line width=.6pt,black](-0.1,0.1)--(-0.4,0.4);
			\draw[-latex,line width=.6pt,black](0,0)--(0.3,0.3);
			\draw[line width=.6pt,black](0.1,0.1)--(0.4,0.4);
			\node[ line width=0.6pt, dashed, draw opacity=0.5] (a) at (-0.5,0.6){$a$};
			\node[ line width=0.6pt, dashed, draw opacity=0.5] (a) at (0.5,0.6){$b$};
			\node[ line width=0.6pt, dashed, draw opacity=0.5] (a) at (0,-0.7){$c$};
			\node[ line width=0.6pt, dashed, draw opacity=0.5] (a) at (0.3,-0.1){$\alpha$};
		\end{tikzpicture}
	\end{aligned}\big{\rangle}\,.
\end{equation}
The condition $Q_v=1$ corresponds to satisfying Gauss's law at the vertex $v$ in an analogy with lattice gauge theory.
The face operator $B_f$ is more involved and is expressed as a linear combination of elementary face operators $B^k_f$:
\begin{equation}
    \notag
    B_f=\sum_{k\in \Irr(\eD)}w_k B^k_f,
\end{equation}
where $w_k=Y^{k^* k}_{\one}/\sum_{l\in \Irr(\eD)} d_l^2$. 
The operator $B_f^k$ acts on the face $f$ by inserting a closed $k$-loop into it, which is depicted diagrammatically as
\begin{gather}
    \notag
    B_f^k\, \Big{|} 
    \begin{aligned}
        \begin{tikzpicture}[scale=0.9]
            \draw[line width=.6pt,black] (0,0) -- (0.87,0.5);
            \draw[line width=.6pt,black] (0.87,1.5) -- (0.87,0.5);
            \draw[line width=.6pt,black] (0.87,1.5) -- (0,2);
            \draw[line width=.6pt,black] (0,0) -- (-0.87,0.5);
            \draw[line width=.6pt,black] (-0.87,1.5) -- (-0.87,0.5);
            \draw[line width=.6pt,black] (-0.87,1.5) -- (0,2);
            \draw[-latex,line width=.6pt,black] (0,0) -- (0.5742,0.33);
            \draw[-latex,line width=.6pt,black] (0.87,0.5) -- (0.87,1.1);
            \draw[-latex,line width=.6pt,black] (0.87,1.5) -- (0.2958,1.83);
            \draw[-latex,line width=.6pt,black] (0,0) -- (-0.5742,0.33);
            \draw[-latex,line width=.6pt,black] (-0.87,0.5) -- (-0.87,1.1);
            \draw[-latex,line width=.6pt,black] (-0.87,1.5) -- (-0.2958,1.83);
            \draw[line width=.6pt,black] (0,-0.6) -- (0,0);
            \draw[line width=.6pt,black] (0,2) -- (0,2.6);
            \draw[line width=.6pt,black] (0.87,0.5) -- (1.5225,0.125);
            \draw[line width=.6pt,black] (0.87,1.5) -- (1.5225,1.875);
            \draw[line width=.6pt,black] (-0.87,0.5) -- (-1.5225,0.125);
            \draw[line width=.6pt,black] (-0.87,1.5) -- (-1.5225,1.875);
            \draw[-latex,line width=.6pt,black] (0,-0.6) -- (0,-0.2);
            \draw[-latex,line width=.6pt,black] (0,2) -- (0,2.4);
            \draw[-latex,line width=.6pt,black] (1.5225,0.125) -- (1.1,0.3678);
            \draw[-latex,line width=.6pt,black] (0.87,1.5) -- (1.2925,1.7428);
            \draw[-latex,line width=.6pt,black] (-1.5225,0.125) -- (-1.1,0.3678);
            \draw[-latex,line width=.6pt,black] (-0.87,1.5) -- (-1.2925,1.7428);
            \node[ line width=0.6pt, dashed, draw opacity=0.5] (a) at (1.2,1){$j_1$};
            \node[ line width=0.6pt, dashed, draw opacity=0.5] (a) at (0.6,2.1){$j_2$};
            \node[ line width=0.6pt, dashed, draw opacity=0.5] (a) at (-0.6,2.05){$j_3$};
            \node[ line width=0.6pt, dashed, draw opacity=0.5] (a) at (-1.2,1){$j_4$};
            \node[ line width=0.6pt, dashed, draw opacity=0.5] (a) at (-0.55,0){$j_5$};
            \node[ line width=0.6pt, dashed, draw opacity=0.5] (a) at (0.55,0){$j_6$};
            \node[ line width=0.6pt, dashed, draw opacity=0.5] (a) at (1.8,2){$i_1$};
            \node[ line width=0.6pt, dashed, draw opacity=0.5] (a) at (0,2.9){$i_2$};
            \node[ line width=0.6pt, dashed, draw opacity=0.5] (a) at (-1.8,2){$i_3$};
            \node[ line width=0.6pt, dashed, draw opacity=0.5] (a) at (-1.8,0){$i_4$};
            \node[ line width=0.6pt, dashed, draw opacity=0.5] (a) at (0,-0.9){$i_5$};
            \node[ line width=0.6pt, dashed, draw opacity=0.5] (a) at (1.8,0){$i_6$};
            \node[ line width=0.6pt, dashed, draw opacity=0.5] (a) at (0.65,1.35){$\alpha_1$};
            \node[ line width=0.6pt, dashed, draw opacity=0.5] (a) at (0,1.7){$\alpha_2$};
            \node[ line width=0.6pt, dashed, draw opacity=0.5] (a) at (-0.55,1.35){$\alpha_3$};
            \node[ line width=0.6pt, dashed, draw opacity=0.5] (a) at (-0.55,0.55){$\alpha_4$};
            \node[ line width=0.6pt, dashed, draw opacity=0.5] (a) at (0,0.3){$\alpha_5$};
            \node[ line width=0.6pt, dashed, draw opacity=0.5] (a) at (0.62,0.65){$\alpha_6$};
        \end{tikzpicture}
    \end{aligned}
    \Big{\rangle} = 
    \Big{|}
    \begin{aligned}
        \begin{tikzpicture}[scale=0.9]
            \draw[line width=.6pt,black] (0,0) -- (0.87,0.5);
            \draw[line width=.6pt,black] (0.87,1.5) -- (0.87,0.5);
            \draw[line width=.6pt,black] (0.87,1.5) -- (0,2);
            \draw[line width=.6pt,black] (0,0) -- (-0.87,0.5);
            \draw[line width=.6pt,black] (-0.87,1.5) -- (-0.87,0.5);
            \draw[line width=.6pt,black] (-0.87,1.5) -- (0,2);
            \draw[-latex,line width=.6pt,black] (0,0) -- (0.5742,0.33);
            \draw[-latex,line width=.6pt,black] (0.87,0.5) -- (0.87,1.1);
            \draw[-latex,line width=.6pt,black] (0.87,1.5) -- (0.2958,1.83);
            \draw[-latex,line width=.6pt,black] (0,0) -- (-0.5742,0.33);
            \draw[-latex,line width=.6pt,black] (-0.87,0.5) -- (-0.87,1.1);
            \draw[-latex,line width=.6pt,black] (-0.87,1.5) -- (-0.2958,1.83);
            \draw[line width=.6pt,black] (0,-0.6) -- (0,0);
            \draw[line width=.6pt,black] (0,2) -- (0,2.6);
            \draw[line width=.6pt,black] (0.87,0.5) -- (1.5225,0.125);
            \draw[line width=.6pt,black] (0.87,1.5) -- (1.5225,1.875);
            \draw[line width=.6pt,black] (-0.87,0.5) -- (-1.5225,0.125);
            \draw[line width=.6pt,black] (-0.87,1.5) -- (-1.5225,1.875);
            \draw[-latex,line width=.6pt,black] (0,-0.6) -- (0,-0.2);
            \draw[-latex,line width=.6pt,black] (0,2) -- (0,2.4);
            \draw[-latex,line width=.6pt,black] (1.5225,0.125) -- (1.1,0.3678);
            \draw[-latex,line width=.6pt,black] (0.87,1.5) -- (1.2925,1.7428);
            \draw[-latex,line width=.6pt,black] (-1.5225,0.125) -- (-1.1,0.3678);
            \draw[-latex,line width=.6pt,black] (-0.87,1.5) -- (-1.2925,1.7428);
            \node[ line width=0.6pt, dashed, draw opacity=0.5] (a) at (1.2,1){$j_1$};
            \node[ line width=0.6pt, dashed, draw opacity=0.5] (a) at (0.6,2.1){$j_2$};
            \node[ line width=0.6pt, dashed, draw opacity=0.5] (a) at (-0.6,2.05){$j_3$};
            \node[ line width=0.6pt, dashed, draw opacity=0.5] (a) at (-1.2,1){$j_4$};
            \node[ line width=0.6pt, dashed, draw opacity=0.5] (a) at (-0.55,0){$j_5$};
            \node[ line width=0.6pt, dashed, draw opacity=0.5] (a) at (0.55,0){$j_6$};
            \node[ line width=0.6pt, dashed, draw opacity=0.5] (a) at (1.8,2){$i_1$};
            \node[ line width=0.6pt, dashed, draw opacity=0.5] (a) at (0,2.9){$i_2$};
            \node[ line width=0.6pt, dashed, draw opacity=0.5] (a) at (-1.8,2){$i_3$};
            \node[ line width=0.6pt, dashed, draw opacity=0.5] (a) at (-1.8,0){$i_4$};
            \node[ line width=0.6pt, dashed, draw opacity=0.5] (a) at (0,-0.9){$i_5$};
            \node[ line width=0.6pt, dashed, draw opacity=0.5] (a) at (1.8,0){$i_6$};
            \node[ line width=0.6pt, dashed, draw opacity=0.5] (a) at (0.65,1.35){$\alpha_1$};
            \node[ line width=0.6pt, dashed, draw opacity=0.5] (a) at (0,1.7){$\alpha_2$};
            \node[ line width=0.6pt, dashed, draw opacity=0.5] (a) at (-0.55,1.35){$\alpha_3$};
            \node[ line width=0.6pt, dashed, draw opacity=0.5] (a) at (-0.55,0.55){$\alpha_4$};
            \node[ line width=0.6pt, dashed, draw opacity=0.5] (a) at (0,0.3){$\alpha_5$};
            \node[ line width=0.6pt, dashed, draw opacity=0.5] (a) at (0.62,0.65){$\alpha_6$};
            \draw[line width=0.8pt,black] (0.35,1.05) .. controls +(0,0.4) and +(0,0.4) .. (-0.35,1.05);
            \draw[line width=0.8pt,black] (0.35,0.95) .. controls +(0,-0.4) and +(0,-0.4) .. (-0.35,0.95);
            \draw[line width=0.8pt,black](0.35,1.05) -- (0.35,0.95);
            \draw[line width=0.8pt,black](-0.35,1.05) -- (-0.35,0.95);
            \draw[-latex,line width=0.8pt,black](0.35,0.9) -- (0.35,1.12);
            \node[ line width=0.1pt, dashed, draw opacity=0.5] (a) at (0.1,1){$k$};
        \end{tikzpicture}
    \end{aligned}
    \Big{\rangle}\;.
\end{gather}
The condition $B_f=1$ can be interpreted as enforcing that no flux threads through the face $f$.
By replacing the relevant data with those of a module category and a bimodule category, the string-net construction can be extended to incorporate boundaries and domain walls.

\section{Some technical results on factorization and antipode-like maps} 
\label{app:tech_proofs}

In this section, we provide a detailed proof of some properties of the factorization map and the antipode-like map for defect tube algebras. We divide it into two parts. The first part proves the relevant properties of these maps in the boundary case, while the second part establishes the corresponding results in the domain wall case.

\subsubsection*{Boundary case.}

\textit{Proof of} \eqref{eq:coassociativity}. For a basis tube $X\in \Tube({}_\eC\eN;{}_\eC\eM)$, both expressions, applied to $X$, are equal to 
    \begin{equation}
        \notag
        \sum_{i,j,\sigma} \sqrt{\frac{d_i}{d_ad_j}} \sum_{t,s,\varrho} \sqrt{\frac{d_s}{d_ad_t}} \; \begin{aligned}
            \begin{tikzpicture}[scale=0.65]
                \fill[gray!15] (0,1.5) arc[start angle=90, end angle=270, radius=1.5] -- (0,-0.5) arc[start angle=270, end angle=90, radius=0.5] -- cycle;
        \draw[dotted] (0,1.5) arc[start angle=90, end angle=270, radius=1.5]; 
        \draw[dotted] (0,0.5) arc[start angle=90, end angle=270, radius=0.5];
            \draw[line width=1pt,violet] (0,0.5)--(0,1.5);
            \draw[line width=1pt,cyan] (0,-0.5)--(0,-1.5);
             \draw[red] (0,0.8) arc[start angle=90, end angle=270, radius=0.8];
            \node[ line width=0.6pt, dashed, draw opacity=0.5] (a) at (-1,0){$\scriptstyle a$};
            \node[ line width=0.6pt, dashed, draw opacity=0.5] (a) at (-0.2,-1.1){$\scriptstyle i$};
            \node[ line width=0.6pt, dashed, draw opacity=0.5] (a) at (0.2,-0.8){$\scriptstyle \sigma$};
            \node[ line width=0.6pt, dashed, draw opacity=0.5] (a) at (0,-0.25){$\scriptstyle j$};
            \node[ line width=0.6pt, dashed, draw opacity=0.5] (a) at (0,0.3){$\scriptstyle y$};
            \node[ line width=0.6pt, dashed, draw opacity=0.5] (a) at (-0.2,1.1){$\scriptstyle z$};
            \node[ line width=0.6pt, dashed, draw opacity=0.5] (a) at (0.2,0.8){$\scriptstyle \nu$};
        \end{tikzpicture}
    \end{aligned} \otimes \; \begin{aligned}
            \begin{tikzpicture}[scale=0.65]
                \fill[gray!15] (0,1.5) arc[start angle=90, end angle=270, radius=1.5] -- (0,-0.5) arc[start angle=270, end angle=90, radius=0.5] -- cycle;
        \draw[dotted] (0,1.5) arc[start angle=90, end angle=270, radius=1.5]; 
        \draw[dotted] (0,0.5) arc[start angle=90, end angle=270, radius=0.5];
            \draw[line width=1pt,cyan] (0,0.5)--(0,1.5);
            \draw[line width=1pt,violet] (0,-0.5)--(0,-1.5);
             \draw[red] (0,0.8) arc[start angle=90, end angle=270, radius=0.8];
            \node[ line width=0.6pt, dashed, draw opacity=0.5] (a) at (-1,0){$\scriptstyle a$};
            \node[ line width=0.6pt, dashed, draw opacity=0.5] (a) at (-0.2,-1.1){$\scriptstyle s$};
            \node[ line width=0.6pt, dashed, draw opacity=0.5] (a) at (0.2,-0.8){$\scriptstyle \varepsilon'$};
            \node[ line width=0.6pt, dashed, draw opacity=0.5] (a) at (0,-0.3){$\scriptstyle t$};
            \node[ line width=0.6pt, dashed, draw opacity=0.5] (a) at (0,0.3){$\scriptstyle j$};
            \node[ line width=0.6pt, dashed, draw opacity=0.5] (a) at (-0.2,1.1){$\scriptstyle i$};
            \node[ line width=0.6pt, dashed, draw opacity=0.5] (a) at (0.2,0.8){$\scriptstyle \sigma$};
        \end{tikzpicture}
    \end{aligned} \otimes \; \begin{aligned}
            \begin{tikzpicture}[scale=0.65]
                \fill[gray!15] (0,1.5) arc[start angle=90, end angle=270, radius=1.5] -- (0,-0.5) arc[start angle=270, end angle=90, radius=0.5] -- cycle;
        \draw[dotted] (0,1.5) arc[start angle=90, end angle=270, radius=1.5]; 
        \draw[dotted] (0,0.5) arc[start angle=90, end angle=270, radius=0.5];
            \draw[line width=1pt,violet] (0,0.5)--(0,1.5);
            \draw[line width=1pt,cyan] (0,-0.5)--(0,-1.5);
             \draw[red] (0,0.8) arc[start angle=90, end angle=270, radius=0.8];
            \node[ line width=0.6pt, dashed, draw opacity=0.5] (a) at (-1,0){$\scriptstyle a$};
            \node[ line width=0.6pt, dashed, draw opacity=0.5] (a) at (-0.2,-1.1){$\scriptstyle u$};
            \node[ line width=0.6pt, dashed, draw opacity=0.5] (a) at (0.2,-0.8){$\scriptstyle \zeta$};
            \node[ line width=0.6pt, dashed, draw opacity=0.5] (a) at (0,-0.3){$\scriptstyle v$};
            \node[ line width=0.6pt, dashed, draw opacity=0.5] (a) at (0,0.3){$\scriptstyle t$};
            \node[ line width=0.6pt, dashed, draw opacity=0.5] (a) at (-0.2,1.1){$\scriptstyle s$};
            \node[ line width=0.6pt, dashed, draw opacity=0.5] (a) at (0.2,0.8){$\scriptstyle \varrho$};
        \end{tikzpicture}
    \end{aligned}  \;. 
    \end{equation}

\noindent 
\textit{Proof of} \eqref{eq:unit-trifold}. By a direct calculation, both expressions are equal to  
    \begin{equation}
        \notag
        \sum_{x,z\in\Irr(\eM)}\sum_{u,v\in \Irr(\eN)}\;  \begin{aligned}
        \begin{tikzpicture}[scale=0.65]
            \fill[gray!15] (0,1.5) arc[start angle=90, end angle=270, radius=1.5] -- (0,-0.5) arc[start angle=270, end angle=90, radius=0.5] -- cycle;
            \draw[dotted] (0,1.5) arc[start angle=90, end angle=270, radius=1.5]; 
            \draw[dotted] (0,0.5) arc[start angle=90, end angle=270, radius=0.5]; 
            \draw[line width=1pt,cyan] (0,0.5)--(0,1.5);
            \draw[line width=1pt,violet] (0,-0.5)--(0,-1.5);
            \draw[red,dotted,thick] (0,0.8) arc[start angle=90, end angle=270, radius=0.8];
            \node[ line width=0.6pt, dashed, draw opacity=0.5] (a) at (-1,0){$\scriptstyle \mathbb{1}$};
            \node[ line width=0.6pt, dashed, draw opacity=0.5] (a) at (0.25,-1.2){$\scriptstyle x$};
            \node[ line width=0.6pt, dashed, draw opacity=0.5] (a) at (0.25,1.2){$\scriptstyle u$};
        \end{tikzpicture}
    \end{aligned} \otimes \; \begin{aligned}
        \begin{tikzpicture}[scale=0.65]
            \fill[gray!15] (0,1.5) arc[start angle=90, end angle=270, radius=1.5] -- (0,-0.5) arc[start angle=270, end angle=90, radius=0.5] -- cycle;
            \draw[dotted] (0,1.5) arc[start angle=90, end angle=270, radius=1.5]; 
            \draw[dotted] (0,0.5) arc[start angle=90, end angle=270, radius=0.5]; 
            \draw[line width=1pt,violet] (0,0.5)--(0,1.5);
            \draw[line width=1pt,cyan] (0,-0.5)--(0,-1.5);
            \draw[red,dotted,thick] (0,0.8) arc[start angle=90, end angle=270, radius=0.8];
            \node[ line width=0.6pt, dashed, draw opacity=0.5] (a) at (-1,0){$\scriptstyle \mathbb{1}$};
            \node[ line width=0.6pt, dashed, draw opacity=0.5] (a) at (0.25,-1.2){$\scriptstyle v$};
            \node[ line width=0.6pt, dashed, draw opacity=0.5] (a) at (0.25,1.2){$\scriptstyle x$};
        \end{tikzpicture}
    \end{aligned} \otimes \; \begin{aligned}
        \begin{tikzpicture}[scale=0.65]
            \fill[gray!15] (0,1.5) arc[start angle=90, end angle=270, radius=1.5] -- (0,-0.5) arc[start angle=270, end angle=90, radius=0.5] -- cycle;
            \draw[dotted] (0,1.5) arc[start angle=90, end angle=270, radius=1.5]; 
            \draw[dotted] (0,0.5) arc[start angle=90, end angle=270, radius=0.5]; 
            \draw[line width=1pt,cyan] (0,0.5)--(0,1.5);
            \draw[line width=1pt,violet] (0,-0.5)--(0,-1.5);
            \draw[red,dotted,thick] (0,0.8) arc[start angle=90, end angle=270, radius=0.8];
            \node[ line width=0.6pt, dashed, draw opacity=0.5] (a) at (-1,0){$\scriptstyle \mathbb{1}$};
            \node[ line width=0.6pt, dashed, draw opacity=0.5] (a) at (0.25,-1.2){$\scriptstyle z$};
            \node[ line width=0.6pt, dashed, draw opacity=0.5] (a) at (0.25,1.2){$\scriptstyle v$};
        \end{tikzpicture}
    \end{aligned}\;.
    \end{equation}


\noindent
\textit{Proof of} \eqref{eq:preserve-product}. Take $X, Y$ as Eq.~\eqref{eq:basis-diagram-XY}. First, by Eq.~\eqref{eq:product-basis} we have 
    \begin{align}
    \delta(XY) & = \delta_{u,v'}\delta_{z,y'}\sum_{t,\beta}\sum_{\alpha}[{_{\eN}}F_{aa'u'}^{v}]_{u\zeta\zeta'}^{t\alpha\beta}\sum_{\theta}[{_{\eM}}F_y^{aa'z'}]_{z\nu\nu'}^{t\theta\beta}\sqrt{\frac{d_{a}d_{a'}}{d_t}} \nonumber \\
    & \quad \quad \sum_{i,k\in\Irr(\eK)}\sum_{\sigma}\sqrt{\frac{d_k}{d_td_i}}\;\begin{aligned}
        \begin{tikzpicture}[scale=0.65]
        \begin{scope}
            \fill[gray!15]
                (0,1.5) arc[start angle=90, end angle=270, radius=1.5] -- 
                (0,-0.5) arc[start angle=270, end angle=90, radius=0.5] -- cycle;
        \end{scope}
        \draw[dotted] (0,1.5) arc[start angle=90, end angle=270, radius=1.5]; 
        \draw[dotted] (0,0.5) arc[start angle=90, end angle=270, radius=0.5]; 
            \draw[line width=1pt,cyan] (0,0.5)--(0,1.5);
            \draw[line width=1pt,teal] (0,-0.5)--(0,-1.5);
             \draw[red] (0,0.8) arc[start angle=90, end angle=270, radius=0.8];
            \node[ line width=0.6pt, dashed, draw opacity=0.5] (a) at (-1,0){$\scriptstyle t$};
            \node[ line width=0.6pt, dashed, draw opacity=0.5] (a) at (-0.2,-1.1){$\scriptstyle k$};
            \node[ line width=0.6pt, dashed, draw opacity=0.5] (a) at (0.25,-0.8){$\scriptstyle \sigma$};
            \node[ line width=0.6pt, dashed, draw opacity=0.5] (a) at (0,-0.3){$\scriptstyle i$};
            \node[ line width=0.6pt, dashed, draw opacity=0.5] (a) at (0,0.3){$\scriptstyle u'$};
            \node[ line width=0.6pt, dashed, draw opacity=0.5] (a) at (-0.2,1.1){$\scriptstyle v$};
            \node[ line width=0.6pt, dashed, draw opacity=0.5] (a) at (0.25,0.8){$\scriptstyle \alpha$};
        \end{tikzpicture}
    \end{aligned}\otimes \;
    \begin{aligned}
        \begin{tikzpicture}[scale=0.65]
        \begin{scope}
            \fill[gray!15]
                (0,1.5) arc[start angle=90, end angle=270, radius=1.5] -- 
                (0,-0.5) arc[start angle=270, end angle=90, radius=0.5] -- cycle;
        \end{scope}
        \draw[dotted] (0,1.5) arc[start angle=90, end angle=270, radius=1.5]; 
        \draw[dotted] (0,0.5) arc[start angle=90, end angle=270, radius=0.5]; 
            \draw[line width=1pt,teal] (0,0.5)--(0,1.5);
            \draw[line width=1pt,violet] (0,-0.5)--(0,-1.5);
             \draw[red] (0,0.8) arc[start angle=90, end angle=270, radius=0.8];
            \node[ line width=0.6pt, dashed, draw opacity=0.5] (a) at (-1,0){$\scriptstyle t$};
            \node[ line width=0.6pt, dashed, draw opacity=0.5] (a) at (-0.2,-1.1){$\scriptstyle y$};
            \node[ line width=0.6pt, dashed, draw opacity=0.5] (a) at (0.25,-0.8){$\scriptstyle \theta$};
            \node[ line width=0.6pt, dashed, draw opacity=0.5] (a) at (0,-0.25){$\scriptstyle z'$};
            \node[ line width=0.6pt, dashed, draw opacity=0.5] (a) at (0,0.3){$\scriptstyle i$};
            \node[ line width=0.6pt, dashed, draw opacity=0.5] (a) at (-0.2,1.1){$\scriptstyle k$};
            \node[ line width=0.6pt, dashed, draw opacity=0.5] (a) at (0.25,0.8){$\scriptstyle \sigma$};
        \end{tikzpicture}
    \end{aligned}.     \notag
\end{align} 
On the other hand, one obtains 
\begin{align}
    &\quad \delta(X)\delta(Y) \nonumber \\
    & = \sum_{i,k\in\Irr(\eK)}\sum_{\sigma}\sqrt{\frac{d_k}{d_ad_i}}\sum_{i',k'\in\Irr(\eK)}\sum_{\sigma'}\sqrt{\frac{d_{k'}}{d_{a'}d_{i'}}}\;\begin{aligned}
        \begin{tikzpicture}[scale=0.65]
        \begin{scope}
            \fill[gray!15]
                (0,1.5) arc[start angle=90, end angle=270, radius=1.5] -- 
                (0,-0.5) arc[start angle=270, end angle=90, radius=0.5] -- cycle;
        \end{scope}
        \draw[dotted] (0,1.5) arc[start angle=90, end angle=270, radius=1.5]; 
        \draw[dotted] (0,0.5) arc[start angle=90, end angle=270, radius=0.5]; 
            \draw[line width=1pt,cyan] (0,0.5)--(0,1.5);
            \draw[line width=1pt,teal] (0,-0.5)--(0,-1.5);
             \draw[red] (0,0.8) arc[start angle=90, end angle=270, radius=0.8];
            \node[ line width=0.6pt, dashed, draw opacity=0.5] (a) at (-1,0){$\scriptstyle a$};
            \node[ line width=0.6pt, dashed, draw opacity=0.5] (a) at (-0.2,-1.1){$\scriptstyle k$};
            \node[ line width=0.6pt, dashed, draw opacity=0.5] (a) at (0.25,-0.8){$\scriptstyle \sigma$};
            \node[ line width=0.6pt, dashed, draw opacity=0.5] (a) at (0,-0.3){$\scriptstyle i$};
            \node[ line width=0.6pt, dashed, draw opacity=0.5] (a) at (0,0.3){$\scriptstyle u$};
            \node[ line width=0.6pt, dashed, draw opacity=0.5] (a) at (-0.2,1.1){$\scriptstyle v$};
            \node[ line width=0.6pt, dashed, draw opacity=0.5] (a) at (0.25,0.8){$\scriptstyle \zeta$};
        \end{tikzpicture}
    \end{aligned}
    \begin{aligned}
        \begin{tikzpicture}[scale=0.65]
        \begin{scope}
            \fill[gray!15]
                (0,1.5) arc[start angle=90, end angle=270, radius=1.5] -- 
                (0,-0.5) arc[start angle=270, end angle=90, radius=0.5] -- cycle;
        \end{scope}
        \draw[dotted] (0,1.5) arc[start angle=90, end angle=270, radius=1.5]; 
        \draw[dotted] (0,0.5) arc[start angle=90, end angle=270, radius=0.5]; 
            \draw[line width=1pt,cyan] (0,0.5)--(0,1.5);
            \draw[line width=1pt,teal] (0,-0.5)--(0,-1.5);
             \draw[red] (0,0.8) arc[start angle=90, end angle=270, radius=0.8];
            \node[ line width=0.6pt, dashed, draw opacity=0.5] (a) at (-1,0){$\scriptstyle a'$};
            \node[ line width=0.6pt, dashed, draw opacity=0.5] (a) at (-0.2,-1.1){$\scriptstyle k'$};
            \node[ line width=0.6pt, dashed, draw opacity=0.5] (a) at (0.25,-0.8){$\scriptstyle \sigma'$};
            \node[ line width=0.6pt, dashed, draw opacity=0.5] (a) at (0,-0.25){$\scriptstyle i'$};
            \node[ line width=0.6pt, dashed, draw opacity=0.5] (a) at (0,0.3){$\scriptstyle u'$};
            \node[ line width=0.6pt, dashed, draw opacity=0.5] (a) at (-0.2,1.1){$\scriptstyle v'$};
            \node[ line width=0.6pt, dashed, draw opacity=0.5] (a) at (0.25,0.8){$\scriptstyle \zeta'$};
        \end{tikzpicture}
    \end{aligned}\otimes \;
    \begin{aligned}
        \begin{tikzpicture}[scale=0.65]
        \begin{scope}
            \fill[gray!15]
                (0,1.5) arc[start angle=90, end angle=270, radius=1.5] -- 
                (0,-0.5) arc[start angle=270, end angle=90, radius=0.5] -- cycle;
        \end{scope}
        \draw[dotted] (0,1.5) arc[start angle=90, end angle=270, radius=1.5]; 
        \draw[dotted] (0,0.5) arc[start angle=90, end angle=270, radius=0.5]; 
            \draw[line width=1pt,teal] (0,0.5)--(0,1.5);
            \draw[line width=1pt,violet] (0,-0.5)--(0,-1.5);
             \draw[red] (0,0.8) arc[start angle=90, end angle=270, radius=0.8];
            \node[ line width=0.6pt, dashed, draw opacity=0.5] (a) at (-1,0){$\scriptstyle a$};
            \node[ line width=0.6pt, dashed, draw opacity=0.5] (a) at (-0.2,-1.1){$\scriptstyle y$};
            \node[ line width=0.6pt, dashed, draw opacity=0.5] (a) at (0.25,-0.8){$\scriptstyle \nu$};
            \node[ line width=0.6pt, dashed, draw opacity=0.5] (a) at (0,-0.3){$\scriptstyle z$};
            \node[ line width=0.6pt, dashed, draw opacity=0.5] (a) at (0,0.3){$\scriptstyle i$};
            \node[ line width=0.6pt, dashed, draw opacity=0.5] (a) at (-0.2,1.1){$\scriptstyle k$};
            \node[ line width=0.6pt, dashed, draw opacity=0.5] (a) at (0.25,0.8){$\scriptstyle \sigma$};
        \end{tikzpicture}
    \end{aligned}
    \begin{aligned}
        \begin{tikzpicture}[scale=0.65]
        \begin{scope}
            \fill[gray!15]
                (0,1.5) arc[start angle=90, end angle=270, radius=1.5] -- 
                (0,-0.5) arc[start angle=270, end angle=90, radius=0.5] -- cycle;
        \end{scope}
        \draw[dotted] (0,1.5) arc[start angle=90, end angle=270, radius=1.5]; 
        \draw[dotted] (0,0.5) arc[start angle=90, end angle=270, radius=0.5]; 
            \draw[line width=1pt,teal] (0,0.5)--(0,1.5);
            \draw[line width=1pt,violet] (0,-0.5)--(0,-1.5);
             \draw[red] (0,0.8) arc[start angle=90, end angle=270, radius=0.8];
            \node[ line width=0.6pt, dashed, draw opacity=0.5] (a) at (-1,0){$\scriptstyle a'$};
            \node[ line width=0.6pt, dashed, draw opacity=0.5] (a) at (-0.2,-1.1){$\scriptstyle y'$};
            \node[ line width=0.6pt, dashed, draw opacity=0.5] (a) at (0.3,-0.8){$\scriptstyle \nu'$};
            \node[ line width=0.6pt, dashed, draw opacity=0.5] (a) at (0,-0.25){$\scriptstyle z'$};
            \node[ line width=0.6pt, dashed, draw opacity=0.5] (a) at (0,0.3){$\scriptstyle i'$};
            \node[ line width=0.6pt, dashed, draw opacity=0.5] (a) at (-0.2,1.1){$\scriptstyle k'$};
            \node[ line width=0.6pt, dashed, draw opacity=0.5] (a) at (0.3,0.8){$\scriptstyle \sigma'$};
        \end{tikzpicture}
    \end{aligned} \nonumber \\
    & = \sum_{i,k\in\Irr(\eK)}\sum_{\sigma}\sqrt{\frac{d_k}{d_ad_i}}\sum_{i',k'\in\Irr(\eK)}\sum_{\sigma'}\sqrt{\frac{d_{k'}}{d_{a'}d_{i'}}}\delta_{u,v'}\delta_{z,y'}\delta_{i,k'} \sum_{t,\beta}\sum_{\alpha}[{_{\eN}}F_{aa'u'}^{v}]_{u\zeta\zeta'}^{t\alpha\beta}  \nonumber \\
    & \quad\sum_{\theta}[{_{\eK}}F_k^{aa'i'}]_{i\sigma\sigma'}^{t\theta\beta}\sqrt{\frac{d_{a}d_{a'}}{d_t}}
     \sum_{t',\beta'}\sum_{\alpha'}[{_{\eK}}F_{aa'i'}^{k}]_{i\sigma\sigma'}^{t'\alpha'\beta'}\sum_{\theta'}[{_{\eM}}F_y^{aa'z'}]_{z\nu\nu'}^{t'\theta'\beta'}\sqrt{\frac{d_{a}d_{a'}}{d_{t'}}} \;\begin{aligned}
        \begin{tikzpicture}[scale=0.65]
        \begin{scope}
            \fill[gray!15]
                (0,1.5) arc[start angle=90, end angle=270, radius=1.5] -- 
                (0,-0.5) arc[start angle=270, end angle=90, radius=0.5] -- cycle;
        \end{scope}
        \draw[dotted] (0,1.5) arc[start angle=90, end angle=270, radius=1.5]; 
        \draw[dotted] (0,0.5) arc[start angle=90, end angle=270, radius=0.5]; 
            \draw[line width=1pt,cyan] (0,0.5)--(0,1.5);
            \draw[line width=1pt,teal] (0,-0.5)--(0,-1.5);
             \draw[red] (0,0.8) arc[start angle=90, end angle=270, radius=0.8];
            \node[ line width=0.6pt, dashed, draw opacity=0.5] (a) at (-1,0){$\scriptstyle t$};
            \node[ line width=0.6pt, dashed, draw opacity=0.5] (a) at (-0.2,-1.1){$\scriptstyle k$};
            \node[ line width=0.6pt, dashed, draw opacity=0.5] (a) at (0.25,-0.8){$\scriptstyle \theta$};
            \node[ line width=0.6pt, dashed, draw opacity=0.5] (a) at (0,-0.25){$\scriptstyle i'$};
            \node[ line width=0.6pt, dashed, draw opacity=0.5] (a) at (0,0.3){$\scriptstyle u'$};
            \node[ line width=0.6pt, dashed, draw opacity=0.5] (a) at (-0.2,1.1){$\scriptstyle v$};
            \node[ line width=0.6pt, dashed, draw opacity=0.5] (a) at (0.25,0.8){$\scriptstyle \alpha$};
        \end{tikzpicture}
    \end{aligned}\otimes \; 
    \begin{aligned}
        \begin{tikzpicture}[scale=0.65]
        \begin{scope}
            \fill[gray!15]
                (0,1.5) arc[start angle=90, end angle=270, radius=1.5] -- 
                (0,-0.5) arc[start angle=270, end angle=90, radius=0.5] -- cycle;
        \end{scope}
        \draw[dotted] (0,1.5) arc[start angle=90, end angle=270, radius=1.5]; 
        \draw[dotted] (0,0.5) arc[start angle=90, end angle=270, radius=0.5]; 
            \draw[line width=1pt,teal] (0,0.5)--(0,1.5);
            \draw[line width=1pt,violet] (0,-0.5)--(0,-1.5);
             \draw[red] (0,0.8) arc[start angle=90, end angle=270, radius=0.8];
            \node[ line width=0.6pt, dashed, draw opacity=0.5] (a) at (-1,0){$\scriptstyle t'$};
            \node[ line width=0.6pt, dashed, draw opacity=0.5] (a) at (-0.2,-1.1){$\scriptstyle y$};
            \node[ line width=0.6pt, dashed, draw opacity=0.5] (a) at (0.3,-0.8){$\scriptstyle \theta'$};
            \node[ line width=0.6pt, dashed, draw opacity=0.5] (a) at (0,-0.25){$\scriptstyle z'$};
            \node[ line width=0.6pt, dashed, draw opacity=0.5] (a) at (0,0.3){$\scriptstyle i'$};
            \node[ line width=0.6pt, dashed, draw opacity=0.5] (a) at (-0.2,1.1){$\scriptstyle k$};
            \node[ line width=0.6pt, dashed, draw opacity=0.5] (a) at (0.3,0.8){$\scriptstyle \alpha'$};
        \end{tikzpicture}
    \end{aligned} \nonumber \\
    & = \delta_{u,v'}\delta_{z,y'}\sum_{k}\sum_{i'}\sum_{t,\alpha,\beta}\sum_{t',\alpha',\beta'}\sum_{\theta}\sum_{\theta'}\underbrace{\sum_{i,\sigma,\sigma'}[{_{\eK}}F_k^{aa'i'}]_{i\sigma\sigma'}^{t\theta\beta}[{_{\eK}}F_{aa'i'}^{k}]_{i\sigma\sigma'}^{t'\alpha'\beta'}}_{\delta_{t,t'}\delta_{\theta,\alpha'}\delta_{\beta,\beta'}} \nonumber  \\
    & \quad\;
     [{_{\eN}}F_{aa'u'}^{v}]_{u\zeta\zeta'}^{t\alpha\beta}[{_{\eM}}F_y^{aa'z'}]_{z\nu\nu'}^{t'\theta'\beta'}\sqrt{\frac{d_{a}d_{a'}}{d_t}}\sqrt{\frac{d_{a}d_{a'}}{d_{t'}}}\sqrt{\frac{d_k}{d_ad_i}}\sqrt{\frac{d_i}{d_{a'}d_{i'}}} \;\begin{aligned}
        \begin{tikzpicture}[scale=0.65]
        \begin{scope}
            \fill[gray!15]
                (0,1.5) arc[start angle=90, end angle=270, radius=1.5] -- 
                (0,-0.5) arc[start angle=270, end angle=90, radius=0.5] -- cycle;
        \end{scope}
        \draw[dotted] (0,1.5) arc[start angle=90, end angle=270, radius=1.5]; 
        \draw[dotted] (0,0.5) arc[start angle=90, end angle=270, radius=0.5]; 
            \draw[line width=1pt,cyan] (0,0.5)--(0,1.5);
            \draw[line width=1pt,teal] (0,-0.5)--(0,-1.5);
             \draw[red] (0,0.8) arc[start angle=90, end angle=270, radius=0.8];
            \node[ line width=0.6pt, dashed, draw opacity=0.5] (a) at (-1,0){$\scriptstyle t$};
            \node[ line width=0.6pt, dashed, draw opacity=0.5] (a) at (-0.2,-1.1){$\scriptstyle k$};
            \node[ line width=0.6pt, dashed, draw opacity=0.5] (a) at (0.25,-0.8){$\scriptstyle \theta$};
            \node[ line width=0.6pt, dashed, draw opacity=0.5] (a) at (0,-0.25){$\scriptstyle i'$};
            \node[ line width=0.6pt, dashed, draw opacity=0.5] (a) at (0,0.3){$\scriptstyle u'$};
            \node[ line width=0.6pt, dashed, draw opacity=0.5] (a) at (-0.2,1.1){$\scriptstyle v$};
            \node[ line width=0.6pt, dashed, draw opacity=0.5] (a) at (0.25,0.8){$\scriptstyle \alpha$};
        \end{tikzpicture}
    \end{aligned}\otimes \; 
    \begin{aligned}
        \begin{tikzpicture}[scale=0.65]
        \begin{scope}
            \fill[gray!15]
                (0,1.5) arc[start angle=90, end angle=270, radius=1.5] -- 
                (0,-0.5) arc[start angle=270, end angle=90, radius=0.5] -- cycle;
        \end{scope}
        \draw[dotted] (0,1.5) arc[start angle=90, end angle=270, radius=1.5]; 
        \draw[dotted] (0,0.5) arc[start angle=90, end angle=270, radius=0.5]; 
            \draw[line width=1pt,teal] (0,0.5)--(0,1.5);
            \draw[line width=1pt,violet] (0,-0.5)--(0,-1.5);
             \draw[red] (0,0.8) arc[start angle=90, end angle=270, radius=0.8];
            \node[ line width=0.6pt, dashed, draw opacity=0.5] (a) at (-1,0){$\scriptstyle t'$};
            \node[ line width=0.6pt, dashed, draw opacity=0.5] (a) at (-0.2,-1.1){$\scriptstyle y$};
            \node[ line width=0.6pt, dashed, draw opacity=0.5] (a) at (0.3,-0.8){$\scriptstyle \theta'$};
            \node[ line width=0.6pt, dashed, draw opacity=0.5] (a) at (0,-0.25){$\scriptstyle z'$};
            \node[ line width=0.6pt, dashed, draw opacity=0.5] (a) at (0,0.3){$\scriptstyle i'$};
            \node[ line width=0.6pt, dashed, draw opacity=0.5] (a) at (-0.2,1.1){$\scriptstyle k$};
            \node[ line width=0.6pt, dashed, draw opacity=0.5] (a) at (0.3,0.8){$\scriptstyle \alpha'$};
        \end{tikzpicture}
    \end{aligned} \nonumber \\
    & = \delta_{u,v'}\delta_{z,y'}\sum_{t,\beta}\sum_{\alpha}\sum_{\theta'} [{_{\eN}}F_{aa'u'}^{v}]_{u\zeta\zeta'}^{t\alpha\beta}[{_{\eM}}F_y^{aa'z'}]_{z\nu\nu'}^{t\theta'\beta}\sqrt{\frac{d_ad_{a'}}{d_t}}  \nonumber  \\
    & \quad \sum_{i',k\in\Irr(\eK)}\sum_{\theta}\sqrt{\frac{d_k}{d_td_{i'}}} \;\begin{aligned}
        \begin{tikzpicture}[scale=0.65]
        \begin{scope}
            \fill[gray!15]
                (0,1.5) arc[start angle=90, end angle=270, radius=1.5] -- 
                (0,-0.5) arc[start angle=270, end angle=90, radius=0.5] -- cycle;
        \end{scope}
        \draw[dotted] (0,1.5) arc[start angle=90, end angle=270, radius=1.5]; 
        \draw[dotted] (0,0.5) arc[start angle=90, end angle=270, radius=0.5]; 
            \draw[line width=1pt,cyan] (0,0.5)--(0,1.5);
            \draw[line width=1pt,teal] (0,-0.5)--(0,-1.5);
             \draw[red] (0,0.8) arc[start angle=90, end angle=270, radius=0.8];
            \node[ line width=0.6pt, dashed, draw opacity=0.5] (a) at (-1,0){$\scriptstyle t$};
            \node[ line width=0.6pt, dashed, draw opacity=0.5] (a) at (-0.2,-1.1){$\scriptstyle k$};
            \node[ line width=0.6pt, dashed, draw opacity=0.5] (a) at (0.25,-0.8){$\scriptstyle \theta$};
            \node[ line width=0.6pt, dashed, draw opacity=0.5] (a) at (0,-0.25){$\scriptstyle i'$};
            \node[ line width=0.6pt, dashed, draw opacity=0.5] (a) at (0,0.3){$\scriptstyle u'$};
            \node[ line width=0.6pt, dashed, draw opacity=0.5] (a) at (-0.2,1.1){$\scriptstyle v$};
            \node[ line width=0.6pt, dashed, draw opacity=0.5] (a) at (0.25,0.8){$\scriptstyle \alpha$};
        \end{tikzpicture}
    \end{aligned}\otimes \; 
    \begin{aligned}
        \begin{tikzpicture}[scale=0.65]
        \begin{scope}
            \fill[gray!15]
                (0,1.5) arc[start angle=90, end angle=270, radius=1.5] -- 
                (0,-0.5) arc[start angle=270, end angle=90, radius=0.5] -- cycle;
        \end{scope}
        \draw[dotted] (0,1.5) arc[start angle=90, end angle=270, radius=1.5]; 
        \draw[dotted] (0,0.5) arc[start angle=90, end angle=270, radius=0.5]; 
            \draw[line width=1pt,teal] (0,0.5)--(0,1.5);
            \draw[line width=1pt,violet] (0,-0.5)--(0,-1.5);
             \draw[red] (0,0.8) arc[start angle=90, end angle=270, radius=0.8];
            \node[ line width=0.6pt, dashed, draw opacity=0.5] (a) at (-1,0){$\scriptstyle t$};
            \node[ line width=0.6pt, dashed, draw opacity=0.5] (a) at (-0.2,-1.1){$\scriptstyle y$};
            \node[ line width=0.6pt, dashed, draw opacity=0.5] (a) at (0.3,-0.8){$\scriptstyle \theta'$};
            \node[ line width=0.6pt, dashed, draw opacity=0.5] (a) at (0,-0.25){$\scriptstyle z'$};
            \node[ line width=0.6pt, dashed, draw opacity=0.5] (a) at (0,0.3){$\scriptstyle i'$};
            \node[ line width=0.6pt, dashed, draw opacity=0.5] (a) at (-0.2,1.1){$\scriptstyle k$};
            \node[ line width=0.6pt, dashed, draw opacity=0.5] (a) at (0.3,0.8){$\scriptstyle \theta$};
        \end{tikzpicture}
    \end{aligned} = \delta(XY).     \notag
\end{align}

\noindent
\textit{Proof of} \eqref{eq:epsilon-rho-1} and \eqref{eq:epsilon-rho-2}. We prove only the first line; the second line follows similarly. 
    For basis tube $X$, one has 
    \begin{equation}
        \notag
        \varepsilon_R(X) = \delta_{z,u}\delta_{y,v}\delta_{\nu,\zeta} \sqrt{\frac{d_ad_u}{d_v}} \sum_k\;  
        \begin{aligned}
            \begin{tikzpicture}[scale=0.65]
                \fill[gray!15] (0,1.2) arc[start angle=90, end angle=270, radius=1.2] -- (0,-0.5) arc[start angle=270, end angle=90, radius=0.5] -- cycle;
        \draw[dotted] (0,1.2) arc[start angle=90, end angle=270, radius=1.2]; 
        \draw[dotted] (0,0.5) arc[start angle=90, end angle=270, radius=0.5];
            \draw[line width=1pt,violet] (0,0.5)--(0,1.2);
            \draw[line width=1pt,violet] (0,-0.5)--(0,-1.2);
            \draw[red] (0,0.8) arc[start angle=90, end angle=270, radius=0.8];
            \node[ line width=0.6pt, dashed, draw opacity=0.5] (a) at (-1,0){$\scriptstyle {\mathbb{1}}$};
            \node[ line width=0.6pt, dashed, draw opacity=0.5] (a) at (0.25,-0.8){$\scriptstyle u$};
            \node[ line width=0.6pt, dashed, draw opacity=0.5] (a) at (0.25,0.8){$\scriptstyle k$};
        \end{tikzpicture}
    \end{aligned}\;. 
    \end{equation}
    A similar expression holds for $\varepsilon'_R(X)$, but the tube element on the right-hand side lies in $\Tube({}_\eC\eM;{}_\eC\eN)$.  Hence by definition, we obtain 
    \begin{equation}
        \delta(\varepsilon_R(X))  =  \delta_{z,u}\delta_{y,v}\delta_{\nu,\zeta} \sqrt{\frac{d_ad_u}{d_v}} \sum_k\sum_t\;  
        \begin{aligned}
            \begin{tikzpicture}[scale=0.65]
                \fill[gray!15] (0,1.2) arc[start angle=90, end angle=270, radius=1.2] -- (0,-0.5) arc[start angle=270, end angle=90, radius=0.5] -- cycle;
        \draw[dotted] (0,1.2) arc[start angle=90, end angle=270, radius=1.2]; 
        \draw[dotted] (0,0.5) arc[start angle=90, end angle=270, radius=0.5];
            \draw[line width=1pt,violet] (0,0.5)--(0,1.2);
            \draw[line width=1pt,cyan] (0,-0.5)--(0,-1.2);
            \draw[red] (0,0.8) arc[start angle=90, end angle=270, radius=0.8];
            \node[ line width=0.6pt, dashed, draw opacity=0.5] (a) at (-1,0){$\scriptstyle {\mathbb{1}}$};
            \node[ line width=0.6pt, dashed, draw opacity=0.5] (a) at (0.25,-0.8){$\scriptstyle t$};
            \node[ line width=0.6pt, dashed, draw opacity=0.5] (a) at (0.25,0.8){$\scriptstyle k$};
        \end{tikzpicture}
    \end{aligned}\otimes \; \begin{aligned}
            \begin{tikzpicture}[scale=0.65]
                \fill[gray!15] (0,1.2) arc[start angle=90, end angle=270, radius=1.2] -- (0,-0.5) arc[start angle=270, end angle=90, radius=0.5] -- cycle;
        \draw[dotted] (0,1.2) arc[start angle=90, end angle=270, radius=1.2]; 
        \draw[dotted] (0,0.5) arc[start angle=90, end angle=270, radius=0.5];
            \draw[line width=1pt,cyan] (0,0.5)--(0,1.2);
            \draw[line width=1pt,violet] (0,-0.5)--(0,-1.2);
            \draw[red] (0,0.8) arc[start angle=90, end angle=270, radius=0.8];
            \node[ line width=0.6pt, dashed, draw opacity=0.5] (a) at (-1,0){$\scriptstyle {\mathbb{1}}$};
            \node[ line width=0.6pt, dashed, draw opacity=0.5] (a) at (0.25,-0.8){$\scriptstyle u$};
            \node[ line width=0.6pt, dashed, draw opacity=0.5] (a) at (0.25,0.8){$\scriptstyle t$};
        \end{tikzpicture}
    \end{aligned}\;, 
    \end{equation}
    which is clearly equal to $1^{\langle 1 \rangle}\otimes 1^{\langle 2 \rangle}\varepsilon'_R(X)$.


\noindent
\textit{Proof of} \eqref{eq:unit-counit-1}-\eqref{eq:unit-counit-4}. Let us prove the first line. The second equality of \eqref{eq:unit-counit-1} is just from definition. It suffices to show the first equality:
    \begin{equation} \label{eq:delta-unit-identity}
        1^{\langle 1\rangle} X\otimes 1^{\langle 2\rangle}= X^{\langle 1\rangle} \otimes X^{\langle 2\rangle} s(X^{\langle 3\rangle}). 
    \end{equation}
    Consider the basis tube 
    \begin{equation}
        \notag
        X = \begin{aligned}

    \end{aligned}\;. 
    \end{equation}
    Then the right-hand side of \eqref{eq:delta-unit-identity} is equal to
\begin{align}
\delta(X^{[-1]})\bigl(1\otimes s(X^{[0]})\bigr) & = \sum_{i,j,\sigma}\sqrt{\frac{d_i}{d_{a}d_j}} \,\delta\left(\begin{aligned}
        %
    \end{aligned}\;. \label{eq:delta-unit-expanded}
\end{align}
By inverse parallel move, the multiplication is 
\begin{align}
    \sum_{j,\sigma}\sqrt{\frac{d_j}{d_{a}d_i}}\;
    \begin{aligned}
        %
    \end{aligned}\;. \label{eq:inverse-parallel}
\end{align}
Plugging this into Eq.~\eqref{eq:delta-unit-expanded}, we obtain
\begin{align}
    \delta(X^{[-1]})\bigl(1\otimes s(X^{[0]})\bigr) & = \sum_{s,t,\theta}\sqrt{\frac{d_s}{d_{a}d_t}} \sum_i \delta_{t,v}\delta_{s,u}\delta_{\theta,\zeta} \sqrt{\frac{d_{a}d_{v}}{d_{u}}} \; \begin{aligned}
        %
    \end{aligned}\;. \nonumber
\end{align}
It is ready to see that the expression above is just $\delta(1)(X\otimes 1)$ as desired. The remaining three lines can be showed similarly.

\noindent
\textit{Proof of} \eqref{eq:antipode-antimultiplicative}. The first assertion $s(1)=1$ is immediate. Let us prove the second one. 
    For general $X,Y$, using Eq.~\eqref{eq:product-basis}, one has
    \begin{equation}
        \notag
        s(XY)=\delta_{u,v'}\delta_{z,y'}\sum_{t,\beta}\sum_{\alpha}[{_{\eN}}F_{aa'u'}^{v}]_{u\zeta\zeta'}^{t\alpha\beta}\sum_{\theta}[{_{\eM}}F_y^{aa'z'}]_{z\nu\nu'}^{t\theta\beta}\sqrt{\frac{d_{a}d_{a'}}{d_t}} \frac{d_{u'}}{d_v} \;\begin{aligned}

    \end{aligned}\;. 
    \end{equation}
    On the other hand, one computes 
    \begin{align} 
    s(Y)s(X)&=\frac{d_{u'}}{d_{v'}}\;\begin{aligned}%
    \end{aligned}, \nonumber
\end{align}
which is equal to $s(XY)$.

\noindent
\textit{Proof of} \eqref{eq:delta-antipode-compatibility}. The assertion follows by direct computation:
    \begin{align}
        \delta(s(X)) & = \frac{d_u}{d_v}\sum_{t,s,\theta} \sqrt{\frac{d_s}{d_ad_t}}\;\begin{aligned}%
    \end{aligned}\;.     \notag
    \end{align}
    Two expressions are equal by inspection.

\noindent
\textit{Proof of} \eqref{eq:decomp-antipode}. By definition, $s(X^{\langle 1\rangle})X^{\langle 2\rangle}s(X^{\langle 3\rangle}) = \varepsilon'_R(X^{[-1]})S(X^{[0]})$.
    It is easy to compute that 
    \begin{equation}
        \notag
        \varepsilon_R'\left(\begin{aligned}
            %
    \end{aligned}\;, 
    \end{equation}
    whose proof is similar to that of \eqref{eq:inverse-parallel}. Therefore, one has 
    \begin{align}
        \varepsilon'_R(X^{[-1]})s(X^{[0]}) & = \sum_{i,j,\sigma}\sqrt{\frac{d_i}{d_ad_j}}  \; \varepsilon_R'\left(\begin{aligned}
            %
    \end{aligned} = s(X).     \notag
    \end{align}

\noindent
\textit{Proof of} \eqref{eq:antipode-counital}. It follows from $\delta(s(X)) = \tau\circ (s\otimes s)\circ\delta(X)$ that $\delta\bigl(s(X)\bigr) = s\bigl(X^{\langle 2 \rangle}\bigr) \otimes s\bigl(X^{\langle 1 \rangle}\bigr)$. Then we have 
    \begin{align}
        \varepsilon'_L(s(X)) &= s\bigl(X^{\langle 2 \rangle}\bigr) s\bigl(s\bigl(X^{\langle 1 \rangle}\bigr)\bigr) = s\bigl(s\bigl(X^{\langle 1 \rangle}\bigr)X^{\langle 2 \rangle}\bigr) = s\bigl(\varepsilon'_R(X)\bigr),     \notag\\
        \varepsilon'_R(s(X)) &= s\bigl(s\bigl(X^{\langle 2 \rangle}\bigr)\bigr) s\bigl(X^{\langle 1 \rangle}\bigr) = s\bigl(X^{\langle 1 \rangle}s\bigl(X^{\langle 2 \rangle}\bigr)\bigr) = s\bigl(\varepsilon'_L(X)\bigr),     \notag
    \end{align}
    for any $X\in \mathbf{Tube}({}_\eC\eM)$. Hence the assertion follows.

\subsubsection*{Domain wall case.}

\noindent
\textit{Proof of} \eqref{eq:coassociativity-DW}. For a basis tube $X\in \Tube({}_\eC\eN_\eD;{}_\eC\eM_\eD)$, both expressions, applied to $X$, are equal to 
    \begin{equation}
        \notag
        \sum_{i,j,k,\sigma,\varrho} \sqrt{\frac{d_i}{d_ad_kd_b}} \sum_{r,t,s,\phi,\psi} \sqrt{\frac{d_r}{d_ad_td_b}} \; \begin{aligned}
        \begin{tikzpicture}[scale=0.65]
        \filldraw[black!60, fill=gray!15, dotted, even odd rule] (0,0) circle[radius=0.5] (0,0) circle[radius=1.5];
             \draw[line width=1pt,violet] (0,0.5)--(0,1.5);
             \draw[line width=1pt,cyan] (0,-0.5)--(0,-1.5);
             \draw[red] (0,0.8) arc[start angle=90, end angle=270, radius=0.8];
             \draw[blue] (0,1.3) arc[start angle=90, end angle=-90, radius=1.3];
            \node[ line width=0.6pt, dashed, draw opacity=0.5] (a) at (0,1.7){$\scriptstyle w$};
             \node[ line width=0.6pt, dashed, draw opacity=0.5] (a) at (0,-1.7){$\scriptstyle i$};
            \node[ line width=0.6pt, dashed, draw opacity=0.5] (a) at (-1,0){$\scriptstyle a$};
            \node[ line width=0.6pt, dashed, draw opacity=0.5] (a) at (1.1,0){$\scriptstyle b$};
            \node[ line width=0.6pt, dashed, draw opacity=0.5] (a) at (-0.2,-1){$\scriptstyle j$};
            \node[ line width=0.6pt, dashed, draw opacity=0.5] (a) at (-0.2,-1.35){$\scriptstyle \sigma$};
            \node[ line width=0.6pt, dashed, draw opacity=0.5] (a) at (0.2,-0.8){$\scriptstyle \varrho$};
            \node[ line width=0.6pt, dashed, draw opacity=0.5] (a) at (0,-0.3){$\scriptstyle k$};
            \node[ line width=0.6pt, dashed, draw opacity=0.5] (a) at (0,0.3){$\scriptstyle u$};
            \node[ line width=0.6pt, dashed, draw opacity=0.5] (a) at (-0.2,1){$\scriptstyle v$};
            \node[ line width=0.6pt, dashed, draw opacity=0.5] (a) at (-0.2,1.3){$\scriptstyle \gamma$};
            \node[ line width=0.6pt, dashed, draw opacity=0.5] (a) at (0.2,0.8){$\scriptstyle \zeta$};
        \end{tikzpicture}
    \end{aligned} \otimes \begin{aligned}
        \begin{tikzpicture}[scale=0.65]
        \filldraw[black!60, fill=gray!15, dotted, even odd rule] (0,0) circle[radius=0.5] (0,0) circle[radius=1.5];
             \draw[line width=1pt,cyan] (0,0.5)--(0,1.5);
             \draw[line width=1pt,violet] (0,-0.5)--(0,-1.5);
             \draw[red] (0,0.8) arc[start angle=90, end angle=270, radius=0.8];
             \draw[blue] (0,1.3) arc[start angle=90, end angle=-90, radius=1.3];
            \node[ line width=0.6pt, dashed, draw opacity=0.5] (a) at (0,1.7){$\scriptstyle i$};
             \node[ line width=0.6pt, dashed, draw opacity=0.5] (a) at (0,-1.7){$\scriptstyle r$};
            \node[ line width=0.6pt, dashed, draw opacity=0.5] (a) at (-1,0){$\scriptstyle a$};
            \node[ line width=0.6pt, dashed, draw opacity=0.5] (a) at (1.1,0){$\scriptstyle b$};
            \node[ line width=0.6pt, dashed, draw opacity=0.5] (a) at (-0.2,-1){$\scriptstyle s$};
            \node[ line width=0.6pt, dashed, draw opacity=0.5] (a) at (-0.2,-1.35){$\scriptstyle \psi$};
            \node[ line width=0.6pt, dashed, draw opacity=0.5] (a) at (0.2,-0.8){$\scriptstyle \phi$};
            \node[ line width=0.6pt, dashed, draw opacity=0.5] (a) at (0,-0.3){$\scriptstyle t$};
            \node[ line width=0.6pt, dashed, draw opacity=0.5] (a) at (0,0.3){$\scriptstyle k$};
            \node[ line width=0.6pt, dashed, draw opacity=0.5] (a) at (-0.2,1){$\scriptstyle j$};
            \node[ line width=0.6pt, dashed, draw opacity=0.5] (a) at (-0.2,1.3){$\scriptstyle \sigma$};
            \node[ line width=0.6pt, dashed, draw opacity=0.5] (a) at (0.2,0.8){$\scriptstyle \varrho$};
        \end{tikzpicture}
    \end{aligned} \otimes \begin{aligned}
        \begin{tikzpicture}[scale=0.65]
        \filldraw[black!60, fill=gray!15, dotted, even odd rule] (0,0) circle[radius=0.5] (0,0) circle[radius=1.5];
             \draw[line width=1pt,violet] (0,0.5)--(0,1.5);
             \draw[line width=1pt,cyan] (0,-0.5)--(0,-1.5);
             \draw[red] (0,0.8) arc[start angle=90, end angle=270, radius=0.8];
             \draw[blue] (0,1.3) arc[start angle=90, end angle=-90, radius=1.3];
            \node[ line width=0.6pt, dashed, draw opacity=0.5] (a) at (0,1.7){$\scriptstyle r$};
             \node[ line width=0.6pt, dashed, draw opacity=0.5] (a) at (0,-1.7){$\scriptstyle x$};
            \node[ line width=0.6pt, dashed, draw opacity=0.5] (a) at (-1,0){$\scriptstyle a$};
            \node[ line width=0.6pt, dashed, draw opacity=0.5] (a) at (1.1,0){$\scriptstyle b$};
            \node[ line width=0.6pt, dashed, draw opacity=0.5] (a) at (-0.2,-1){$\scriptstyle y$};
            \node[ line width=0.6pt, dashed, draw opacity=0.5] (a) at (-0.2,-1.35){$\scriptstyle \mu$};
            \node[ line width=0.6pt, dashed, draw opacity=0.5] (a) at (0.2,-0.8){$\scriptstyle \nu$};
            \node[ line width=0.6pt, dashed, draw opacity=0.5] (a) at (0,-0.3){$\scriptstyle z$};
            \node[ line width=0.6pt, dashed, draw opacity=0.5] (a) at (0,0.3){$\scriptstyle t$};
            \node[ line width=0.6pt, dashed, draw opacity=0.5] (a) at (-0.2,1){$\scriptstyle s$};
            \node[ line width=0.6pt, dashed, draw opacity=0.5] (a) at (-0.2,1.3){$\scriptstyle \psi$};
            \node[ line width=0.6pt, dashed, draw opacity=0.5] (a) at (0.2,0.8){$\scriptstyle \phi$};
        \end{tikzpicture}
    \end{aligned} \;. 
    \end{equation}

\noindent
\textit{Proof of} \eqref{eq:unit-trifold-DW}. By a direct calculation, both expressions are equal to  
    \begin{equation}
        \notag
        \sum_{x,z\in\Irr(\eM)}\sum_{u,v\in \Irr(\eN)}\;  \begin{aligned}
        \begin{tikzpicture}[scale=0.65]
            \filldraw[black!60, fill=gray!15, dotted, even odd rule] (0,0) circle[radius=0.5] (0,0) circle[radius=1.5];
            \draw[line width=1pt,cyan] (0,0.5)--(0,1.5);
            \draw[line width=1pt,violet] (0,-0.5)--(0,-1.5);
            \draw[red, dotted, thick] (0,0.8) arc[start angle=90, end angle=270, radius=0.8];
            \draw[blue, dotted, thick] (0,1.3) arc[start angle=90, end angle=-90, radius=1.3];
            \node[ line width=0.6pt, dashed, draw opacity=0.5] (a) at (0,-0.3){$\scriptstyle z$};
            \node[ line width=0.6pt, dashed, draw opacity=0.5] (a) at (0,0.3){$\scriptstyle u$};
        \end{tikzpicture}
    \end{aligned} \otimes \begin{aligned}
        \begin{tikzpicture}[scale=0.65]
            \filldraw[black!60, fill=gray!15, dotted, even odd rule] (0,0) circle[radius=0.5] (0,0) circle[radius=1.5];
            \draw[line width=1pt,violet] (0,0.5)--(0,1.5);
            \draw[line width=1pt,cyan] (0,-0.5)--(0,-1.5);
            \draw[red, dotted, thick] (0,0.8) arc[start angle=90, end angle=270, radius=0.8];
            \draw[blue, dotted, thick] (0,1.3) arc[start angle=90, end angle=-90, radius=1.3];
            \node[ line width=0.6pt, dashed, draw opacity=0.5] (a) at (0,-0.3){$\scriptstyle v$};
            \node[ line width=0.6pt, dashed, draw opacity=0.5] (a) at (0,0.3){$\scriptstyle z$};
        \end{tikzpicture}
    \end{aligned} \otimes \begin{aligned}
        \begin{tikzpicture}[scale=0.65]
            \filldraw[black!60, fill=gray!15, dotted, even odd rule] (0,0) circle[radius=0.5] (0,0) circle[radius=1.5];
            \draw[line width=1pt,cyan] (0,0.5)--(0,1.5);
            \draw[line width=1pt,violet] (0,-0.5)--(0,-1.5);
            \draw[red, dotted, thick] (0,0.8) arc[start angle=90, end angle=270, radius=0.8];
            \draw[blue, dotted, thick] (0,1.3) arc[start angle=90, end angle=-90, radius=1.3];
            \node[ line width=0.6pt, dashed, draw opacity=0.5] (a) at (0,-0.3){$\scriptstyle x$};
            \node[ line width=0.6pt, dashed, draw opacity=0.5] (a) at (0,0.3){$\scriptstyle v$};
        \end{tikzpicture}
    \end{aligned}\;.
    \end{equation}

\noindent
\textit{Proof of} \eqref{eq:preserve-product-DW}. 
Using definition of $\delta$ and expansion of multiplication as in \eqref{eq:bi_tube_prod_2}, $\delta(X)\delta(Y)$ is equal to 
\begin{align}
        & \; \sum_{i,j,k\in \Irr(\eK)}\sum_{\sigma,\rho}\sqrt{\frac{d_i}{d_ad_kd_b}} \sum_{i',j',k'\in \Irr(\eK)}\sum_{\sigma',\rho'}\sqrt{\frac{d_{i'}}{d_{a'}d_{k'}d_{b'}}} \; \nonumber\\
        &\quad \quad \begin{aligned}
        \begin{tikzpicture}[scale=0.65]
        \filldraw[black!60, fill=gray!15, dotted, even odd rule] (0,0) circle[radius=0.5] (0,0) circle[radius=1.5];
             \draw[line width=1pt,cyan] (0,0.5)--(0,1.5);
             \draw[line width=1pt,teal] (0,-0.5)--(0,-1.5);
             \draw[red] (0,0.8) arc[start angle=90, end angle=270, radius=0.8];
             \draw[blue] (0,1.3) arc[start angle=90, end angle=-90, radius=1.3];
            \node[ line width=0.6pt, dashed, draw opacity=0.5] (a) at (0,1.7){$\scriptstyle w$};
             \node[ line width=0.6pt, dashed, draw opacity=0.5] (a) at (0,-1.7){$\scriptstyle i$};
            \node[ line width=0.6pt, dashed, draw opacity=0.5] (a) at (-1,0){$\scriptstyle a$};
            \node[ line width=0.6pt, dashed, draw opacity=0.5] (a) at (1.1,0){$\scriptstyle b$};
            \node[ line width=0.6pt, dashed, draw opacity=0.5] (a) at (-0.2,-1){$\scriptstyle j$};
            \node[ line width=0.6pt, dashed, draw opacity=0.5] (a) at (-0.2,-1.35){$\scriptstyle \sigma$};
            \node[ line width=0.6pt, dashed, draw opacity=0.5] (a) at (0.2,-0.8){$\scriptstyle \rho$};
            \node[ line width=0.6pt, dashed, draw opacity=0.5] (a) at (0,-0.3){$\scriptstyle k$};
            \node[ line width=0.6pt, dashed, draw opacity=0.5] (a) at (0,0.3){$\scriptstyle u$};
            \node[ line width=0.6pt, dashed, draw opacity=0.5] (a) at (-0.2,1){$\scriptstyle v$};
            \node[ line width=0.6pt, dashed, draw opacity=0.5] (a) at (-0.2,1.3){$\scriptstyle \gamma$};
            \node[ line width=0.6pt, dashed, draw opacity=0.5] (a) at (0.2,0.8){$\scriptstyle \zeta$};
        \end{tikzpicture}
    \end{aligned} \;
    \begin{aligned}
        \begin{tikzpicture}[scale=0.65]
        \filldraw[black!60, fill=gray!15, dotted, even odd rule] (0,0) circle[radius=0.5] (0,0) circle[radius=1.5];
             \draw[line width=1pt,cyan] (0,0.5)--(0,1.5);
             \draw[line width=1pt,teal] (0,-0.5)--(0,-1.5);
             \draw[red] (0,0.8) arc[start angle=90, end angle=270, radius=0.8];
             \draw[blue] (0,1.3) arc[start angle=90, end angle=-90, radius=1.3];
            \node[ line width=0.6pt, dashed, draw opacity=0.5] (a) at (0,1.76){$\scriptstyle w'$};
             \node[ line width=0.6pt, dashed, draw opacity=0.5] (a) at (0,-1.7){$\scriptstyle i'$};
            \node[ line width=0.6pt, dashed, draw opacity=0.5] (a) at (-1,0){$\scriptstyle a'$};
            \node[ line width=0.6pt, dashed, draw opacity=0.5] (a) at (1.1,0){$\scriptstyle b'$};
            \node[ line width=0.6pt, dashed, draw opacity=0.5] (a) at (-0.2,-1){$\scriptstyle j'$};
            \node[ line width=0.6pt, dashed, draw opacity=0.5] (a) at (-0.2,-1.35){$\scriptstyle \sigma'$};
            \node[ line width=0.6pt, dashed, draw opacity=0.5] (a) at (0.3,-0.8){$\scriptstyle \rho'$};
            \node[ line width=0.6pt, dashed, draw opacity=0.5] (a) at (0,-0.25){$\scriptstyle k'$};
            \node[ line width=0.6pt, dashed, draw opacity=0.5] (a) at (0,0.3){$\scriptstyle u'$};
            \node[ line width=0.6pt, dashed, draw opacity=0.5] (a) at (-0.2,1){$\scriptstyle v'$};
            \node[ line width=0.6pt, dashed, draw opacity=0.5] (a) at (-0.2,1.3){$\scriptstyle \gamma'$};
            \node[ line width=0.6pt, dashed, draw opacity=0.5] (a) at (0.3,0.8){$\scriptstyle \zeta'$};
        \end{tikzpicture}
    \end{aligned} \otimes \begin{aligned}
        \begin{tikzpicture}[scale=0.65]
        \filldraw[black!60, fill=gray!15, dotted, even odd rule] (0,0) circle[radius=0.5] (0,0) circle[radius=1.5];
             \draw[line width=1pt,teal] (0,0.5)--(0,1.5);
             \draw[line width=1pt,violet] (0,-0.5)--(0,-1.5);
             \draw[red] (0,0.8) arc[start angle=90, end angle=270, radius=0.8];
             \draw[blue] (0,1.3) arc[start angle=90, end angle=-90, radius=1.3];
            \node[ line width=0.6pt, dashed, draw opacity=0.5] (a) at (0,1.7){$\scriptstyle i$};
             \node[ line width=0.6pt, dashed, draw opacity=0.5] (a) at (0,-1.7){$\scriptstyle x$};
            \node[ line width=0.6pt, dashed, draw opacity=0.5] (a) at (-1,0){$\scriptstyle a$};
            \node[ line width=0.6pt, dashed, draw opacity=0.5] (a) at (1.1,0){$\scriptstyle b$};
            \node[ line width=0.6pt, dashed, draw opacity=0.5] (a) at (-0.2,-1){$\scriptstyle y$};
            \node[ line width=0.6pt, dashed, draw opacity=0.5] (a) at (-0.2,-1.35){$\scriptstyle \mu$};
            \node[ line width=0.6pt, dashed, draw opacity=0.5] (a) at (0.2,-0.8){$\scriptstyle \nu$};
            \node[ line width=0.6pt, dashed, draw opacity=0.5] (a) at (0,-0.3){$\scriptstyle z$};
            \node[ line width=0.6pt, dashed, draw opacity=0.5] (a) at (0,0.3){$\scriptstyle k$};
            \node[ line width=0.6pt, dashed, draw opacity=0.5] (a) at (-0.2,1){$\scriptstyle j$};
            \node[ line width=0.6pt, dashed, draw opacity=0.5] (a) at (-0.2,1.3){$\scriptstyle \sigma$};
            \node[ line width=0.6pt, dashed, draw opacity=0.5] (a) at (0.2,0.8){$\scriptstyle \rho$};
        \end{tikzpicture}
    \end{aligned}\;
    \begin{aligned}
        \begin{tikzpicture}[scale=0.65]
        \filldraw[black!60, fill=gray!15, dotted, even odd rule] (0,0) circle[radius=0.5] (0,0) circle[radius=1.5];
             \draw[line width=1pt,teal] (0,0.5)--(0,1.5);
             \draw[line width=1pt,violet] (0,-0.5)--(0,-1.5);
             \draw[red] (0,0.8) arc[start angle=90, end angle=270, radius=0.8];
             \draw[blue] (0,1.3) arc[start angle=90, end angle=-90, radius=1.3];
            \node[ line width=0.6pt, dashed, draw opacity=0.5] (a) at (0,1.75){$\scriptstyle i'$};
             \node[ line width=0.6pt, dashed, draw opacity=0.5] (a) at (0,-1.7){$\scriptstyle x'$};
            \node[ line width=0.6pt, dashed, draw opacity=0.5] (a) at (-1,0){$\scriptstyle a'$};
            \node[ line width=0.6pt, dashed, draw opacity=0.5] (a) at (1.1,0){$\scriptstyle b'$};
            \node[ line width=0.6pt, dashed, draw opacity=0.5] (a) at (-0.2,-1){$\scriptstyle y'$};
            \node[ line width=0.6pt, dashed, draw opacity=0.5] (a) at (-0.2,-1.35){$\scriptstyle \mu'$};
            \node[ line width=0.6pt, dashed, draw opacity=0.5] (a) at (0.3,-0.8){$\scriptstyle \nu'$};
            \node[ line width=0.6pt, dashed, draw opacity=0.5] (a) at (0,-0.25){$\scriptstyle z'$};
            \node[ line width=0.6pt, dashed, draw opacity=0.5] (a) at (0,0.3){$\scriptstyle k'$};
            \node[ line width=0.6pt, dashed, draw opacity=0.5] (a) at (-0.2,1){$\scriptstyle j'$};
            \node[ line width=0.6pt, dashed, draw opacity=0.5] (a) at (-0.2,1.3){$\scriptstyle \sigma'$};
            \node[ line width=0.6pt, dashed, draw opacity=0.5] (a) at (0.3,0.8){$\scriptstyle \rho'$};
        \end{tikzpicture}
    \end{aligned} \nonumber  \\
    & = \sum_{i,j,k\in \Irr(\eK)}\sum_{\sigma,\rho}\sqrt{\frac{d_i}{d_ad_kd_b}} \sum_{i',j',k'\in \Irr(\eK)}\sum_{\sigma',\rho'}\sqrt{\frac{d_{i'}}{d_{a'}d_{k'}d_{b'}}} \delta_{k,i'}\delta_{u,w'} \delta_{k,i'}\delta_{z,x'}\nonumber\\
    & \quad \quad \times \sum_{l,\alpha,\beta}[{_{\eN}}F_{av'b'}^v]^{l\alpha\beta}_{u\gamma'\zeta}\sum_{n,\kappa,\psi}[{_{\eK}}F_{j}^{aj'b'}]^{n\kappa\psi}_{k\sigma'\rho} \sum_{m,\tau,\varepsilon} [{_{\eN}}F^l_{aa'u'}]^{m\tau\varepsilon}_{v'\zeta'\alpha}\sum_{\theta}[{_{\eK}}F_n^{aa'k'}]^{m\tau\theta}_{j'\rho'\kappa} \nonumber \\
    & \quad \quad \times  \sum_{s,\delta,\phi}[{_{\eN}}F^w_{lb'b}]^{s\delta\phi}_{v\beta\gamma}\sum_{\xi}[{_{\eK}}F_i^{nb'b}]^{s\delta\xi}_{j\psi\sigma} \sqrt{\frac{d_ad_{a'}}{d_m}}\sqrt{\frac{d_bd_{b'}}{d_s}} \;\begin{aligned}
    \begin{tikzpicture}[scale=0.65]
    \filldraw[black!60, fill=gray!15, dotted, even odd rule] (0,0) circle[radius=0.5] (0,0) circle[radius=1.5];
         \draw[line width=1pt,cyan] (0,0.5)--(0,1.5);
         \draw[line width=1pt,teal] (0,-0.5)--(0,-1.5);
         \draw[red] (0,0.8) arc[start angle=90, end angle=270, radius=0.8];
         \draw[blue] (0,1.3) arc[start angle=90, end angle=-90, radius=1.3];
        \node[ line width=0.6pt, dashed, draw opacity=0.5] (a) at (0,1.7){$\scriptstyle w$};
         \node[ line width=0.6pt, dashed, draw opacity=0.5] (a) at (0,-1.7){$\scriptstyle i$};
        \node[ line width=0.6pt, dashed, draw opacity=0.5] (a) at (-1.1,0){$\scriptstyle m$};
        \node[ line width=0.6pt, dashed, draw opacity=0.5] (a) at (1.1,0){$\scriptstyle s$};
        \node[ line width=0.6pt, dashed, draw opacity=0.5] (a) at (-0.2,-1){$\scriptstyle n$};
        \node[ line width=0.6pt, dashed, draw opacity=0.5] (a) at (-0.2,-1.4){$\scriptstyle \xi$};
        \node[ line width=0.6pt, dashed, draw opacity=0.5] (a) at (0.2,-0.8){$\scriptstyle \theta$};
        \node[ line width=0.6pt, dashed, draw opacity=0.5] (a) at (0,-0.2){$\scriptstyle k'$};
        \node[ line width=0.6pt, dashed, draw opacity=0.5] (a) at (0,0.3){$\scriptstyle u'$};
        \node[ line width=0.6pt, dashed, draw opacity=0.5] (a) at (-0.2,1){$\scriptstyle l$};
        \node[ line width=0.6pt, dashed, draw opacity=0.5] (a) at (-0.2,1.3){$\scriptstyle \phi$};
        \node[ line width=0.6pt, dashed, draw opacity=0.5] (a) at (0.2,0.8){$\scriptstyle \varepsilon$};
    \end{tikzpicture}
    \end{aligned}  \nonumber \\ 
    & \quad \quad \otimes  \sum_{l',\alpha',\beta'}[{_{\eK}}F_{aj'b'}^j]^{l'\alpha'\beta'}_{k\sigma'\rho}\sum_{n',\kappa',\psi'}[{_{\eM}}F_{y}^{ay'b'}]^{n'\kappa'\psi'}_{z\mu'\nu} \sum_{m',\tau',\varepsilon'} [{_{\eK}}F^{l'}_{aa'k'}]^{m'\tau'\varepsilon'}_{j'\rho'\alpha'}\sum_{\theta'}[{_{\eM}}F_{n'}^{aa'z'}]^{m'\tau'\theta'}_{y'\nu'\kappa'} \nonumber \\
    & \quad \quad \times  \sum_{s',\delta',\phi'}[{_{\eK}}F^i_{l'b'b}]^{s'\delta'\phi'}_{j\beta'\sigma}\sum_{\xi'}[{_{\eM}}F_x^{n'b'b}]^{s'\delta'\xi'}_{y\psi'\mu} \sqrt{\frac{d_ad_{a'}}{d_{m'}}}\sqrt{\frac{d_bd_{b'}}{d_{s'}}} \;\begin{aligned}
    \begin{tikzpicture}[scale=0.65]
    \filldraw[black!60, fill=gray!15, dotted, even odd rule] (0,0) circle[radius=0.5] (0,0) circle[radius=1.5];
         \draw[line width=1pt,teal] (0,0.5)--(0,1.5);
         \draw[line width=1pt,violet] (0,-0.5)--(0,-1.5);
         \draw[red] (0,0.8) arc[start angle=90, end angle=270, radius=0.8];
         \draw[blue] (0,1.3) arc[start angle=90, end angle=-90, radius=1.3];
        \node[ line width=0.6pt, dashed, draw opacity=0.5] (a) at (0,1.7){$\scriptstyle i$};
         \node[ line width=0.6pt, dashed, draw opacity=0.5] (a) at (0,-1.7){$\scriptstyle x$};
        \node[ line width=0.6pt, dashed, draw opacity=0.5] (a) at (-1.1,0){$\scriptstyle m'$};
        \node[ line width=0.6pt, dashed, draw opacity=0.5] (a) at (1.1,0){$\scriptstyle s'$};
        \node[ line width=0.6pt, dashed, draw opacity=0.5] (a) at (-0.2,-1){$\scriptstyle n'$};
        \node[ line width=0.6pt, dashed, draw opacity=0.5] (a) at (-0.2,-1.4){$\scriptstyle \xi'$};
        \node[ line width=0.6pt, dashed, draw opacity=0.5] (a) at (0.3,-0.8){$\scriptstyle \theta'$};
        \node[ line width=0.6pt, dashed, draw opacity=0.5] (a) at (0,-0.2){$\scriptstyle z'$};
        \node[ line width=0.6pt, dashed, draw opacity=0.5] (a) at (0,0.3){$\scriptstyle k'$};
        \node[ line width=0.6pt, dashed, draw opacity=0.5] (a) at (-0.2,1){$\scriptstyle l'$};
        \node[ line width=0.6pt, dashed, draw opacity=0.5] (a) at (-0.2,1.3){$\scriptstyle \phi'$};
        \node[ line width=0.6pt, dashed, draw opacity=0.5] (a) at (0.3,0.8){$\scriptstyle \varepsilon'$};
    \end{tikzpicture}
    \end{aligned}\;. \nonumber 
    \end{align}
    Recollecting coefficients and using property of F-symbols yields 
    \begin{align} 
    & = \sum_{i,j\in \Irr(\eK)}\sum_{\sigma} \sum_{j',k'\in \Irr(\eK)}\sum_{\rho'}\delta_{u,w'} \delta_{z,x'} \sqrt{\frac{d_i}{d_{k'}}}\sqrt{\frac{d_ad_{a'}}{d_{m'}}}\sqrt{\frac{d_bd_{b'}}{d_{s'}}} \sqrt{\frac{1}{d_md_s}}\nonumber\\
    & \quad \quad \times \sum_{l,\alpha,\beta}[{_{\eN}}F_{av'b'}^v]^{l\alpha\beta}_{u\gamma'\zeta} \sum_{n,\kappa,\psi}\sum_{l',\alpha',\beta'} \underbrace{\sum_{k,\sigma',\rho}[{_{\eK}}F_{j}^{aj'b'}]^{n\kappa\psi}_{k\sigma'\rho}[{_{\eK}}F_{aj'b'}^j]^{l'\alpha'\beta'}_{k\sigma'\rho}}_{\delta_{n,l'}\delta_{\kappa,\alpha'}\delta_{\psi,\beta'}} \nonumber \\
    & \quad \quad \times  \sum_{m,\tau,\varepsilon} [{_{\eN}}F^l_{aa'u'}]^{m\tau\varepsilon}_{v'\zeta'\alpha}\sum_{\theta}[{_{\eK}}F_n^{aa'k'}]^{m\tau\theta}_{j'\rho'\kappa}   \sum_{s,\delta,\phi}[{_{\eN}}F^w_{lb'b}]^{s\delta\phi}_{v\beta\gamma}\sum_{\xi}[{_{\eK}}F_i^{nb'b}]^{s\delta\xi}_{j\psi\sigma}   \nonumber \\ 
    & \quad \quad \times  \sum_{n',\kappa',\psi'}[{_{\eM}}F_{y}^{ay'b'}]^{n'\kappa'\psi'}_{z\mu'\nu} \sum_{m',\tau',\varepsilon'} [{_{\eK}}F^{l'}_{aa'k'}]^{m'\tau'\varepsilon'}_{j'\rho'\alpha'}\sum_{\theta'}[{_{\eM}}F_{n'}^{aa'z'}]^{m'\tau'\theta'}_{y'\nu'\kappa'} \nonumber \\
    & \quad \quad \times  \sum_{s',\delta',\phi'}[{_{\eK}}F^i_{l'b'b}]^{s'\delta'\phi'}_{j\beta'\sigma}\sum_{\xi'}[{_{\eM}}F_x^{n'b'b}]^{s'\delta'\xi'}_{y\psi'\mu}  \;\begin{aligned}
    \begin{tikzpicture}[scale=0.65]
    \filldraw[black!60, fill=gray!15, dotted, even odd rule] (0,0) circle[radius=0.5] (0,0) circle[radius=1.5];
         \draw[line width=1pt,cyan] (0,0.5)--(0,1.5);
         \draw[line width=1pt,teal] (0,-0.5)--(0,-1.5);
         \draw[red] (0,0.8) arc[start angle=90, end angle=270, radius=0.8];
         \draw[blue] (0,1.3) arc[start angle=90, end angle=-90, radius=1.3];
        \node[ line width=0.6pt, dashed, draw opacity=0.5] (a) at (0,1.7){$\scriptstyle w$};
         \node[ line width=0.6pt, dashed, draw opacity=0.5] (a) at (0,-1.7){$\scriptstyle i$};
        \node[ line width=0.6pt, dashed, draw opacity=0.5] (a) at (-1.1,0){$\scriptstyle m$};
        \node[ line width=0.6pt, dashed, draw opacity=0.5] (a) at (1.1,0){$\scriptstyle s$};
        \node[ line width=0.6pt, dashed, draw opacity=0.5] (a) at (-0.2,-1){$\scriptstyle n$};
        \node[ line width=0.6pt, dashed, draw opacity=0.5] (a) at (-0.2,-1.4){$\scriptstyle \xi$};
        \node[ line width=0.6pt, dashed, draw opacity=0.5] (a) at (0.2,-0.8){$\scriptstyle \theta$};
        \node[ line width=0.6pt, dashed, draw opacity=0.5] (a) at (0,-0.2){$\scriptstyle k'$};
        \node[ line width=0.6pt, dashed, draw opacity=0.5] (a) at (0,0.3){$\scriptstyle u'$};
        \node[ line width=0.6pt, dashed, draw opacity=0.5] (a) at (-0.2,1){$\scriptstyle l$};
        \node[ line width=0.6pt, dashed, draw opacity=0.5] (a) at (-0.2,1.3){$\scriptstyle \phi$};
        \node[ line width=0.6pt, dashed, draw opacity=0.5] (a) at (0.2,0.8){$\scriptstyle \varepsilon$};
    \end{tikzpicture}
    \end{aligned} \otimes \begin{aligned}
    \begin{tikzpicture}[scale=0.65]
    \filldraw[black!60, fill=gray!15, dotted, even odd rule] (0,0) circle[radius=0.5] (0,0) circle[radius=1.5];
         \draw[line width=1pt,teal] (0,0.5)--(0,1.5);
         \draw[line width=1pt,violet] (0,-0.5)--(0,-1.5);
         \draw[red] (0,0.8) arc[start angle=90, end angle=270, radius=0.8];
         \draw[blue] (0,1.3) arc[start angle=90, end angle=-90, radius=1.3];
        \node[ line width=0.6pt, dashed, draw opacity=0.5] (a) at (0,1.7){$\scriptstyle i$};
         \node[ line width=0.6pt, dashed, draw opacity=0.5] (a) at (0,-1.7){$\scriptstyle x$};
        \node[ line width=0.6pt, dashed, draw opacity=0.5] (a) at (-1.1,0){$\scriptstyle m'$};
        \node[ line width=0.6pt, dashed, draw opacity=0.5] (a) at (1.1,0){$\scriptstyle s'$};
        \node[ line width=0.6pt, dashed, draw opacity=0.5] (a) at (-0.2,-1){$\scriptstyle n'$};
        \node[ line width=0.6pt, dashed, draw opacity=0.5] (a) at (-0.2,-1.4){$\scriptstyle \xi'$};
        \node[ line width=0.6pt, dashed, draw opacity=0.5] (a) at (0.3,-0.8){$\scriptstyle \theta'$};
        \node[ line width=0.6pt, dashed, draw opacity=0.5] (a) at (0,-0.2){$\scriptstyle z'$};
        \node[ line width=0.6pt, dashed, draw opacity=0.5] (a) at (0,0.3){$\scriptstyle k'$};
        \node[ line width=0.6pt, dashed, draw opacity=0.5] (a) at (-0.2,1){$\scriptstyle l'$};
        \node[ line width=0.6pt, dashed, draw opacity=0.5] (a) at (-0.2,1.3){$\scriptstyle \phi'$};
        \node[ line width=0.6pt, dashed, draw opacity=0.5] (a) at (0.3,0.8){$\scriptstyle \varepsilon'$};
    \end{tikzpicture}
    \end{aligned}  \nonumber \\
    & = \sum_{i\in \Irr(\eK)}\sum_{k'\in \Irr(\eK)}\delta_{u,w'} \delta_{z,x'} \sqrt{\frac{d_i}{d_{k'}}}\sqrt{\frac{d_ad_{a'}}{d_{m'}}}  \sqrt{\frac{d_bd_{b'}}{d_{s'}}} \sqrt{\frac{1}{d_md_s}}\nonumber\\
    & \quad \quad \times \sum_{l,\alpha,\beta}[{_{\eN}}F_{av'b'}^v]^{l\alpha\beta}_{u\gamma'\zeta} \sum_{n}  \sum_{\theta} \sum_{m,\tau,\varepsilon} [{_{\eN}}F^l_{aa'u'}]^{m\tau\varepsilon}_{v'\zeta'\alpha} \sum_{m',\tau',\varepsilon'}\underbrace{\sum_{j',\rho',\kappa}[{_{\eK}}F_n^{aa'k'}]^{m\tau\theta}_{j'\rho'\kappa} [{_{\eK}}F^{n}_{aa'k'}]^{m'\tau'\varepsilon'}_{j'\rho'\kappa}}_{\delta_{m,m'}\delta_{\tau,\tau'}\delta_{\theta,\varepsilon'}}   \nonumber \\
    & \quad \quad \times   \sum_{s,\delta,\phi}[{_{\eN}}F^w_{lb'b}]^{s\delta\phi}_{v\beta\gamma}\sum_{\xi} \sum_{s',\delta',\phi'}\underbrace{\sum_{j,\psi,\sigma}[{_{\eK}}F_i^{nb'b}]^{s\delta\xi}_{j\psi\sigma} [{_{\eK}}F^i_{nb'b}]^{s'\delta'\phi'}_{j\psi\sigma}}_{\delta_{s,s'}\delta_{\delta,\delta'}\delta_{\xi,\phi'}}  \sum_{n',\kappa',\psi'}[{_{\eM}}F_{y}^{ay'b'}]^{n'\kappa'\psi'}_{z\mu'\nu}   \nonumber \\
    & \quad \quad \times  \sum_{\theta'}[{_{\eM}}F_{n'}^{aa'z'}]^{m'\tau'\theta'}_{y'\nu'\kappa'}\sum_{\xi'}[{_{\eM}}F_x^{n'b'b}]^{s'\delta'\xi'}_{y\psi'\mu}  \;\begin{aligned}
    \begin{tikzpicture}[scale=0.65]
    \filldraw[black!60, fill=gray!15, dotted, even odd rule] (0,0) circle[radius=0.5] (0,0) circle[radius=1.5];
         \draw[line width=1pt,cyan] (0,0.5)--(0,1.5);
         \draw[line width=1pt,teal] (0,-0.5)--(0,-1.5);
         \draw[red] (0,0.8) arc[start angle=90, end angle=270, radius=0.8];
         \draw[blue] (0,1.3) arc[start angle=90, end angle=-90, radius=1.3];
        \node[ line width=0.6pt, dashed, draw opacity=0.5] (a) at (0,1.7){$\scriptstyle w$};
         \node[ line width=0.6pt, dashed, draw opacity=0.5] (a) at (0,-1.7){$\scriptstyle i$};
        \node[ line width=0.6pt, dashed, draw opacity=0.5] (a) at (-1.1,0){$\scriptstyle m$};
        \node[ line width=0.6pt, dashed, draw opacity=0.5] (a) at (1.1,0){$\scriptstyle s$};
        \node[ line width=0.6pt, dashed, draw opacity=0.5] (a) at (-0.2,-1){$\scriptstyle n$};
        \node[ line width=0.6pt, dashed, draw opacity=0.5] (a) at (-0.2,-1.4){$\scriptstyle \xi$};
        \node[ line width=0.6pt, dashed, draw opacity=0.5] (a) at (0.2,-0.8){$\scriptstyle \theta$};
        \node[ line width=0.6pt, dashed, draw opacity=0.5] (a) at (0,-0.2){$\scriptstyle k'$};
        \node[ line width=0.6pt, dashed, draw opacity=0.5] (a) at (0,0.3){$\scriptstyle u'$};
        \node[ line width=0.6pt, dashed, draw opacity=0.5] (a) at (-0.2,1){$\scriptstyle l$};
        \node[ line width=0.6pt, dashed, draw opacity=0.5] (a) at (-0.2,1.3){$\scriptstyle \phi$};
        \node[ line width=0.6pt, dashed, draw opacity=0.5] (a) at (0.2,0.8){$\scriptstyle \varepsilon$};
    \end{tikzpicture}
    \end{aligned} \otimes \begin{aligned}
    \begin{tikzpicture}[scale=0.65]
    \filldraw[black!60, fill=gray!15, dotted, even odd rule] (0,0) circle[radius=0.5] (0,0) circle[radius=1.5];
         \draw[line width=1pt,teal] (0,0.5)--(0,1.5);
         \draw[line width=1pt,violet] (0,-0.5)--(0,-1.5);
         \draw[red] (0,0.8) arc[start angle=90, end angle=270, radius=0.8];
         \draw[blue] (0,1.3) arc[start angle=90, end angle=-90, radius=1.3];
        \node[ line width=0.6pt, dashed, draw opacity=0.5] (a) at (0,1.7){$\scriptstyle i$};
         \node[ line width=0.6pt, dashed, draw opacity=0.5] (a) at (0,-1.7){$\scriptstyle x$};
        \node[ line width=0.6pt, dashed, draw opacity=0.5] (a) at (-1.1,0){$\scriptstyle m'$};
        \node[ line width=0.6pt, dashed, draw opacity=0.5] (a) at (1.1,0){$\scriptstyle s'$};
        \node[ line width=0.6pt, dashed, draw opacity=0.5] (a) at (-0.2,-1){$\scriptstyle n'$};
        \node[ line width=0.6pt, dashed, draw opacity=0.5] (a) at (-0.2,-1.4){$\scriptstyle \xi'$};
        \node[ line width=0.6pt, dashed, draw opacity=0.5] (a) at (0.3,-0.8){$\scriptstyle \theta'$};
        \node[ line width=0.6pt, dashed, draw opacity=0.5] (a) at (0,-0.2){$\scriptstyle z'$};
        \node[ line width=0.6pt, dashed, draw opacity=0.5] (a) at (0,0.3){$\scriptstyle k'$};
        \node[ line width=0.6pt, dashed, draw opacity=0.5] (a) at (-0.2,1){$\scriptstyle n$};
        \node[ line width=0.6pt, dashed, draw opacity=0.5] (a) at (-0.2,1.3){$\scriptstyle \phi'$};
        \node[ line width=0.6pt, dashed, draw opacity=0.5] (a) at (0.3,0.8){$\scriptstyle \varepsilon'$};
    \end{tikzpicture}
    \end{aligned}  \nonumber \\
    & = \delta_{u,w'} \delta_{z,x'}\sum_{l,\alpha,\beta}[{_{\eN}}F_{av'b'}^v]^{l\alpha\beta}_{u\gamma'\zeta} \sum_{n',\kappa',\psi'}[{_{\eM}}F_{y}^{ay'b'}]^{n'\kappa'\psi'}_{z\mu'\nu}  \sum_{m,\tau,\varepsilon} [{_{\eN}}F^l_{aa'u'}]^{m\tau\varepsilon}_{v'\zeta'\alpha} \sum_{\theta'}[{_{\eM}}F_{n'}^{aa'z'}]^{m\tau\theta'}_{y'\nu'\kappa'} \nonumber \\
    & \quad \quad \times \sum_{s,\delta,\phi}[{_{\eN}}F^w_{lb'b}]^{s\delta\phi}_{v\beta\gamma} \sum_{\xi'}[{_{\eM}}F_x^{n'b'b}]^{s\delta\xi'}_{y\psi'\mu} \sqrt{\frac{d_ad_{a'}}{d_m}}\sqrt{\frac{d_bd_{b'}}{d_s}} \nonumber  \\
    & \quad \quad \times \sum_{i,n,k'\in \Irr(\eK)} \sum_{\theta,\xi} \sqrt{\frac{d_i}{d_md_{k'}d_s}}  \;\begin{aligned}
    \begin{tikzpicture}[scale=0.65]
    \filldraw[black!60, fill=gray!15, dotted, even odd rule] (0,0) circle[radius=0.5] (0,0) circle[radius=1.5];
         \draw[line width=1pt,cyan] (0,0.5)--(0,1.5);
         \draw[line width=1pt,teal] (0,-0.5)--(0,-1.5);
         \draw[red] (0,0.8) arc[start angle=90, end angle=270, radius=0.8];
         \draw[blue] (0,1.3) arc[start angle=90, end angle=-90, radius=1.3];
        \node[ line width=0.6pt, dashed, draw opacity=0.5] (a) at (0,1.7){$\scriptstyle w$};
         \node[ line width=0.6pt, dashed, draw opacity=0.5] (a) at (0,-1.7){$\scriptstyle i$};
        \node[ line width=0.6pt, dashed, draw opacity=0.5] (a) at (-1.1,0){$\scriptstyle m$};
        \node[ line width=0.6pt, dashed, draw opacity=0.5] (a) at (1.1,0){$\scriptstyle s$};
        \node[ line width=0.6pt, dashed, draw opacity=0.5] (a) at (-0.2,-1){$\scriptstyle n$};
        \node[ line width=0.6pt, dashed, draw opacity=0.5] (a) at (-0.2,-1.4){$\scriptstyle \xi$};
        \node[ line width=0.6pt, dashed, draw opacity=0.5] (a) at (0.2,-0.8){$\scriptstyle \theta$};
        \node[ line width=0.6pt, dashed, draw opacity=0.5] (a) at (0,-0.2){$\scriptstyle k'$};
        \node[ line width=0.6pt, dashed, draw opacity=0.5] (a) at (0,0.3){$\scriptstyle u'$};
        \node[ line width=0.6pt, dashed, draw opacity=0.5] (a) at (-0.2,1){$\scriptstyle l$};
        \node[ line width=0.6pt, dashed, draw opacity=0.5] (a) at (-0.2,1.3){$\scriptstyle \phi$};
        \node[ line width=0.6pt, dashed, draw opacity=0.5] (a) at (0.2,0.8){$\scriptstyle \varepsilon$};
    \end{tikzpicture}
    \end{aligned} \otimes \begin{aligned}
    \begin{tikzpicture}[scale=0.65]
    \filldraw[black!60, fill=gray!15, dotted, even odd rule] (0,0) circle[radius=0.5] (0,0) circle[radius=1.5];
         \draw[line width=1pt,teal] (0,0.5)--(0,1.5);
         \draw[line width=1pt,violet] (0,-0.5)--(0,-1.5);
         \draw[red] (0,0.8) arc[start angle=90, end angle=270, radius=0.8];
         \draw[blue] (0,1.3) arc[start angle=90, end angle=-90, radius=1.3];
        \node[ line width=0.6pt, dashed, draw opacity=0.5] (a) at (0,1.7){$\scriptstyle i$};
         \node[ line width=0.6pt, dashed, draw opacity=0.5] (a) at (0,-1.7){$\scriptstyle x$};
        \node[ line width=0.6pt, dashed, draw opacity=0.5] (a) at (-1,0){$\scriptstyle m$};
        \node[ line width=0.6pt, dashed, draw opacity=0.5] (a) at (1.1,0){$\scriptstyle s$};
        \node[ line width=0.6pt, dashed, draw opacity=0.5] (a) at (-0.2,-1){$\scriptstyle n'$};
        \node[ line width=0.6pt, dashed, draw opacity=0.5] (a) at (-0.2,-1.4){$\scriptstyle \xi'$};
        \node[ line width=0.6pt, dashed, draw opacity=0.5] (a) at (0.3,-0.8){$\scriptstyle \theta'$};
        \node[ line width=0.6pt, dashed, draw opacity=0.5] (a) at (0,-0.2){$\scriptstyle z'$};
        \node[ line width=0.6pt, dashed, draw opacity=0.5] (a) at (0,0.3){$\scriptstyle k'$};
        \node[ line width=0.6pt, dashed, draw opacity=0.5] (a) at (-0.2,0.94){$\scriptstyle n$};
        \node[ line width=0.6pt, dashed, draw opacity=0.5] (a) at (-0.2,1.3){$\scriptstyle \xi$};
        \node[ line width=0.6pt, dashed, draw opacity=0.5] (a) at (0.2,0.8){$\scriptstyle \theta$};
    \end{tikzpicture}
    \end{aligned}\; , \nonumber 
    \end{align}
    which is equal to $\delta(XY)$ by applying $\delta$ to \eqref{eq:bi_tube_prod_2}.

\noindent
\textit{Proof of} \eqref{eq:epsilon-rho-DW-1} and \eqref{eq:epsilon-rho-DW-2}. We prove only the first line; the second line follows similarly. 
    For basis tube $X$, one has 
    \begin{equation}
        \notag
        \varepsilon_R(X) = \delta_{z,u}\delta_{y,v}\delta_{x,w}\delta_{\mu,\gamma}\delta_{\nu,\zeta} \sqrt{\frac{d_ad_zd_b}{d_x}} \sum_k\;  
        \begin{aligned}
        \begin{tikzpicture}[scale=0.65]
            \filldraw[black!60, fill=gray!15, dotted, even odd rule] (0,0) circle[radius=0.5] (0,0) circle[radius=1.5];
            \draw[line width=1pt,violet] (0,0.5)--(0,1.5);
            \draw[line width=1pt,violet] (0,-0.5)--(0,-1.5);
            \draw[red, dotted, thick] (0,0.8) arc[start angle=90, end angle=270, radius=0.8];
            \draw[blue, dotted, thick] (0,1.3) arc[start angle=90, end angle=-90, radius=1.3];
            \node[ line width=0.6pt, dashed, draw opacity=0.5] (a) at (0,-0.3){$\scriptstyle u$};
            \node[ line width=0.6pt, dashed, draw opacity=0.5] (a) at (0,0.3){$\scriptstyle k$};
        \end{tikzpicture}
    \end{aligned}\;. 
    \end{equation}
    Hence by definition, we obtain 
    \begin{equation}
        \notag
        \delta(\varepsilon_R(X))  =  \delta_{z,u}\delta_{y,v}\delta_{x,w}\delta_{\mu,\gamma}\delta_{\nu,\zeta} \sqrt{\frac{d_ad_zd_b}{d_x}} \sum_k\sum_t\;  \begin{aligned}
        \begin{tikzpicture}[scale=0.65]
            \filldraw[black!60, fill=gray!15, dotted, even odd rule] (0,0) circle[radius=0.5] (0,0) circle[radius=1.5];
            \draw[line width=1pt,violet] (0,0.5)--(0,1.5);
            \draw[line width=1pt,cyan] (0,-0.5)--(0,-1.5);
            \draw[red, dotted, thick] (0,0.8) arc[start angle=90, end angle=270, radius=0.8];
            \draw[blue, dotted, thick] (0,1.3) arc[start angle=90, end angle=-90, radius=1.3];
            \node[ line width=0.6pt, dashed, draw opacity=0.5] (a) at (0,-0.3){$\scriptstyle t$};
            \node[ line width=0.6pt, dashed, draw opacity=0.5] (a) at (0,0.3){$\scriptstyle k$};
        \end{tikzpicture}
    \end{aligned} \otimes 
    \begin{aligned}
        \begin{tikzpicture}[scale=0.65]
            \filldraw[black!60, fill=gray!15, dotted, even odd rule] (0,0) circle[radius=0.5] (0,0) circle[radius=1.5];
            \draw[line width=1pt,cyan] (0,0.5)--(0,1.5);
            \draw[line width=1pt,violet] (0,-0.5)--(0,-1.5);
            \draw[red, dotted, thick] (0,0.8) arc[start angle=90, end angle=270, radius=0.8];
            \draw[blue, dotted, thick] (0,1.3) arc[start angle=90, end angle=-90, radius=1.3];
            \node[ line width=0.6pt, dashed, draw opacity=0.5] (a) at (0,-0.3){$\scriptstyle u$};
            \node[ line width=0.6pt, dashed, draw opacity=0.5] (a) at (0,0.3){$\scriptstyle t$};
        \end{tikzpicture}
    \end{aligned}\;, 
    \end{equation}
    which is clearly equal to $1^{\langle 1 \rangle}\otimes 1^{\langle 2 \rangle}\varepsilon'_R(X)$.

\noindent
\textit{Proof of} \eqref{eq:unit-counit-DW-1}-\eqref{eq:unit-counit-DW-4}. Let us prove the first line. The second equality of \eqref{eq:unit-counit-DW-1} is just from definition. It suffices to show the first equality:
    \begin{equation} \label{eq:delta-unit-identity-DW}
        1^{\langle 1\rangle} X\otimes 1^{\langle 2\rangle}= X^{\langle 1\rangle} \otimes X^{\langle 2\rangle} s(X^{\langle 3\rangle}). 
    \end{equation}
    Consider the basis tube 
    \begin{equation}
        \notag
        X = \begin{aligned}

    \end{aligned}\;. 
    \end{equation}
    Then the right-hand side of \eqref{eq:delta-unit-identity-DW} is equal to
    \begin{align}
    &\quad \delta(X^{[-1]})\bigl(1\otimes s(X^{[0]})\bigr) \nonumber \\
     & = \sum_{i,j,k,\sigma,\varrho}\sqrt{\frac{d_i}{d_{a}d_kd_b}} \,\delta\left(\begin{aligned}
        %
    \end{aligned}\;. \label{eq:delta-unit-expanded-DW}
\end{align}
Recollecting the coefficients and using inverse parallel move, the multiplication is 
\begin{align}
    & \sum_{j,\sigma}\sqrt{\frac{d_j}{d_id_b}}\sum_{k,\varrho}\sqrt{\frac{d_k}{d_ad_j}}\;
    \begin{aligned}
        %
    \end{aligned}\;. \nonumber
\end{align}
Plugging this into Eq.~\eqref{eq:delta-unit-expanded-DW}, we obtain
\begin{align}
    & \quad \delta(X^{[-1]})\bigl(1\otimes s(X^{[0]})\bigr) \nonumber\\
    & = \sum_{r,s,t,\phi,\psi}\sqrt{\frac{d_r}{d_{a}d_td_b}} \sum_i \delta_{t,z}\delta_{r,x}\delta_{s,y}\delta_{\phi,\mu}\delta_{\psi,\nu} \sqrt{\frac{d_{a}d_{z}d_b}{d_{x}}} \; \begin{aligned}
        %
    \end{aligned}\;.  \nonumber
\end{align}
It is readily to see that the expression above is just $\delta(1)(X\otimes 1)$ as desired. The remaining three lines can be showed similarly.

\noindent
\textit{Proof of} \eqref{eq:antipode-antimultiplicative-DW}. The first assertion $s(1)=1$ is immediate. Let us prove the second one. 
    For general $X,Y$, using Eq.~\eqref{eq:bi_tube_prod_2}, one has
    \begin{align} 
        s(XY)& = \delta_{u,w'}\delta_{z,x'}\sum_{l,\alpha,\beta}[{_{\eN}}F_{av'b'}^v]^{l\alpha\beta}_{u\gamma'\zeta}\sum_{n,\rho,\sigma}[{_{\eM}}F_{y}^{ay'b'}]^{n\rho\sigma}_{z\mu'\nu} \sum_{m,\tau,\varepsilon} [{_{\eN}}F^l_{aa'u'}]^{m\tau\varepsilon}_{v'\zeta'\alpha}\sum_{\theta}[{_{\eM}}F_n^{aa'z'}]^{m\tau\theta}_{y'\nu'\rho} \nonumber \\
        & \quad \quad \times  \sum_{s,\delta,\phi}[{_{\eN}}F^w_{lb'b}]^{s\delta\phi}_{v\beta\gamma}\sum_{\xi}[{_{\eM}}F_x^{nb'b}]^{s\delta\xi}_{y\sigma\mu} \sqrt{\frac{d_ad_{a'}}{d_m}}\sqrt{\frac{d_bd_{b'}}{d_s}}\frac{d_{u'}}{d_w}\; 
        \begin{aligned}
        \begin{tikzpicture}[scale=0.65]
            \filldraw[black!60, fill=gray!15, dotted, even odd rule] (0,0) circle[radius=0.5] (0,0) circle[radius=1.5];
            \draw[line width=1pt,violet] (0,0.5)--(0,1.5);
            \draw[line width=1pt,cyan] (0,-0.5)--(0,-1.5);
            \draw[blue] (0,0.8) arc[start angle=90, end angle=-90, radius=0.8];
            \draw[red] (0,1.3) arc[start angle=90, end angle=270, radius=1.3];
            \node[ line width=0.6pt, dashed, draw opacity=0.5] (a) at (0,1.7){$\scriptstyle z'$};
            \node[ line width=0.6pt, dashed, draw opacity=0.5] (a) at (0,-1.7){$\scriptstyle u'$};
            \node[ line width=0.6pt, dashed, draw opacity=0.5] (a) at (-1,0){$\scriptstyle \bar{m}$};
            \node[ line width=0.6pt, dashed, draw opacity=0.5] (a) at (1,0){$\scriptstyle \bar{s}$};
            \node[ line width=0.6pt, dashed, draw opacity=0.5] (a) at (0.2,-1){$\scriptstyle l$};
            \node[ line width=0.6pt, dashed, draw opacity=0.5] (a) at (0.2,-1.35){$\scriptstyle \varepsilon$};
            \node[ line width=0.6pt, dashed, draw opacity=0.5] (a) at (-0.2,-0.8){$\scriptstyle \phi$};
            \node[ line width=0.6pt, dashed, draw opacity=0.5] (a) at (0,-0.3){$\scriptstyle w$};
            \node[ line width=0.6pt, dashed, draw opacity=0.5] (a) at (0,0.3){$\scriptstyle x$};
            \node[ line width=0.6pt, dashed, draw opacity=0.5] (a) at (0.2,1){$\scriptstyle n$};
            \node[ line width=0.6pt, dashed, draw opacity=0.5] (a) at (0.2,1.3){$\scriptstyle \theta$};
            \node[ line width=0.6pt, dashed, draw opacity=0.5] (a) at (-0.2,0.8){$\scriptstyle \xi$};
        \end{tikzpicture}
    \end{aligned}\;.  \nonumber 
    \end{align}
    On the other hand, one computes 
    \begin{align} 
    s(Y)s(X)&=\frac{d_{u'}}{d_{w'}}\;\begin{aligned}
        \begin{tikzpicture}[scale=0.65]
            \filldraw[black!60, fill=gray!15, dotted, even odd rule] (0,0) circle[radius=0.5] (0,0) circle[radius=1.5];
            \draw[line width=1pt,violet] (0,0.5)--(0,1.5);
            \draw[line width=1pt,cyan] (0,-0.5)--(0,-1.5);
            \draw[blue] (0,0.8) arc[start angle=90, end angle=-90, radius=0.8];
            \draw[red] (0,1.3) arc[start angle=90, end angle=270, radius=1.3];
            \node[ line width=0.6pt, dashed, draw opacity=0.5] (a) at (0,1.7){$\scriptstyle z'$};
            \node[ line width=0.6pt, dashed, draw opacity=0.5] (a) at (0,-1.65){$\scriptstyle u'$};
            \node[ line width=0.6pt, dashed, draw opacity=0.5] (a) at (-1,0){$\scriptstyle \bar{a}'$};
            \node[ line width=0.6pt, dashed, draw opacity=0.5] (a) at (1.06,0){$\scriptstyle \bar{b}'$};
            \node[ line width=0.6pt, dashed, draw opacity=0.5] (a) at (0.3,-1){$\scriptstyle v'$};
            \node[ line width=0.6pt, dashed, draw opacity=0.5] (a) at (0.3,-1.35){$\scriptstyle \zeta'$};
            \node[ line width=0.6pt, dashed, draw opacity=0.5] (a) at (-0.3,-0.8){$\scriptstyle \gamma'$};
            \node[ line width=0.6pt, dashed, draw opacity=0.5] (a) at (0,-0.2){$\scriptstyle w'$};
            \node[ line width=0.6pt, dashed, draw opacity=0.5] (a) at (0,0.3){$\scriptstyle x'$};
            \node[ line width=0.6pt, dashed, draw opacity=0.5] (a) at (0.3,1.06){$\scriptstyle y'$};
            \node[ line width=0.6pt, dashed, draw opacity=0.5] (a) at (0.3,1.4){$\scriptstyle \nu'$};
            \node[ line width=0.6pt, dashed, draw opacity=0.5] (a) at (-0.3,0.8){$\scriptstyle \mu'$};
        \end{tikzpicture}
    \end{aligned}\cdot \frac{d_u}{d_w}\;\begin{aligned}
        \begin{tikzpicture}[scale=0.65]
            \filldraw[black!60, fill=gray!15, dotted, even odd rule] (0,0) circle[radius=0.5] (0,0) circle[radius=1.5];
            \draw[line width=1pt,violet] (0,0.5)--(0,1.5);
            \draw[line width=1pt,cyan] (0,-0.5)--(0,-1.5);
            \draw[blue] (0,0.8) arc[start angle=90, end angle=-90, radius=0.8];
            \draw[red] (0,1.3) arc[start angle=90, end angle=270, radius=1.3];
            \node[ line width=0.6pt, dashed, draw opacity=0.5] (a) at (0,1.7){$\scriptstyle z$};
            \node[ line width=0.6pt, dashed, draw opacity=0.5] (a) at (0,-1.7){$\scriptstyle u$};
            \node[ line width=0.6pt, dashed, draw opacity=0.5] (a) at (-1.1,0){$\scriptstyle \bar{a}$};
            \node[ line width=0.6pt, dashed, draw opacity=0.5] (a) at (1,0){$\scriptstyle \bar{b}$};
            \node[ line width=0.6pt, dashed, draw opacity=0.5] (a) at (0.2,-1){$\scriptstyle v$};
            \node[ line width=0.6pt, dashed, draw opacity=0.5] (a) at (0.2,-1.35){$\scriptstyle \zeta$};
            \node[ line width=0.6pt, dashed, draw opacity=0.5] (a) at (-0.2,-0.8){$\scriptstyle \gamma$};
            \node[ line width=0.6pt, dashed, draw opacity=0.5] (a) at (0,-0.3){$\scriptstyle w$};
            \node[ line width=0.6pt, dashed, draw opacity=0.5] (a) at (0,0.3){$\scriptstyle x$};
            \node[ line width=0.6pt, dashed, draw opacity=0.5] (a) at (0.2,1){$\scriptstyle y$};
            \node[ line width=0.6pt, dashed, draw opacity=0.5] (a) at (0.2,1.3){$\scriptstyle \nu$};
            \node[ line width=0.6pt, dashed, draw opacity=0.5] (a) at (-0.2,0.8){$\scriptstyle \mu$};
        \end{tikzpicture}
    \end{aligned}  =\frac{d_{u'}}{d_{w'}}\frac{d_u}{d_w}\delta_{x',z}\delta_{w',u}\; \begin{aligned}
        \begin{tikzpicture}[scale=0.6]
            \filldraw[black!60, fill=gray!15, dotted, even odd rule] (0,0) circle[radius=0.5] (0,0) circle[radius=2.5];
            \draw[line width=1pt,violet] (0,0.5)--(0,2.5);
            \draw[line width=1pt,cyan] (0,-0.5)--(0,-2.5);
            \draw[blue] (0,0.8) arc[start angle=90, end angle=-90, radius=0.8];
            \draw[red] (0,1.2) arc[start angle=90, end angle=270, radius=1.2];
            \node[ line width=0.6pt, dashed, draw opacity=0.5] (a) at (-0.2,1.4){$\scriptstyle z$};
            \node[ line width=0.6pt, dashed, draw opacity=0.5] (a) at (-0.2,-1.4){$\scriptstyle u$}; 
            \node[ line width=0.6pt, dashed, draw opacity=0.5] (a) at (-1,0){$\scriptstyle \bar{a}$};
            \node[ line width=0.6pt, dashed, draw opacity=0.5] (a) at (1,0){$\scriptstyle \bar{b}$};
            \node[ line width=0.6pt, dashed, draw opacity=0.5] (a) at (0.2,-0.95){$\scriptstyle v$};
            \node[ line width=0.6pt, dashed, draw opacity=0.5] (a) at (0.2,-1.3){$\scriptstyle \zeta$};
            \node[ line width=0.6pt, dashed, draw opacity=0.5] (a) at (-0.25,-0.8){$\scriptstyle \gamma$};
            \node[ line width=0.6pt, dashed, draw opacity=0.5] (a) at (0,-0.25){$\scriptstyle w$};
            \node[ line width=0.6pt, dashed, draw opacity=0.5] (a) at (0,0.3){$\scriptstyle x$};
            \node[ line width=0.6pt, dashed, draw opacity=0.5] (a) at (0.2,1){$\scriptstyle y$};
            \node[ line width=0.6pt, dashed, draw opacity=0.5] (a) at (0.25,1.3){$\scriptstyle \nu$};
            \node[ line width=0.6pt, dashed, draw opacity=0.5] (a) at (-0.25,0.8){$\scriptstyle \mu$}; 
            \draw[blue] (0,1.6) arc[start angle=90, end angle=-90, radius=1.6];
            \node[ line width=0.6pt, dashed, draw opacity=0.5] (a) at (-0.25,1.75){$\scriptstyle \mu'$};
            \node[ line width=0.6pt, dashed, draw opacity=0.5] (a) at (-0.3,-1.77){$\scriptstyle \gamma'$};
            \draw[red] (0,2.1) arc[start angle=90, end angle=270, radius=2.1];
            \node[ line width=0.6pt, dashed, draw opacity=0.5] (a) at (-1.75,0.4){$\scriptstyle \bar{a}'$};
            \node[ line width=0.6pt, dashed, draw opacity=0.5] (a) at (1.85,0.4){$\scriptstyle \bar{b}'$};
            \node[ line width=0.6pt, dashed, draw opacity=0.5] (a) at (0.3,1.9){$\scriptstyle y'$};
            \node[ line width=0.6pt, dashed, draw opacity=0.5] (a) at (0.3,-1.8){$\scriptstyle v'$};
            \node[ line width=0.6pt, dashed, draw opacity=0.5] (a) at (0.3,2.2){$\scriptstyle \nu'$};
            \node[ line width=0.6pt, dashed, draw opacity=0.5] (a) at (0.3,-2.2){$\scriptstyle \zeta'$};
            \node[ line width=0.6pt, dashed, draw opacity=0.5] (a) at (0.3,2.6){$\scriptstyle z'$};
            \node[ line width=0.6pt, dashed, draw opacity=0.5] (a) at (0.3,-2.6){$\scriptstyle u'$};
        \end{tikzpicture}
    \end{aligned} \nonumber \\
    & =\frac{d_{u'}}{d_{w}}\delta_{x',z}\delta_{w',u}\sum_{l,\alpha,\beta}[{_{\eN}}F_{av'b'}^v]^{l\alpha\beta}_{u\gamma'\zeta}\sum_{n,\rho,\sigma}[{_{\eM}}F_{y}^{ay'b'}]^{n\rho\sigma}_{z\mu'\nu} \sum_{m,\tau,\varepsilon} [{_{\eN}}F^l_{aa'u'}]^{m\tau\varepsilon}_{v'\zeta'\alpha}\sum_{\theta}[{_{\eM}}F_n^{aa'z'}]^{m\tau\theta}_{y'\nu'\rho} \nonumber \\
    & \quad \quad \times  \sum_{s,\delta,\phi}[{_{\eN}}F^w_{lb'b}]^{s\delta\phi}_{v\beta\gamma}\sum_{\xi}[{_{\eM}}F_x^{nb'b}]^{s\delta\xi}_{y\sigma\mu} \sqrt{\frac{d_ad_{a'}}{d_m}}\sqrt{\frac{d_bd_{b'}}{d_s}}\; 
    \begin{aligned}
    \begin{tikzpicture}[scale=0.65]
        \filldraw[black!60, fill=gray!15, dotted, even odd rule] (0,0) circle[radius=0.5] (0,0) circle[radius=1.5];
        \draw[line width=1pt,violet] (0,0.5)--(0,1.5);
        \draw[line width=1pt,cyan] (0,-0.5)--(0,-1.5);
        \draw[blue] (0,0.8) arc[start angle=90, end angle=-90, radius=0.8];
        \draw[red] (0,1.3) arc[start angle=90, end angle=270, radius=1.3];
        \node[ line width=0.6pt, dashed, draw opacity=0.5] (a) at (0,1.7){$\scriptstyle z'$};
        \node[ line width=0.6pt, dashed, draw opacity=0.5] (a) at (0,-1.7){$\scriptstyle u'$};
        \node[ line width=0.6pt, dashed, draw opacity=0.5] (a) at (-1,0){$\scriptstyle \bar{m}$};
        \node[ line width=0.6pt, dashed, draw opacity=0.5] (a) at (1,0){$\scriptstyle \bar{s}$};
        \node[ line width=0.6pt, dashed, draw opacity=0.5] (a) at (0.2,-1){$\scriptstyle l$};
        \node[ line width=0.6pt, dashed, draw opacity=0.5] (a) at (0.2,-1.35){$\scriptstyle \varepsilon$};
        \node[ line width=0.6pt, dashed, draw opacity=0.5] (a) at (-0.2,-0.8){$\scriptstyle \phi$};
        \node[ line width=0.6pt, dashed, draw opacity=0.5] (a) at (0,-0.3){$\scriptstyle w$};
        \node[ line width=0.6pt, dashed, draw opacity=0.5] (a) at (0,0.3){$\scriptstyle x$};
        \node[ line width=0.6pt, dashed, draw opacity=0.5] (a) at (0.2,1){$\scriptstyle n$};
        \node[ line width=0.6pt, dashed, draw opacity=0.5] (a) at (0.2,1.3){$\scriptstyle \theta$};
        \node[ line width=0.6pt, dashed, draw opacity=0.5] (a) at (-0.2,0.8){$\scriptstyle \xi$};
    \end{tikzpicture}
    \end{aligned}\;,  \nonumber 
\end{align}
which is equal to $s(XY)$.

\noindent
\textit{Proof of} \eqref{eq:delta-antipode-compatibility-DW}. The assertion follows by direct computation:
    \begin{align}
        \delta(s(X)) & = \frac{d_u}{d_w}\sum_{r,s,t,\phi,\psi } \sqrt{\frac{d_r}{d_ad_td_b}}\;\begin{aligned}
        \begin{tikzpicture}[scale=0.65]
            \filldraw[black!60, fill=gray!15, dotted, even odd rule] (0,0) circle[radius=0.5] (0,0) circle[radius=1.5];
            \draw[line width=1pt,violet] (0,0.5)--(0,1.5);
            \draw[line width=1pt,cyan] (0,-0.5)--(0,-1.5);
            \draw[blue] (0,0.8) arc[start angle=90, end angle=-90, radius=0.8];
            \draw[red] (0,1.3) arc[start angle=90, end angle=270, radius=1.3];
            \node[ line width=0.6pt, dashed, draw opacity=0.5] (a) at (0,1.7){$\scriptstyle z$};
            \node[ line width=0.6pt, dashed, draw opacity=0.5] (a) at (0,-1.7){$\scriptstyle r$};
            \node[ line width=0.6pt, dashed, draw opacity=0.5] (a) at (-1.1,0){$\scriptstyle \bar{a}$};
            \node[ line width=0.6pt, dashed, draw opacity=0.5] (a) at (1,0){$\scriptstyle \bar{b}$};
            \node[ line width=0.6pt, dashed, draw opacity=0.5] (a) at (0.2,-1){$\scriptstyle s$};
            \node[ line width=0.6pt, dashed, draw opacity=0.5] (a) at (0.2,-1.35){$\scriptstyle \phi$};
            \node[ line width=0.6pt, dashed, draw opacity=0.5] (a) at (-0.2,-0.8){$\scriptstyle \psi$};
            \node[ line width=0.6pt, dashed, draw opacity=0.5] (a) at (0,-0.3){$\scriptstyle t$};
            \node[ line width=0.6pt, dashed, draw opacity=0.5] (a) at (0,0.3){$\scriptstyle x$};
            \node[ line width=0.6pt, dashed, draw opacity=0.5] (a) at (0.2,1){$\scriptstyle y$};
            \node[ line width=0.6pt, dashed, draw opacity=0.5] (a) at (0.2,1.3){$\scriptstyle \nu$};
            \node[ line width=0.6pt, dashed, draw opacity=0.5] (a) at (-0.2,0.8){$\scriptstyle \mu$};
        \end{tikzpicture}
    \end{aligned} \otimes 
    \begin{aligned}
        \begin{tikzpicture}[scale=0.65]
            \filldraw[black!60, fill=gray!15, dotted, even odd rule] (0,0) circle[radius=0.5] (0,0) circle[radius=1.5];
            \draw[line width=1pt,cyan] (0,0.5)--(0,1.5);
            \draw[line width=1pt,violet] (0,-0.5)--(0,-1.5);
            \draw[blue] (0,0.8) arc[start angle=90, end angle=-90, radius=0.8];
            \draw[red] (0,1.3) arc[start angle=90, end angle=270, radius=1.3];
            \node[ line width=0.6pt, dashed, draw opacity=0.5] (a) at (0,1.7){$\scriptstyle r$};
            \node[ line width=0.6pt, dashed, draw opacity=0.5] (a) at (0,-1.7){$\scriptstyle u$};
            \node[ line width=0.6pt, dashed, draw opacity=0.5] (a) at (-1.1,0){$\scriptstyle \bar{a}$};
            \node[ line width=0.6pt, dashed, draw opacity=0.5] (a) at (1,0){$\scriptstyle \bar{b}$};
            \node[ line width=0.6pt, dashed, draw opacity=0.5] (a) at (0.2,-1){$\scriptstyle v$};
            \node[ line width=0.6pt, dashed, draw opacity=0.5] (a) at (0.2,-1.35){$\scriptstyle \zeta$};
            \node[ line width=0.6pt, dashed, draw opacity=0.5] (a) at (-0.2,-0.8){$\scriptstyle \gamma$};
            \node[ line width=0.6pt, dashed, draw opacity=0.5] (a) at (0,-0.3){$\scriptstyle w$};
            \node[ line width=0.6pt, dashed, draw opacity=0.5] (a) at (0,0.3){$\scriptstyle t$};
            \node[ line width=0.6pt, dashed, draw opacity=0.5] (a) at (0.2,1){$\scriptstyle s$};
            \node[ line width=0.6pt, dashed, draw opacity=0.5] (a) at (0.2,1.3){$\scriptstyle \phi$};
            \node[ line width=0.6pt, dashed, draw opacity=0.5] (a) at (-0.2,0.8){$\scriptstyle \psi$};
        \end{tikzpicture}
    \end{aligned}\;,     \notag \\
        \tau(s\otimes s)\circ \delta(X)  & = \sum_{r,s,t,\phi,\psi} \sqrt{\frac{d_r}{d_ad_td_b}} \frac{d_t}{d_r}\;\begin{aligned}
        \begin{tikzpicture}[scale=0.65]
            \filldraw[black!60, fill=gray!15, dotted, even odd rule] (0,0) circle[radius=0.5] (0,0) circle[radius=1.5];
            \draw[line width=1pt,violet] (0,0.5)--(0,1.5);
            \draw[line width=1pt,cyan] (0,-0.5)--(0,-1.5);
            \draw[blue] (0,0.8) arc[start angle=90, end angle=-90, radius=0.8];
            \draw[red] (0,1.3) arc[start angle=90, end angle=270, radius=1.3];
            \node[ line width=0.6pt, dashed, draw opacity=0.5] (a) at (0,1.7){$\scriptstyle z$};
            \node[ line width=0.6pt, dashed, draw opacity=0.5] (a) at (0,-1.7){$\scriptstyle t$};
            \node[ line width=0.6pt, dashed, draw opacity=0.5] (a) at (-1.1,0){$\scriptstyle \bar{a}$};
            \node[ line width=0.6pt, dashed, draw opacity=0.5] (a) at (1,0){$\scriptstyle \bar{b}$};
            \node[ line width=0.6pt, dashed, draw opacity=0.5] (a) at (0.2,-1){$\scriptstyle s$};
            \node[ line width=0.6pt, dashed, draw opacity=0.5] (a) at (0.2,-1.35){$\scriptstyle \psi$};
            \node[ line width=0.6pt, dashed, draw opacity=0.5] (a) at (-0.2,-0.8){$\scriptstyle \phi$};
            \node[ line width=0.6pt, dashed, draw opacity=0.5] (a) at (0,-0.3){$\scriptstyle r$};
            \node[ line width=0.6pt, dashed, draw opacity=0.5] (a) at (0,0.3){$\scriptstyle x$};
            \node[ line width=0.6pt, dashed, draw opacity=0.5] (a) at (0.2,1){$\scriptstyle y$};
            \node[ line width=0.6pt, dashed, draw opacity=0.5] (a) at (0.2,1.3){$\scriptstyle \nu$};
            \node[ line width=0.6pt, dashed, draw opacity=0.5] (a) at (-0.2,0.8){$\scriptstyle \mu$};
        \end{tikzpicture}
    \end{aligned} \otimes 
    \frac{d_u}{d_w}\;\begin{aligned}
        \begin{tikzpicture}[scale=0.65]
            \filldraw[black!60, fill=gray!15, dotted, even odd rule] (0,0) circle[radius=0.5] (0,0) circle[radius=1.5];
            \draw[line width=1pt,cyan] (0,0.5)--(0,1.5);
            \draw[line width=1pt,violet] (0,-0.5)--(0,-1.5);
            \draw[blue] (0,0.8) arc[start angle=90, end angle=-90, radius=0.8];
            \draw[red] (0,1.3) arc[start angle=90, end angle=270, radius=1.3];
            \node[ line width=0.6pt, dashed, draw opacity=0.5] (a) at (0,1.7){$\scriptstyle t$};
            \node[ line width=0.6pt, dashed, draw opacity=0.5] (a) at (0,-1.7){$\scriptstyle u$};
            \node[ line width=0.6pt, dashed, draw opacity=0.5] (a) at (-1.1,0){$\scriptstyle \bar{a}$};
            \node[ line width=0.6pt, dashed, draw opacity=0.5] (a) at (1,0){$\scriptstyle \bar{b}$};
            \node[ line width=0.6pt, dashed, draw opacity=0.5] (a) at (0.2,-1){$\scriptstyle v$};
            \node[ line width=0.6pt, dashed, draw opacity=0.5] (a) at (0.2,-1.35){$\scriptstyle \zeta$};
            \node[ line width=0.6pt, dashed, draw opacity=0.5] (a) at (-0.2,-0.8){$\scriptstyle \gamma$};
            \node[ line width=0.6pt, dashed, draw opacity=0.5] (a) at (0,-0.3){$\scriptstyle w$};
            \node[ line width=0.6pt, dashed, draw opacity=0.5] (a) at (0,0.3){$\scriptstyle r$};
            \node[ line width=0.6pt, dashed, draw opacity=0.5] (a) at (0.2,1){$\scriptstyle s$};
            \node[ line width=0.6pt, dashed, draw opacity=0.5] (a) at (0.2,1.3){$\scriptstyle \psi$};
            \node[ line width=0.6pt, dashed, draw opacity=0.5] (a) at (-0.2,0.8){$\scriptstyle \phi$};
        \end{tikzpicture}
    \end{aligned}\;.     \notag
    \end{align}
    Two expressions are equal by inspection.

\noindent
\textit{Proof of} \eqref{eq:decomp-antipode-DW}. By definition, $s(X^{\langle 1\rangle})X^{\langle 2\rangle}s(X^{\langle 3\rangle}) = \varepsilon'_R(X^{[-1]})S(X^{[0]})$.
    It is easy to compute that 
    \begin{equation}
        \notag
        \varepsilon_R'\left(\begin{aligned}
        \begin{tikzpicture}[scale=0.65]
            \filldraw[black!60, fill=gray!15, dotted, even odd rule] (0,0) circle[radius=0.5] (0,0) circle[radius=1.5];
            \draw[line width=1pt,cyan] (0,0.5)--(0,1.5);
            \draw[line width=1pt,cyan] (0,-0.5)--(0,-1.5);
            \draw[red] (0,0.8) arc[start angle=90, end angle=270, radius=0.8];
            \draw[blue] (0,1.3) arc[start angle=90, end angle=-90, radius=1.3];
            \node[ line width=0.6pt, dashed, draw opacity=0.5] (a) at (0,1.7){$\scriptstyle w$};
            \node[ line width=0.6pt, dashed, draw opacity=0.5] (a) at (0,-1.7){$\scriptstyle i$};
            \node[ line width=0.6pt, dashed, draw opacity=0.5] (a) at (-1,0){$\scriptstyle a$};
            \node[ line width=0.6pt, dashed, draw opacity=0.5] (a) at (1.1,0){$\scriptstyle b$};
            \node[ line width=0.6pt, dashed, draw opacity=0.5] (a) at (-0.2,-1){$\scriptstyle j$};
            \node[ line width=0.6pt, dashed, draw opacity=0.5] (a) at (-0.2,-1.35){$\scriptstyle \sigma$};
            \node[ line width=0.6pt, dashed, draw opacity=0.5] (a) at (0.2,-0.8){$\scriptstyle \varrho$};
            \node[ line width=0.6pt, dashed, draw opacity=0.5] (a) at (0,-0.3){$\scriptstyle k$};
            \node[ line width=0.6pt, dashed, draw opacity=0.5] (a) at (0,0.3){$\scriptstyle u$};
            \node[ line width=0.6pt, dashed, draw opacity=0.5] (a) at (-0.2,1){$\scriptstyle v$};
            \node[ line width=0.6pt, dashed, draw opacity=0.5] (a) at (-0.2,1.3){$\scriptstyle \gamma$};
            \node[ line width=0.6pt, dashed, draw opacity=0.5] (a) at (0.2,0.8){$\scriptstyle \zeta$};
        \end{tikzpicture}
    \end{aligned}\right) =\delta_{k,u}\delta_{j,v}\delta_{i,w}\delta_{\varrho,\zeta}\delta_{\sigma,\gamma} \sqrt{\frac{d_ad_ud_b}{d_w}} \sum_t\;  
        \begin{aligned}
        \begin{tikzpicture}[scale=0.65]
            \filldraw[black!60, fill=gray!15, dotted, even odd rule] (0,0) circle[radius=0.5] (0,0) circle[radius=1.5];
            \draw[line width=1pt,violet] (0,0.5)--(0,1.5);
            \draw[line width=1pt,cyan] (0,-0.5)--(0,-1.5);
            \draw[red, dotted, thick] (0,0.8) arc[start angle=90, end angle=270, radius=0.8];
            \draw[blue, dotted, thick] (0,1.3) arc[start angle=90, end angle=-90, radius=1.3];
            \node[ line width=0.6pt, dashed, draw opacity=0.5] (a) at (0,-0.3){$\scriptstyle u$};
            \node[ line width=0.6pt, dashed, draw opacity=0.5] (a) at (0,0.3){$\scriptstyle t$};
        \end{tikzpicture}
    \end{aligned}\;. 
    \end{equation}
    Therefore, one has 
    \begin{align}
        &\quad \varepsilon'_R(X^{[-1]})s(X^{[0]}) \nonumber \\
        & = \sum_{i,j,k,\sigma,\varrho}\sqrt{\frac{d_i}{d_ad_kd_b}}  \; \varepsilon_R'\left(\begin{aligned}
        \begin{tikzpicture}[scale=0.65]
            \filldraw[black!60, fill=gray!15, dotted, even odd rule] (0,0) circle[radius=0.5] (0,0) circle[radius=1.5];
            \draw[line width=1pt,cyan] (0,0.5)--(0,1.5);
            \draw[line width=1pt,cyan] (0,-0.5)--(0,-1.5);
            \draw[red] (0,0.8) arc[start angle=90, end angle=270, radius=0.8];
            \draw[blue] (0,1.3) arc[start angle=90, end angle=-90, radius=1.3];
            \node[ line width=0.6pt, dashed, draw opacity=0.5] (a) at (0,1.7){$\scriptstyle w$};
            \node[ line width=0.6pt, dashed, draw opacity=0.5] (a) at (0,-1.7){$\scriptstyle i$};
            \node[ line width=0.6pt, dashed, draw opacity=0.5] (a) at (-1,0){$\scriptstyle a$};
            \node[ line width=0.6pt, dashed, draw opacity=0.5] (a) at (1.1,0){$\scriptstyle b$};
            \node[ line width=0.6pt, dashed, draw opacity=0.5] (a) at (-0.2,-1){$\scriptstyle j$};
            \node[ line width=0.6pt, dashed, draw opacity=0.5] (a) at (-0.2,-1.35){$\scriptstyle \sigma$};
            \node[ line width=0.6pt, dashed, draw opacity=0.5] (a) at (0.2,-0.8){$\scriptstyle \varrho$};
            \node[ line width=0.6pt, dashed, draw opacity=0.5] (a) at (0,-0.3){$\scriptstyle k$};
            \node[ line width=0.6pt, dashed, draw opacity=0.5] (a) at (0,0.3){$\scriptstyle u$};
            \node[ line width=0.6pt, dashed, draw opacity=0.5] (a) at (-0.2,1){$\scriptstyle v$};
            \node[ line width=0.6pt, dashed, draw opacity=0.5] (a) at (-0.2,1.3){$\scriptstyle \gamma$};
            \node[ line width=0.6pt, dashed, draw opacity=0.5] (a) at (0.2,0.8){$\scriptstyle \zeta$};
        \end{tikzpicture}
    \end{aligned}\right) s\left(\begin{aligned}
        \begin{tikzpicture}[scale=0.65]
            \filldraw[black!60, fill=gray!15, dotted, even odd rule] (0,0) circle[radius=0.5] (0,0) circle[radius=1.5];
            \draw[line width=1pt,cyan] (0,0.5)--(0,1.5);
            \draw[line width=1pt,violet] (0,-0.5)--(0,-1.5);
            \draw[red] (0,0.8) arc[start angle=90, end angle=270, radius=0.8];
            \draw[blue] (0,1.3) arc[start angle=90, end angle=-90, radius=1.3];
            \node[ line width=0.6pt, dashed, draw opacity=0.5] (a) at (0,1.7){$\scriptstyle i$};
            \node[ line width=0.6pt, dashed, draw opacity=0.5] (a) at (0,-1.7){$\scriptstyle x$};
            \node[ line width=0.6pt, dashed, draw opacity=0.5] (a) at (-1,0){$\scriptstyle a$};
            \node[ line width=0.6pt, dashed, draw opacity=0.5] (a) at (1.1,0){$\scriptstyle b$};
            \node[ line width=0.6pt, dashed, draw opacity=0.5] (a) at (-0.2,-1){$\scriptstyle y$};
            \node[ line width=0.6pt, dashed, draw opacity=0.5] (a) at (-0.2,-1.35){$\scriptstyle \mu$};
            \node[ line width=0.6pt, dashed, draw opacity=0.5] (a) at (0.2,-0.8){$\scriptstyle \nu$};
            \node[ line width=0.6pt, dashed, draw opacity=0.5] (a) at (0,-0.3){$\scriptstyle z$};
            \node[ line width=0.6pt, dashed, draw opacity=0.5] (a) at (0,0.3){$\scriptstyle k$};
            \node[ line width=0.6pt, dashed, draw opacity=0.5] (a) at (-0.2,1){$\scriptstyle j$};
            \node[ line width=0.6pt, dashed, draw opacity=0.5] (a) at (-0.2,1.3){$\scriptstyle \sigma$};
            \node[ line width=0.6pt, dashed, draw opacity=0.5] (a) at (0.2,0.8){$\scriptstyle \varrho$};
        \end{tikzpicture}
    \end{aligned}\right) \nonumber \\
        & = \sum_{i,j,k,\sigma,\varrho}\sqrt{\frac{d_i}{d_ad_kd_b}} \delta_{k,u}\delta_{j,v}\delta_{i,w}\delta_{\varrho,\zeta}\delta_{\sigma,\gamma} \sqrt{\frac{d_ad_ud_b}{d_w}} \sum_t\;  
        \begin{aligned}
        \begin{tikzpicture}[scale=0.65]
            \filldraw[black!60, fill=gray!15, dotted, even odd rule] (0,0) circle[radius=0.5] (0,0) circle[radius=1.5];
            \draw[line width=1pt,violet] (0,0.5)--(0,1.5);
            \draw[line width=1pt,cyan] (0,-0.5)--(0,-1.5);
            \draw[red, dotted, thick] (0,0.8) arc[start angle=90, end angle=270, radius=0.8];
            \draw[blue, dotted, thick] (0,1.3) arc[start angle=90, end angle=-90, radius=1.3];
            \node[ line width=0.6pt, dashed, draw opacity=0.5] (a) at (0,-0.3){$\scriptstyle u$};
            \node[ line width=0.6pt, dashed, draw opacity=0.5] (a) at (0,0.3){$\scriptstyle t$};
        \end{tikzpicture}
    \end{aligned} \cdot  \frac{d_k}{d_i}\; \begin{aligned}
        \begin{tikzpicture}[scale=0.65]
            \filldraw[black!60, fill=gray!15, dotted, even odd rule] (0,0) circle[radius=0.5] (0,0) circle[radius=1.5];
            \draw[line width=1pt,violet] (0,0.5)--(0,1.5);
            \draw[line width=1pt,cyan] (0,-0.5)--(0,-1.5);
            \draw[blue] (0,0.8) arc[start angle=90, end angle=-90, radius=0.8];
            \draw[red] (0,1.3) arc[start angle=90, end angle=270, radius=1.3];
            \node[ line width=0.6pt, dashed, draw opacity=0.5] (a) at (0,1.7){$\scriptstyle z$};
            \node[ line width=0.6pt, dashed, draw opacity=0.5] (a) at (0,-1.7){$\scriptstyle k$};
            \node[ line width=0.6pt, dashed, draw opacity=0.5] (a) at (-1.1,0){$\scriptstyle \bar{a}$};
            \node[ line width=0.6pt, dashed, draw opacity=0.5] (a) at (1,0){$\scriptstyle \bar{b}$};
            \node[ line width=0.6pt, dashed, draw opacity=0.5] (a) at (0.2,-1){$\scriptstyle j$};
            \node[ line width=0.6pt, dashed, draw opacity=0.5] (a) at (0.2,-1.35){$\scriptstyle \varrho$};
            \node[ line width=0.6pt, dashed, draw opacity=0.5] (a) at (-0.2,-0.8){$\scriptstyle \sigma$};
            \node[ line width=0.6pt, dashed, draw opacity=0.5] (a) at (0,-0.3){$\scriptstyle i$};
            \node[ line width=0.6pt, dashed, draw opacity=0.5] (a) at (0,0.3){$\scriptstyle x$};
            \node[ line width=0.6pt, dashed, draw opacity=0.5] (a) at (0.2,1){$\scriptstyle y$};
            \node[ line width=0.6pt, dashed, draw opacity=0.5] (a) at (0.2,1.3){$\scriptstyle \nu$};
            \node[ line width=0.6pt, dashed, draw opacity=0.5] (a) at (-0.2,0.8){$\scriptstyle \mu$};
        \end{tikzpicture}
    \end{aligned} \nonumber \\
    & = \frac{d_u}{d_w}\; \begin{aligned}
        \begin{tikzpicture}[scale=0.65]
            \filldraw[black!60, fill=gray!15, dotted, even odd rule] (0,0) circle[radius=0.5] (0,0) circle[radius=1.5];
            \draw[line width=1pt,violet] (0,0.5)--(0,1.5);
            \draw[line width=1pt,cyan] (0,-0.5)--(0,-1.5);
            \draw[blue] (0,0.8) arc[start angle=90, end angle=-90, radius=0.8];
            \draw[red] (0,1.3) arc[start angle=90, end angle=270, radius=1.3];
            \node[ line width=0.6pt, dashed, draw opacity=0.5] (a) at (0,1.7){$\scriptstyle z$};
            \node[ line width=0.6pt, dashed, draw opacity=0.5] (a) at (0,-1.7){$\scriptstyle u$};
            \node[ line width=0.6pt, dashed, draw opacity=0.5] (a) at (-1.1,0){$\scriptstyle \bar{a}$};
            \node[ line width=0.6pt, dashed, draw opacity=0.5] (a) at (1,0){$\scriptstyle \bar{b}$};
            \node[ line width=0.6pt, dashed, draw opacity=0.5] (a) at (0.2,-1){$\scriptstyle v$};
            \node[ line width=0.6pt, dashed, draw opacity=0.5] (a) at (0.2,-1.35){$\scriptstyle \zeta$};
            \node[ line width=0.6pt, dashed, draw opacity=0.5] (a) at (-0.2,-0.8){$\scriptstyle \gamma$};
            \node[ line width=0.6pt, dashed, draw opacity=0.5] (a) at (0,-0.3){$\scriptstyle w$};
            \node[ line width=0.6pt, dashed, draw opacity=0.5] (a) at (0,0.3){$\scriptstyle x$};
            \node[ line width=0.6pt, dashed, draw opacity=0.5] (a) at (0.2,1){$\scriptstyle y$};
            \node[ line width=0.6pt, dashed, draw opacity=0.5] (a) at (0.2,1.3){$\scriptstyle \nu$};
            \node[ line width=0.6pt, dashed, draw opacity=0.5] (a) at (-0.2,0.8){$\scriptstyle \mu$};
        \end{tikzpicture}
    \end{aligned} = s(X).  \nonumber 
    \end{align}

\noindent
\textit{Proof of} \eqref{eq:antipode-counital-DW}. The proof is the same as that for boundary case.

\section{Proof of Proposition~\ref{prop:Hom-Nat}}
\label{sec:Hom-Nat}

For simplicity, denote $A:={\mathbf{Tube}({}_\eC\eM;{}_\eC\eN)}$. 
In this part, we will show that there is a natural vector space isomorphism 
\begin{equation}
    \notag
    \Hom_A(\mathcal{H}^{\mathfrak{f}},\mathcal{H}^{\mathfrak{g}})  \simeq \operatorname{Nat}_{\eC}(\mathfrak{g},\mathfrak{f}),
\end{equation}
where $\operatorname{Nat}_{\eC}(\mathfrak{g},\mathfrak{f})$ is the set of natural transformations of $\eC$-module functors. The strategy is to construct two linear maps 
\begin{equation}
    \notag
    \Phi:\operatorname{Nat}_{\eC}(\mathfrak{g},\mathfrak{f}) \to \Hom_A(\mathcal{H}^{\mathfrak{f}},\mathcal{H}^{\mathfrak{g}}),\quad \Psi:\Hom_A(\mathcal{H}^{\mathfrak{f}},\mathcal{H}^{\mathfrak{g}})\to \operatorname{Nat}_{\eC}(\mathfrak{g},\mathfrak{f}),
\end{equation}
such that $\Psi\circ \Phi = \id$, $\Phi\circ \Psi = \id$. 
As usual, we denote 
\begin{equation}
    \notag
    \mathcal{H}^{\mathfrak f} = \bigoplus_{s\in\operatorname{Irr}(\eM)}\bigoplus_{t\in\operatorname{Irr}(\eN)}V^t_{s\mathfrak f}, \quad V^t_{s\mathfrak f} := \operatorname{Hom}_{\eN}(\mathfrak f(s),t).
\end{equation}
Similar for $\mathcal{H}^\mathfrak{g}$. 

We first construct $\Phi$. Let $\eta:\mathfrak{g} \Rightarrow \mathfrak{f}$ be a natural transformations of $\eC$-module functors. For every $z\in \Irr(\eM)$, denote $\eta_z:\mathfrak{g}(z)\to \mathfrak{f}(z)$. 

\begin{lemma}
    Define
    \begin{equation}
        \notag
        T_\eta:\mathcal{H}^{\mathfrak f}\to \mathcal{H}^{\mathfrak g},\quad V^u_{z\mathfrak f} \ni \alpha \mapsto \alpha\circ\eta_s\in V^u_{z\mathfrak g}. 
    \end{equation}
    Then $T_\eta$ is an $A$-module intertwiner. 
\end{lemma}

\begin{proof}
    Let us denote 
    \begin{equation}
        \notag
         X:=\begin{aligned}
        \begin{tikzpicture}[scale=0.65]
            \fill[gray!15] (0,1.5) arc[start angle=90, end angle=270, radius=1.5] -- (0,-0.5) arc[start angle=270, end angle=90, radius=0.5] -- cycle;
            \draw[dotted] (0,1.5) arc[start angle=90, end angle=270, radius=1.5]; 
            \draw[dotted] (0,0.5) arc[start angle=90, end angle=270, radius=0.5];
            \draw[line width=1pt,cyan] (0,0.5)--(0,1.5);
            \draw[line width=1pt,violet] (0,-0.5)--(0,-1.5);
            \draw[red] (0,0.8) arc[start angle=90, end angle=270, radius=0.8];
            \node[ line width=0.6pt, dashed, draw opacity=0.5] (a) at (-1,0){$\scriptstyle a$};
            \node[ line width=0.6pt, dashed, draw opacity=0.5] (a) at (-0.2,-1.1){$\scriptstyle y$};
            \node[ line width=0.6pt, dashed, draw opacity=0.5] (a) at (0.2,-0.8){$\scriptstyle \nu$};
            \node[ line width=0.6pt, dashed, draw opacity=0.5] (a) at (0,-0.3){$\scriptstyle z$};
            \node[ line width=0.6pt, dashed, draw opacity=0.5] (a) at (0,0.3){$\scriptstyle u$};
            \node[ line width=0.6pt, dashed, draw opacity=0.5] (a) at (-0.2,1.1){$\scriptstyle v$};
            \node[ line width=0.6pt, dashed, draw opacity=0.5] (a) at (0.2,0.8){$\scriptstyle \zeta$};
        \end{tikzpicture}
        \end{aligned}\;, \quad |z,u,\alpha \rangle_{\mathfrak{f}}:=
        \begin{aligned}
        \begin{tikzpicture}[scale=0.85]
            \filldraw[black!60, fill=gray!15, dotted] (0,0) -- (0,1) arc[start angle=90, end angle=270, radius=1] -- cycle;
            \draw[cyan, line width=1pt] (0,0) -- (0,1); 
            \draw[violet, line width=1pt] (0,-1) -- (0,0); 
            \draw[decorate, line width=.6pt, decoration={snake, amplitude=2pt, segment length=6pt}] (0,0) -- (1,0); 
            \node[ line width=0.6pt, dashed, draw opacity=0.5] at (-0.2,0.5){$\scriptstyle u$};
            \node[ line width=0.6pt, dashed, draw opacity=0.5] at (-0.2,-0.6){$\scriptstyle z$};
            \node[ line width=0.6pt, dashed, draw opacity=0.5] at (0.6,-0.3){$\scriptstyle \mathfrak{f}$};
            \node[ line width=0.6pt, dashed, draw opacity=0.5] at (-0.3,0){$\scriptstyle \alpha$};
            \node[ line width=0.6pt, dashed, draw opacity=0.5] at (0,0){$\scriptstyle \bullet$};
        \end{tikzpicture}
    \end{aligned} \quad = 
    \begin{aligned}
    \begin{tikzpicture}[scale=0.78]
    \draw[violet,line width=1pt] (-0.85,0) to[out=90,in=215] node[pos=0.25,left] {$\scriptstyle z$} (0,1.35);
    \draw[teal,line width=1pt] (0.85,0) to[out=90,in=-35] node[pos=0.28,right] {$\scriptstyle\mathfrak f$} (0,1.35);
    \draw[cyan,line width=1pt] (0,1.35)--(0,2.35);
    \node[circle,fill=black,inner sep=1.5pt] at (0,1.35) {};
    \node[left=3pt] at (0,1.35) {$\scriptstyle\alpha$};
    \node[above=2pt] at (0,2.35) {$\scriptstyle u$};
    \end{tikzpicture}
    \end{aligned}\;. 
    \end{equation}
    Therefore, $T_\eta|z,u,\alpha \rangle_{\mathfrak{f}} = |z,u,\alpha\circ\eta_z \rangle_{\mathfrak{g}}$, which we denote as
    \begin{equation}
        \notag
        |z,u,\alpha\circ\eta_z \rangle_{\mathfrak{g}}:= \begin{aligned}
        \begin{tikzpicture}[scale=0.78]
        \draw[violet,line width=1pt] (-0.85,0) to[out=90,in=215] node[pos=0.25,left] {$\scriptstyle z$} (0,1.35);
        \draw[teal,line width=1pt] (0.85,0) to[out=90,in=-35] (0,1.35);
        \draw[cyan,line width=1pt] (0,1.35)--(0,2.35);
        \node[circle,fill=black,inner sep=1.5pt] at (0,1.35) {};
        \node[left=3pt] at (0,1.35) {$\scriptstyle\alpha$};
        \node[above=2pt] at (0,2.35) {$\scriptstyle u$};
        \node[right,teal] at (0.8,0.3) {$\scriptstyle\mathfrak g$};
        \node[right,teal] at (0.3,1.3) {$\scriptstyle\mathfrak f$};
        \draw[fill=white, line width=0.6pt] (0.45,0.5) rectangle (0.95,0.9);
    \node at (0.7,0.7) {$\scriptstyle\eta_z$};
        \end{tikzpicture}
        \end{aligned}\;.
    \end{equation}
    Let
    \begin{equation}
        \notag
        \varphi^{\mathfrak f}_{a,z}:\mathfrak f(a\otimes z) \to a\otimes\mathfrak f(z), \quad \varphi^{\mathfrak g}_{a,z}: \mathfrak g(a\otimes z) \to a\otimes\mathfrak g(z)
    \end{equation}
    be the structure isomorphisms of module functors. For the labels
    $\nu,\alpha,\zeta$, define
    \begin{equation}
        \notag
        \Theta_{\mathfrak f} (\nu,\alpha,\zeta) := \zeta \circ (\id_a\otimes\,\alpha) \circ \varphi^{\mathfrak f}_{a,z} \circ \mathfrak f(\nu) = \begin{aligned}
    \begin{tikzpicture}[scale=0.7]
    \draw[violet,line width=1pt](-1.2,0)--(-1.2,1.0);
    \draw[red,line width=1pt] (-1.2,1.0) to[out=155,in=200]node[pos=0.48,left] {$\scriptstyle a$}(-0.15,3.35);
    \draw[violet,line width=1pt] (-1.2,1.0)to[out=25,in=210]
    node[pos=0.53,above] {$\scriptstyle z$} (0.35,2.25);
    \draw[teal,line width=1pt] (1.2,0) to[out=90,in=-25]
    node[pos=0,right] {$\scriptstyle\mathfrak f$}(0.35,2.25);
    \draw[cyan,line width=1pt] (0.35,2.25) to[out=90,in=-20] node[pos=0.55,right] {$\scriptstyle u$} (-0.15,3.35);
    \draw[cyan,line width=1pt] (-0.15,3.35)--(-0.15,4.25);
    \node[circle,fill=black,inner sep=1.5pt] at (-1.2,1.0) {};
    \node[circle,fill=black,inner sep=1.5pt] at (0.35,2.25) {};
    \node[circle,fill=black,inner sep=1.5pt] at (-0.15,3.35) {};
    \node[left=3pt] at (-1.2,1.0) {$\scriptstyle\nu$};
    \node[right=3pt] at (0.35,2.25) {$\scriptstyle\alpha$};
    \node[right=3pt] at (-0.15,3.35) {$\scriptstyle\zeta$};
    \node[left] at (-1.2,0) {$\scriptstyle y$};
    \node[left] at (-0.15,4.25) {$\scriptstyle v$};
    \end{tikzpicture}
    \end{aligned} \in \operatorname{Hom}_{\eN}(\mathfrak f(y),v). 
    \end{equation}
    Equivalently, in the multifusion category $\eP$, this is the     composite
    \begin{equation}
        \notag
        y\otimes\mathfrak f
        \xrightarrow{\nu\otimes\id_{\mathfrak f}}
        (a\otimes z)\otimes\mathfrak f
        \xrightarrow{\mathfrak a_{a,z,\mathfrak f}}
        a\otimes(z\otimes\mathfrak f)
        \xrightarrow{\id_a\otimes\alpha}
        a\otimes u
        \xrightarrow{\zeta}
        v, 
    \end{equation}
    where $y\otimes\mathfrak f=\mathfrak f(y)$, $z\otimes\mathfrak f=\mathfrak f(z)$, and the associator
    $\mathfrak a_{a,z,\mathfrak f}$ is identified with the structure isomorphism $\varphi^{\mathfrak f}_{a,z}$. Similarly we can define $\Theta_{\mathfrak g} (\nu,\beta,\zeta)$.

    Because $\eta$ is an ordinary natural transformation, we have
    \begin{equation}
        \notag
        \mathfrak f(\nu)\circ\eta_y
        =
        \eta_{a\otimes z}\circ\mathfrak g(\nu).
    \end{equation}
    Because $\eta$ is a $\eC$-module natural transformation, we also have
    \begin{equation}
        \notag
        \varphi^{\mathfrak f}_{a,z}
        \circ
        \eta_{a\otimes z}
        =
        (\id_a\otimes\eta_z)
        \circ
        \varphi^{\mathfrak g}_{a,z}. 
    \end{equation}
    It follows that
    \begin{align}
        \Theta_{\mathfrak f}(\nu,\alpha,\zeta)
        \circ\eta_y
        &=
        \zeta
        \circ(\id_a\otimes\alpha)
        \circ\varphi^{\mathfrak f}_{a,z}
        \circ\mathfrak f(\nu)
        \circ\eta_y
        \nonumber\\
        &=
        \zeta
        \circ(\id_a\otimes\alpha)
        \circ\varphi^{\mathfrak f}_{a,z}
        \circ\eta_{a\otimes z}
        \circ\mathfrak g(\nu)
        \nonumber\\
        &=
        \zeta
        \circ(\id_a\otimes\alpha)
        \circ(\id_a\otimes\eta_z)
        \circ\varphi^{\mathfrak g}_{a,z}
        \circ\mathfrak g(\nu)
        \nonumber\\
        &=
        \zeta
        \circ
        \bigl(
        \id_a\otimes(\alpha\circ\eta_z)
        \bigr)
        \circ
        \varphi^{\mathfrak g}_{a,z}
        \circ
        \mathfrak g(\nu)
        \nonumber\\
        &=
        \Theta_{\mathfrak g}
        \bigl(\nu,\alpha\circ\eta_z,\zeta\bigr). \nonumber
    \end{align}
    In diagrammatic representation, this is equivalent to
    \begin{equation}
        \notag
        \begin{aligned}
        \begin{tikzpicture}[scale=0.7]
        \draw[violet,line width=1pt](-1.2,0)--(-1.2,1.0);
        \draw[red,line width=1pt] (-1.2,1.0) to[out=155,in=200]node[pos=0.48,left] {$\scriptstyle a$}(-0.15,3.35);
        \draw[violet,line width=1pt] (-1.2,1.0)to[out=25,in=210]
        node[pos=0.53,above] {$\scriptstyle z$} (0.35,2.25);
        \draw[teal,line width=1pt] (1.2,0) to[out=90,in=-25]
        node[pos=0,right] {$\scriptstyle\mathfrak g$}(0.35,2.25);
        \draw[cyan,line width=1pt] (0.35,2.25) to[out=90,in=-20] node[pos=0.55,right] {$\scriptstyle u$} (-0.15,3.35);
        \draw[cyan,line width=1pt] (-0.15,3.35)--(-0.15,4.25);
        \node[circle,fill=black,inner sep=1.5pt] at (-1.2,1.0) {};
        \node[circle,fill=black,inner sep=1.5pt] at (0.35,2.25) {};
        \node[circle,fill=black,inner sep=1.5pt] at (-0.15,3.35) {};
        \node[left=3pt] at (-1.2,1.0) {$\scriptstyle\nu$};
        \node[right=3pt] at (0.35,2.25) {$\scriptstyle\alpha$};
        \node[right=3pt] at (-0.15,3.35) {$\scriptstyle\zeta$};
        \node[left] at (-1.2,0) {$\scriptstyle y$};
        \node[left] at (-0.15,4.25) {$\scriptstyle v$};
        \draw[fill=white, line width=0.6pt] (0.95,0.4) rectangle (1.45,0.8);
        \node at (1.2,0.6) {$\scriptstyle\eta_y$};
        \node[teal] at (1.3,1.3) {$\scriptstyle\mathfrak f$};
        \end{tikzpicture}
        \end{aligned} 
        \quad = \quad 
        \begin{aligned}
        \begin{tikzpicture}[scale=0.7]
        \draw[violet,line width=1pt](-1.2,0)--(-1.2,1.0);
        \draw[red,line width=1pt] (-1.2,1.0) to[out=155,in=200]node[pos=0.48,left] {$\scriptstyle a$}(-0.15,3.35);
        \draw[violet,line width=1pt] (-1.2,1.0)to[out=25,in=210]
        node[pos=0.53,above] {$\scriptstyle z$} (0.35,2.25);
        \draw[teal,line width=1pt] (1.2,0) to[out=90,in=-25]
        node[pos=0,right] {$\scriptstyle\mathfrak g$}(0.35,2.25);
        \draw[cyan,line width=1pt] (0.35,2.25) to[out=90,in=-20] node[pos=0.55,right] {$\scriptstyle u$} (-0.15,3.35);
        \draw[cyan,line width=1pt] (-0.15,3.35)--(-0.15,4.25);
        \node[circle,fill=black,inner sep=1.5pt] at (-1.2,1.0) {};
        \node[circle,fill=black,inner sep=1.5pt] at (0.35,2.25) {};
        \node[circle,fill=black,inner sep=1.5pt] at (-0.15,3.35) {};
        \node[left=3pt] at (-1.2,1.0) {$\scriptstyle\nu$};
        \node[right=3pt] at (0.35,2.25) {$\scriptstyle\alpha$};
        \node[right=3pt] at (-0.15,3.35) {$\scriptstyle\zeta$};
        \node[left] at (-1.2,0) {$\scriptstyle y$};
        \node[left] at (-0.15,4.25) {$\scriptstyle v$};
        \draw[fill=white, line width=0.6pt] (0.75,1.3) rectangle (1.25,1.7);
        \node at (1,1.5) {$\scriptstyle\eta_z$};
        \node[teal] at (0.5,1.8) {$\scriptstyle\mathfrak f$};
        \end{tikzpicture}
        \end{aligned} \;.
    \end{equation}  
    This is precisely 
    \begin{equation}
        \notag
        T_\eta\rho_{\mathfrak f}(X)|z,u,\alpha \rangle_{\mathfrak{f}} = \rho_{\mathfrak g}(X)T_\eta|z,u,\alpha \rangle_{\mathfrak{f}}.
    \end{equation}
    Hence $T_\eta \in \operatorname{Hom}_{A}(\mathcal{H}^{\mathfrak f},\mathcal{H}^{\mathfrak g})$. 
\end{proof}

We have constructed a linear map
\begin{equation}
    \notag
    \Phi: \operatorname{Nat}_{\eC}(\mathfrak g,\mathfrak f) \to
    \operatorname{Hom}_{A}(\mathcal{H}^{\mathfrak f},\mathcal{H}^{\mathfrak g}), \quad
    \Phi(\eta)=T_\eta\,.
\end{equation}
Next we reconstruction of a natural transformation from an
intertwiner. Let $T:\mathcal{H}^{\mathfrak f}\to \mathcal{H}^{\mathfrak g}$ be an intertwiner. Then $T$ decomposes into linear maps $T_{s,t}: V^t_{s\mathfrak f} \to V^t_{s\mathfrak g}$. 

For each fixed $s$, choose matrix-unit morphisms
\[
p^{\mathfrak f}_{s,t,r}:
\mathfrak f(s)\longrightarrow t,
\qquad
\iota^{\mathfrak f}_{s,t,r}:
t\longrightarrow\mathfrak f(s),
\]
where $t\in\operatorname{Irr}(\eN)$ and $r$ runs over the multiplicity
of $t$ in $\mathfrak f(s)$, such that
\begin{equation}
p^{\mathfrak f}_{s,t,r}
\circ
\iota^{\mathfrak f}_{s,t',r'}
=
\delta_{t,t'}\delta_{r,r'}\id_t,\quad \sum_{t,r}
\iota^{\mathfrak f}_{s,t,r}
\circ
p^{\mathfrak f}_{s,t,r}
=
\id_{\mathfrak f(s)}.
\label{eq:antiFF-matrix-unit-orthogonality}
\end{equation}
Such matrix units exist because $\eN$ is finite semisimple.
Define
\begin{equation}
\eta^T_s
:=
\sum_{t,r}
\iota^{\mathfrak f}_{s,t,r}
\circ
T_{s,t}
\bigl(p^{\mathfrak f}_{s,t,r}\bigr).
\label{eq:antiFF-eta-from-T}
\end{equation}
Since $T_{s,t} (p^{\mathfrak f}_{s,t,r}) \in \operatorname{Hom}_{\eN}(\mathfrak g(s),t)$, every summand above has type $\mathfrak g(s)\to t\to\mathfrak f(s)$. 
Thus $\eta^T_s:\mathfrak g(s)\to \mathfrak f(s)$. 

We now prove that
\begin{equation}
    T_{s,t}(\alpha) = \alpha\circ\eta^T_s,\quad \forall~\alpha\in \operatorname{Hom}_{\eN}(\mathfrak f(s),t). 
    \label{eq:antiFF-T-is-precomposition}
\end{equation} 
It suffices to verify this for
$\alpha=p^{\mathfrak f}_{s,t,r_0}$. Using
\eqref{eq:antiFF-eta-from-T} and
\eqref{eq:antiFF-matrix-unit-orthogonality}, we obtain
\begin{equation}
    \notag
    p^{\mathfrak f}_{s,t,r_0}\circ\eta^T_s = \sum_{t',r}
    p^{\mathfrak f}_{s,t,r_0}
    \circ
    \iota^{\mathfrak f}_{s,t',r}
    \circ
    T_{s,t'}
    \bigl(p^{\mathfrak f}_{s,t',r}\bigr)
    =
    \sum_{t',r}
    \delta_{t,t'}\delta_{r_0,r}
    T_{s,t'}
    \bigl(p^{\mathfrak f}_{s,t',r}\bigr)
    =
    T_{s,t}
    \bigl(p^{\mathfrak f}_{s,t,r_0}\bigr).
\end{equation}
By linearity, \eqref{eq:antiFF-T-is-precomposition} holds for every
$\alpha\in V^t_{s\mathfrak f}$.

The components $\eta^T_s$ defined on simple objects extend uniquely,
by additivity, to an ordinary natural transformation $\eta^T:\mathfrak g\Rightarrow\mathfrak f$. 
Indeed, if $x\simeq\oplus_s m_ss$, $x'\simeq\oplus_s n_ss$, then a morphism $h:x\to x'$ is, under these decompositions, a matrix
whose entries are of the form $h^s_{\mu,\nu} \id_s$. 
Defining
\[
    \eta^T_x := \bigoplus_s m_s \bigl(\eta^T_s\bigr),
\]
we have, entrywise,
\begin{equation}
    \notag
    \bigl[
    \mathfrak f(h)\circ\eta^T_x
    \bigr]^s_{\mu,\nu}
    =
    h^s_{\mu,\nu}\eta^T_s
    =
    \bigl[
    \eta^T_{x'}\circ\mathfrak g(h)
    \bigr]^s_{\mu,\nu}. 
\end{equation}
Therefore
\begin{equation}
    \notag
    \mathfrak f(h)\circ\eta^T_x = \eta^T_{x'}\circ\mathfrak g(h). 
\end{equation}

\begin{lemma}
The transformation $\eta^T:\mathfrak g\Rightarrow\mathfrak f$ constructed above is a $\eC$-module natural transformation.
\end{lemma}

\begin{proof}
Write
\begin{equation}
    \notag
    T_{z,u}(\alpha)
    =
    \sum_{\beta\in V^u_{z\mathfrak g}}
    T^{z,u}_{\beta\alpha}\,\beta. 
\end{equation}
Since $T$ is an $A$-module map, the relation $T\rho_{\mathfrak f}(X) = \rho_{\mathfrak g}(X)T$ and the action formula \eqref{eq:left-action-rep} imply
\begin{equation}
    \sum_{\kappa\in V^v_{y\mathfrak f}}
    T^{y,v}_{\lambda\kappa}
    \bigl[F^v_{az\mathfrak f}\bigr]
     ^{\,y\nu\kappa}_{u\alpha\zeta}
    =
    \sum_{\beta\in V^u_{z\mathfrak g}}
    \bigl[F^v_{az\mathfrak g}\bigr]
     ^{\,y\nu\lambda}_{u\beta\zeta}
    T^{z,u}_{\beta\alpha}.
    \label{eq:T-F-intertwining}
\end{equation}
Here the common factor $\sqrt{{d_a d_z}/{d_y}}$ on the two sides has canceled.

By the construction of $\eta^T$, we have
\begin{equation}
    \notag
    T_{z,u}(\alpha)=\alpha\circ\eta_z^T,
    \qquad
    T_{y,v}(\kappa)=\kappa\circ\eta_y^T. 
\end{equation}
Consequently, after reversing the corresponding F-symbol
expansions, Eq.~\eqref{eq:T-F-intertwining} is equivalent to
\begin{equation}
    \notag
    \Theta_{\mathfrak f}(\nu,\alpha,\zeta)
    \circ\eta_y^T
    =
    \Theta_{\mathfrak g}
    \bigl(\nu,\alpha\circ\eta_z^T,\zeta\bigr). 
\end{equation}
Expanding the definition of $\Theta$ gives
\begin{equation}
    \zeta \circ (\id_a\otimes\alpha) \circ \varphi^{\mathfrak f}_{a,z} \circ \mathfrak f(\nu) \circ \eta_y^T = \zeta \circ (\id_a\otimes\alpha) \circ (\id_a\otimes\eta_z^T) \circ \varphi^{\mathfrak g}_{a,z} \circ \mathfrak g(\nu). \label{eq:expanded-module-naturality-short}
\end{equation}
Since $\eta^T$ is already an ordinary natural transformation,
\begin{equation}
    \notag
    \mathfrak f(\nu)\circ\eta_y^T = \eta_{a\otimes z}^T\circ\mathfrak g(\nu).
\end{equation}
Therefore \eqref{eq:expanded-module-naturality-short} becomes
\begin{equation}
    \zeta\circ(\id_a\otimes\alpha)\circ D_{a,z}\circ\mathfrak g(\nu) = 0,\label{eq:D-tested-short}
\end{equation}
where
\begin{equation*}
    D_{a,z} := \varphi^{\mathfrak f}_{a,z} \circ \eta_{a\otimes z}^T - (\id_a\otimes\eta_z^T) \circ \varphi^{\mathfrak g}_{a,z}. 
\end{equation*}

As $y,\nu$ range over the fusion basis of $a\otimes z$, the
morphisms $\mathfrak g(\nu): \mathfrak g(y)\to \mathfrak g(a\otimes z)$ are jointly epimorphic. 
Indeed, for dual fusion bases one has
\[
\sum_{y,\nu}
\mathfrak g(\nu)\circ\mathfrak g(\widehat\nu)
=
\id_{\mathfrak g(a\otimes z)},
\qquad
\sum_{u,v,\alpha,\zeta}
\widehat q_{u,v;\alpha,\zeta}
\circ
q_{u,v;\alpha,\zeta}
=
\id_{a\otimes\mathfrak f(z)},
\]
where $q_{u,v;\alpha,\zeta} = \zeta\circ(\id_a\otimes\alpha)$.  Similarly, as
$u,v,\alpha,\zeta$ range over the corresponding fusion bases, the
morphisms $\zeta\circ(\id_a\otimes\alpha): a\otimes\mathfrak f(z)\to v$ are jointly monomorphic. 
Hence Eq.~\eqref{eq:D-tested-short} implies $D_{a,z}=0$. Equivalently,
\begin{equation*}
    \varphi^{\mathfrak f}_{a,z} \circ \eta_{a\otimes z}^T = (\id_a\otimes\eta_z^T) \circ \varphi^{\mathfrak g}_{a,z}. 
\end{equation*}
Thus $\eta^T$ is a $\eC$-module natural transformation.
\end{proof}

We have constructed a linear map
\begin{equation}
    \notag
    \Psi:\operatorname{Hom}_{A}(\mathcal{H}^{\mathfrak f},\mathcal{H}^{\mathfrak g})\to\operatorname{Nat}_{\eC}(\mathfrak g,\mathfrak f),\quad\Psi(T)=\eta^T.
\end{equation}

It remains to show the two constructions are mutually inverse.
Let $T\in \Hom_A(H_{\mathfrak f},H_{\mathfrak g})$. For every
$\alpha\in V^t_{s\mathfrak f}$,
Eq.~\eqref{eq:antiFF-T-is-precomposition} gives
\begin{equation}
    \notag
    T_{\eta^T}(\alpha) = \alpha\circ\eta^T_s = T_{s,t}(\alpha).
\end{equation}
Hence $\Phi \bigl(\Psi(T)\bigr)=T$. 
Conversely, let $\eta\in \operatorname{Nat}_{\eC}(\mathfrak g,\mathfrak f)$.  Using the definition of $T_\eta$ and Eq.\eqref{eq:antiFF-matrix-unit-orthogonality}, we find
\begin{equation}
    \notag
    \eta^{T_\eta}_s = \sum_{t,r} \iota^{\mathfrak f}_{s,t,r} \circ T_{\eta;s,t} \bigl(p^{\mathfrak f}_{s,t,r}\bigr) = \sum_{t,r} \iota^{\mathfrak f}_{s,t,r} \circ p^{\mathfrak f}_{s,t,r} \circ \eta_s  = \eta_s. 
\end{equation}
Therefore $\Psi\bigl(\Phi(\eta)\bigr)=\eta$. It follows that $\Phi$ and $\Psi$ are mutually inverse. Hence Proposition~\ref{prop:Hom-Nat} has been established.